\documentclass[a4paper,12pt]{amsart}

\usepackage{latexsym}
\usepackage{amsmath}
\usepackage{amsthm}
\usepackage{amssymb}
\usepackage{enumerate}
\usepackage{multicol}
\usepackage{graphicx}
\usepackage{scalefnt}
\usepackage{fancyhdr}
\usepackage{longtable}

\allowdisplaybreaks[4]

\renewcommand{\headrulewidth}{0.0pt}

\theoremstyle{plain}
\newtheorem{theorem}{Theorem}[section]

\newtheorem{lemma}[theorem]{Lemma}
\newtheorem{corollary}[theorem]{Corollary}

\theoremstyle{definition}

\newtheorem{example}[theorem]{Example}
\newtheorem{remark}[theorem]{Remark}
\numberwithin{equation}{section}

\def\ba{{\mathbf{a}}}

\def\bb{{\mathbf{b}}}

\def\cB{{\mathcal{B}}}

\def\C{{\mathbb{C}}}

\def\be{{\mathbf{e}}}
\def\cE{{\mathcal{E}}}

\def\sH{{\mathsf{H}}}

\def\cI{{\mathcal{I}}}

\def\bJ{{\mathbf{J}}}

\def\bk{{\mathbf{k}}}

\def\cK{{\mathcal{K}}}

\def\cM{{\mathcal{M}}}

\def\N{{\mathbb{N}}}

\def\cP{{\mathcal{P}}}
\def\bp{{\mathbf{p}}}

\def\bq{{\mathbf{q}}}

\def\R{{\mathbb{R}}}
\def\cR{{\mathcal{R}}}

\def\S{{\mathbb{S}}}
\def\cS{{\mathcal{S}}}
\def\bs{{\mathbf{s}}}

\def\T{{\mathbb{T}}}
\def\bt{{\mathbf{t}}}

\def\bU{{\mathbf{U}}}

\def\cV{{\mathcal{V}}}
\def\sV{{\mathsf{V}}}
\def\bv{{\mathbf{v}}}

\def\fw{{\mathrm{w}}}
\def\bW{{\mathbf{W}}}

\def\bx{{\mathbf{x}}}
\def\bX{{\mathbf{X}}}

\def\by{{\mathbf{y}}}
\def\bY{{\mathbf{Y}}}

\def\Z{{\mathbb{Z}}}
\def\bz{{\mathbf{z}}}
\def\bZ{{\mathbf{Z}}}

\def\opsi{{\overline{\psi}}}

\def\btheta{{\mbox{\boldmath$\theta$}}}

\def\bpi{{\mbox{\boldmath$\pi$}}}

\def\spin{{\{\uparrow,\downarrow\}}}
\def\ua{{\uparrow}}
\def\da{{\downarrow}}

\def\O{{\Omega}}
\def\o{{\omega}}

\def\eps{{\varepsilon}}
\def\beps{{\mbox{\boldmath$\varepsilon$}}}
\def\g{{\gamma}}
\def\bga{{\mbox{\boldmath$\gamma$}}}
\def\G{{\Gamma}}
\def\s{{\sigma}}

\def\D{{\Delta}}

\def\<{{\langle}}
\def\>{{\rangle}}

\def\Tr{\mathop{\mathrm{Tr}}}

\def\Mat{\mathop{\mathrm{Mat}}}
\def\Map{\mathop{\mathrm{Map}}}

\def\sgn{\mathop{\mathrm{sgn}}\nolimits}

\def\Re{\mathop{\mathrm{Re}}}
\def\Ref{\mathop{\mathrm{Ref}}\nolimits}

\def\b0{{\mathbf{0}}}

\begin{document}
\scalefont{1.1}

\thispagestyle{fancy}
\lhead{}
\chead{}
\rhead{}
\cfoot{\thepage}
\renewcommand{\headrulewidth}{0.0pt}

\begin{center}\Large\bf
The Zero-Temperature Limit\\
of the Free Energy Density in Many-Electron
 Systems at Half-Filling
\end{center}
\bigskip 

\begin{center}\large
Yohei Kashima \medskip \\
Graduate School of Mathematical Sciences, University of Tokyo,\\
Komaba, Tokyo, 153-8914, Japan\\ 
kashima@ms.u-tokyo.ac.jp
\end{center}
\bigskip

\begin{quotation}
\small {\bf  Abstract. }
We prove by means of a renormalization group method that 
in weakly interacting many-electron systems at half-filling on
 a periodic hyper-cubic lattice, the free energy
 density uniformly converges to an analytic
 function of the coupling constants in the infinite-volume,
 zero-temperature limit if the external magnetic field has a
 chessboard-like flux configuration. The spatial dimension is allowed to
 be any number larger than 1. The system covers the Hubbard model with
 a nearest-neighbor hopping term, on-site interactions, exponentially decaying density-density interactions and exponentially decaying
 spin-spin interactions. The magnetic field must be included in the
 kinetic term by the Peierls substitution. The flux configuration and
 the sign of the nearest-neighbor density-density/spin-spin interactions can be
 adjusted so that the free energy density is minimum among all the flux
 configurations. Consequently, the minimum free energy density is proved to converge to an analytic function of the
 coupling constants in the infinite-volume, zero-temperature limit. 
 These are extension of the results on a square lattice in the preceding work 
([Kashima, Y., ``The special issue for the 20th anniversary'',
	J. Math. Sci. Univ. Tokyo. {\bf 23} (2016), 1--288]).
We refer to lemmas proved in the reference in order to complete the
 proof of the main results of this paper. So this work is a continuation
 of the preceding work. 
\end{quotation}
\medskip

\begin{quotation}\small 
\end{quotation}

\tableofcontents

\section{Many-electron systems and the main results}\label{sec_system}

\subsection{Introduction}\label{subsec_introduction}
Rigorous construction of many-electron systems in low temperature is a
frontier of mathematical physics. Especially reaching the
infinite-volume, zero-temperature limit from a formulation in
finite volume and positive temperature appears to be a mathematical
challenge. As considered as the simplest possible model of interacting
electrons, the Hubbard models have been the central objectives in the
constructive theories based on multi-scale Grassmann
integration. Among them, substantial progress has been made in the
zero-temperature construction of the 1-dimensional models. See
\cite{BFM1}, \cite{BFM2} for the latest results. As for the
2-dimensional Hubbard models, there have been attempts to develop
low-temperature theories since the 2000s (see \cite{Ri}, \cite{AMRi1},
\cite{AMRi2}, \cite{BGM}, \cite{P}). There was also a thorough construction
of 2-dimensional Fermion systems in spatial continuum at zero
temperature by Feldman,  Kn\"orrer and Trubowitz \cite{FKT1},
\cite{FKT2}. As yet we have seen few examples of reaching the zero-temperature limit in
the concrete lattice models in spatial dimension $\ge 2$. One pioneering example of taking the zero-temperature limit in 2
dimension was reported by Giuliani and Mastropietro in \cite{GM} where
the half-filled Hubbard model on the honeycomb lattice was specifically
considered. Beneath the model-dependent details, the work of Giuliani
and Mastropietro seems to suggest an effective remedy for the
temperature-dependency of the constructive theories. The
hint from \cite{GM} was explored and another example of the
2-dimensional Hubbard model which admits the infinite-volume,
zero-temperature limit was given in our previous work \cite{K15}. In
more detail the model studied in \cite{K15} was the half-filled
Hubbard model on a square lattice, containing an external magnetic
field whose flux is $\pi$ (mod $2\pi$) per plaquette and $0$ (mod
$2\pi$) through the large circles around the periodic lattice. 
Recently, Giuliani and Jauslin reported a zero-temperature construction of
the free energy density and the two-point Schwinger function of an
interacting Fermion model on a bilayer honeycomb lattice in \cite{GJ}.

Since the
focus of \cite{K15} was on presenting a pile of lemmas leading to the
zero-temperature limit in a self-contained manner, possibility of
applying its framework to other models was not fully investigated there. As
 a continuation of \cite{K15}, here we focus on providing other
examples of many-electron systems where the analyticity at
zero-temperature can be proven essentially within the same framework. 
The main results of this paper can be seen as a
generalization of the results of \cite{K15}. We will establish a theorem
stating that the free energy density of a weakly interacting
many-electron system at half-filling uniformly converges with respect to
the amplitude of interaction in the infinite-volume, zero-temperature
limit. Here we allow the spatial dimension to be any number larger than
1. The system is defined on a periodic hyper-cubic lattice. The kinetic
term of the Hamiltonian is determined by the nearest-neighbor hopping of
electrons and contains an external magnetic field by means of the
Peierls substitution. The magnetic flux is assumed to change its sign at
plaquette alternately like a chessboard. The flux $\pi$ (mod $2\pi$) per
plaquette is a special case of such configurations. The magnetic flux
through the large circles winding around the periodic lattice is assumed
to be either uniformly 0 (mod $2\pi$) or uniformly $\pi$ (mod
$2\pi$). The interacting part of the Hamiltonian has a general form
satisfying a number of invariant properties and a decay property which is faster than any polynomial order and slower than an
exponential order. The interaction covers on-site interactions, exponentially decaying density-density interactions and exponentially
decaying spin-spin interactions as special cases. The whole Hamiltonian
has a symmetry which ensures that the system is at half-filling. The
magnetic flux and the interacting term can be chosen so that the free
energy density of the system is minimum among all flux
configurations, according to Lieb's result on the flux phase problem
(\cite{L}). Thus, it follows that the minimum free energy density in the
flux phase problem on a hyper-cubic lattice uniformly converges in the
infinite-volume, zero-temperature limit. We will explain how these
results generalize the main results of \cite{K15} in Remark
\ref{rem_how_generalized} after officially stating the main theorem and
its corollary in Subsection \ref{subsec_results}. 

The key strategy of our construction is to view the hyper-cubic lattice
as a composition of some sparser hyper-cubic lattices. The original one-band Hamiltonian is accordingly formulated into a
multi-band Hamiltonian. More precisely, we transform the one-band
Hamiltonian on a $d$-dimensional hyper-cubic lattice into a $2^d$-band
Hamiltonian. This procedure is a generalization of the formulation in
\cite{K15} where the one-band Hamiltonian on a square lattice was
formulated into a 4-band Hamiltonian. The multi-band formulation makes
it feasible to study symmetric properties and spectral properties of the
hopping matrix. We prove that the modulus of the band spectrum
of the hopping matrix is bounded from below by a non-negative function
of momentum variable vanishing at a single point. In fact the hopping matrix in momentum space
fails to be invertible only at the point. Therefore, this point times zero
time-momentum is the only singular point of the free covariance in
the zero-temperature limit.
The Hamiltonian has sufficient
symmetries to guarantee that the singular point of the free covariance
remains to be the singular point of the effective covariance during
infrared (IR) integration. Therefore, the same renormalization technique
as in \cite{K15}, which was motivated by \cite{GM}, applies to this
model as well. The power-counting in the IR integration depends on the spatial
dimension quantitatively. The power in the norm estimation of Grassmann
polynomials contains the spatial dimension $d$ as a parameter. By
substituting $d=2$ we can recover the same power-counting as in the IR
integration process \cite[\mbox{Section 7}]{K15}. However, our
multi-scale integration is qualitatively unaffected by the generalization of
the spatial dimension in the sense that Grassmann monomials of degree
$\ge 4$ are irrelevant at every iteration of the IR integration if the 
spatial dimension is larger than 1. We follow
steps, which are seen essentially parallel to the stories of \cite{K15}
in the eyes of abstraction, to complete the proof of the main theorem. We
 will refer to the relevant parts of \cite{K15} from time to time to
 fill the proofs of necessary lemmas. For
this reason this work should be strictly considered as a continuation of
\cite{K15}.

Nonetheless the generalization of the spatial dimension and 
the generalization of the interaction cause some technical details to be
different from the previous construction in \cite{K15}. 
The generalization in terms of the spatial dimension requires the
multi-band formulation to be constructed inductively. This part is
explained in Subsection \ref{subsec_multi_band_hamiltonian}.
In addition to
the new $2^d$-band formulation procedure in Subsection
\ref{subsec_multi_band_hamiltonian}, we will present other sections which
are largely affected by the generalization of the interaction without
significant omission. These are the symmetric Grassmann
integral formulation in Subsection \ref{subsec_grassmann}, the Matsubara
ultra-violet (UV) integration in Section \ref{sec_UV} and the
time-continuum, infinite-volume limit of the truncated Grassmann
integral formulation in Appendix \ref{app_h_L_limit}. Moreover, in the
belief that the inductive arguments in \cite[\mbox{Section 7}]{K15} are
not seen trivial at present, we make this occasion to present a more
organized version of the IR integration process than \cite[\mbox{Section
7}]{K15} in order to convince the readers of the true validity of
the mathematical renormalization group method. 

As for a relevance to the contemporary physical research, one can find
the Fermionic Hamiltonian with magnetic flux in a mean-field
theory of the Heisenberg-Hubbard model simulating the high-Tc
superconducting materials (\cite{AM}). More recently, the half-filled
Hubbard model with flux $\pi$ per plaquette together with the
half-filled Hubbard model on the honeycomb lattice tends to be studied
by means of numerical computation in order to describe the
metal-insulator transition driven by the electron-electron interaction
(\cite{OH}, \cite{CS}, \cite{IAS}, \cite{THAH}, \cite{EKO}, \cite{OYS}
and so on). These numerical studies commonly start with a speculation
that in the $\pi$-flux Hubbard model at half-filling, unlike in the 0-flux
Hubbard model at half-filling, the semi-metal phase remains in a
weak-coupling region so that the metal-insulator transition is
detectable in a middle (not the edge) of the phase diagram with the
horizontal axis of the coupling strength. The main result of this paper
suggests that there is no phase transition caused by the weak electron
interaction not only in the $\pi$-flux Hubbard model but also in a class
of electron models with staggered flux. This should provide a theoretical
support for numerical studies into the metal-insulator transition away
from the edge of the phase diagram in
these models yet to appear in physical literature. 

The contents of this paper are outlined as follows. In the rest of
this section we define the Hamiltonian operators, see what kind of
interaction is actually covered by our general definition
and state the main results of this paper. In Section
\ref{sec_formulation} we transform the one-band Hamiltonian into a
multi-band Hamiltonian and formulate the multi-band Hamiltonian by means
of finite-dimensional Grassmann integration. In Section
\ref{sec_UV} we construct the Matsubara UV integration both at a fixed
temperature and at 2 different temperatures. In Section \ref{sec_IR} we
carry out the IR integration and complete the proof of the main
theorem. In Appendix \ref{app_normal_order} we provide a lemma
concerning reordering in a non-commutative $\C$-algebra, which is
conveniently used in the proof that our many-electron system is at
half-filling in Subsection \ref{subsec_hamiltonian}. In Appendix
\ref{app_flux_phase} we restate Lieb's result on a $d$-dimensional flux
phase problem in order to facilitate the derivation of the corollary
about the minimum free energy density from the main
theorem. Finally in Appendix \ref{app_h_L_limit} we prove that each
truncation of the Taylor series of the Grassmann integral formulation of
the free energy density converges in the time-continuum, infinite-volume
limit. A flow chart of our construction showing the dependency between
the sections of this paper and the lemmas of the previous work \cite{K15} is given in
Figure \ref{fig_flow_chart}. We also attach a list of notations for
sake of the readers in the end. However, this list only contains notations which
were not used in \cite{K15} or were used in \cite{K15} with different
meanings and thus need additional remarks. The readers should refer to
the more comprehensive list in \cite{K15} for notations which are not 
contained in the supplementary list of this paper.  
 
\begin{figure}
\begin{center}
\begin{picture}(355,315)(0,0)

\put(20,300){\line(1,0){80}}
\put(20,320){\line(1,0){80}}
\put(20,300){\line(0,1){20}}
\put(100,300){\line(0,1){20}}
\put(60,300){\vector(0,-1){90}}
\put(33,305){Section \ref{sec_system}}

\put(20,310){\line(-1,0){20}}
\put(0,310){\line(0,-1){305}}
\put(0,5){\vector(1,0){20}}

\put(150,300){\line(1,0){95}}
\put(150,320){\line(1,0){95}}
\put(150,300){\line(0,1){20}}
\put(245,300){\line(0,1){20}}
\put(150,310){\vector(-1,0){50}}
\put(160,305){Appendix \ref{app_normal_order}}

\put(150,270){\line(1,0){95}}
\put(150,290){\line(1,0){95}}
\put(150,270){\line(0,1){20}}
\put(245,270){\line(0,1){20}}
\put(150,280){\vector(-2,1){50}}
\put(150,280){\vector(-2,-3){50}}
\put(160,275){Appendix \ref{app_flux_phase}}
 
\put(150,235){\line(1,0){145}}
\put(150,260){\line(1,0){145}}
\put(150,235){\line(0,1){25}}
\put(295,235){\line(0,1){25}}
\put(200,260){\vector(0,1){10}}
\put(155,245){\tiny \cite{K15} $\left\{\begin{array}{ll} 
\text{Lemma } A.2,
 &\text{Lemma } A.4,\\
                            \text{Theorem A.5}. & \end{array}\right.$}

\put(20,210){\line(1,0){80}}
\put(20,190){\line(1,0){80}}
\put(20,190){\line(0,1){20}}
\put(100,190){\line(0,1){20}}
\put(60,190){\vector(0,-1){35}}
\put(33,195){Section \ref{sec_formulation}}

\put(20,200){\line(-1,0){10}}
\put(10,200){\line(0,-1){185}}
\put(10,15){\vector(1,0){10}}

\put(150,225){\line(1,0){145}}
\put(150,180){\line(1,0){145}}
\put(150,180){\line(0,1){45}}
\put(295,180){\line(0,1){45}}
\put(150,200){\vector(-1,0){50}}
\put(155,200){\tiny \cite{K15} $\left\{\begin{array}{ll} 
\text{Lemma } 2.1, &\text{Lemma } 2.2,\\
\text{Lemma } 2.4, &\text{Lemma } 2.8,\\
\text{Lemma } 2.10, &\text{Lemma } B.1,\\
\text{Lemma } B.2, &\text{Lemma } B.3. 
 \end{array}\right.$}

\put(20,155){\line(1,0){80}}
\put(20,135){\line(1,0){80}}
\put(20,135){\line(0,1){20}}
\put(100,135){\line(0,1){20}}
\put(150,145){\vector(-1,0){50}}
\put(60,135){\vector(0,-1){115}}
\put(33,140){Section \ref{sec_UV}}

\put(150,115){\line(1,0){145}}
\put(150,170){\line(1,0){145}}
\put(150,115){\line(0,1){55}}
\put(295,115){\line(0,1){55}}

\put(155,140){\tiny \cite{K15} $\left\{\begin{array}{ll} 
\text{Lemma } 3.1, &\text{Lemma } 3.8,\\
\text{Lemma } 3.9, &\text{Lemma } 4.1,\\
\text{Lemma } 4.6, &\text{Lemma } 5.1,\\
\text{Lemma } 5.3, &\text{Lemma } 6.1,\\
\text{Lemma } 6.2, &\text{Lemma } 6.3.  
 \end{array}\right.$}

\put(150,85){\line(1,0){95}}
\put(150,105){\line(1,0){95}}
\put(150,85){\line(0,1){20}}
\put(245,85){\line(0,1){20}}
\put(150,95){\vector(-2,-3){50}}
\put(160,90){Appendix \ref{app_h_L_limit}}

\put(20,0){\line(1,0){80}}
\put(20,20){\line(1,0){80}}
\put(20,0){\line(0,1){20}}
\put(100,0){\line(0,1){20}}

\put(33,5){Section \ref{sec_IR}}

\put(150,0){\line(1,0){205}}
\put(150,75){\line(1,0){205}}
\put(150,0){\line(0,1){75}}
\put(355,0){\line(0,1){75}}
\put(150,35){\vector(-2,-1){50}}
\put(155,35){\tiny \cite{K15} $\left\{\begin{array}{ll} 
\text{Lemma } 3.9,
 &\text{Proposition } 5.6,\\
                            \text{Proposition 5.9}, &\text{Proposition
			     6.4 (2),(3)},\\
\text{Lemma 7.4}, &\text{Lemma 7.5},\\
\text{Lemma 7.6 (1),(2)}, &\text{Lemma 7.13 (3),(4)},\\
\text{Lemma 7.18 (3)}, &\text{Proof of Theorem 1.1},\\
\text{Lemma C.3}, &\text{Lemma E.1},\\
\text{Lemma E.2}. &  \end{array}\right.$}
\end{picture}
 \caption{Flow chart of our construction, where the arrows mean major dependency.}\label{fig_flow_chart}
\end{center}
\end{figure}
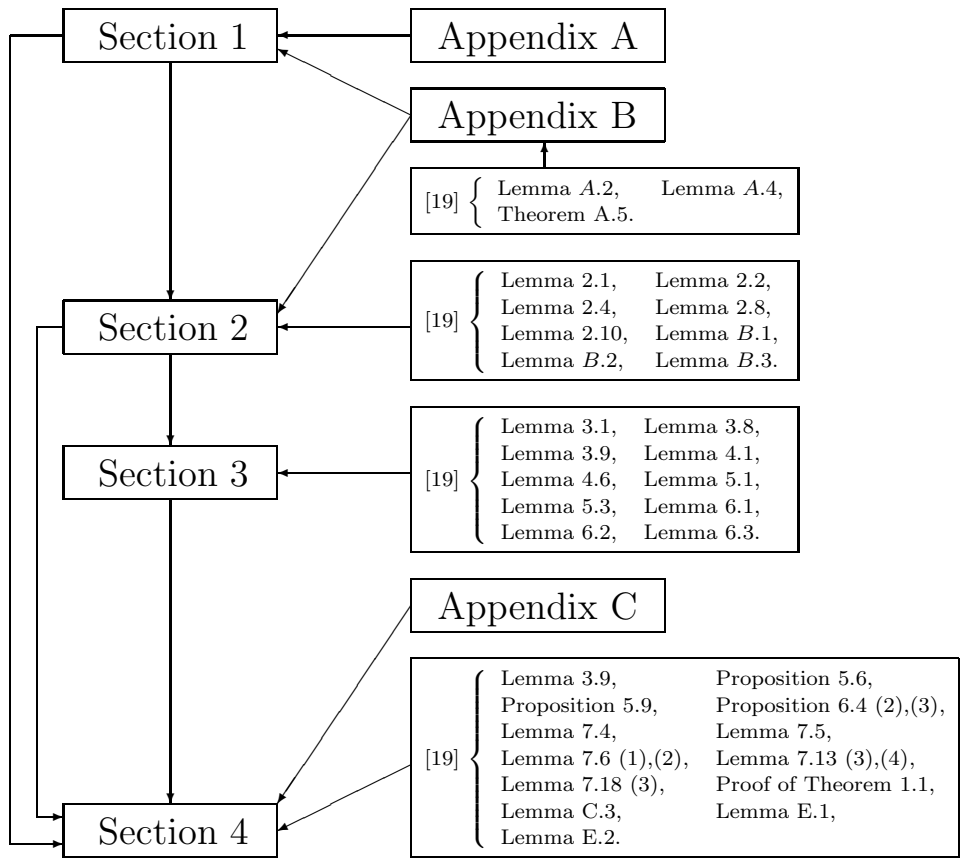

\subsection{Hamiltonians}\label{subsec_hamiltonian}

We let the number $d\ (\ge 2)$ denote the spatial dimension throughout
this paper. For $L\in\N$ we define the $d$-dimensional spatial lattice
$\G(L)$ by $\G(L):=\{0,1,\cdots,L-1\}^d$. In this subsection we
introduce a class of Hamiltonians on the Fermionic Fock space
$F_f(L^2(\G(2L)\times\spin))$. For a technical reason we define the
Hamiltonians in the spatial lattice of even length $2L$. Our Hamiltonians
contain an external magnetic field by means of the Peierls
substitution. The phase $\theta_L:\Z^d\times\Z^d\to\R$ is
assumed to satisfy that
\begin{align}
&\theta_L(\bx,\by)=-\theta_L(\by,\bx)\quad(\text{mod }2\pi),\label{eq_phase_condition_general}\\
&\theta_L\left(\bx+2L\sum_{j=1}^dm_j\be_j,\by\right)=\theta_L(\bx,\by)\quad(\text{mod }2\pi),\notag\\
&(\forall \bx,\by\in \Z^d,m_j\in \Z\ (j=1,2,\cdots,d)).\notag
\end{align}
Here $\be_j$ is the vector of $\Z^d$ whose j-th entry is 1 and the other
entries are 0. The free energy of the system with the periodic boundary
condition is known to be dependent on the magnetic field only by
the flux per plaquette and the flux through large circles winding around the periodic
lattice. Thus, it is important to specify these fluxes in advance. 
Let $\theta_{j,k}\in\R$,
$\eps_l^L\in\{0,1\}$ for 
$j,k,l\in\{1,2,\cdots,d\}$ with $j<k$. We allow $\eps_l^L$ to change its value depending
on $L$ and assume that $\eps^1_l=0$ $(\forall l\in\{1,2,\cdots,d\})$. We
assume that
\begin{align}
&\theta_L(\bx+\be_j,\bx)+\theta_L(\bx+\be_j+\be_k,\bx+\be_j)\notag\\
&+\theta_L(\bx+\be_k,\bx+\be_j+\be_k)+\theta_L(\bx,\bx+\be_k)=(-1)^{x_j+x_k}\theta_{j,k}\quad(\text{mod }2\pi),\label{eq_flux_per_plaquette}\\
&\sum_{m=0}^{2L-1}\theta_L(\bx+(m+1)\be_l,\bx+m\be_l)=\eps_l^L\pi
\quad(\text{mod }2\pi),\label{eq_flux_per_circle}\\
&(\forall \bx=(x_1,x_2,\cdots,x_d)\in\Z^d,j,k,l\in\{1,2,\cdots,d\}\text{
 with }j<k).\notag
\end{align}
The condition \eqref{eq_flux_per_plaquette} determines the flux per plaquette.
When $d=2$, the condition \eqref{eq_flux_per_plaquette} requires the flux per
plaquette to be arranged like a chessboard as pictured in Figure
\ref{fig_flux_per_plaquette}.
The condition \eqref{eq_flux_per_circle} states that the flux through
the closed contour parallel to $\be_l$ is $\eps_l^L\pi$ (mod $2\pi$) for
any $l\in \{1,2,\cdots,d\}$.

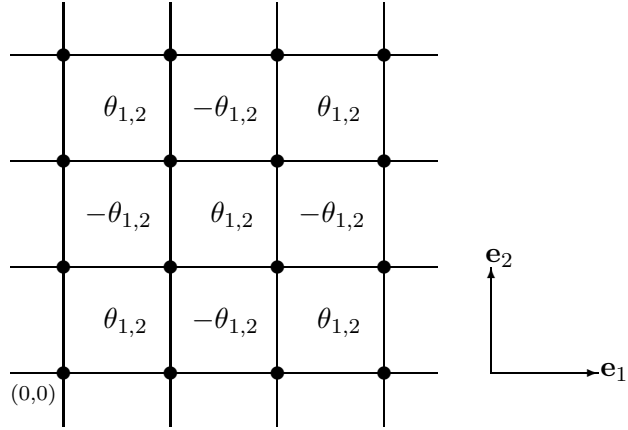
\begin{figure}
\begin{center}
\begin{picture}(230,160)(0,0)

\put(0,10){\tiny (0,0)}

\put(20,20){\circle*{5}}
\put(60,20){\circle*{5}}
\put(100,20){\circle*{5}}
\put(140,20){\circle*{5}}
\put(20,60){\circle*{5}}
\put(60,60){\circle*{5}}
\put(100,60){\circle*{5}}
\put(140,60){\circle*{5}}
\put(20,100){\circle*{5}}
\put(60,100){\circle*{5}}
\put(100,100){\circle*{5}}
\put(140,100){\circle*{5}}
\put(20,140){\circle*{5}}
\put(60,140){\circle*{5}}
\put(100,140){\circle*{5}}
\put(140,140){\circle*{5}}

\put(0,20){\line(1,0){160}}
\put(0,60){\line(1,0){160}}
\put(0,100){\line(1,0){160}}
\put(0,140){\line(1,0){160}}
\put(20,0){\line(0,1){160}}
\put(60,0){\line(0,1){160}}
\put(100,0){\line(0,1){160}}
\put(140,0){\line(0,1){160}}

\put(115,37){\small$\theta_{1,2}$}
\put(108,77){\small$-\theta_{1,2}$}
\put(115,117){\small$\theta_{1,2}$}

\put(68,37){\small$-\theta_{1,2}$}
\put(75,77){\small$\theta_{1,2}$}
\put(68,117){\small$-\theta_{1,2}$}

\put(35,37){\small$\theta_{1,2}$}
\put(28,77){\small$-\theta_{1,2}$}
\put(35,117){\small$\theta_{1,2}$}

\put(180,20){\vector(1,0){40}}
\put(180,20){\vector(0,1){40}}
\put(221,19){\small$\be_1$}
\put(178,62){\small$\be_2$}
\end{picture}
 \caption{The chessboard-like flux configuration.}\label{fig_flux_per_plaquette}
\end{center}
\end{figure}

Our analysis will be made on the quantitative assumption that 
\begin{align}
\frac{1}{2}\max_{m\in\{1,2,\cdots,d\}}\left(\sum_{j=1}^{m-1}|1+e^{i\theta_{j,m}}|+\sum_{j=m+1}^d|1+e^{i\theta_{m,j}}|\right)<1.\label{eq_flux_quantitative_assumption}
\end{align}

Let $t_j\in\R_{>0}$ $(j=1,2,\cdots,d)$ be the hopping amplitudes. The 
free Hamiltonian $\sH_0$ is defined by
\begin{align*}
\sH_0:=\sum_{\bx,\by\in\G(2L)}\sum_{\s\in \spin}
1_{\exists j\in \{1,2,\cdots,d\}\text{ s.t. }\bx-\by=\be_j\text{ or }-\be_j
\text{ in
 }(\Z/2L\Z)^d}t_je^{i\theta_L(\bx,\by)}\psi_{\bx\s}^*\psi_{\by\s},
\end{align*}
where $\psi_{\bx\s}$ is the annihilation operator and $\psi_{\bx\s}^*$
is its adjoint operator called creation operator. The function $1_P$
returns 1 if a proposition $P$ is true, 0 otherwise.
For any $\bx\in\Z^d$ we
define $\psi_{\bx\s}$, $\psi_{\bx\s}^*$ by identifying $\bx$ with the site
$\bx'$ of $\G(2L)$ satisfying that $\bx'=\bx$ in $(\Z/2L\Z)^d$. The condition
\eqref{eq_phase_condition_general} ensures that $\sH_0$ is self-adjoint.

To define the interacting part, we introduce the kernel functions. 
For any set $A,B$ let $\Map(A,B)$ denote the set of maps from $A$ to 
$B$. Take $n_v\in\N$, $N_v\in\N_{\ge 2}$. We assume
that $V_0^L\in\Map(\C^{n_v},\C)$, 
$V_m^L\in \Map(\C^{n_v},\Map((\Z^d\times \spin)^m\times (\Z^d\times
\spin)^m,\C))$ $(m=1,2,\cdots,N_v)$ satisfy the following conditions.
\begin{enumerate}[(i)]
\item\label{item_linearity} 
\begin{align*}
&\bU\mapsto V_0^L(\bU):\C^{n_v}\to\C,\\ 
&\bU\mapsto V_m^L(\bU):\C^{n_v}\to  \Map((\Z^d\times \spin)^m\times (\Z^d\times
\spin)^m,\C)
\end{align*}
are linear.
\item\label{item_bi_anti_symmetric}
\begin{align}
&V_m^L(\bU)((X_1,X_2,\cdots,X_m),(Y_1,Y_2,\cdots,Y_m))\label{eq_bi_anti_symmetric}\\
&=\sgn(\eta)\sgn(\xi)\notag\\
&\quad\cdot V_m^L(\bU)((X_{\eta(1)},X_{\eta(2)},\cdots,X_{\eta(m)}),
                          (Y_{\xi(1)},Y_{\xi(2)},\cdots,Y_{\xi(m)})),\notag\\
&(\forall X_j,Y_j\in \Z^d\times\spin\
 (j=1,2,\cdots,m),\bU\in\C^{n_v},\eta,\xi\in\S_m),\notag
\end{align}
where $\S_m$ is the set of all permutations over $\{1,2,\cdots,m\}$.
\item\label{item_spin_parity}
\begin{align}
&V_m^L(\bU)(((\bx_1,\s_1),\cdots,(\bx_m,\s_m)),((\by_1,\tau_1),\cdots,(\by_m,\tau_m)))\label{eq_spin_parity}\\
&=(-1)^{\sum_{j=1}^m(1_{\s_j=\ua}+1_{\tau_j=\ua})}\notag\\
&\quad\cdot
 V_m^L(\bU)(((\bx_1,\s_1),\cdots,(\bx_m,\s_m)),((\by_1,\tau_1),\cdots,(\by_m,\tau_m))),\notag\\
&(\forall (\bx_j,\s_j),(\by_j,\tau_j)\in \Z^d\times\spin\
 (j=1,2,\cdots,m),\bU\in\C^{n_v}).\notag
\end{align}
\item\label{item_spin_reflection}
\begin{align}
&V_m^L(\bU)(((\bx_1,\s_1),\cdots,(\bx_m,\s_m)),((\by_1,\tau_1),\cdots,(\by_m,\tau_m)))\label{eq_spin_reflection}\\
&=V_m^L(\bU)(((\bx_1,-\s_1),\cdots,(\bx_m,-\s_m)),((\by_1,-\tau_1),\cdots,(\by_m,-\tau_m))),\notag\\
&(\forall (\bx_j,\s_j),(\by_j,\tau_j)\in \Z^d\times\spin\
 (j=1,2,\cdots,m),\bU\in\C^{n_v}).\notag
\end{align}
\item\label{item_periodicity}
\begin{align}
&V_m^L(\bU)(((\bx_1^1,\s_1),\cdots,(\bx_m^1,\s_m)),((\by_1^1,\tau_1),\cdots,(\by_m^1,\tau_m)))\label{eq_periodicity}\\
&=V_m^L(\bU)(((\bx_1^2,\s_1),\cdots,(\bx_m^2,\s_m)),((\by_1^2,\tau_1),\cdots,(\by_m^2,\tau_m))),\notag\\
&(\forall \bx_j^1,\bx_j^2,\by_j^1,\by_j^2\in \Z^d \text{ satisfying
 }\bx_j^1=\bx_j^2,\by_j^1=\by_j^2\text{ in }(\Z/2L\Z)^d
,\notag\\
&\quad \s_j,\tau_j\in\spin\
 (j=1,2,\cdots,m),\bU\in\C^{n_v}).\notag
\end{align}
\item\label{item_translation}
\begin{align}
&V_m^L(\bU)(((\bx_1+2\bz,\s_1),\cdots,(\bx_m+2\bz,\s_m)),\label{eq_translation}\\
&\qquad\qquad ((\by_1+2\bz,\tau_1),\cdots,(\by_m+2\bz,\tau_m)))\notag\\
&=V_m^L(\bU)(((\bx_1,\s_1),\cdots,(\bx_m,\s_m)),((\by_1,\tau_1),\cdots,(\by_m,\tau_m))),\notag\\
&(\forall (\bx_j,\s_j),(\by_j,\tau_j)\in \Z^d\times\spin\
 (j=1,2,\cdots,m),\bz\in\Z^d,\bU\in\C^{n_v}).\notag
\end{align}
\item\label{item_inversion}
\begin{align}
&V_m^L(\bU)(((\bx_1,\s_1),\cdots,(\bx_m,\s_m)),((\by_1,\tau_1),\cdots,(\by_m,\tau_m)))\label{eq_inversion}\\
&=V_m^L(\bU)(((-\bx_1,\s_1),\cdots,(-\bx_m,\s_m)),((-\by_1,\tau_1),\cdots,(-\by_m,\tau_m))),\notag\\
&(\forall (\bx_j,\s_j),(\by_j,\tau_j)\in \Z^d\times\spin\
 (j=1,2,\cdots,m),\bU\in\C^{n_v}).\notag
\end{align}
\item\label{item_U1_invariance}
For any $\theta\in \Map(\Z^d,\R)$ satisfying that
     $\theta(\bx)=\theta(\by)$ ($\forall \bx,\by\in\Z^d$ with $\bx=\by$
     in $(\Z/2L\Z)^d$),
\begin{align}
&V_m^L(\bU)(((\bx_1,\s_1),\cdots,(\bx_m,\s_m)),((\by_1,\tau_1),\cdots,(\by_m,\tau_m)))\label{eq_U1_invariance}\\
&=e^{i(\sum_{j=1}^m\theta(\bx_j)-\sum_{j=1}^m\theta(\by_j))}\notag\\
&\quad \cdot V_m^L(\bU)(((\bx_1,\s_1),\cdots,(\bx_m,\s_m)),((\by_1,\tau_1),\cdots,(\by_m,\tau_m))),\notag\\
&(\forall (\bx_j,\s_j),(\by_j,\tau_j)\in \Z^d\times\spin\
 (j=1,2,\cdots,m),\bU\in\C^{n_v}).\notag
\end{align}
\item\label{item_hermiticity}
\begin{align}
& V_0^L(\bU)=\overline{V_0^L(\overline{\bU})},\quad V_m^L(\bU)(\bX,\bY)=\overline{V_m^L(\overline{\bU})(\bY,\bX)},
\label{eq_hermiticity}\\
&(\forall \bX,\bY\in (\Z^d\times\spin)^m,\bU\in\C^{n_v}).\notag
\end{align}
\item\label{item_particle_hole}
\begin{align}
&V_m^L(\bU)(\bX,\bY)=(-1)^m\sum_{l=0}^{N_v-m}\left(\begin{array}{c} m+l\\ l \end{array}
 \right)^2l!\sum_{(\bz_j,\eta_j)\in\G(2L)\times\spin\atop
 (j=1,2,\cdots,l)}\label{eq_particle_hole}\\
&\qquad\qquad\qquad
\quad\cdot
 V_{m+l}^L(\bU)((\bX,((\bz_1,\eta_1),(\bz_2,\eta_2),\cdots,(\bz_l,\eta_l))),\notag\\
&\qquad\qquad\qquad\qquad\qquad\qquad
                (((\bz_l,\eta_l),(\bz_{l-1},\eta_{l-1}),\cdots,(\bz_1,\eta_1)),\bY)),\notag\\
&(\forall \bX,\bY\in (\Z^d\times\spin)^m,\bU\in\C^{n_v}).\notag
\end{align}
\item\label{item_infinite_volume_limit}
For any $j\in\{1,2,\cdots,n_v\}$, $\bX,\bY\in(\Z^d\times\spin)^m$, 
\begin{align*}
\lim_{L\to\infty\atop L\in\N}\frac{1}{L^d}\frac{\partial}{\partial
 U_j}V_0^L(\bU),\quad \lim_{L\to\infty\atop L\in\N}\frac{\partial}{\partial
 U_j}V_m^L(\bU)(\bX,\bY)
\end{align*}
converge.
\item\label{item_decay_bound} For any $c\in\R_{\ge 0}$,
\begin{align}
&\sup_{L\in\N}\sup_{\bU\in\C^{n_v}\text{ with }|U_j|\le 1\atop
 (j=1,2,\cdots,n_v)}\sup_{p,q\in\{1,2,\cdots,2m-1\}\atop
 k\in\{1,2,\cdots,d\}}\sup_{(\bx,\s)\in\G(2L)\times
 \spin}\sum_{(\bx_j,\s_j)\in\G(2L)\times\spin\atop
 (j=1,2,\cdots,2m-1)}\label{eq_decay_bound}\\
&\cdot\left(\frac{L}{\pi}|e^{i\frac{\pi}{L}\<\bx-\bx_q,\be_k\>}-1|+1\right)e^{\sum_{j=1}^d(c\frac{L}{\pi}|e^{i\frac{\pi}{L}\<\bx-\bx_p,\be_j\>}-1|)^{1/2}}\notag\\
&\cdot|V_m^L(\bU)(((\bx,\s),(\bx_1,\s_1),\cdots,(\bx_{m-1},\s_{m-1})),\notag\\
&\qquad\qquad\quad ((\bx_m,\s_m),(\bx_{m+1},\s_{m+1}),\cdots,(\bx_{2m-1},\s_{2m-1})))|<\infty,\notag       
\end{align}
where $\<\cdot,\cdot\>$ is the standard inner product.
\end{enumerate}
For $\bU\in\R^{n_v}$ we define the interacting part of the Hamiltonian
     by 
\begin{align}
\sV:=&\sum_{m=0}^{N_v}\sum_{(\bx_j,\s_j),(\by_j,\tau_j)\in\G(2L)\times\spin\atop(j=1,2,\cdots,m)}\label{eq_one_band_interaction}\\
&\cdot V_m^L(\bU)(((\bx_1,\s_1),\cdots,(\bx_m,\s_m)),
                              ((\by_1,\tau_1),\cdots,(\by_m,\tau_m)))\notag\\
&\cdot \psi_{\bx_1\s_1}^*\cdots
 \psi_{\bx_m\s_m}^*\psi_{\by_1\tau_1}\cdots\psi_{\by_m\tau_m}.\notag
\end{align}
By the property \eqref{eq_hermiticity}
the operator $\sV$ is self-adjoint. The Hamiltonian $\sH$ is defined by
$\sH:=\sH_0+\sV$. Note that $\sH$ is a self-adjoint operator in the
Fermionic Fock space $F_f(L^2(\G(2L)\times \spin))$. 

The main results of
this paper concern analyticity and convergent properties of the
free energy density 
$$
-\frac{1}{\beta (2L)^d}\log(\Tr e^{-\beta \sH}),
$$                                
where $\beta\in \R_{>0}$ is the inverse temperature. 
Since the phase is an important parameter, we sometimes write
$\sH_0(\theta_L)$, $\sH(\theta_L)$ in place of $\sH_0$, $\sH$ respectively.
The many-electron
system is half-filled in the following sense.
\begin{lemma}\label{lem_half_filled}
For any $(\bx,\s)\in\G(2L)\times \spin$,
$$
\frac{\Tr(e^{-\beta \sH}\psi_{\bx\s}^*\psi_{\bx\s})}{\Tr e^{-\beta \sH}}=\frac{1}{2}.
$$
\end{lemma}

\begin{proof}
Let $\O_{2L}$ denote the vacuum of the Fock space 
$F_f(L^2(\G(2L)\times \spin))$. Define the transform $A$ on $F_f(L^2(\G(2L)\times \spin))$ by
\begin{align*}
&A\O_{2L}:=\prod_{\bx\in\G(2L)}(\psi_{\bx\ua}^*\psi_{\bx\da}^*)\O_{2L},\\
&A(\psi_{\bx_1\s_1}^*\cdots\psi_{\bx_n\s_n}^*\O_{2L})\\
&:=(-1)^{\sum_{j=1}^n\sum_{k=1}^d\<\bx_j,\be_k\>}\psi_{\bx_1\s_1}\cdots\psi_{\bx_n\s_n}\prod_{\bx\in\G(2L)}(\psi_{\bx\ua}^*\psi_{\bx\da}^*)
\O_{2L}
\end{align*}
for any $(\bx_j,\s_j)\in\G(2L)\times\spin$ $(j=1,2,\cdots,n)$,
 and by linearity. We can see that $A$ is a unitary transform and
 $A\sH_0(\theta_L)A^*=\sH_0(-\theta_L)$. Moreover, by using the properties
 \eqref{eq_U1_invariance}, \eqref{eq_particle_hole},
 \eqref{eq_hermiticity}, \eqref{eq_bi_anti_symmetric}
 and Lemma
 \ref{lem_normal_order} proved in Appendix \ref{app_normal_order} in
 this order,
\begin{align*}
&A\sV A^*\\
&=\sum_{m=0}^{N_v}\sum_{(\bx_j,\s_j),(\by_j,\tau_j)\in\G(2L)\times
 \spin\atop
 (j=1,2,\cdots,m)}(-1)^{\sum_{j=1}^m\sum_{k=1}^d\<\bx_j+\by_j,\be_k\>}\\
&\quad\cdot V_m^L(\bU)(((\bx_1,\s_1),\cdots,(\bx_m,\s_m)),((\by_1,\tau_1),\cdots,(\by_m,\tau_m)))\\
&\quad\cdot\psi_{\bx_1\s_1}\cdots\psi_{\bx_m\s_m}\psi_{\by_1\tau_1}^*\cdots\psi_{\by_m\tau_m}^*\\
&=\sum_{m=0}^{N_v}\sum_{l=0}^{N_v-m}\sum_{(\bx_j,\s_j),(\by_j,\tau_j)\in\G(2L)\times
 \spin\atop
 (j=1,2,\cdots,m)}\sum_{(\bz_j,\eta_j)\in\G(2L)\times\spin\atop
 (j=1,2,\cdots,l)}
(-1)^m
\left(\begin{array}{c} m+l\\ l\end{array}\right)^2l!\\
&\quad\cdot
 V_{m+l}^L(\bU)((((\bx_1,\s_1),\cdots,(\bx_m,\s_m)),((\bz_1,\eta_1),\cdots,(\bz_l,\eta_l))),\\
&\qquad\qquad\qquad\quad
 (((\bz_l,\eta_l),\cdots,(\bz_1,\eta_1)),((\by_1,\tau_1),\cdots,(\by_m,\tau_m))))\\
&\quad\cdot\psi_{\bx_1\s_1}\cdots\psi_{\bx_m\s_m}\psi_{\by_1\tau_1}^*\cdots\psi_{\by_m\tau_m}^*\\
&=\sum_{m=0}^{N_v}\sum_{l=0}^{m}\sum_{(\bx_j,\s_j),(\by_j,\tau_j)\in\G(2L)\times \spin\atop
 (j=1,2,\cdots,m-l)}\sum_{(\bz_j,\eta_j)\in\G(2L)\times\spin\atop
 (j=1,2,\cdots,l)}
(-1)^{m-l}\left(\begin{array}{c} m\\ l\end{array}\right)^2l!\\
&\quad\cdot
\overline{V_{m}^L(\bU)((((\by_1,\tau_1),\cdots,(\by_{m-l},\tau_{m-l})),((\bz_1,\eta_1),\cdots,(\bz_l,\eta_l))),}\\
&\qquad\qquad\qquad
 \overline{(((\bz_l,\eta_l),\cdots,(\bz_1,\eta_1)),((\bx_1,\s_1),\cdots,(\bx_{m-l},\s_{m-l}))))}\\
&\quad\cdot\psi_{\bx_1\s_1}\cdots\psi_{\bx_{m-l}\s_{m-l}}\psi_{\by_1\tau_1}^*\cdots\psi_{\by_{m-l}\tau_{m-l}}^*\\
&=\overline{\sV},
\end{align*}
where we set
\begin{align*}
\overline{\sV}:=&\sum_{m=0}^{N_v}\sum_{(\bx_j,\s_j),(\by_j,\tau_j)\in\G(2L)\times \spin\atop
 (j=1,2,\cdots,m)}\\
&\cdot
\overline{V_{m}^L(\bU)((((\bx_1,\s_1),\cdots,(\bx_{m},\s_{m})),((\by_1,\tau_1),\cdots,(\by_m,\tau_m)))}\\
& \cdot\psi_{\bx_1\s_1}^*\cdots\psi_{\bx_m\s_m}^*\psi_{\by_1\tau_1}\cdots\psi_{\by_m\tau_m}.
\end{align*}
Thus, we have for any $(\bx,\s)\in\G(2L)\times\spin$ that
\begin{align*}
\frac{\Tr(e^{-\beta \sH}\psi_{\bx\s}^*\psi_{\bx\s})}{\Tr e^{-\beta
 \sH}}&=\frac{\Tr(e^{-\beta(\sH_0(-\theta_L)+\overline{\sV})}A\psi_{\bx\s}^*\psi_{\bx\s}A^*)}{\Tr
 e^{-\beta(\sH_0(-\theta_L)+\overline{\sV})}}\\
&=1-\frac{\Tr(e^{-\beta (\sH_0(-\theta_L)+\overline{\sV})}\psi_{\bx\s}^*\psi_{\bx\s})}{\Tr e^{-\beta
 (\sH_0(-\theta_L)+\overline{\sV})}}.
\end{align*}
Then, by considering that 
\begin{align*}
&\Tr e^{-\beta(\sH_0(-\theta_L)+\overline{\sV})}=\overline{\Tr e^{-\beta
 \sH}}=\Tr e^{-\beta \sH},\\
&\Tr
 (e^{-\beta(\sH_0(-\theta_L)+\overline{\sV})}\psi_{\bx\s}^*\psi_{\bx\s})=\overline{\Tr
 (e^{-\beta \sH}\psi_{\bx\s}^*\psi_{\bx\s})}=\Tr
 (e^{-\beta \sH}\psi_{\bx\s}^*\psi_{\bx\s}),
\end{align*}
we obtain the claimed equality.
\end{proof}

\begin{remark}
There was unfortunately a flaw in the definition of the unitary
 transform in \cite[\mbox{Remark 1.4}]{K15} which was intended to
 demonstrate a proof of the same claim as the above lemma. By using the
 unitary transform $A$ we can correct \cite[\mbox{Remark 1.4}]{K15}.
It is simpler to confirm the equalities
 $A\sH_0(\theta_L)A^*=\sH_0(-\theta_L)$, $A\sV A^*=\sV$ for the free
 Hamiltonian $\sH_0(\theta_L)$ and the on-site interaction $\sV$ of
 \cite{K15}. Then, the conclusion of \cite[\mbox{Remark 1.4}]{K15}
 follows from the same argument as the last part of the above proof.
\end{remark}

\subsection{Examples}\label{subsec_examples}

Let us see that the interaction $\sV$ covers some relevant models of
interacting electrons. To shorten formulas, let $v_m(c)$ denote the left-hand side of the
inequality \eqref{eq_decay_bound} for $m\in\{1,2,\cdots,N_v\}$
and $c\in\R_{\ge 0}$. Moreover, set 
\begin{align*}
v_0:=\sup_{L\in\N}\sup_{\bU\in\C^{n_v}\text{ with }\atop |U_j|\le 1(j=1,2,\cdots,n_v)}\frac{1}{L^d}|V_0^L(\bU)|.
\end{align*}
\begin{example}[The on-site interaction]\label{eg_on_site}
Let $g\in \Map(\{1,-1\}^d,\{1,2,$ $3,\cdots,2^d\})$. With the coupling
 constants $\bU_o=(U_o(1),U_o(2),U_o(3),\cdots,$ $U_o(2^d))$ the on-site interaction
 $\sV_o$ is defined by
\begin{align*}
&\sV_o:=\\
&\sum_{\bx\in\G(2L)}U_o(g((-1)^{x_1},(-1)^{x_2},\cdots,(-1)^{x_d}))\left(\psi_{\bx\ua}^*\psi_{\bx\ua}-\frac{1}{2}\right)
\left(\psi_{\bx\da}^*\psi_{\bx\da}-\frac{1}{2}\right).
\end{align*}
The operator $\sV_o$ is equivalently written as follows. 
\begin{align*}
\sV_o=&\sum_{X_j,Y_j\in\G(2L)\times\spin\atop
 (j=1,2)}V_{o,2}^L(\bU_o)((X_1,X_2),(Y_1,Y_2))\psi_{X_1}^*\psi_{X_2}^*\psi_{Y_1}\psi_{Y_2}\\
&+\sum_{X,Y\in\G(2L)\times
 \spin}V_{o,1}^L(\bU_o)(X,Y)\psi_X^*\psi_Y+V_{o,0}^L(\bU_o)
\end{align*}
with
\begin{align*}
&V_{o,2}^L(\bU_o)(((\bx_1,\s_1),(\bx_2,\s_2)),((\by_1,\tau_1),(\by_2,\tau_2)))\\
&:=\frac{1}{4}U_o(g((-1)^{x_{1,1}},(-1)^{x_{1,2}},\cdots,(-1)^{x_{1,d}}))1_{\bx_1=\bx_2=\by_1=\by_2\text{
 in }(\Z/2L\Z)^d}\\
&\quad\cdot (1_{(\s_1,\s_2)=(\ua,\da)}-1_{(\s_1,\s_2)=(\da,\ua)})
          (1_{(\tau_1,\tau_2)=(\da,\ua)}-1_{(\tau_1,\tau_2)=(\ua,\da)}),\\
&V_{o,1}^L(\bU_o)((\bx,\s),(\by,\tau))\\
&:=-\frac{1}{2}U_o(g((-1)^{x_1},(-1)^{x_2},\cdots,(-1)^{x_d}))1_{(\bx,\s)=(\by,\tau)\text{
 in }(\Z/2L\Z)^d\times \spin},\\
&V_{o,0}^L(\bU_o):=\frac{L^d}{4}\sum_{\bx\in\{1,-1\}^d}U_o(g(\bx)).
\end{align*}
We can check that the kernels $V_{o,j}^L$ $(j=0,1,2)$ satisfy the
 conditions
 \eqref{item_linearity}, \eqref{item_bi_anti_symmetric}, $\cdots$, \eqref{item_infinite_volume_limit}
 with $N_v=2$, $n_v=2^d$. We can estimate the factors $v_0$, $v_m(c)$
 $(m=1,2)$ for this interaction as follows.
\begin{align*}
v_2(c)\le\frac{1}{2},\quad v_1(c)\le\frac{1}{2},\quad v_0\le 2^{d-2}.
\end{align*}
The operator $\sV_o-V_{o,0}^L(\bU_o)$ is also one example of the
 interaction $\sV$ and it is equal to the interaction treated in
 \cite{K15} when $d=2$ and $g$ is bijective.
\end{example}

\begin{example}[The density-density interaction]\label{eg_density_density}
Let $f_d$ be a real-valued continuous function on $\R^d$ satisfying that 
\begin{align*}
&f_d(\b0)=0,\quad |f_d(\bx)|\le c_1e^{-c_2\sum_{j=1}^d|x_j|},\ (\forall
 \bx\in\R^d),
\end{align*}
where $c_1,c_2\in\R_{>0}$ are fixed constants. We define the
 periodic function $f_d^L$ on $\R^d$ by
\begin{align*}
&f_{d}^L(\bx):=f_d\left(\frac{L}{\pi}|e^{i\frac{\pi}{L}x_1}-1|,\frac{L}{\pi}|e^{i\frac{\pi}{L}x_2}-1|,\cdots,\frac{L}{\pi}|e^{i\frac{\pi}{L}x_d}-1|\right)
\end{align*}
and the density-density interaction $\sV_d$ by
\begin{align*}
\sV_d:=U_d\sum_{\bx,\by\in\G(2L)}f_d^L(\bx-\by)(\psi_{\bx\ua}^*\psi_{\bx\ua}+\psi_{\bx\da}^*\psi_{\bx\da}-1)
(\psi_{\by\ua}^*\psi_{\by\ua}+\psi_{\by\da}^*\psi_{\by\da}-1),
\end{align*}
where $U_d$ is the coupling constant.  
We can write as follow.
\begin{align*}
\sV_d=&\sum_{X_j,Y_j\in\G(2L)\times\spin\atop
 (j=1,2)}V_{d,2}^L(U_d)((X_1,X_2),(Y_1,Y_2))\psi_{X_1}^*\psi_{X_2}^*\psi_{Y_1}\psi_{Y_2}\\
&+\sum_{X,Y\in\G(2L)\times
 \spin}V_{d,1}^L(U_d)(X,Y)\psi_X^*\psi_Y+V_{d,0}^L(U_d)
\end{align*}
with the bi-anti-symmetric kernels $V_{d,j}^L$ $(j=0,1,2)$ defined by
\begin{align*}
&V_{d,2}^L(U_d)(((\bx_1,\s_1),(\bx_2,\s_2)),((\by_1,\tau_1),(\by_2,\tau_2)))\\
&:=\frac{1}{4}U_d
 f_d^L(\bx_1-\bx_2)\sum_{\eta,\xi\in\S_2}\sgn(\eta)\sgn(\xi)
1_{(\bx_{\eta(1)},\s_{\eta(1)})=(\by_{\xi(2)},\tau_{\xi(2)})
\text{
 in }(\Z/2L\Z)^d\times \spin}\\
&\qquad\qquad\qquad\qquad\qquad\qquad\qquad\qquad\cdot 1_{(\bx_{\eta(2)},\s_{\eta(2)})=(\by_{\xi(1)},\tau_{\xi(1)})\text{
 in }(\Z/2L\Z)^d\times \spin},\\
&V_{d,1}^L(U_d)((\bx,\s),(\by,\tau)):=-2U_d1_{(\bx,\s)=(\by,\tau)\text{
 in }(\Z/2L\Z)^d\times \spin}\sum_{\bz\in\G(2L)}f_d^L(\bz),\\
&V_{d,0}^L(U_d):=(2L)^dU_d\sum_{\bz\in\G(2L)}f_d^L(\bz).
\end{align*}
The kernels $V_{d,j}^L$ $(j=0,1,2)$ satisfy the
 conditions
 \eqref{item_linearity}, \eqref{item_bi_anti_symmetric}, $\cdots$, \eqref{item_infinite_volume_limit}
 with $N_v=2$, $n_v=1$. The factors $v_0$, $v_m(c)$
 $(m=1,2)$ can be estimated as follows.
\begin{align*}
&v_2(c)
\le 2\sup_{L\in\N}\sup_{j'\in\{1,2,\cdots,d\}}\sum_{\bx\in\G(2L)}\\
&\qquad\quad\cdot\left(\frac{L}{\pi}|e^{i\frac{\pi}{L}x_{j'}}-1|+1\right)e^{\sum_{j=1}^d(c\frac{L}{\pi}|e^{i\frac{\pi}{L}x_{j}}-1|)^{1/2}}|f_d^L(\bx)|\\
&\qquad\le
 2c_1\left(\sum_{x\in\Z}(|x|+1)e^{(c|x|)^{1/2}-c_2\frac{2}{\pi}|x|}\right)^d,\\
&v_1(c)\le 2\sum_{\bz\in\G(2L)}|f_d^L(\bz)|\le
 2c_1\left(\sum_{x\in\Z}e^{-c_2\frac{2}{\pi}|x|}\right)^d,\\
&v_0\le 2^dc_1\left(\sum_{x\in\Z}e^{-c_2\frac{2}{\pi}|x|}\right)^d.
\end{align*}
The density-density interaction only between nearest-neighbor sites has
 particular importance for the flux phase problem, since it can be
 dealt within the framework of repeated reflection. Such a model is 
one special case of the interactions introduced above. To see this, let us choose a
 continuous function $f$ on $[0,\infty)$ satisfying that 
\begin{align}
&f(x)\in[0,1]\ (\forall x\in [0,\infty)),\label{eq_function_for_n_n}\\
&f(x)=\left\{\begin{array}{ll}1 &\text{if }x\in [\frac{2}{\pi},1],\\
0 &\text{if }x\in \{0\}\cup[\frac{4}{\pi},\infty)\end{array}
\right.  \notag
\end{align}
and set 
$$ f_d(\bx):=f\left(\sum_{j=1}^d|x_j|\right),\ (\bx\in\Z^d).$$
It follows that
\begin{align}\label{eq_nearest_neighbor_constants}
&f_d(\b0)=0,\quad |f_d(\bx)|\le e^{\frac{4}{\pi}}e^{-\sum_{j=1}^d|x_j|},\ (\forall
 \bx\in\Z^d).
\end{align}
Moreover, for any $\bx\in \Z^d$,
\begin{align*}
&f_d^L(\bx)=\left\{\begin{array}{ll}1 & \text{if }\exists
	    j\in\{1,2,\cdots,d\}\text{ s.t. }\bx=\be_j\text{ or
	     }-\be_j\text{ in }(\Z/2L\Z)^d,\\
0 &\text{ otherwise }
\end{array}
\right.
\end{align*}
and thus,
\begin{align*}
\sV_d=U_d&\sum_{\bx,\by\in\G(2L)}1_{\exists j\in\{1,2,\cdots,d\}\text{
 s.t. }\bx-\by=\be_j\text{ or }-\be_j\text{ in
 }(\Z/2L\Z)^d}\\
&\cdot(\psi_{\bx\ua}^*\psi_{\bx\ua}+\psi_{\bx\da}^*\psi_{\bx\da}-1)
(\psi_{\by\ua}^*\psi_{\by\ua}+\psi_{\by\da}^*\psi_{\by\da}-1).
\end{align*}
In this case the above estimation of $v_2(c)$, $v_1(c)$, $v_0$ holds
 with $c_1=e^{\frac{4}{\pi}}$, $c_2=1$.
\end{example}

\begin{example}[The spin-spin interaction]\label{eg_spin_spin}
Let us choose real-valued continuous functions $f_{s,j}$
 $(j=1,2,3)$ on $\R^d$ satisfying that 
\begin{align*}
&f_{s,j}(\b0)=0,\quad |f_{s,j}(\bx)|\le c_1e^{-c_2\sum_{k=1}^d|x_k|},\ (\forall
 \bx\in\R^d),
\end{align*}
with constants $c_1,c_2\in\R_{>0}$. Then, set 
\begin{align*}
&f_{s,j}^L(\bx):=f_{s,j}\left(\frac{L}{\pi}|e^{i\frac{\pi}{L}x_1}-1|,\frac{L}{\pi}|e^{i\frac{\pi}{L}x_2}-1|,\cdots,\frac{L}{\pi}|e^{i\frac{\pi}{L}x_d}-1|\right)
\end{align*}
for $\bx\in\R^d$. With the Pauli matrices 
\begin{align*}
&P^{(1)}=\left(\begin{array}{cc}0 & 1 \\ 1 & 0 \end{array}\right),\quad
P^{(2)}=\left(\begin{array}{cc}0 & -i \\ i & 0 \end{array}\right),\quad
P^{(3)}=\left(\begin{array}{cc}1 & 0 \\ 0 & -1 \end{array}\right)
\end{align*}
and the coupling constants $U_{s,j}$ $(j=1,2,3)$, the spin-spin
 interaction $\sV_s$ is defined as follows.
\begin{align}
&\sV_s:=\sum_{j=1}^3\sV_{s,j},\notag\\
&\sV_{s,j}:=U_{s,j}\sum_{\bx,\by\in\G(2L)}\sum_{\s,\tau,\mu,\lambda\in\spin}f_{s,j}^L(\bx-\by)(\psi_{\bx\s}^*P^{(j)}_{\s,\tau}\psi_{\bx\tau})
                    (\psi_{\by\mu}^*P^{(j)}_{\mu,\lambda}\psi_{\by\lambda}),\label{eq_spin_each_interaction}\\
&(j=1,2,3).\notag
\end{align}
The operators $\sV_{s,j}$ $(j=1,2,3)$ can be rewritten as 
\begin{align*}
&\sV_{s,j}=\sum_{X_k,Y_k\in\G(2L)\times\spin\atop
 (k=1,2)}V_{s,j,2}^L(U_{s,j})((X_1,X_2),(Y_1,Y_2))\psi_{X_1}^*\psi_{X_2}^*\psi_{Y_1}\psi_{Y_2}
\end{align*}
with the kernels $V_{s,j,2}^L$  $(j=1,2,3)$ defined by
\begin{align*}
&V_{s,j,2}^L(U_{s,j})(((\bx_1,\s_1),(\bx_2,\s_2)),((\by_1,\tau_1),(\by_2,\tau_2)))\\
&:=\frac{1}{4}U_{s,j}f_{s,j}^L(\bx_1-\bx_2)\\
&\quad\cdot
 \sum_{\eta,\xi\in\S_2}\sgn(\eta)\sgn(\xi)1_{\bx_{\eta(1)}=\by_{\xi(2)}\text{ in }(\Z/2L\Z)^d}
1_{\bx_{\eta(2)}=\by_{\xi(1)}\text{ in }(\Z/2L\Z)^d}\\
&\quad\cdot P_{\s_{\eta(1)},\tau_{\xi(2)}}^{(j)}
                              P_{\s_{\eta(2)},\tau_{\xi(1)}}^{(j)}.
\end{align*}
The kernel $V_{s,j,2}^L$ satisfies
 \eqref{item_linearity}, \eqref{item_bi_anti_symmetric}, $\cdots$, \eqref{item_infinite_volume_limit}
 with $N_v=2$, $n_v=1$ and so does the whole kernel
 $V_{s,2}^L=\sum_{j=1}^3V_{s,j,2}^{L}$ with $N_v=2$, $n_v=3$. For
 $V_{s,2}^L$ an upper bound on $v_2(c)$ is obtained as follows.
\begin{align*}
&v_2(c)\\
&\le 6\sup_{L\in\N}\sup_{j'\in\{1,\cdots,d\}\atop k\in\{1,2,3\}}\sum_{\bx\in\G(2L)}
\left(\frac{L}{\pi}|e^{i\frac{\pi}{L}x_{j'}}-1|+1\right)e^{\sum_{j=1}^d(c\frac{L}{\pi}|e^{i\frac{\pi}{L}x_{j}}-1|)^{1/2}}|f_{s,k}^L(\bx)|\\
&\le
 6c_1\left(\sum_{x\in\Z}(|x|+1)e^{(c|x|)^{\frac{1}{2}}-c_2\frac{2}{\pi}|x|}\right)^d.
\end{align*}
Again by using a continuous non-negative function $f$ on $[0,\infty)$
 satisfying \eqref{eq_function_for_n_n} we can formulate the spin-spin
 interaction between nearest-neighbor sites. By setting  
\begin{align*}
&f_{s,j}(\bx)=f\left(\sum_{k=1}^d|x_k|\right),\\
&f_{s,j}^L(\bx)=f_{s,j}\left(\frac{L}{\pi}|e^{i\frac{\pi}{L}x_1}-1|,\frac{L}{\pi}|e^{i\frac{\pi}{L}x_2}-1|,\cdots,\frac{L}{\pi}|e^{i\frac{\pi}{L}x_d}-1|\right)\\
&(j=1,2,3),
\end{align*}
the operator $\sV_{s,j}$ defined in \eqref{eq_spin_each_interaction}
 reads
\begin{align*}
\sV_{s,j}=&U_{s,j}\sum_{\bx,\by\in\G(2L)}\sum_{\s,\tau,\mu,\lambda\in\spin}1_{\exists k\in\{1,2,\cdots,d\}\text{
 s.t. }\bx-\by=\be_k\text{ or }-\be_k\text{ in
 }(\Z/2L\Z)^d}\\
&\qquad\qquad\qquad\qquad\cdot (\psi_{\bx\s}^*P^{(j)}_{\s,\tau}\psi_{\bx\tau})
                    (\psi_{\by\mu}^*P^{(j)}_{\mu,\lambda}\psi_{\by\lambda}).
\end{align*}
Moreover, the upper bound on $v_2(c)$ derived above holds with
 $c_1=e^{\frac{4}{\pi}}$, $c_2=1$ since $f_{s,j}$ satisfies
 \eqref{eq_nearest_neighbor_constants} in this case. 
\end{example}

In summary, the operator $\sV_o+\sV_d+\sV_s$ is one example of the
interactions treated in this paper.

\subsection{The main results}\label{subsec_results}
For $c\in\R_{>0}$ let $D(c)$ denote the disk $\{z\in\C\ |\
|z|<c\}$. Recall that for $m\in\{1,2,\cdots,N_v\}$, $v_m(c)$ denotes the
left-hand side of the inequality \eqref{eq_decay_bound}. For any
non-empty compact set $K$ of $\C^{n_v}$, $C(K;\C)$ denotes the Banach
space of all complex-valued continuous functions on $K$, equipped with
the uniform norm. Remind us that the norm of $f\in C(K;\C)$ is equal to
$\sup_{\bz\in K}|f(\bz)|$.
The following theorem is the main result of this
paper.

\begin{theorem}\label{thm_main_theorem}
There exists a constant $c(d,N_v)\in\R_{>0}$ depending only on $d$ and
 $N_v$ such that the following statements hold true with the quantity
 $R$ defined by
\begin{align*}
R:=&\left(\sum_{l=1}^{N_v}c(d,N_v)^lv_l(c(d,N_v))\right)^{-1}\\
&\cdot\left(1- 
\frac{1}{2}\max_{m\in\{1,2,\cdots,d\}}\left(\sum_{j=1}^{m-1}|1+e^{i\theta_{j,m}}|+\sum_{j=m+1}^d|1+e^{i\theta_{m,j}}|\right)\right)^{\frac{N_vd}{2}}\\
&\cdot\left(\min_{j\in\{1,2,\cdots,d\}}t_j\right)^{N_vd}
\left(\max_{j\in\{1,2,\cdots,d\}}t_j\right)^{1-N_vd}.
\end{align*}
\begin{enumerate}
\item There exists $F(\beta,L)\in C(\overline{D(R)}^{n_v};\C)$
      parameterized by $\beta\in\R_{>0}$ and $L\in \N$ satisfying 
$L\ge
      \max\{t_1,t_2,\cdots,t_d\}\beta$ such that $F(\beta,L)$
      is analytic in $D(R)^{n_v}$ and 
\begin{align*}
&F(\beta,L)(\bU)=-\frac{1}{\beta (2L)^d}\log(\Tr e^{-\beta \sH}),\\
&(\forall \bU\in \overline{D(R)}^{n_v}\cap \R^{n_v},\beta\in\R_{>0},\\
&\quad L\in \N\text{ satisfying }L\ge \max\{t_1,t_2,\cdots,t_d\}
 \beta).
\end{align*}
\item There exists $F(\beta)\in C(\overline{D(R)}^{n_v};\C)$
      parameterized by $\beta \in\R_{>0}$ such that 
\begin{align*}
\lim_{L\to\infty,L\in\N\atop \text{with }L\ge
 \max\{t_1,t_2,\cdots,t_d\}\beta}F(\beta,L)=F(\beta)\text{ in
 }C(\overline{D(R)}^{n_v};\C).
\end{align*}
\item There exists $F\in C(\overline{D(R)}^{n_v};\C)$ such that 
\begin{align*}
\lim_{\beta\to\infty,\beta\in\R_{>0}}F(\beta)=F\text{ in
 }C(\overline{D(R)}^{n_v};\C).
\end{align*}
\end{enumerate}
\end{theorem}

If we restrict the interaction $\sV$ to have a special form and choose
the phase $\theta_L$ to satisfy a certain condition, the free energy
density considered in Theorem \ref{thm_main_theorem} becomes the minimum free
energy in the flux phase problem. More precisely, we assume that
\begin{align}
\sV=&U_o\sum_{\bx\in\G(2L)}
\left(\psi_{\bx\ua}^*\psi_{\bx\ua}-\frac{1}{2}\right)
\left(\psi_{\bx\da}^*\psi_{\bx\da}-\frac{1}{2}\right)\label{eq_reflection_positive_interaction}\\
&+U_d\sum_{\bx,\by\in\G(2L)}
1_{\exists k\in\{1,2,\cdots,d\}\text{
 s.t. }\bx-\by=\be_k\text{ or }-\be_k\text{ in
 }(\Z/2L\Z)^d}\notag\\
&\qquad\cdot(\psi_{\bx\ua}^*\psi_{\bx\ua}+\psi_{\bx\da}^*\psi_{\bx\da}-1)
(\psi_{\by\ua}^*\psi_{\by\ua}+\psi_{\by\da}^*\psi_{\by\da}-1)\notag
\\
&+\sum_{j=1}^3U_{s,j}\sum_{\bx,\by\in\G(2L)}\sum_{\s,\tau,\mu,\lambda\in\spin}
1_{\exists k\in\{1,2,\cdots,d\}\text{
 s.t. }\bx-\by=\be_k\text{ or }-\be_k\text{ in
 }(\Z/2L\Z)^d}\notag\\
&\qquad\cdot(\psi_{\bx\s}^*P_{\s,\tau}^{(j)}\psi_{\bx\tau})
            (\psi_{\by\mu}^*P_{\mu,\lambda}^{(j)}\psi_{\by\lambda})\notag
\end{align}
with the Pauli matrices $P^{(j)}$ $(j=1,2,3)$
 and $U_o\in\R$, $U_d,U_{s,j}\in\R_{\ge 0}$ $(j=1,2,3)$. 
The interaction $\sV$ has a form to which the reflection positivity
 lemma \cite[\mbox{Lemma}]{L} is applicable.
As studied in the
 previous subsection, the factors $v_0$, $v_1(c)$, $v_2(c)$ for this
 interaction are bounded from above by a constant depending only on $c$
 and $d$.

Recall that for a phase $\varphi:\Z^d\times\Z^d\to \R$ satisfying
\eqref{eq_phase_condition_general} we set 
\begin{align}
\sH_0(\varphi)=\sum_{\bx,\by\in\G(2L)}\sum_{\s\in \spin}
1_{\exists j\in \{1,2,\cdots,d\}\text{ s.t. }\bx-\by=\be_j\text{ or }-\be_j
\text{ in
 }(\Z/2L\Z)^d}t_je^{i\varphi(\bx,\by)}\psi_{\bx\s}^*\psi_{\by\s},\label{eq_one_band_free_hamiltonian}
\end{align}
and $\sH(\varphi)=\sH_0(\varphi)+\sV$. The flux phase problem is to find a
phase $\varphi$ which minimizes the free energy $-(1/\beta)\log(\Tr
e^{-\beta \sH(\varphi)})$. Theorem \ref{thm_flux_phase}, which is a  
simple extension of Lieb's theorem \cite{L}, stated in Appendix \ref{app_flux_phase} implies
that if the phase $\theta_L$ satisfies \eqref{eq_phase_condition_general},
\eqref{eq_flux_per_plaquette}
with $\theta_{j,k}=\pi$ for all $j,k\in\{1,2,\cdots,d\}$ with $j<k$ and
\eqref{eq_flux_per_circle} with $\eps_l^L=1_{L\in 2\Z}$ for all $l\in
\{1,2,\cdots,d\}$, then
\begin{align*}
&-\frac{1}{\beta}\log(\Tr e^{-\beta \sH(\theta_L)})\\
& =\min\left\{-\frac{1}{\beta}\log(\Tr e^{-\beta \sH(\varphi)})\
 \Big|\ \varphi:\Z^d\times\Z^d\to \R\text{ satisfying
 }\eqref{eq_phase_condition_general}\right\}.
\end{align*}
Combined with Theorem \ref{thm_main_theorem}, we obtain the
following corollary.
\begin{corollary}\label{cor_minimum_energy}
There exists a constant $c(d)\in\R_{>0}$ depending only on $d$ such that
the following statements hold with the quantity $R$ defined by
\begin{align*}
R:=c(d)\left(\min_{j\in\{1,2,\cdots,d\}}t_j\right)^{2d}
 \left(\max_{j\in\{1,2,\cdots,d\}}t_j\right)^{1-2d}.
\end{align*}
\begin{enumerate}
\item There exists $F(\beta,L)\in C(\overline{D(R)}^5;\C)$ parameterized
      by $\beta\in\R_{>0}$ and $L\in \N$ satisfying $L\ge
      \max\{t_1,t_2,\cdots,t_d\}\beta$ such that $F(\beta,L)$ is
      analytic in $D(R)^5$ and 
\begin{align*}
&F(\beta,L)(U_o,U_d,U_{s,1},U_{s,2},U_{s,3})\\
&=\min\left\{-\frac{1}{\beta(2L)^d}\log(\Tr e^{-\beta \sH(\varphi)})\
 \Big|\ \varphi:\Z^d\times\Z^d\to \R\text{ satisfying
 }\eqref{eq_phase_condition_general}\right\},\\
&(\forall U_o\in\overline{D(R)}\cap \R,
 (U_d,U_{s,1},U_{s,2},U_{s,3})\in\overline{D(R)}^4\cap \R_{\ge 0}^4,\\
&\quad \beta\in\R_{>0},L\in \N\text{ satisfying }L\ge \max\{t_1,t_2,\cdots,t_d\} \beta).
\end{align*}
\item There exists $F(\beta)\in C(\overline{D(R)}^{5};\C)$
      parameterized by $\beta \in\R_{>0}$ such that 
\begin{align*}
\lim_{L\to\infty,L\in\N\atop \text{with }L\ge
 \max\{t_1,t_2,\cdots,t_d\}\beta}F(\beta,L)=F(\beta)\text{ in
 }C(\overline{D(R)}^{5};\C).
\end{align*}
\item There exists $F\in C(\overline{D(R)}^{5};\C)$ such that 
\begin{align*}
\lim_{\beta\to\infty,\beta\in\R_{>0}}F(\beta)=F\text{ in
 }C(\overline{D(R)}^{5};\C).
\end{align*}
\end{enumerate}
\end{corollary}

\begin{remark}\label{rem_how_generalized}
Let us explain how Theorem \ref{thm_main_theorem} and Corollary
 \ref{cor_minimum_energy} generalize \cite[\mbox{Theorem 1.1, Corollary
 1.2}]{K15}. Both in Theorem \ref{thm_main_theorem} and Corollary
 \ref{cor_minimum_energy} the spatial dimension $d$ is any number larger
 than 1, while it was fixed to be 2 in \cite[\mbox{Theorem 1.1, Corollary
 1.2}]{K15}. In Theorem \ref{thm_main_theorem} we assume the flux
 conditions \eqref{eq_flux_per_plaquette}, \eqref{eq_flux_per_circle},
 which are more general than the conditions \cite[\mbox{(1,2)}]{K15}
 requiring that the flux per plaquette is $\pi$ (mod $2\pi$) and the
 flux through the large circles around the periodic square lattice is
 $0$ (mod $2\pi$). As we saw in Example \ref{eg_on_site}, the
 interaction $\sV$ covers the on-site interaction considered in
 \cite[\mbox{Theorem 1.1}]{K15} as a special case. Concerning the
 spatial dimension and the flux configuration, therefore, Theorem
 \ref{thm_main_theorem} 
is more general than 
 \cite[\mbox{Theorem 1.1}]{K15}. However, here the hopping amplitude
 depends only on the direction and thus the whole hopping
 amplitudes are described by the $d$ parameters $t_1,t_2,\cdots,t_d$,
 while in \cite[\mbox{Theorem 1.1}]{K15} the hopping amplitude is
 constant in each direction and is 
allowed to vary alternately and thus the whole
 hopping amplitudes are described by the 4 parameters ``$t_{h,e},t_{h,o},t_{v,e},t_{v,o}$''
as it was
 2-dimensional. See \cite[\mbox{Figure 2}]{K15} for the configuration of
 the hopping amplitudes.
Theorem \ref{thm_main_theorem} is less general than 
\cite[\mbox{Theorem 1.1}]{K15} only in this sense. In this paper we do
 not stick to the generalization of the hopping amplitudes in the
 interest of simplicity. If we assume that the hopping amplitude depends
 only on the direction in \cite[\mbox{Theorem 1.1}]{K15}, then the
 factor ``$f_{\bt}^2$'' determining the possible magnitude of the coupling
 in \cite[\mbox{Theorem 1.1}]{K15} becomes the
 factor $(\min\{t_1,t_2\})^4(\max\{t_1,t_2\})^{-3}$ included in $R$ in
 Theorem \ref{thm_main_theorem}. In this setting, therefore, Theorem
 \ref{thm_main_theorem} naturally extends \cite[\mbox{Theorem
 1.1}]{K15}. As for Corollary \ref{cor_minimum_energy}, the apparent
 generality is that the interaction includes not only the on-site
 interaction but also the density-density interaction and the spin-spin
 interaction as defined in
 \eqref{eq_reflection_positive_interaction}. Moreover, the number
 $L$ can be both odd and even, while it was restricted to be odd in
 \cite[\mbox{Corollary 1.2}]{K15}. This generalization is due to the
 fact that here the magnetic flux through the large circles around the
 lattice can be uniformly 0 (mod $2\pi$) or uniformly $\pi$ (mod
 $2\pi$) depending on the parity of $L$ and thus the free energy density in Theorem
 \ref{thm_main_theorem} can be the minimum in the flux phase problem in
 both cases, according to the known sufficient condition restated in
 Theorem \ref{thm_flux_phase}.
\end{remark}

\begin{remark}
It is not trivial to make explicit the dependency of the constants
 $c(d,N_v)$, $c(d)$ on $d,N_v$. We can see from our construction that it
 would require a wide range of additional calculations to do
 so. Not to lengthen the paper further, we decide not to tackle this clarification.
\end{remark}

\begin{remark} 
The condition \eqref{eq_flux_quantitative_assumption} requires the flux
 per plaquette  $\theta_{j,k}$ not to vanish for any
 $j,k\in\{1,2,\cdots,d\}$ with $j<k$. In 2-dimensional case the
 constraint \eqref{eq_flux_quantitative_assumption} is fulfilled if
 $\theta_{1,2}\neq 0$ (mod $2\pi$). This means that the infinite-volume,
 zero-temperature limit of the free energy density can be taken if 
the system contains an arbitrarily thin magnetic field having a
 chessboard-like flux pattern over the square lattice and the interaction is
 accordingly weak.
\end{remark}

\begin{remark}
The exponent $1/2$ in \eqref{eq_decay_bound} stems from the fact
 that we use a Gevrey-class cut-off function $\phi$ satisfying that 
$$
\sup_{x\in\R}|\phi^{(n)}(x)|\le 2^n(n!)^2,\ (\forall n\in\N\cup\{0\})
$$
(see the beginning of Subsection \ref{subsec_UV_covariance}). We can
 prove the similar results for the interactions satisfying
 \eqref{eq_decay_bound} with the exponent $r\in (0,1)$ in place of
 $1/2$ by using a cut-off function $\phi$ satisfying that 
$$
\sup_{x\in\R}|\phi^{(n)}(x)|\le const^n(n!)^\frac{1}{r},\ (\forall n\in\N\cup\{0\}).
$$
However, this generalization will bring the extra parameter $r$ into
 the major part of the construction since other parameters need to
 be tuned depending on $r$. In this paper we choose not to pursue this
 generalization for simplicity.
\end{remark}
 
\section{Multi-band formulation}\label{sec_formulation}
In this section we introduce a $2^d$-band Hamiltonian operator whose
free energy density is equal to that governed by the 1-band Hamiltonian
$\sH$. Then, we will focus on the $2^d$-band model and derive the
finite-dimensional Grassmann integral formulation of the partition
function. The Grassmann integral formulation of the $2^d$-band
model will be the major objective of our multi-scale analysis in
the following sections. 

\subsection{Multi-band Hamiltonian}\label{subsec_multi_band_hamiltonian}
We will define the hopping matrix of the multi-band Hamiltonian by induction
with respect to the spatial dimension. To this end,
we need some notations. For $n\in \N$ let $\Mat(n,\C)$ denote the set
of all $n\times n$ complex matrices and let $I_n$ denote the $n\times n$
unit matrix. Set 
\begin{align*}
\G_n(L):=\{0,1,\cdots,L-1\}^n,\quad \cB_n:=\{1,2,3,\cdots,2^n\}.
\end{align*}
Note that for any $\rho\in\cB_n$ there uniquely exists
$(\rho_1,\rho_2,\cdots,\rho_n)\in \{0,1\}^n$ such that
$\rho=\sum_{j=0}^{n-1}\rho_{j+1}2^j+1$. Thus, we can define $b_n\in
\Map(\cB_n,\{0,1\}^n)$ by $b_n(\rho):=(\rho_1,\rho_2,\cdots,\rho_n)$.
The map $b_n$ is bijective. 
We will suppress the index $n$ of $\G_n(L)$, $\cB_n$, $b_n$ after fixing
$n$ to be the spatial dimension $d$. We keep showing the dependency on $n$ while
we argue inductively with respect to $n$.
For $n\in \N$ and $(\xi_j)_{1\le j\le n}\in\R^n$ we define the matrix 
$U_n((\xi_j)_{1\le j\le
n})\in \Mat(2^n,\C)$ parameterized by $(\xi_j)_{1\le j\le
n}$ as follows. Set
$$U_1(\xi_1):=\left(\begin{array}{cc}1 & 0 \\
                                        0 &
		       e^{i\xi_1}\end{array}\right).$$
Assume that we have defined $U_m((\xi_j)_{1\le j\le m})\in
\Mat(2^m,\C)$. Then, define $U_{m+1}((\xi_j)_{1\le j\le m+1})\in
\Mat(2^{m+1},\C)$ by
$$U_{m+1}((\xi_j)_{1\le j\le m+1}):=\left(\begin{array}{cc} U_m((\xi_j)_{1\le j\le m}) & O \\
                                        O &
		       e^{i\xi_{m+1}} U_m((\xi_j)_{1\le j\le m}) \end{array}\right).$$

\begin{lemma}\label{lem_unitary_characterization}
For any $n\in \N$,
\begin{align*}
U_n((\xi_j)_{1\le j\le
 n})(\rho,\eta)=e^{i\sum_{j=1}^nb_n(\rho)(j)\xi_j}\delta_{\rho,\eta},\
 (\forall \rho,\eta\in\cB_n).
\end{align*}
\end{lemma}
\begin{proof}
The claim holds for $n=1$ by definition. Assume that it holds for some
 $n\in\N$. Let $\rho,\eta\in \cB_{n+1}$. If $b_{n+1}(\rho)(n+1)\neq
 b_{n+1}(\eta)(n+1)$, $U_{n+1}((\xi_j)_{1\le j\le n+1})(\rho,\eta)$ $=0$ by
 definition. If $b_{n+1}(\rho)(n+1)=b_{n+1}(\eta)(n+1)=0$, by the hypothesis of
 induction, 
\begin{align*}
U_{n+1}((\xi_j)_{1\le j\le n+1})(\rho,\eta)&=U_n((\xi_j)_{1\le j\le
 n})(\rho,\eta)=e^{i\sum_{j=1}^nb_n(\rho)(j)\xi_j}\delta_{\rho,\eta}\\
&=e^{i\sum_{j=1}^{n+1}b_{n+1}(\rho)(j)\xi_j}\delta_{\rho,\eta}.
\end{align*}
If $b_{n+1}(\rho)(n+1)=b_{n+1}(\eta)(n+1)=1$, by the hypothesis of induction, 
\begin{align*}
U_{n+1}((\xi_j)_{1\le j\le n+1})(\rho,\eta)&=e^{i\xi_{n+1}}U_n((\xi_j)_{1\le j\le
 n})(\rho-2^n,\eta-2^n)\\
&=e^{i\sum_{j=1}^nb_n(\rho-2^n)(j)\xi_j+i\xi_{n+1}}\delta_{\rho-2^{n},\eta-2^n}\\
&=e^{i\sum_{j=1}^{n+1}b_{n+1}(\rho)(j)\xi_j}\delta_{\rho,\eta}.
\end{align*}
Thus, the result holds for $n+1$. By induction, the claim holds for any
 $n\in\N$. 
\end{proof}

Let $\g_{j,k}\in\R$ for $j,k\in\{1,2,\cdots,n\}$ with $j<k$. Then, let
$(\g_{j,k})_{1\le j<k\le n}$ denote the vector 
$$
(\g_{1,2},\g_{1,3},\g_{2,3},\g_{1,4},\g_{2,4},\g_{3,4},\cdots,\g_{n-1,n})\in\R^{\frac{n(n-1)}{2}}.
$$
For $n\in\N$ we define $M_n((a_j)_{1\le j\le n}, (\g_{j,k})_{1\le
j<k\le n})\in \Mat(2^n,\C)$ parameterized by 
$(a_j)_{1\le j\le n}\in\C^n$,
$(\g_{j,k})_{1\le j<k\le n}\in\R^{\frac{n(n-1)}{2}}$ as follows.
$$M_1(a_1):=\left(\begin{array}{cc}0 & a_1 \\ \overline{a_1} & 0
		  \end{array}\right).$$
Assume that we have defined $M_m((a_j)_{1\le j\le m}, (\g_{j,k})_{1\le
j<k\le m})\in \Mat(2^m,\C)$. Then, define $M_{m+1}((a_j)_{1\le j\le m+1}, (\g_{j,k})_{1\le
j<k\le m+1})\in \Mat(2^{m+1},\C)$ by
\begin{align*}
&M_{m+1}((a_j)_{1\le j\le m+1}, (\g_{j,k})_{1\le j<k\le m+1})\\
&:=\left(\begin{array}{cc} M_m((a_j)_{1\le j\le m}, (\g_{j,k})_{1\le
j<k\le m}) & a_{m+1}U_m((\g_{j,m+1})_{1\le j\le m}) \\
\overline{a_{m+1}}U_m((\g_{j,m+1})_{1\le j\le m})^* &
 M_m((a_j)_{1\le j\le m}, (\g_{j,k})_{1\le j<k\le m}) \end{array}
 \right).
\end{align*}
We can see from the definition that $M_n((a_j)_{1\le j\le n}, (\g_{j,k})_{1\le
j<k\le n})$ is hermitian. The matrix $M_n((a_j)_{1\le j\le n}, (\g_{j,k})_{1\le
j<k\le n})$ is meant to be a generalization of the hopping matrix in
momentum space. Before substituting the physical parameters, let us
summarize its general properties. For any $M\in \Mat(n,\C)$ let $\|M\|_{n\times n}$ denote
its operator norm $\sup_{\bv\in\C^n\text{ with
}\|\bv\|_{\C^n}=1}\|M\bv\|_{\C^n}$. 

\begin{lemma}\label{lem_general_hopping_properties}
\begin{enumerate}
\item\label{item_general_hopping_characterization}
For any $\rho,\eta\in\cB_n$, 
\begin{align}
&M_n((a_j)_{1\le j\le n},(\g_{j,k})_{1\le j<k\le n})(\rho,\eta)\label{eq_general_hopping_characterization}\\
&=1_{\exists j\in\{1,2,\cdots,n\}\text{
 s.t. }b_n(\rho)(j)<b_n(\eta)(j)\bigwedge b_n(\rho)(k)=b_n(\eta)(k)\
 (\forall k\in\{1,2,\cdots,n\}\backslash \{j\})}\notag\\
&\qquad\cdot e^{i1_{j\ge
 2}\sum_{l=1}^{j-1}b_n(\rho)(l)\g_{l,j}}a_j\notag\\
&\quad + 1_{\exists j\in\{1,2,\cdots,n\}\text{
 s.t. }b_n(\rho)(j)>b_n(\eta)(j)\bigwedge b_n(\rho)(k)=b_n(\eta)(k)\
 (\forall k\in\{1,2,\cdots,n\}\backslash \{j\})}\notag\\
&\qquad\quad\cdot e^{-i1_{j\ge
 2}\sum_{l=1}^{j-1}b_n(\eta)(l)\g_{l,j}}\overline{a_j}.\notag
\end{align}
\item\label{item_unitary_general_hopping_unitary}
For any $(\xi_j)_{1\le j\le n}\in\R^n$,
\begin{align*}
&U_n((\xi_j)_{1\le j\le n})M_n((a_j)_{1\le j\le n},(\g_{j,k})_{1\le
 j<k\le n})U_n((\xi_j)_{1\le j\le n})^*\\
&=M_n((e^{-i\xi_j}a_j)_{1\le j\le n},(\g_{j,k})_{1\le
 j<k\le n}).
\end{align*}
\item\label{item_general_hopping_upper}
\begin{align*}
\|M_n((a_j)_{1\le j\le n},(\g_{j,k})_{1\le j<k\le n})\|_{2^n\times
 2^n}\le\sum_{j=1}^n|a_j|.
\end{align*}
\item\label{item_general_hopping_lower}
\begin{align*}
&\inf_{\bv\in\C^{2^n}\text{ with }\atop
\|\bv\|_{\C^{2^n}}=1}\|M_n((a_j)_{1\le j\le n},(\g_{j,k})_{1\le
 j<k\le n})\bv\|_{\C^{2^n}}^2\\
&\ge \left(1-1_{n\ge
 2}\frac{1}{2}\max_{m\in\{1,2,\cdots,n\}}\left(\sum_{j=1}^{m-1}|1+e^{i\g_{j,m}}|+\sum_{j=m+1}^n|1+e^{i\g_{m,j}}|\right)\right)
\sum_{j=1}^n|a_j|^2.
\end{align*}
\end{enumerate}
\end{lemma}
\begin{proof}
\eqref{item_general_hopping_characterization}: Assume that the result is
 true for $\rho,\eta\in\cB_n$ with $\rho\le \eta$. Then, the result for
 $\rho,\eta\in\cB_n$ with $\rho>\eta$ follows from the hermiticity
of $M_n$. Thus, it suffices to prove the equality for $\rho,\eta\in
 \cB_n$ with $\rho\le \eta$. It holds for $n=1$ by definition. Assume
 that it is true for some $n\in\N$. Take $\rho,\eta\in\cB_{n+1}$
 satisfying $\rho\le \eta$. It follows that $b_{n+1}(\rho)(n+1)\le
 b_{n+1}(\eta)(n+1)$. If $b_{n+1}(\rho)(n+1)=b_{n+1}(\eta)(n+1)$, by
 setting $m:=b_{n+1}(\rho)(n+1)2^n$ we see that
\begin{align*}
&M_{n+1}((a_j)_{1\le j\le n+1},(\g_{j,k})_{1\le j<k\le
 n+1})(\rho,\eta)\\
&= M_{n}((a_j)_{1\le j\le n},(\g_{j,k})_{1\le j<k\le
 n})(\rho-m,\eta-m)\\
&=1_{\exists j\in\{1,2,\cdots,n\}\text{
 s.t. }b_n(\rho-m)(j)<b_n(\eta-m)(j)\bigwedge
 b_n(\rho-m)(k)=b_n(\eta-m)(k)\ (\forall
 k\in\{1,2,\cdots,n\}\backslash\{j\})}\\
&\quad\cdot e^{i1_{j\ge
 2}\sum_{l=1}^{j-1}b_n(\rho-m)(l)\g_{l,j}}a_j\\
&=1_{\exists j\in\{1,2,\cdots,n+1\}\text{
 s.t. }b_{n+1}(\rho)(j)<b_{n+1}(\eta)(j)\bigwedge
 b_{n+1}(\rho)(k)=b_{n+1}(\eta)(k)\ (\forall
 k\in\{1,2,\cdots,n+1\}\backslash\{j\})}\\
&\quad\cdot e^{i1_{j\ge
 2}\sum_{l=1}^{j-1}b_{n+1}(\rho)(l)\g_{l,j}}a_j\\
&=(\text{the right-hand side of
 }\eqref{eq_general_hopping_characterization}).
\end{align*}
If $b_{n+1}(\rho)(n+1)< b_{n+1}(\eta)(n+1)$, by Lemma
 \ref{lem_unitary_characterization}, 
\begin{align*}
&M_{n+1}((a_j)_{1\le j\le n+1},(\g_{j,k})_{1\le j<k\le
 n+1})(\rho,\eta)\\
&=a_{n+1}U_n((\g_{j,n+1})_{1\le j\le
 n})(\rho,\eta-b_{n+1}(\eta)(n+1)2^n)\\
&=e^{i\sum_{l=1}^nb_n(\rho)(l)\g_{l,n+1}}\delta_{\rho,\eta-b_{n+1}(\eta)(n+1)2^n}
a_{n+1}\\
&=1_{b_{n+1}(\rho)(k)=b_{n+1}(\eta)(k)\
 (\forall
 k\in\{1,2,\cdots,n\})}e^{i\sum_{l=1}^nb_n(\rho)(l)\g_{l,n+1}}a_{n+1}\\
&=(\text{the right-hand side of
 }\eqref{eq_general_hopping_characterization}).
\end{align*}
Thus, the results hold for $n+1$. The induction with $n$ proves the claim
 for any $n\in\N$.

\eqref{item_unitary_general_hopping_unitary}: The equality for $n=1$ can
 be confirmed by a direct calculation. Assume that it is true for some
 $n\in\N$. By the definition and the hypothesis of induction, 
\begin{align*}
 &U_{n+1}((\xi_j)_{1\le j\le n+1})
M_{n+1}((a_j)_{1\le j\le n+1},(\g_{j,k})_{1\le j<k\le
 n+1})U_{n+1}((\xi_j)_{1\le j\le n+1})^*\\
&=\Bigg(\begin{array}{c} U_{n}((\xi_j)_{1\le j\le n})
M_{n}((a_j)_{1\le j\le n},(\g_{j,k})_{1\le j<k\le
 n})U_{n}((\xi_j)_{1\le j\le n})^* \\
e^{i\xi_{n+1}}\overline{a_{n+1}}U_n((\g_{j,n+1})_{1\le j\le
  n})^*
\end{array}\\
&\qquad\qquad\qquad\ \begin{array}{c} e^{-i\xi_{n+1}}a_{n+1}U_n((\g_{j,n+1})_{1\le j\le n}) \\
  U_{n}((\xi_j)_{1\le j\le n})
M_{n}((a_j)_{1\le j\le n},(\g_{j,k})_{1\le j<k\le
 n})U_{n}((\xi_j)_{1\le j\le n})^*
\end{array}\Bigg)\\
&=M_{n+1}((e^{-i\xi_j}a_j)_{1\le j\le n+1},(\g_{j,k})_{1\le
 j<k\le n+1}).
\end{align*}
Thus, by induction the equality holds for any $n\in\N$. 

\eqref{item_general_hopping_upper}: We can see from the definition that 
the inequality holds for $n=1$. Assume that it holds for some $n\in\N$.
By the unitary property of $U_n((\g_{j,n+1})_{1\le j\le n})$ and the claim
 \eqref{item_unitary_general_hopping_unitary} we have that
\begin{align*}
&M_{n+1}((a_j)_{1\le j\le n+1},(\g_{j,k})_{1\le j<k\le n+1})^2=\\
&\Bigg(\begin{array}{c} M_n((a_j)_{1\le j\le n},(\g_{j,k})_{1\le
 j<k\le n})^2+|a_{n+1}|^2I_{2^n} \\
\overline{a_{n+1}}U_n((\g_{j,n+1})_{1\le j\le
 n})^*M_n(((1+e^{-i\g_{j,n+1}})a_j)_{1\le j\le n},(\g_{j,k})_{1\le j<k\le n})
\end{array}\\
&\qquad\quad \begin{array}{c}a_{n+1}M_n(((1+e^{-i\g_{j,n+1}})a_j)_{1\le j\le
  n},(\g_{j,k})_{1\le j<k\le n})U_n((\g_{j,n+1})_{1\le j\le
 n})\\
M_n((a_j)_{1\le j\le n},(\g_{j,k})_{1\le
 j<k\le n})^2+|a_{n+1}|^2I_{2^n} 
\end{array}\Bigg).
\end{align*}
Thus, for any $\bv_1,\bv_2\in\C^{2^n}$,
\begin{align}
&\left\|M_{n+1}((a_j)_{1\le j\le n+1},(\g_{j,k})_{1\le j<k\le
 n+1})\left(\begin{array}{c}\bv_1 \\\bv_2\end{array}
\right)\right\|^2_{\C^{2^{n+1}}}=\label{eq_key_matrix_equality}\\
&\sum_{l=1}^2\<\bv_l,M_{n}((a_j)_{1\le j\le n},(\g_{j,k})_{1\le j<k\le
 n})^2\bv_l+|a_{n+1}|^2\bv_l\>_{\C^{2^n}}\notag\\
&+2\Re \<\bv_1,\notag\\
&\qquad a_{n+1}M_n(((1+e^{-i\g_{j,n+1}})a_j)_{1\le
 j\le n},(\g_{j,k})_{1\le j<k\le
 n})U_n((\g_{j,n+1})_{1\le j\le n})\bv_2\>_{\C^{2^n}}.\notag
\end{align}
It follows from \eqref{eq_key_matrix_equality} and the hypothesis of
 induction that 
\begin{align*}
&\|M_{n+1}((a_j)_{1\le j\le n+1},(\g_{j,k})_{1\le j<k\le
 n+1})\|_{2^{n+1}\times 2^{n+1}}^2\\
&\le \|M_{n}((a_j)_{1\le j\le n},(\g_{j,k})_{1\le j<k\le
 n})\|_{2^{n}\times 2^{n}}^2+|a_{n+1}|^2\\
&\quad+2|a_{n+1}|\|M_n(((1+e^{-i\g_{j,n+1}})a_j)_{1\le
 j\le n},(\g_{j,k})_{1\le j<k\le n})\|_{2^n\times
 2^n}\\
&\qquad\cdot\sup_{\bv_1,\bv_2\in \C^{2^n}\text{ with }\atop
 \|\bv_1\|^2_{\C^{2^n}}+\|\bv_2\|^2_{\C^{2^n}}=1}\|\bv_1\|_{\C^{2^n}}\|\bv_2\|_{\C^{2^n}}\\
&\le \|M_{n}((a_j)_{1\le j\le n},(\g_{j,k})_{1\le j<k\le
 n})\|_{2^{n}\times 2^{n}}^2+|a_{n+1}|^2\\
&\quad+|a_{n+1}|\|M_n(((1+e^{-i\g_{j,n+1}})a_j)_{1\le
 j\le n},(\g_{j,k})_{1\le j<k\le n})\|_{2^n\times
 2^n}\\
&\le
 \left(\sum_{j=1}^n|a_j|\right)^2+|a_{n+1}|^2+|a_{n+1}|\sum_{j=1}^n|1+e^{i\g_{j,n+1}}||a_j|\\
&\le \left(\sum_{j=1}^{n+1}|a_j|\right)^2.
\end{align*}
Thus, the inequality holds for $n+1$. The induction with $n$ ensures the
 result.

\eqref{item_general_hopping_lower}: First let us prove that
\begin{align}
&\inf_{\bv\in\C^{2^n}\text{ with }\atop
\|\bv\|_{\C^{2^n}}=1}\|M_n((a_j)_{1\le j\le n},(\g_{j,k})_{1\le
 j<k\le n})\bv\|_{\C^{2^n}}^2\label{eq_pre_hopping_lower_bound}\\
&\ge \sum_{j=1}^n|a_j|^2-1_{n\ge
 2}\sum_{m=2}^n|a_m|\sum_{j=1}^{m-1}|1+e^{i\g_{j,m}}||a_j|.\notag
\end{align}
We can check that the inequality \eqref{eq_pre_hopping_lower_bound}
 holds for $n=1$. Assume that it holds for some $n\in\N$. By
 \eqref{eq_key_matrix_equality}, the induction hypothesis and the claim
 \eqref{item_general_hopping_upper}, 
\begin{align*}
&\inf_{\bv\in\C^{2^{n+1}}\text{ with }\atop
\|\bv\|_{\C^{2^{n+1}}}=1}\|M_{n+1}((a_j)_{1\le j\le n+1},(\g_{j,k})_{1\le
 j<k\le n+1})\bv\|_{\C^{2^{n+1}}}^2\\
&\ge \inf_{\bv\in\C^{2^n}\text{ with }\atop
\|\bv\|_{\C^{2^n}}=1}\|M_n((a_j)_{1\le j\le n},(\g_{j,k})_{1\le
 j<k\le n})\bv\|_{\C^{2^n}}^2\\
&\quad +|a_{n+1}|^2-|a_{n+1}|\|M_{n}(((1+e^{-i\g_{j,n+1}})a_j)_{1\le j\le n},(\g_{j,k})_{1\le
 j<k\le n})\|_{2^n\times 2^n}\\
 &\ge \sum_{j=1}^n|a_j|^2-1_{n\ge
 2}\sum_{m=2}^n|a_m|\sum_{j=1}^{m-1}|1+e^{i\g_{j,m}}||a_j|\\
&\quad +|a_{n+1}|^2-|a_{n+1}|\sum_{j=1}^n|1+e^{i\g_{j,n+1}}||a_j|\\
&=\sum_{j=1}^{n+1}|a_j|^2-\sum_{m=2}^{n+1}|a_m|\sum_{j=1}^{m-1}|1+e^{i\g_{j,m}}||a_j|.
\end{align*}
Thus, the inequality \eqref{eq_pre_hopping_lower_bound} holds for
 $n+1$. By induction it holds true for any $n\in\N$.

Define $S\in \Mat(n,\C)$ by
\begin{align*}
S(j,k):=\left\{\begin{array}{ll}\frac{1}{2}|1+e^{i\g_{j,k}}| &
	 \text{if }j<k,\\
\frac{1}{2}|1+e^{i\g_{k,j}}| &
	 \text{if }j>k,\\
0 & \text{if }j=k.
\end{array}\right.
\end{align*}
It follows from the inequality \eqref{eq_pre_hopping_lower_bound}
 that
\begin{align*}
&\inf_{\bv\in\C^{2^n}\text{ with }\atop
\|\bv\|_{\C^{2^n}}=1}\|M_n((a_j)_{1\le j\le n},(\g_{j,k})_{1\le
 j<k\le n})\bv\|_{\C^{2^n}}^2\\
&\ge \sum_{j=1}^n|a_j|^2-1_{n\ge
 2}\sum_{j=1}^n\sum_{k=1}^nS(j,k)|a_j||a_k|\ge (1-1_{n\ge 2}\|S\|_{n\times n})\sum_{j=1}^n|a_j|^2.
\end{align*}
It remains to prove that
\begin{align}
\|S\|_{n\times n}\le \frac{1}{2}\max_{m\in\{1,2,\cdots,n\}}\left(
\sum_{j=1}^{m-1}|1+e^{i\g_{j,m}}|+\sum_{j=m+1}^n|1+e^{i\g_{m,j}}|
\right).\label{eq_inside_hopping_upper_bound}
\end{align}
Though the inequality of this form is well-known (see e.g. \\
\cite[\mbox{Lemma 3.1.1}]{BR}), we give the
 proof for completeness.
Let $\alpha\in\R$ be an eigen value of $S$ such that
 $|\alpha|=\|S\|_{n\times n}$. Let $\bv=(v_1,\cdots,v_n)^t\in \C^n$ be its eigen vector. We
 can choose $l\in\{1,2,\cdots,n\}$ so that
 $|v_l|=\max_{j\in\{1,2,\cdots,n\}}|v_j|$. Then,
\begin{align*}
\|S\|_{n\times n}=\left|\frac{1}{v_l}\sum_{j=1}^nS(l,j)v_j\right|\le
 \sum_{j=1}^n|S(l,j)|\le \max_{m\in\{1,2,\cdots,n\}}\sum_{j=1}^n|S(m,j)|,
\end{align*}
which is \eqref{eq_inside_hopping_upper_bound}.
\end{proof}

Now we fix $d\in\N_{\ge 2}$ and use
the notations $\G(L)$, $\cB$, $b$ instead of $\G_d(L)$, $\cB_d$, $b_d$ respectively. 
Here we formulate the hopping matrix of our multi-band model. Set $\bpi:=(\pi,\pi,\cdots,\pi)\in \R^d$. For parameters
$\beps=(\eps_j)_{1\le j\le d}\in\R^d$, $\bga=(\g_{j,k})_{1\le j<k\le
d}\in \R^{d(d-1)/2}$ we define $E(\beps,\bga)\in\Map(\R^d,\Mat(2^d,\C))$ by
\begin{align*}
E(\beps,\bga)(\bk):=M_d\left(\left(t_j(1+e^{i\frac{\pi}{L}\eps_j-ik_j})\left(\frac{1}{2}\right)^{1_{L=1}}\right)_{1\le
 j\le d}, -\bga\right),\quad (\bk\in\R^d).
\end{align*}
We will see that $E(\beps,\bga)$ is equal to the hopping matrix of our multi-band
Hamiltonian in momentum
space if we replace $\beps, \bga$ by the actual parameters.
 The next lemma follows from Lemma
\ref{lem_general_hopping_properties} and the definition of
$E(\beps,\bga)$.

\begin{lemma}\label{lem_hopping_properties} The following statements
 hold for any $\bk\in\R^d$, $\beps\in\R^d$, $\bga\in \R^{d(d-1)/2}$.
\begin{enumerate}
\item\label{item_unitary_hopping_minus}
\begin{align*}
U_d(\bpi)E(\beps,\bga)(\bk)U_d(\bpi)^*=-E(\beps,\bga)(\bk).
\end{align*}
\item\label{item_unitary_hopping_invariance}
\begin{align*}
&U_d\left(-\frac{\pi}{L}\beps\right)U_d(\bk)E(\beps,\bga)\left(-\bk+\frac{2\pi}{L}\beps\right)U_d(\bk)^*U_d\left(-\frac{\pi}{L}\beps\right)^*\\
&=E(\beps,\bga)(\bk).
\end{align*}
\item\label{item_hopping_upper}
\begin{align*}
\|E(\beps,\bga)(\bk)\|_{2^d\times 2^d}\le 2\sum_{j=1}^dt_j.
\end{align*}
\item\label{item_derivative_hopping_upper}
\begin{align*}
&\left\|\left(\frac{\partial}{\partial
 k_j}\right)^mE(\beps,\bga)(\bk)\right\|_{2^d\times 2^d}\le t_j,\
 (\forall j\in\{1,2,\cdots,d\},m\in\N_{\ge 1}).
\end{align*}
\item\label{item_hopping_lower}
\begin{align*}
&\inf_{\bv\in\C^{2^d}\text{ with }\atop
\|\bv\|_{\C^{2^d}}=1}\|E(\beps,\bga)(\bk)\bv\|_{\C^{2^d}}\\
&\ge \left(1-\frac{1}{2}\max_{m\in\{1,2,\cdots,d\}}\left(\sum_{j=1}^{m-1}|1+e^{i\g_{j,m}}|+\sum_{j=m+1}^d|1+e^{i\g_{m,j}}|\right)\right)^{\frac{1}{2}}\\
&\quad\cdot \frac{1}{2}\min_{j\in\{1,2,\cdots,d\}}t_j\cdot\left(\sum_{l=1}^d|1+e^{i\frac{\pi}{L}\eps_l-ik_l}|^2\right)^{\frac{1}{2}}.
\end{align*}
\end{enumerate}
\end{lemma}

We define $\nu\in\Map(\cB\times\G(L),\G(2L))$ by
$\nu((\rho,\bx)):=2\bx+b(\rho)$. Note that $\nu$ is
bijective. The momentum lattice $\G(L)^*$, dual to $\G(L)$, is defined by 
$$
\G(L)^*:=\left\{0,\frac{2\pi}{L},\cdots,\frac{2\pi}{L}(L-1)\right\}^d.
$$
With the physical parameters $\eps_l^L\in\{0,1\}$, $\theta_{j,k}\in\R$ ($l,j,k\in
\{1,2,\cdots,d\}$ with $j<k$) introduced in Subsection
\ref{subsec_hamiltonian}, we set $\beps^L:=(\eps_j^L)_{1\le j\le d}$,
$\btheta:=(\theta_{j,k})_{1\le j<k\le d}$.
Then, we define $F(\beps^L,\btheta)\in\Map((\cB\times
\G(L))^2,\C)$, $G(\beps^L,\btheta)\in$ 
$\Map(\G(2L)^2,\C)$, which formulate
the hopping matrices, as follows.
\begin{align*}
&F(\beps^L,\btheta)((\rho,\bx),(\eta,\by)):=\frac{1}{L^d}\sum_{\bk\in\G(L)^*}e^{i\<\bx-\by,\bk\>}E(\beps^L,\btheta)(\bk)(\rho,\eta),\\
&(\forall (\rho,\bx),(\eta,\by)\in\cB\times\G(L)),\\
&G(\beps^L,\btheta)(\bx,\by):=F(\beps^L,\btheta)(\nu^{-1}(\bx),\nu^{-1}(\by)),\quad
 (\forall \bx,\by\in\G(2L)).
\end{align*}
Moreover, we define $\varphi\in\Map(\Z^d\times\Z^d,\R)$ by 
\begin{align*}
\varphi(\bx,\by):=\left\{\begin{array}{l} (-1)^{x_j+1}1_{j\ge
  2}\sum_{l=1}^{j-1}1_{x_l\in2\Z+1}\theta_{l,j}+1_{x_j\in2\Z}\frac{\pi}{L}\eps_{j}^L\\
\quad\text{ if }\exists j\in\{1,2,\cdots,d\}\text{ s.t. }\bx-\by=\be_j\text{
 in }(\Z/2L\Z)^d,\\
 (-1)^{x_j+1}1_{j\ge
  2}\sum_{l=1}^{j-1}1_{x_l\in2\Z+1}\theta_{l,j}-1_{x_j\in2\Z+1}\frac{\pi}{L}\eps_{j}^L\\
\quad\text{ if }\exists j\in\{1,2,\cdots,d\}\text{ s.t. }\bx-\by=-\be_j\text{
 in }(\Z/2L\Z)^d,\\
0\text{ otherwise.}
\end{array}
\right.
\end{align*}
Note that $\varphi$ satisfies \eqref{eq_phase_condition_general}.

\begin{lemma}\label{lem_position_hopping_properties}
\begin{enumerate}
\item\label{item_hopping_amplitude_phase}
\begin{align*}
&G(\beps^L,\btheta)(\bx,\by)=|G(\beps^L,\btheta)(\bx,\by)|e^{i\varphi(\bx,\by)},\
 (\forall \bx,\by\in\G(2L)).
\end{align*}
\item\label{item_hopping_amplitude}
\begin{align*}
&|G(\beps^L,\btheta)(\bx,\by)|\\
&=\left\{\begin{array}{ll} t_j &\text{if
				}\exists j\in\{1,2,\cdots,d\}\text{
				 s.t. }\bx-\by=\be_j\text{ or
				 }-\be_j\text{ in }(\Z/2L\Z)^d,\\
0 &\text{otherwise,}\end{array}\right.\\
&(\forall \bx,\by\in\G(2L)).
\end{align*}
\item\label{item_hopping_flux_per_plaquette}
\begin{align*}
&\varphi(\bx+\be_j,\bx)+\varphi(\bx+\be_j+\be_k,\bx+\be_j)+\varphi(\bx+\be_k,\bx+\be_j+\be_k)\\
&\quad+\varphi(\bx,\bx+\be_k)\\
&=(-1)^{x_j+x_k}\theta_{j,k},\quad (\forall \bx\in\Z^d,j,k\in\{1,2,\cdots,d\}\text{ with }j<k).
\end{align*}
\item\label{item_hopping_flux_per_circle}
\begin{align*}
\sum_{m=0}^{2L-1}\varphi(\bx+(m+1)\be_j,\bx+m\be_j)=\eps_j^L\pi,\quad(\forall
 \bx\in\Z^d,j\in\{1,2,\cdots,d\}).
\end{align*}
\end{enumerate}
\end{lemma}
\begin{proof}
\eqref{item_hopping_amplitude_phase}, \eqref{item_hopping_amplitude}:
Take $\bx,\by\in\G(2L)$. Let
 $(\rho,\hat{\bx}),(\eta,\hat{\by})$ $(\in\cB\times\G(L))$ be such that
$(\rho,\hat{\bx})=\nu^{-1}(\bx)$, $(\eta,\hat{\by})=\nu^{-1}(\by)$.
Moreover, let $b(\rho)=(\rho_1,\rho_2,\cdots,\rho_d)$,
 $b(\eta)=(\eta_1,\eta_2,\cdots,\eta_d)$. By Lemma
 \ref{lem_general_hopping_properties}
 \eqref{item_general_hopping_characterization} and the assumption that
 $\eps_l^1=0$ $(\forall l\in \{1,2,\cdots,d\})$,
\begin{align*}
&G(\beps^L,\btheta)(\bx,\by)\\
&=\frac{1}{L^d}\sum_{\bk\in\G(L)^*}e^{i\<\hat{\bx}-\hat{\by},\bk\>}\Bigg(
1_{\exists j\in\{1,2,\cdots,d\}\text{ s.t. }\rho_j<\eta_j\bigwedge
 \rho_m=\eta_m\ (\forall m\in\{1,2,\cdots,d\}\backslash\{j\})}\\
&\qquad\cdot e^{-i1_{j\ge
 2}\sum_{l=1}^{j-1}\rho_l\theta_{l,j}}t_j(1+e^{i\frac{\pi}{L}\eps_{j}^L-ik_j})\left(\frac{1}{2}\right)^{1_{L=1}}\\
&\quad + 1_{\exists j\in\{1,2,\cdots,d\}\text{ s.t. }\rho_j>\eta_j\bigwedge
 \rho_m=\eta_m\ (\forall m\in\{1,2,\cdots,d\}\backslash\{j\})}\\
&\qquad \cdot e^{i1_{j\ge
 2}\sum_{l=1}^{j-1}\eta_l\theta_{l,j}}t_j(1+e^{-i\frac{\pi}{L}\eps_{j}^L+ik_j})\left(\frac{1}{2}\right)^{1_{L=1}}\Bigg)\\
&=1_{\exists j\in\{1,2,\cdots,d\}\text{ s.t. }\rho_j<\eta_j\bigwedge
 \rho_m=\eta_m\ (\forall m\in\{1,2,\cdots,d\}\backslash\{j\})}\\
&\qquad\cdot e^{-i1_{j\ge
 2}\sum_{l=1}^{j-1}\rho_l\theta_{l,j}}t_j
\left(1_{\hat{x}_j=\hat{y}_j}+e^{i\frac{\pi}{L}\eps_j^L}1_{\hat{x}_j=\hat{y}_j+1\
 (\text{mod} L)}\right)\left(\frac{1}{2}\right)^{1_{L=1}}\prod_{m=1\atop
 m\neq j}^d1_{\hat{x}_m=\hat{y}_m}\\
&\quad +
1_{\exists j\in\{1,2,\cdots,d\}\text{ s.t. }\rho_j>\eta_j\bigwedge
 \rho_m=\eta_m\ (\forall m\in\{1,2,\cdots,d\}\backslash\{j\})}\\
&\qquad\cdot e^{i1_{j\ge
 2}\sum_{l=1}^{j-1}\eta_l\theta_{l,j}}t_j
\left(1_{\hat{x}_j=\hat{y}_j}+e^{-i\frac{\pi}{L}\eps_j^L}1_{\hat{x}_j=\hat{y}_j-1\
 (\text{mod} L)}\right)\left(\frac{1}{2}\right)^{1_{L=1}}\prod_{m=1\atop
 m\neq j}^d1_{\hat{x}_m=\hat{y}_m}\\
&=1_{\exists j\in\{1,2,\cdots,d\}\text{ s.t. }\bx-\by=\be_j\text{ or }-\be_j
\text{ in }(\Z/2L\Z)^d\bigwedge x_j\in 2\Z}e^{-i1_{j\ge
 2}\sum_{l=1}^{j-1}1_{x_l\in2\Z+1}\theta_{l,j}}t_j\\
&\quad\cdot\left(1_{\bx-\by=-\be_j\text{ in }(\Z/2L\Z)^d}+
1_{\bx-\by=\be_j\text{ in
 }(\Z/2L\Z)^d}e^{i\frac{\pi}{L}\eps_{j}^L}\right)\left(\frac{1}{2}\right)^{1_{L=1}}\\
&\quad+1_{\exists j\in\{1,2,\cdots,d\}\text{ s.t. }\bx-\by=\be_j\text{ or }-\be_j
\text{ in }(\Z/2L\Z)^d\bigwedge x_j\in 2\Z+1}e^{i1_{j\ge
 2}\sum_{l=1}^{j-1}1_{x_l\in2\Z+1}\theta_{l,j}}t_j\\
&\qquad\cdot\left(1_{\bx-\by=\be_j\text{ in }(\Z/2L\Z)^d}+
1_{\bx-\by=-\be_j\text{ in
 }(\Z/2L\Z)^d}e^{-i\frac{\pi}{L}\eps_{j}^L}\right)\left(\frac{1}{2}\right)^{1_{L=1}}\\
&=1_{L=1}1_{\exists j\in\{1,2,\cdots,d\}\text{ s.t. }\bx-\by=\be_j
\text{ in }(\Z/2L\Z)^d}e^{i(-1)^{x_j+1}1_{j\ge
 2}\sum_{l=1}^{j-1}1_{x_l\in2\Z+1}\theta_{l,j}}t_j\\
&\quad + 1_{L\ge 2}
\Big(1_{\exists j\in \{1,2,\cdots,d\}\text{ s.t. }\bx-\by=\be_j\text{ in
 }(\Z/2L\Z)^d}\\
&\qquad\qquad\quad\cdot e^{i(-1)^{x_j+1}1_{j\ge
 2}\sum_{l=1}^{j-1}1_{x_l\in2\Z+1}\theta_{l,j}+i1_{x_j\in2\Z}\frac{\pi}{L}\eps_j^L}t_j\\
&\qquad\qquad+1_{\exists j\in \{1,2,\cdots,d\}\text{ s.t. }\bx-\by=-\be_j\text{ in
 }(\Z/2L\Z)^d}\\
&\qquad\qquad\quad\cdot e^{i(-1)^{x_j+1}1_{j\ge
 2}\sum_{l=1}^{j-1}1_{x_l\in2\Z+1}\theta_{l,j}-i1_{x_j\in2\Z+1}\frac{\pi}{L}\eps_j^L}t_j\Big).
\end{align*}
This implies the claims \eqref{item_hopping_amplitude_phase},
 \eqref{item_hopping_amplitude}.

\eqref{item_hopping_flux_per_plaquette}:
Take $j,k\in\{1,2,\cdots,d\}$ with $j<k$ and $\bx\in\Z^d$. By
 definition,
\begin{align*}
&\varphi(\bx+\be_j,\bx)+\varphi(\bx+\be_j+\be_k,\bx+\be_j)+\varphi(\bx+\be_k,\bx+\be_j+\be_k)\\
&\quad+\varphi(\bx,\bx+\be_k)\\
&=(-1)^{x_j+2}1_{j\ge 2}\sum_{l=1}^{j-1}1_{x_l\in 2\Z+1}\theta_{l,j}+
1_{x_j+1\in2\Z}\frac{\pi}{L}\eps_j^L\\
&\quad + (-1)^{x_k+2}\left(\sum_{l=1}^{j-1}1_{x_l\in
 2\Z+1}\theta_{l,k}+1_{x_j+1\in2\Z+1}\theta_{j,k}+\sum_{l=j+1}^{k-1}1_{x_l\in 2\Z+1}\theta_{l,k}\right)\\
&\quad + 1_{x_k+1\in2\Z}\frac{\pi}{L}\eps_k^L +(-1)^{x_j+1}1_{j\ge
 2}\sum_{l=1}^{j-1}1_{x_l\in2\Z+1}\theta_{l,j}-1_{x_j\in2\Z+1}\frac{\pi}{L}\eps_{j}^L\\
&\quad+(-1)^{x_k+1}\sum_{l=1}^{k-1}1_{x_l\in 2\Z+1}\theta_{l,k}
-1_{x_k\in2\Z+1}\frac{\pi}{L}\eps_{k}^L\\
&=(-1)^{x_k+2}1_{x_j+1\in2\Z+1}\theta_{j,k}+(-1)^{x_k+1}1_{x_j\in2\Z+1}\theta_{j,k}\\
&=(-1)^{x_j+x_k}\theta_{j,k}.
\end{align*}
Thus, the claim \eqref{item_hopping_flux_per_plaquette} holds.

\eqref{item_hopping_flux_per_circle}:
The equality follows from the definition of $\varphi$.
\end{proof}

Since we have constructed the hopping matrix, we can readily define the
$2^d$-band Hamiltonian. Using the creation, annihilation operators on
the Fermionic Fock space $F_f(L^2(\cB\times \G(L)\times \spin))$, we set
\begin{align*}
&H_0:=\sum_{(\rho,\bx),(\eta,\by)\in\cB\times
 \G(L)}\sum_{\s\in\spin}F(\beps^L,\btheta)((\rho,\bx),(\eta,\by))\psi_{\rho\bx\s}^*\psi_{\eta\by\s},\\
&V:=\sum_{m=0}^{N_v}\sum_{(\rho_j,\bx_j,\s_j),(\eta_j,\by_j,\tau_j)\in\cB\times\G(L)\times
 \spin\atop
 (j=1,2,\cdots,m)}\\
&\cdot V_m^L(\bU)((\nu(\rho_1,\bx_1)\s_1,\cdots,\nu(\rho_m,\bx_m)\s_m),
                             (\nu(\eta_1,\by_1)\tau_1,\cdots,\nu(\eta_m,\by_m)\tau_m))\\
&\qquad\cdot
 \psi_{\rho_1\bx_1\s_1}^*\cdots\psi_{\rho_m\bx_m\s_m}^*\psi_{\eta_1\by_1\tau_1}\cdots 
\psi_{\eta_m\by_m\tau_m},\\
&H:=H_0+V
\end{align*}
for $\bU\in \R^{n_v}$.
The operator $H$ is defined in $F_f(L^2(\cB\times
\G(L)\times \spin))$ and self-adjoint. The following lemma suggests that we can
focus on the free energy density governed by the Hamiltonian $H$ in
order to prove
Theorem \ref{thm_main_theorem}.

\begin{lemma}\label{lem_free_energy_equivalence}
$$
\Tr e^{-\beta \sH}=\Tr e^{-\beta H}.
$$
\end{lemma}
\begin{proof}
Let us define the operators $\sH'_0$, $\sH'$ on $F_f(L^2(\G(2L)\times
 \spin))$ by
\begin{align*}
&\sH_0':=\sum_{\bx,\by\in\G(2L)}\sum_{\s\in\spin}G(\beps^L,\btheta)(\bx,\by)\psi_{\bx\s}^*\psi_{\by\s},\\
&\sH':=\sH_0'+\sV.
\end{align*}
Moreover, define the map $W$ from $F_f(L^2(\cB\times \G(L)\times
 \spin))$ to \\
$F_f(L^2(\G(2L)\times \spin))$ by
\begin{align*}
&W(\O_L):=\O_{2L},\\
&W(\psi_{\rho_1\bx_1\s_1}^*\cdots\psi_{\rho_n\bx_n\s_n}^*\O_L):=
\psi_{\nu(\rho_1\bx_1)\s_1}^*\cdots\psi_{\nu(\rho_n\bx_n)\s_n}^*\O_{2L}
\end{align*}
and by linearity. Here $\O_L$ denotes the vacuum of $F_f(L^2(\cB\times \G(L)\times
 \spin))$.  We can see that $W$ is unitary, $WHW^*=\sH'$ and thus $\Tr
 e^{-\beta \sH'}=\Tr e^{-\beta H}$. Since the phases $\theta_L$,
 $\varphi$ satisfy \eqref{eq_phase_condition_general},
 \eqref{eq_flux_per_plaquette} and \eqref{eq_flux_per_circle}, Lemma
 \ref{lem_partition_gauge_invariance} in Appendix \ref{app_flux_phase}
 ensures that $\Tr
 e^{-\beta \sH}=\Tr e^{-\beta \sH'}$. Thus, we obtain the claimed
 equality.
\end{proof}

From here until the proof of Theorem \ref{thm_main_theorem} in
Subsection \ref{subsec_completion_IR} we mainly
study $\Tr e^{-\beta H}$ instead of $\Tr e^{-\beta \sH}$.

\subsection{Grassmann integral formulation}\label{subsec_grassmann}
In this subsection we derive finite-dimensional Grassmann integral
formulations of the quantity $\log(\Tr e^{-\beta H}/\Tr e^{-\beta
H_0})$. Most of the lemmas in this subsection are based on the same
ideas as in \cite[\mbox{Subsection 2.2,\ 2.3,\ 2.4,\ 2.5}]{K15}. 
To avoid unnecessary repetition, we only provide parts of the proofs
which need to be clarified.

With the parameter $h\in (2/\beta)\N$ the index set $I$ of the
basis of Grassmann algebra is defined by
$$
I:=\cB\times\G(L)\times\spin\times[0,\beta)_h\times\{1,-1\},
$$
where $[0,\beta)_h:=\{0,1/h,2/h,\cdots,\beta-1/h\}$, a discrete version
of the interval $[0,\beta)$. Let $N$ denote
$2^{d+2}L^d\beta h$, the cardinality of $I$. Let $\cV$ be the complex
vector space spanned by the abstract basis $\{\psi_X\}_{X\in
I}$. Then, let $\bigwedge \cV$ denote the direct sum of anti-symmetric
tensor products of $\cV$. We call $\bigwedge\cV$ Grassmann algebra
generated by $\{\psi_X\}_{X\in I}$. Apart from minor differences between
the index sets, the basic description of finite-dimensional Grassmann
integral in \cite[\mbox{Subsection 2.2}]{K15} applies in this paper as
well. We follow the same notational rules concerning Grassmann
polynomials set in \cite[\mbox{Subsection 2.2}]{K15}. The Grassmann
polynomial $V(\psi)$, the analogue of the interaction $V$ in $\bigwedge
\cV$ is defined by
\begin{align}
&V(\psi):=\label{eq_Grassmann_interaction}\\
&\sum_{m=0}^{N_v}\sum_{(\rho_j,\bx_j,\s_j),(\eta_j,\by_j,\tau_j)\in\cB\times\G(L)\times\spin
 \atop
 (j=1,2,\cdots,m)}\frac{1}{h}\sum_{s\in[0,\beta)_h}\notag\\
&\cdot
 V_m^L(\bU)((\nu(\rho_1,\bx_1)\s_1,\cdots,\nu(\rho_m,\bx_m)\s_m),(\nu(\eta_1,\by_1)\tau_1,\cdots,\nu(\eta_m,\by_m)\tau_m))\notag\\
&\quad\cdot \opsi_{\rho_1\bx_1\s_1
 s}\cdots\opsi_{\rho_m\bx_m\s_ms}\psi_{\eta_1\by_1\tau_1s}\cdots\psi_{\eta_m\by_m\tau_ms}\notag
\end{align}
with $\bU\in\C^{n_v}$. 
We can expand a Grassmann polynomial $f(\psi)\in \bigwedge \cV$ by using
the anti-symmetric kernels $f_m:I^m\to\C$ $(m=1,2,\cdots,N)$ as follows.
$$
f(\psi)=f_0+\sum_{m=1}^N\left(\frac{1}{h}\right)^m\sum_{\bX\in I^m}f_m(\bX)\psi_{\bX},
$$
where $f_0\in\C$, $\psi_{\bX}:=\psi_{X_1}\psi_{X_2}\cdots\psi_{X_m}$ for
$\bX=(X_1,X_2,\cdots,X_m)\in I^m$. For any function $g$ on $I^m$ its
$L^1$-norm $\|g\|_{L^1}$ is defined by 
$$
\|g\|_{L^1}:=\left(\frac{1}{h}\right)^m\sum_{\bX\in I^m}|g(\bX)|.
$$
It will be convenient to let $\|g_0\|_{L^1}$ denote $|g_0|$ for
$g_0\in\C$ as well. Set
$U_{max}:=\max_{j\in\{1,2,\cdots,n_v\}}|U_j|$. The anti-symmetric
kernels of $V(\psi)$ can be estimated as follows.

\begin{lemma}\label{lem_estimations_Grassmann_V}
\begin{align*}
&|V_0|\le \beta L^dU_{max}v_0,\\
&\|V_{2m}\|_{L^1}\le 2^{d+1}\beta L^d U_{max}v_m(0),\quad (\forall
 m\in\{1,2,\cdots,N_v\}),\\
&\|V_{m}\|_{L^1}=0,\quad (\forall
 m\in\{1,2,\cdots,N\}\backslash\{2,4,\cdots,2N_v\}).
\end{align*}
\end{lemma}
\begin{proof}
The bounds on $|V_0|$, $\|V_m\|_{L^1}$
 $(m\in\{1,2,\cdots,N\}\backslash\{2,4,\cdots,2N_v\})$ follow from
 definition. Let $m\in \{1,2,\cdots,N_v\}$. By \cite[\mbox{Lemma
 B.1}]{K15}, the bijectivity of $\nu$ and the
 definition of $v_m(0)$ we have that 
\begin{align*}
&\|V_{2m}\|_{L^1}\\
&\le
 \beta\sum_{(\bx_j,\s_j),(\by_j,\tau_j)\in\G(2L)\times\spin\atop
 (j=1,2,\cdots,m)}|V_m^L(\bU)((\bx_1\s_1,\cdots,\bx_m\s_m),(\by_1\tau_1,\cdots,\by_m\tau_m))|\\
&\le 2^{d+1}\beta L^d U_{max}v_m(0).
\end{align*}
\end{proof}

The free covariance $C$ is defined as follows. 
For $(\rho,\bx,\s,x),(\eta,\by,\tau,y)\in\cB\times \G(L)\times\spin
\times[0,\beta)$,
\begin{align}
C(\rho\bx\s x,\eta\by\tau y):=\frac{\Tr (e^{-\beta H_0}(1_{x\ge
 y}\psi_{\rho\bx\s}^*(x)\psi_{\eta\by\tau}(y)-1_{x<y}\psi_{\eta\by\tau}(y)\psi_{\rho\bx\s}^*(x)))}{\Tr
 e^{-\beta H_0}},\label{eq_covariance_original_definition}
\end{align}
where
$\psi_{\rho\bx\s}^{(*)}(x):=e^{xH_0}\psi_{\rho\bx\s}^{(*)}e^{-xH_0}$.
Let $\cM$ denote the set of the Matsubara frequency
$(\pi/\beta)(2\Z+1)$. We introduce the finite subset $\cM_h$ of $\cM$ by
$$
\cM_h:=\{\o\in\cM\ |\ |\o|<\pi h\}.
$$
If we restrict the time variables to the discrete set $[0,\beta)_h$, the
covariance can be written as a sum over $\cM_h\times\G(L)^*$. Set
$$I_0:=\cB\times\G(L)\times\spin\times[0,\beta)_h.$$
For
$(\rho,\bx,\s,x)$, $(\eta,\by,\tau,y)\in I_0$,
\begin{align}
&C(\rho\bx\s x,\eta\by \tau y)\label{eq_introduction_covariance}\\
&=\frac{\delta_{\s,\tau}}{\beta
 L^d}\sum_{(\o,\bk)\in\cM_h\times\G(L)^*}e^{i\<\bx-\by,\bk\>+i(x-y)\o}
h^{-1}(I_{2^d}-e^{-i\frac{\o}{h}I_{2^d}+\frac{1}{h}\cE(\bk)})^{-1}(\rho,\eta),\notag
\end{align}
where $\cE\in \Map(\R^d,\Mat(2^d,\C))$ is defined by
$$\cE(\bk):=E(-\beps^L,-\btheta)(\bk).$$
Let us briefly explain how to derive \eqref{eq_introduction_covariance}.
It is implied by 
\cite[\mbox{Lemma
 2.1}]{K15} that 
\begin{align*}
&C(\rho\bx\s x,\eta\by \tau y)\\
&=\frac{\delta_{\s,\tau}}{\beta
 L^d}\sum_{(\o,\bk)\in\cM_h\times\G(L)^*}e^{-i\<\bx-\by,\bk\>+i(x-y)\o}
h^{-1}\big(I_{2^d}-e^{-i\frac{\o}{h}I_{2^d}+\frac{1}{h}\overline{E(\beps^L,\btheta)(\bk)}}\big)^{-1}(\rho,\eta).
\end{align*}
Then, by using that
$\overline{E(\beps^L,\btheta)(\bk)}=E(-\beps^L,-\btheta)(-\bk)$ we
obtain \eqref{eq_introduction_covariance}.

The next lemma states that the quantity
$\log(\Tr e^{-\beta H}/\Tr e^{-\beta H_0})$ is equal to the
time-continuum limit of the Grassmann Gaussian integral with the
covariance $C$. Despite the
generalization of the interaction, its proof is parallel to \cite[\mbox{Lemma
 2.2}]{K15}, which was built upon the idea that the discretization of
 the integrals over $[0,\beta)$ inside the perturbative expansion of
 $\Tr e^{-\beta H}/\Tr e^{-\beta H_0}$ converges well as the step size
 is sent to zero. For any $z\in\C\backslash \R_{\le 0}$ we define $\log
 z\in\C$ by the principal value $\log |z|+i\theta$ with
 $\theta\in(-\pi,\pi)$ satisfying $z=|z|e^{i\theta}$.
See \cite[\mbox{Subsection 2.2}]{K15} for the definition of the
Grassmann Gaussian integral $\int\cdot d\mu_{C}(\psi)$.

\begin{lemma}\label{lem_grassmann_formulation}
\begin{enumerate}
\item\label{item_grassmann_formulation_real}
For any $r\in\R_{>0}$ there exists $h_0\in\R_{>0}$ such that
\begin{align*}
&\Re \int e^{-V(\psi)}d\mu_{C}(\psi)>0,\\
&(\forall \bU\in\overline{D(r)}^{n_v}\cap \R^{n_v},h\in
 (2/\beta)\N\text{ with }h\ge h_0).
\end{align*}
\item\label{item_grassmann_formulation_convergence}
For any $r\in\R_{>0}$,
\begin{align*}
&\lim_{h\to\infty\atop
 h\in(2/\beta)\N}\sup_{\bU\in\overline{D(r)}^{n_v}\cap\R^{n_v}}\left|
\log\left(\frac{\Tr e^{-\beta H}}{\Tr e^{-\beta H_0}}\right)
-\log\left(\int e^{-V(\psi)}d\mu_{C}(\psi)\right)\right|=0.
\end{align*}
\end{enumerate}
\end{lemma}

Next we connect the above formulation to another Grassmann integral
formulation which has better symmetric properties from a technical view
point of infrared integration process. The general estimation in
\cite[\mbox{Appendix B}]{K15} underlies the analysis in the rest of this
section. Let $\chi$ be a compactly supported smooth function on $\R$
satisfying that $\chi(x)\in[0,1]$ $(\forall x\in\R)$. This section proceeds
without imposing more conditions on $\chi$. The function
$\chi$ will be specified after this section. Using $\chi$ as a cut-off
function, we introduce the covariances $C_{\le 0}^{+}$, $C_{>0}^+$,
$C_{\le 0}^{\infty}$, $C_{>0}^-$, $C_{>0}^{+(h)}$, $\cI\in\Map(I_0^2,\C)$ as follows. For $(\rho,\bx,\s,x),(\eta,\by,\tau,y)\in I_0$, 
\begin{align}
&C_{\le 0}^+(\rho\bx\s x,\eta\by \tau y):=\frac{\delta_{\s,\tau}}{\beta
 L^d}\sum_{(\o,\bk)\in\cM_h\times\G(L)^*}e^{i\<\bx-\by,\bk\>+i(x-y)\o}\notag\\
&\qquad\qquad\qquad\qquad\qquad\cdot\chi(h|1-e^{i\frac{\o}{h}}|)h^{-1}(I_{2^d}-e^{-i\frac{\o}{h}I_{2^d}+\frac{1}{h}\cE(\bk)})^{-1}(\rho,\eta),\notag\\
&C_{> 0}^+(\rho\bx\s x,\eta\by \tau y):=\frac{\delta_{\s,\tau}}{\beta
 L^d}\sum_{(\o,\bk)\in\cM_h\times\G(L)^*}e^{i\<\bx-\by,\bk\>+i(x-y)\o}\label{eq_UV_covariance_+}\\
&\quad\qquad\qquad\qquad\qquad\cdot(1-\chi(h|1-e^{i\frac{\o}{h}}|))h^{-1}(I_{2^d}-e^{-i\frac{\o}{h}I_{2^d}+\frac{1}{h}\cE(\bk)})^{-1}(\rho,\eta),\notag\\
&C_{\le 0}^{\infty}(\rho\bx\s x,\eta\by \tau y):=\frac{\delta_{\s,\tau}}{\beta
 L^d}\sum_{(\o,\bk)\in\cM_h\times\G(L)^*}e^{i\<\bx-\by,\bk\>+i(x-y)\o}\notag\\
&\qquad\qquad\qquad\qquad\qquad\cdot \chi(|\o|)(i\o I_{2^d}-\cE(\bk))^{-1}(\rho,\eta),\notag\\
&C_{> 0}^{-}(\rho\bx\s x,\eta\by \tau y):=\frac{\delta_{\s,\tau}}{\beta
 L^d}\sum_{(\o,\bk)\in\cM_h\times\G(L)^*}e^{i\<\bx-\by,\bk\>+i(x-y)\o}\label{eq_UV_covariance_-}\\
&\quad\qquad\qquad\qquad\qquad\cdot (1-\chi(h|1-e^{i\frac{\o}{h}}|))h^{-1}(e^{i\frac{\o}{h}I_{2^d}-\frac{1}{h}\cE(\bk)}-I_{2^d})^{-1}(\rho,\eta),\notag\\
&C_{> 0}^{+(h)}(\rho\bx\s x,\eta\by \tau y):= C_{> 0}^{+}(\rho\bx\s x,\eta\by \tau
 y)\notag\\
&\qquad\qquad\qquad\qquad\quad+\frac{1_{(\rho,\bx,\s)=(\eta,\by,\tau)}}{\beta h}
\sum_{\o\in\cM_h}e^{i(x-y)\o}\chi(h|1-e^{i\frac{\o}{h}}|),\notag\\
&\cI(\rho\bx\s x,\eta\by\tau y):=1_{(\rho,\bx,\s, x)=(\eta,\by,\tau,y)}.\notag
\end{align}
One can derive from the definitions that
\begin{align}
&C_{>0}^{+(h)}(\rho\bx\s x,\eta\by\tau y)=C^{-}_{>0}(\rho\bx\s
 x,\eta\by\tau y)+\cI(\rho\bx\s x,\eta\by\tau y),\label{eq_artificial_covariances_relation}\\
&(\forall (\rho,\bx,\s, x),(\eta,\by,\tau, y)\in I_0).\notag
\end{align}
The next lemma can be proved by applying Gram's inequality and the
Cauchy-Binet formula in the same way as in the proof of 
\cite[\mbox{Lemma 2.4}]{K15}.

\begin{lemma}\label{lem_artificial_covariances_bound}
There exist $(\beta,L,d,\chi,\cE)$-dependent, $h$-independent constants
 $h_0$, $c_1\in\R_{>0}$ such that the following inequalities hold for
 any $h\in(2/\beta)\N$ with $h\ge h_0$.
\begin{align*}
&|\det(C_o(X_i,Y_j))_{1\le i,j\le n}|\le c_1^n,\\
&|\det(C_{>0}^{+(h)}(X_i,Y_j)-C_{>0}^{+}(X_i,Y_j))_{1\le i,j\le
 n}|\le\frac{1}{h}c_1^n,\\
&|\det(C_{\le 0}^{+}(X_i,Y_j)-C_{\le 0}^{\infty}(X_i,Y_j))_{1\le i,j\le
 n}|\le\frac{1}{h}c_1^n,\\
&(\forall n\in \N,X_j,Y_j\in I_0\ (j=1,2,\cdots,n))
\end{align*}
for $C_o=C$, $C_{\le 0}^+$, $C_{> 0}^+$, $C_{\le 0}^{\infty}$, $C_{>0}^-$, $C_{>
 0}^{+(h)}$.
\end{lemma}

In the following we assume that $h\ge h_0$ so that the results of Lemma
\ref{lem_artificial_covariances_bound} are available. Define the Grassmann polynomials $V^+(\psi)$, $V^-(\psi)$, $S^+(\psi)$,
$S^-(\psi)$, $S^0(\psi)\in\bigwedge \cV$ by
\begin{align}
&V^+(\psi):=V(\psi),\label{eq_Grassmann_interaction_plus_minus}\\
&V^-(\psi):=\sum_{m=0}^{N_v}\sum_{(\rho_j,\bx_j,\s_j),(\eta_j,\by_j,\tau_j)\in\cB\times
 \G(L)\times\spin\atop
 (j=1,2,\cdots,m)}\frac{1}{h}\sum_{s\in[0,\beta)_h}(-1)^m\notag\\
&\quad\cdot
 V_m^L((\nu(\rho_1,\bx_1)\s_1,\cdots,\nu(\rho_m,\bx_m)\s_m),
       (\nu(\eta_1,\by_1)\tau_1,\cdots,\nu(\eta_m,\by_m)\tau_m))\notag\\
&\qquad\qquad\cdot\opsi_{\rho_1\bx_1\s_1 s}\cdots\opsi_{\rho_m\bx_m\s_m
 s}\psi_{\eta_1\by_1\tau_1 s}\cdots \psi_{\eta_m\by_m\tau_m s},\notag\\
&S^{\delta}(\psi):=\int
 e^{-V^{\delta}(\psi+\psi^1)}d\mu_{C_{>0}^{\delta}}(\psi^1),\ (\delta
 \in \{+,-\}),\notag\\
&S^{0}(\psi):=\int
 e^{-V^{+}(\psi+\psi^1)}d\mu_{C_{>0}^{+(h)}}(\psi^1).\notag
\end{align}
For conciseness let $g(\alpha)$ denote
$$
\beta L^dv_0+2^{d+1}\beta L^d\sum_{m=1}^{N_v}(\alpha+1)^{2m}c_1^mv_m(0)
$$
for $\alpha\in\R_{\ge 0}$, where $c_1$ is the constant appearing in
Lemma \ref{lem_artificial_covariances_bound}.

\begin{lemma}\label{lem_grassmann_interaction_bound}
\begin{enumerate}
\item\label{item_grassmann_0th_bound}
\begin{align*}
|S_0^{\delta}-e^{-V_0}|\le e^{U_{max}g(0)}-e^{U_{max}\beta L^dv_0},\quad
 (\forall \delta\in\{+,-,0\}).
\end{align*}
\item\label{item_grassmann_all_bound}
\begin{align*}
\sum_{m=0}^N\alpha^mc_1^{\frac{m}{2}}\|S_m^{\delta}\|_{L^1}\le
 e^{U_{max} g(\alpha)},\quad (\forall \alpha\in \R_{\ge 0},
 \delta\in\{+,-,0\}).
\end{align*}
\item\label{item_grassmann_all_difference_bound}
\begin{align*}
\sum_{m=0}^N\alpha^mc_1^{\frac{m}{2}}\|S_m^{+}-S_m^0\|_{L^1}\le
 \frac{1}{h}(e^{U_{max}g(\alpha+1)}-e^{U_{max}\beta L^dv_0}),\quad
(\forall \alpha\in\R_{\ge 0}).
\end{align*}
\item\label{item_grassmann_minus_zero_difference_bound}
\begin{align*}
&\sum_{m=0}^N\alpha^mc_1^{\frac{m}{2}}\|S_m^{-}-S_m^0\|_{L^1}\le
 \frac{1}{\beta h}(U_{max}g(\alpha)-U_{max}\beta
 L^dv_0)^2e^{U_{max}g(\alpha)},\\
&(\forall \alpha\in\R_{\ge 0}).
\end{align*}
\end{enumerate}
\end{lemma}

\begin{proof} 
Combination of Lemma \ref{lem_estimations_Grassmann_V}, Lemma
 \ref{lem_artificial_covariances_bound} and \\
\cite[\mbox{Lemma B.2 (1),(2),(4)}]{K15}
 yields the inequalities in
 \eqref{item_grassmann_0th_bound}, \eqref{item_grassmann_all_bound},
 \eqref{item_grassmann_all_difference_bound}. 

Let us prove the inequality in
 \eqref{item_grassmann_minus_zero_difference_bound}, which is a
 generalization of \cite[\mbox{Lemma 2.6}]{K15}. Define the functions
 $W_m^{\delta}$ $(\delta=+,-,$ $m=1,2,\cdots,N_v)$ on
 $(\cB\times\G(L)\times\spin)^m\times (\cB\times\G(L)\times\spin)^m$ 
by 
\begin{align*}
&W_m^{\delta}(((\rho_1,\bx_1,\s_1),\cdots,(\rho_m,\bx_m,\s_m)),
              ((\eta_1,\by_1,\tau_1),\cdots,(\eta_m,\by_m,\tau_m)))\\
&:=(1_{\delta=+}+1_{\delta=-}(-1)^m)\\
&\quad\cdot V_m^{L}((\nu(\rho_1,\bx_1)\s_1,\cdots,\nu(\rho_m,\bx_m)\s_m),(\nu(\eta_1,\by_1)\tau_1,\cdots,\nu(\eta_m,\by_m)\tau_m)),\\
&(\forall
 (\rho_j,\bx_j,\s_j),(\eta_j,\by_j,\tau_j)\in\cB\times\G(L)\times\spin\
(j=1,2,\cdots,m)).
\end{align*}
For any $s\in [0,\beta)_h$, $\bX=(X_1,X_2,\cdots,X_m)\in
 (\cB\times\G(L)\times\spin)^m$ we abbreviate
 $(X_m,X_{m-1},\cdots,X_1)$, $((X_1,s),(X_2,s),\cdots,(X_m,s))$,
 $\opsi_{X_1s}\opsi_{X_2s}$ $\cdots\opsi_{X_ms}$,
 $\psi_{X_1s}\psi_{X_2s}\cdots \psi_{X_ms}$ to $\widetilde{\bX}$, $\bX
 s$, $\opsi_{\bX s}$, $\psi_{\bX s}$ respectively. For
 $s\in[0,\beta)_h$ we define $W_s^{\delta}(\psi)\in\bigwedge\cV$ $(\delta
 =+,-)$ by
\begin{align*}
W_s^{\delta}(\psi):=\sum_{m=1}^{N_v}\sum_{\bX,\bY\in(\cB\times\G(L)\times\spin)^m}
W_m^{\delta}(\bX,\bY)\opsi_{\bX s}\psi_{\bY s}.
\end{align*}
Take $s_1,s_2,\cdots,s_n\in[0,\beta)_h$ satisfying $s_j\neq s_k$
 ($\forall j,k\in\{1,2,\cdots,n\}$ with $j\neq k$). By the invariance
 \eqref{eq_particle_hole}, the equality \eqref{eq_artificial_covariances_relation} and anti-symmetry, 
\begin{align}
&\int\prod_{j=1}^nW_{s_j}^+(\psi+\psi^1)d\mu_{C_{>0}^{+(h)}}(\psi^1)\label{eq_different_time_equality}
\\
&=\int\prod_{j=1}^nW_{s_j}^+(\psi+\psi^1+\psi^2)d\mu_{\cI}(\psi^2)d\mu_{C_{>0}^-}(\psi^1)\notag\\
&=\int\prod_{j=1}^n\Bigg(\sum_{m_j=1}^{N_v}\sum_{l_j=0}^{m_j}\sum_{\bX_j,\bY_j\in(\cB\times\G(L)\times\spin)^{m_j-l_j}\atop
 \bW_j\in (\cB\times\G(L)\times\spin)^{l_j}}\left(\begin{array}{c}m_j\\
						  l_j\end{array}\right)^2l_j! \notag\\
&\quad\cdot W_{m_j}^+((\bX_j,\bW_j),(\widetilde{\bW_j},\bY_j))(\opsi+\opsi^1)_{\bX_js_j}(\psi+\psi^1)_{\bY_js_j}
\opsi_{\bW_j s_j}^2\psi_{\widetilde{\bW_j} s_j}^2\Bigg)\notag\\
&\quad\cdot d\mu_{\cI}(\psi^2)d\mu_{C_{>0}^-}(\psi^1)\notag\\
&=\int\prod_{j=1}^n\Bigg(\sum_{m_j=1}^{N_v}\sum_{l_j=0}^{m_j}\sum_{\bX_j,\bY_j\in(\cB\times\G(L)\times\spin)^{m_j-l_j}\atop
 \bW_j\in (\cB\times\G(L)\times\spin)^{l_j}}\left(\begin{array}{c}m_j\\
						  l_j\end{array}\right)^2l_j!\notag\\
&\quad\cdot W_{m_j}^+((\bX_j,\bW_j),(\widetilde{\bW_j},\bY_j))(\opsi+\opsi^1)_{\bX_js_j}(\psi+\psi^1)_{\bY_js_j}
\Bigg)d\mu_{C_{>0}^-}(\psi^1)\notag\\
&=\int\prod_{j=1}^nW_{s_j}^-(\psi+\psi^1)
d\mu_{C_{>0}^-}(\psi^1).\notag
\end{align}

Define $Q^{\delta}(\psi)\in\bigwedge \cV$ $(\delta=+,-)$ by
\begin{align*}
Q^{\delta}(\psi)
:=e^{-V_0}+e^{-V_0}&\sum_{n=1}^N\frac{(-1)^n}{n!}\prod_{j=1}^n\left(\frac{1}{h}\sum_{s_j\in[0,\beta)_h}\right)\\
&\cdot 1_{\forall
 j\forall k\in\{1,2,\cdots,n\}(j\neq k\to s_j\neq
 s_k)}\prod_{j=1}^nW_{s_j}^{\delta}(\psi).
\end{align*}
Then, the equality \eqref{eq_different_time_equality} implies that
\begin{align*}
\int Q^+(\psi+\psi^1)d\mu_{C_{>0}^{+(h)}}(\psi^1)=
\int Q^-(\psi+\psi^1)d\mu_{C_{>0}^{-}}(\psi^1).
\end{align*}
Therefore,
\begin{align}
S^0(\psi)-S^-(\psi)&=\int(e^{-V^+(\psi+\psi^1)}-Q^+(\psi+\psi^1))d\mu_{C_{>0}^{+(h)}}(\psi^1)\label{eq_important_division}\\
&\quad -
 \int(e^{-V^-(\psi+\psi^1)}-Q^-(\psi+\psi^1))d\mu_{C_{>0}^{-}}(\psi^1).\notag
\end{align}

Let us set 
\begin{align*}
\tilde{S}^0(\psi):=\int(e^{-V^+(\psi+\psi^1)}-Q^+(\psi+\psi^1))d\mu_{C_{>0}^{+(h)}}(\psi^1).
\end{align*}
For any $m\in\{0,1,\cdots,N\}$ we can characterize
 $\tilde{S}^0_m(\psi)$, the $m$-th order part of $\tilde{S}^0(\psi)$ as follows.
\begin{align*}
\tilde{S}^0_m(\psi)=&e^{-V_0}\sum_{n=2}^N\frac{(-1)^n}{n!}\prod_{j=1}^n\Bigg(\frac{1}{h}\sum_{s_j\in[0,\beta)_h}\sum_{m_j=1}^{N_v}\sum_{k_j=0}^{m_j}\sum_{l_j=0}^{m_j}
\left(\begin{array}{c}m_j\\ k_j\end{array}\right)
\left(\begin{array}{c}m_j\\
      l_j\end{array}\right)\\
&\cdot
\sum_{\bX_j\in(\cB\times \G(L)\times\spin)^{m_j-k_j},\bX_j'\in(\cB\times
 \G(L)\times\spin)^{k_j}\atop
\bY_j\in(\cB\times \G(L)\times\spin)^{m_j-l_j},\bY_j'\in(\cB\times
 \G(L)\times\spin)^{l_j}}W_{m_j}^+((\bX_j,\bX_j'),(\bY_j,\bY_j'))\Bigg)\\
&\cdot 1_{\exists j\exists k\in\{1,2,\cdots,n\}(j\neq k\wedge
 s_j=s_k)}1_{\sum_{j=1}^nk_j=\sum_{j=1}^nl_j=\frac{m}{2}}\eps_{\pm}\\
&\cdot \int
 \opsi_{\bX_1s_1}^1\psi_{\bY_1s_1}^1\opsi_{\bX_2s_2}^1\psi_{\bY_2s_2}^1\cdots 
\opsi_{\bX_ns_n}^1\psi_{\bY_ns_n}^1d\mu_{C_{>0}^{+(h)}}(\psi^1)\\
&\cdot \opsi_{\bX_1's_1}\psi_{\bY_1's_1}\opsi_{\bX_2's_2}\psi_{\bY_2's_2}\cdots 
\opsi_{\bX_n's_n}\psi_{\bY_n's_n},
\end{align*}
where the factor $\eps_{\pm}\in\{1,-1\}$ depends only on $m_j,k_j,l_j$
 $(j=1,2,\cdots,$ $n)$. From this equality, Lemma
 \ref{lem_artificial_covariances_bound}, \cite[\mbox{Lemma B.1}]{K15}
 and the inequality that 
\begin{align*}
\prod_{j=1}^n\left(\frac{1}{h}\sum_{s_j\in[0,\beta)_h}\right)
1_{\exists j\exists k\in\{1,2,\cdots,n\}(j\neq k\wedge
 s_j=s_k)}\le \left(\begin{array}{c} n\\
		    2\end{array}\right)\frac{\beta^{n-1}}{h},\quad
 (\forall n\in\N_{\ge 2}),
\end{align*}
we can deduce that
\begin{align*}
\|\tilde{S}_m^0\|_{L^1}
&\le e^{|V_0|}\sum_{n=2}^N\frac{1}{n!}\left(\begin{array}{c}n \\
				   2\end{array}\right)\frac{\beta^{n-1}}{h}
\prod_{j=1}^n\Bigg(\sum_{m_j=1}^{N_v}\sum_{k_j=0}^{m_j}\sum_{l_j=0}^{m_j}
\left(\begin{array}{c}m_j\\ k_j\end{array}\right)
\left(\begin{array}{c}m_j\\
      l_j\end{array}\right)\\
&\quad\cdot \sum_{\bX,\bY\in(\G(2L)\times\spin)^{m_j}}|V_{m_j}^L(\bX,\bY)|c_1^{m_j}\Bigg)c_1^{-\frac{m}{2}}1_{\sum_{j=1}^nk_j=\sum_{j=1}^nl_j=\frac{m}{2}}.
\end{align*}
Thus, for any $\alpha\in\R_{\ge 0}$, 
\begin{align*}
\sum_{m=0}^N\alpha^mc_1^{\frac{m}{2}}\|\tilde{S}_m^0\|_{L^1}
&=\sum_{m=0}^{N/2}\alpha^{2m}c_1^{m}\|\tilde{S}_{2m}^0\|_{L^1}\\
&\le \frac{1}{2\beta
 h}e^{U_{max}\beta L^d v_0}
\sum_{n=2}^{\infty}\frac{1}{(n-2)!}(U_{max}g(\alpha)-U_{max}\beta L^d v_0)^n\\
&=\frac{1}{2\beta h}(U_{max}g(\alpha)-U_{max}\beta L^d v_0)^2e^{U_{max}g(\alpha)}.
\end{align*}
The Grassmann polynomial 
$$
\int(e^{-V^-(\psi+\psi^1)}-Q^-(\psi+\psi^1))d\mu_{C_{>0}^{-}}(\psi^1)
$$
can be estimated in the same way as above. By combining these bounds
 with the equality \eqref{eq_important_division} we can
 derive the claimed inequality.
\end{proof}

\begin{lemma}\label{lem_grassmann_exponential_bound}
Let $\alpha\in\R_{\ge 0}$ and $\eps\in (0,1)$. Assume that 
\begin{align*}
U_{max}g(\alpha)\le \log\left(\frac{2(\eps+1)}{\eps+2}\right).
\end{align*}
Then, the following inequalities hold.
\begin{enumerate}
\item\label{item_grassmann_exponential_0th_bound}
$$
|S_0^{\delta}-e^{-V_0}|\le\frac{\eps}{\eps+2},\quad (\forall \delta\in\{+,-,0\}).
$$
\item\label{item_grassmann_exponential_all_bound}
$$
\sup_{\delta\in\{+,-,0\}}\sum_{m=1}^N\alpha^mc_1^{\frac{m}{2}}\|S_m^{\delta}\|_{L^1}\le
     \eps\inf_{\delta\in\{+,-,0\}}|S_0^{\delta}|.
$$
\end{enumerate}
\end{lemma}
\begin{proof}
\eqref{item_grassmann_exponential_0th_bound}: It follows from Lemma
 \ref{lem_grassmann_interaction_bound} \eqref{item_grassmann_0th_bound}
 and the assumption that for $\delta\in\{+,-,0\}$,
\begin{align*}
|S_0^{\delta}-e^{-V_0}|\le e^{U_{max}g(0)}-1\le\frac{\eps}{\eps+2}.
\end{align*}

\eqref{item_grassmann_exponential_all_bound}:
The assumption implies that
\begin{align}
e^{U_{max}g(\alpha)}\le
 (\eps+1)(2-e^{U_{max}g(\alpha)}).\label{eq_preliminary_long_ineq}
\end{align}
Moreover, by Lemma \ref{lem_grassmann_interaction_bound}
 \eqref{item_grassmann_0th_bound} and the inequality that 
$$
|e^{-V_0}-1|\le e^{U_{max}\beta L^dv_0}-1
$$
we see that
\begin{align}
|S_0^{\delta}-1|\le |S_0^{\delta}-e^{-V_0}|+|e^{-V_0}-1|\le
 e^{U_{max}g(0)}-1 \label{eq_grassmann_0th_with_1}
\end{align}
for any $\delta\in\{+,-,0\}$. Thus, 
\begin{align}
2-e^{U_{max}g(\alpha)}\le\inf_{\delta\in\{+,-,0\}}|S_0^{\delta}|.\label{eq_preliminary_ineq_next}
\end{align}
Using Lemma \ref{lem_grassmann_interaction_bound}
 \eqref{item_grassmann_all_bound}, \eqref{eq_preliminary_long_ineq} and
 \eqref{eq_preliminary_ineq_next}, we have that
\begin{align*}
\sup_{\delta\in\{+,-,0\}}\sum_{m=1}^N\alpha^mc_1^{\frac{m}{2}}\|S_m^{\delta}\|_{L^1}\le
    e^{U_{max}g(\alpha)}-\inf_{\delta\in\{+,-,0\}}|S_0^{\delta}|\le
 \eps\inf_{\delta\in\{+,-,0\}}|S_0^{\delta}|.
 \end{align*}
\end{proof}

\begin{lemma}\label{lem_grassmann_logarithm_bound}
Let $\alpha\in \R_{\ge 0}$, $\eps\in(0,1)$. Assume that
\begin{align*}
U_{max}g(\alpha+1)\le\log\left(\frac{2(\eps+1)}{\eps+2}\right).
\end{align*}
Set $R^{\delta}(\psi):=\log S^{\delta}(\psi)$, $(\delta\in
 \{+,-,0\})$. Then, the following inequalities hold for any $h\in
 (2/\beta)\N$ satisfying $h>\max\{1/2,2/\beta,h_0\}$.
\begin{enumerate}
\item\label{item_grassmann_logarithm_bound_0th}
$$
|R_0^{\delta}|\le \log\left(\frac{\eps+2}{2}\right),\quad (\forall \delta \in \{+,-,0\}).
$$
\item\label{item_grassmann_logarithm_bound}
$$
\sum_{m=1}^N\alpha^mc_1^{\frac{m}{2}}\|R_m^{\delta}\|_{L^1}\le
     -\log(1-\eps),\quad (\forall \delta\in \{+,-,0\}).
$$
\item\label{item_grassmann_logarithm_bound_0th_difference}
\begin{align*}
|R_0^{\delta}-R_0^0|\le
 -\log\left(1-\max\left\{1,\frac{4}{\beta}\right\}\frac{1}{2h}\right),\quad (\forall \delta\in\{+,-\}).
\end{align*}
\item\label{item_grassmann_logarithm_bound_difference}
\begin{align*}
\sum_{m=1}^N\alpha^mc_1^{\frac{m}{2}}\|R_m^{\delta}-R_m^0\|_{L^1}\le
\max\left\{1,\frac{4}{\beta}\right\}\frac{1}{2(1-\eps)h},\quad (\forall \delta\in\{+,-\}).\end{align*}
\end{enumerate}
\end{lemma}

\begin{proof}
It follows from \eqref{eq_grassmann_0th_with_1} and the
 assumption that
$$
|S_0^{\delta}-1|\le\frac{\eps}{\eps+2}<1,\quad(\forall \delta\in\{+,-,0\}).
$$
This means that the assumption of \cite[\mbox{Lemma
 B.3}]{K15} is satisfied and thus we can apply it.
The claims can be proved in a way close to the proof of \cite[\mbox{Lemma
 2.8}]{K15}. We only explain which lemmas are necessary to prove each
 claim. We use the assumption, \eqref{eq_grassmann_0th_with_1} and \cite[\mbox{Lemma
 B.3 (1)}]{K15} to prove the claim
 \eqref{item_grassmann_logarithm_bound_0th}. The assumption and Lemma
 \ref{lem_grassmann_exponential_bound}
 \eqref{item_grassmann_exponential_all_bound} enable us to apply \cite[\mbox{Lemma
 B.3 (2)}]{K15} to prove the claim
 \eqref{item_grassmann_logarithm_bound}. We use the assumption, Lemma
 \ref{lem_grassmann_interaction_bound}
 \eqref{item_grassmann_all_difference_bound},\eqref{item_grassmann_minus_zero_difference_bound},
 \eqref{eq_preliminary_ineq_next} and \cite[\mbox{Lemma
 B.3 (3)}]{K15} to prove the claim
 \eqref{item_grassmann_logarithm_bound_0th_difference}. By combining the
 assumption, Lemma \ref{lem_grassmann_interaction_bound}
 \eqref{item_grassmann_all_difference_bound},\eqref{item_grassmann_minus_zero_difference_bound},
 Lemma \ref{lem_grassmann_exponential_bound}
 \eqref{item_grassmann_exponential_all_bound} and
 \eqref{eq_preliminary_ineq_next} with \cite[\mbox{Lemma
 B.3 (4)}]{K15} we can deduce the claim
 \eqref{item_grassmann_logarithm_bound_difference}.
\end{proof}

Here we reach the lemma stating that the Grassmann integral formulation
in Lemma \ref{lem_grassmann_formulation} can be approximated by another
formulation which will turn out to have a desirable symmetry later in
Section \ref{sec_IR}. We will mainly deal with this formulation in the
infrared multi-scale analysis in Section \ref{sec_IR}.
 
\begin{lemma}\label{lem_grassmann_symmetric}
There exist $(\beta,L,d,g(2),\chi,\cE)$-dependent, $h$-independent constants
 $h_0,c_2,c_3\in\R_{> 0}$ such that the following statements hold
 for any $h\in (2/\beta)\N$ satisfying $h\ge h_0$ and $\bU\in\C^{n_v}$
 satisfying $|U_j|\le c_2$ $(\forall j\in \{1,2,\cdots,n_v\})$.
\begin{enumerate}
\item\label{item_grassmann_symmetric_positive}
\begin{align*}
&\Re \int e^{-V(\psi)}d\mu_C(\psi)>0,\\
& \Re \int e^{\frac{1}{2}(R^+(\psi)+R^-(\psi))}d\mu_{C_{\le
 0}^{\infty}}(\psi)>0.
\end{align*}
\item\label{item_logarithm_final_h_estimate}
\begin{align*}
&\left|\log\left(\int e^{-V(\psi)}d\mu_{C}(\psi)\right)-\log\left(\int e^{\frac{1}{2}(R^+(\psi)+R^-(\psi))}d\mu_{C_{\le
 0}^{\infty}}(\psi)\right)\right|\le \frac{1}{h}c_3.
\end{align*}
 \end{enumerate}
\end{lemma}
\begin{proof} Take $\eps\in (0,2/5)$. Assume that 
$$
U_{max}\le g(3)^{-1}\log\left(\frac{2(\eps+1)}{\eps+2}\right).
$$
Then, all the inequalities claimed in Lemma
 \ref{lem_grassmann_logarithm_bound} hold with $\alpha=2$
 and $h\in (2/\beta)\N$ satisfying $h>\max\{1/2,2/\beta,h_0\}$. 
Note that the inequalities proved in Lemma
 \ref{lem_grassmann_logarithm_bound} have exactly the same form as those
 proved in \cite[\mbox{Lemma 2.8}]{K15}.
Based on these inequalities and \cite[\mbox{Lemma B.2}]{K15}, we
 only need to follow the same argument
 as in the proof of \cite[\mbox{Lemma 2.10}]{K15} to obtain the results.
\end{proof}

\section{The Matsubara ultra-violet integration}\label{sec_UV}

In this section we carry out a multi-scale integration over the large
Matsubara frequency. In the first subsection we summarize properties of
the covariances with the Matsubara UV cut-off. Most of these
properties have already been proved in \cite[\mbox{Lemma 6.2,\ Lemma
6.3}]{K15}. We only provide proofs for claims which are not directly
implied by \cite[\mbox{Lemma 6.2,\ Lemma
6.3}]{K15}. Using these results, we will 
establish upper bounds on Grassmann
polynomials produced by the Matsubara UV integration 
in Subsection \ref{subsec_isothermal} and Subsection
\ref{subsec_anisothermal}. Though these subsections are
aimed at achieving the same goal as in \\
\cite[\mbox{Subsection 5.1,
Subsection 5.2, Section
6}]{K15}, the generalization of the interaction creates different
aspects which cannot be skipped without proof. We will provide the full
construction of the Matsubara UV integration.

\subsection{Covariances with the Matsubara ultra-violet
  cut-off}\label{subsec_UV_covariance}
From now till the proof of Theorem \ref{thm_main_theorem} in Subsection
\ref{subsec_completion_IR} we assume that
\begin{align}
\max_{j\in\{1,2,\cdots,d\}}t_j=1.\label{eq_hopping_amplitude_condition}
\end{align}
Theorem \ref{thm_main_theorem}, the main theorem of this paper, can be deduced from that proved under
this condition. It follows from Lemma \ref{lem_hopping_properties}
\eqref{item_hopping_upper},\eqref{item_derivative_hopping_upper} and
\eqref{eq_hopping_amplitude_condition} that
\begin{align}
\sup_{j\in\{1,2,\cdots,d\}}\sup_{\bk\in\R^d}\left\|\left(\frac{\partial}{\partial
 k_j}\right)^n\cE(\bk)\right\|_{2^d\times 2^d}\le 2d,\quad (\forall
 n\in\N\cup\{0\}).\label{eq_dispersion_derivative_bound}
\end{align}
In \cite[\mbox{Lemma 6.1}]{K15}, which was based on \cite[\mbox{Theorem
1.3.5}]{H}, we introduced a function $\phi\in C^{\infty}(\R)$ satisfying
that
\begin{align*}
&\phi(x)=1,\ (\forall x\in (-\infty,\pi^2/6]),\\
&\phi(x)=0,\ (\forall x\in [\pi^2/3,\infty)),\\
&\frac{d}{dx}\phi(x)\le 0,\ (\forall x\in \R),\\
&\left|\left(\frac{d}{dx}\right)^k\phi(x)\right|\le 2^k (k!)^2,\ (\forall x\in\R,k\in\N\cup
 \{0\}).
\end{align*} 
We keep using this function to construct cut-off functions in this paper
as well.

The inequality \eqref{eq_dispersion_derivative_bound} suggests that the
general results in \\
\cite[\mbox{Subsection 6.1}]{K15} hold with
``$E_1=2d$'', ``$E_2=1$'' for our covariances if we define the cut-off
functions in the same manner as in \cite[\mbox{Subsection 6.1}]{K15}.
Let us do so for simplicity. With $M\in\R_{>1}$, set
\begin{align*}
&M_{UV}:=\frac{2\sqrt{6}}{\pi}(2d+1),\\
&N_h:=\max\Bigg\{\Bigg\lfloor\frac{\log\Big(2h\left(\frac{\pi^2}{6}\right)^{-1/2}M_{UV}^{-1}\Big)}{\log
 M}\Bigg\rfloor+1,1\Bigg\}.
\end{align*}
Here $\lfloor x \rfloor$ denotes the largest integer not exceeding $x$
for $x\in\R$. It follows that
\begin{align*}
\phi(M_{UV}^{-2}M^{-2N_h}h^2|1-e^{i\frac{\o}{h}}|^2)=1,\quad(\forall
 \o\in\R).
\end{align*}
We define the cut-off function $\chi_{h,l}:\R\to \R_{\ge 0}$
$(l=0,1,\cdots,N_h)$ by 
\begin{align*}
&\chi_{h,0}(\o):=\phi(M_{UV}^{-2}h^2|1-e^{i\frac{\o}{h}}|^2),\\
&\chi_{h,l}(\o):=\phi(M_{UV}^{-2}M^{-2l}h^2|1-e^{i\frac{\o}{h}}|^2)
               -\phi(M_{UV}^{-2}M^{-2(l-1)}h^2|1-e^{i\frac{\o}{h}}|^2),\\
&(\o\in\R,l\in\{1,2,\cdots,N_h\}).
\end{align*}
These functions have the properties described in \cite[\mbox{(6.3),\
(6.4)}]{K15}. Using these functions, we define the covariances with the Matsubara UV
cut-off $C_l^+$, $C_l^-:I_0^2\to \C$ $(l=0,1,\cdots,N_h)$ as follows.
\begin{align*}
&C_l^+(\rho\bx\s x,\eta\by \tau y)\\
&:=\frac{\delta_{\s,\tau}}{\beta
 L^d}\sum_{(\o,\bk)\in\cM_h\times\G(L)^*}e^{i\<\bx-\by,\bk\>+i(x-y)\o}\chi_{h,l}(\o)h^{-1}(I_{2^d}-e^{-i\frac{\o}{h}I_{2^d}+\frac{1}{h}\cE(\bk)})^{-1}(\rho,\eta),\\
&C_l^-(\rho\bx\s x,\eta\by \tau y)\\
&:=\frac{\delta_{\s,\tau}}{\beta
 L^d}\sum_{(\o,\bk)\in\cM_h\times\G(L)^*}e^{i\<\bx-\by,\bk\>+i(x-y)\o}\chi_{h,l}(\o)h^{-1}(e^{i\frac{\o}{h}I_{2^d}-\frac{1}{h}\cE(\bk)}-I_{2^d})^{-1}(\rho,\eta),\\
&((\rho,\bx,\s,x),(\eta,\by,\tau,y)\in I_0).
\end{align*}

Here let us introduce some notations which will be used to study the
decay properties of the covariances in this section and for many other
purposes in the rest of this paper. For any
$(\rho,\bx,\s,x,\theta),(\eta,\by,\tau,y,\xi)\in I$,
$j\in\{0,1,\cdots,d\}$, set
\begin{align*}
&d_j((\rho,\bx,\s,x,\theta),(\eta,\by,\tau,y,\xi))\\
&:=\left\{\begin{array}{ll}\frac{\beta}{2\pi}|e^{i\frac{2\pi}{\beta}x}-
e^{i\frac{2\pi}{\beta}y}| & \text{if }j=0,\\
\frac{L}{2\pi}|e^{i\frac{2\pi}{L}\<\bx,\be_j\>}-
e^{i\frac{2\pi}{L}\<\by,\be_j\>}| & \text{if }j\in\{1,2,\cdots,d\}.
\end{array}
\right.
\end{align*}
For any $x\in(1/h)\Z$ let $r_{\beta}(x)\in[0,\beta)_h$,
 $n_{\beta}(x)\in\Z$ be such that $x=n_{\beta}(x)\beta+r_{\beta}(x)$.
This defines the maps $r_{\beta}:(1/h)\Z\to [0,\beta)_h$,
 $n_{\beta}:(1/h)\Z\to \Z$. 
We will assume that
\begin{align}
\beta_1,\beta_2\in\N,\quad \beta_1\le\beta_2,\quad h\in
 4\N,\label{eq_basic_beta_h_assumption}
\end{align}
when we need to estimate differences between anti-symmetric functions
defined at 2 different temperatures. Here $\beta_1$, $\beta_2$ are meant
to be the 2 different inverse-temperatures. Though the inverse temperature originally
belongs to $\R_{>0}$, we will later see that the convergence property of
the free energy density as $\beta\to \infty$ $(\beta\in\R_{>0})$ can be deduced
from the convergent property as $\beta\to\infty$ $(\beta\in\N)$. On the
assumption \eqref{eq_basic_beta_h_assumption}, set
$$
\left[-\frac{\beta_1}{4},\frac{\beta_1}{4}\right)_h:=\left\{-\frac{\beta_1}{4},-\frac{\beta_1}{4}+\frac{1}{h},-\frac{\beta_1}{4}+\frac{2}{h},\cdots,\frac{\beta_1}{4}-\frac{1}{h}\right\}.
$$ 
Note that $0\in[-\beta_1/4,\beta_1/4)_h$. We define the index sets
$\hat{I}_0$, $\hat{I}$, $I_0^0$, $I^0$ by
\begin{align*}
&\hat{I}_0:=\cB\times\G(L)\times\spin\times\left[-\frac{\beta_1}{4},\frac{\beta_1}{4}\right)_h,\quad
 \hat{I}:=\hat{I}_0\times\{1,-1\},\\
&I_0^0:=\cB\times\G(L)\times\spin\times\{0\},\quad
 I^0:=I_0^0\times\{1,-1\}.
\end{align*}
For any $(\rho,\bx,\s,x,\theta)$, $(\eta,\by,\tau,y,\xi)\in\hat{I}$,
$j\in\{0,1,\cdots,d\}$, set 
\begin{align*}
&\hat{d}_j((\rho,\bx,\s,x,\theta),(\eta,\by,\tau,y,\xi))\\
&:=\left\{\begin{array}{ll} |x-y| & \text{if }j=0,\\
\frac{L}{2\pi}|e^{i\frac{2\pi}{L}\<\bx,\be_j\>}-
e^{i\frac{2\pi}{L}\<\by,\be_j\>}| & \text{if }j\in\{1,2,\cdots,d\}.
\end{array}
\right.
\end{align*}
In fact these notations were used in \cite{K15}.
We add the notation $(\beta)$ to the right side of a
temperature-dependent object when we want to show its temperature
dependency explicitly. For example we sometimes write $I_0(\beta)$ instead
 of $I_0$ and $C_l^+(\beta):I_0(\beta)^2\to\C$ instead of 
$C_l^+:I_0^2\to\C$. 

\begin{lemma}\label{lem_UV_covariance_properties}
Assume that $h\ge e^{4d}$. There exists a constant $c_0\in\R_{\ge 1}$,
 which depends only on $d$, $M$, and a constant $c_w\in(0,1]$
 independent of any parameter such that the following statements hold
 for any $\delta\in\{+,-\}$, $l\in\{1,2,\cdots,N_h\}$.
\begin{enumerate} 
\item\label{item_UV_covariance_determinant}
\begin{align}
&|\det(\<\bp_i,\bq_j\>_{\C^m}C_{l}^{\delta}(X_i,Y_j))_{1\le i,j\le n}|\le
 c_0^n,\label{eq_UV_covariance_determinant}\\
&(\forall m,n\in\N,\bp_i,\bq_i\in\C^m\text{ with }
\|\bp_i\|_{\C^m},\|\bq_i\|_{\C^m}\le 1,\notag\\
&\quad X_i,Y_i\in I_0\
 (i=1,2,\cdots,n)).\notag
\end{align}
\item\label{item_UV_covariance_determinant_equal_time}
\begin{align}
&|\det(\<\bp_i,\bq_j\>_{\C^m}C_{l}^{\delta}((X_i,s),(Y_j,s)))_{1\le i,j\le n}|\le (M^{-l}+M^{l-N_h})c_0^n,\label{eq_UV_covariance_determinant_equal_time}\\
&(\forall m,n\in\N,\bp_i,\bq_i\in\C^m\text{ with }
\|\bp_i\|_{\C^m},\|\bq_i\|_{\C^m}\le 1,\notag\\
&\quad X_i,Y_i\in \cB\times\G(L)\times\spin\
 (i=1,2,\cdots,n),s\in[0,\beta)_h).\notag
\end{align}
\item\label{item_UV_covariance_decay}
\begin{align}
&\sup_{j'\in\{0,1,\cdots,d\}}\sup_{X\in I}\frac{1}{h}\sum_{Y\in
 I}(d_{j'}(X,Y)+1)e^{\sum_{j=0}^d(c_w(d+1)^{-2}M^{-2}d_j(X,Y))^{1/2}}|\widetilde{C_l^{\delta}}(X,Y)|\label{eq_UV_covariance_decay}\\
&\le
M^{-l} c_0.\notag
\end{align}
\item\label{item_UV_covariance_determinant_difference}
On the assumption \eqref{eq_basic_beta_h_assumption},
\begin{align}
&|\det(\<\bp_i,\bq_j\>_{\C^m}C_{l}^{\delta}(\beta_1)(\rho_i\bx_i\s_ir_{\beta_1}(x_i),\eta_j\by_j\tau_jr_{\beta_1}(y_j)))_{1\le
 i,j\le n}\label{eq_UV_covariance_determinant_difference}\\
&-
\det(\<\bp_i,\bq_j\>_{\C^m}C_{l}^{\delta}(\beta_2)(\rho_i\bx_i\s_ir_{\beta_2}(x_i),\eta_j\by_j\tau_jr_{\beta_2}(y_j)))_{1\le
 i,j\le n}|\notag\\
&\le\beta_1^{-\frac{1}{2}}M^{-\frac{l}{2}}c_0^n,\notag\\
&(\forall m,n\in\N,\bp_i,\bq_i\in\C^m\text{ with }
\|\bp_i\|_{\C^m},\|\bq_i\|_{\C^m}\le 1,\notag\\
&\quad (\rho_i,\bx_i,\s_i,x_i),(\eta_i,\by_i,\tau_i,y_i)\in \hat{I}_0\
 (i=1,2,\cdots,n)).\notag
\end{align}
\item\label{item_UV_covariance_decay_difference}
On the assumption \eqref{eq_basic_beta_h_assumption},
\begin{align}
&\sup_{X\in I^0}\frac{1}{h}\sum_{(\eta,\by,\tau,y,\xi)\in
 \hat{I}}e^{\sum_{j=0}^d(\frac{1}{\pi}c_w(d+1)^{-2}M^{-2}\hat{d}_j(X,(\eta,\by,\tau,y,\xi)))^{1/2}}\label{eq_UV_covariance_decay_difference}\\
&\qquad\cdot|\widetilde{C_l^{\delta}}(\beta_1)(X,\eta\by\tau r_{\beta_1}(y)\xi)
-\widetilde{C_l^{\delta}}(\beta_2)(X,\eta\by\tau r_{\beta_2}(y)\xi)|\notag\\
&\le
 \beta_{1}^{-\frac{1}{2}}M^{-l}c_0.\notag
\end{align}
\end{enumerate}
In \eqref{item_UV_covariance_decay},
 \eqref{item_UV_covariance_decay_difference},
 $\widetilde{C_l^{\delta}}:I^2\to\C$ denotes the anti-symmetric
 extension of $C_l^{\delta}$ defined by
\begin{align}
&\widetilde{C_l^{\delta}}((X,\theta),(Y,\xi)):=\frac{1}{2}(1_{(\theta,\xi)=(1,-1)}C_l^{\delta}(X,Y)-1_{(\theta,\xi)=(-1,1)}C_l^{\delta}(Y,X)),\label{eq_covariance_anti_symmetrization}\\
&(\forall X,Y\in I_0,\theta,\xi\in\{1,-1\}).\notag
\end{align}
\end{lemma}

\begin{remark}
There are unfortunately insufficiencies in the estimation of the
 difference between the determinants defined at $\beta_1$, $\beta_2$ in the
 proofs of \cite[\mbox{Lemma 6.3,\ Lemma 7.14}]{K15}, though the results
 themselves hold true. Here we prove \eqref{eq_UV_covariance_determinant_difference} in a way that it recovers the
 insufficient parts of the proofs of the related inequalities in 
\cite[\mbox{Lemma 6.3,\ Lemma 7.14}]{K15}. 
\end{remark}

\begin{proof}[Proof of Lemma \ref{lem_UV_covariance_properties}]
 First of all, let us note that the condition ``$h\ge
 e^{2E_1}$'' required in \cite[\mbox{Lemma 6.2,\ Lemma 6.3}]{K15} is
 equal to $h\ge e^{4d}$ in this case because of
 \eqref{eq_dispersion_derivative_bound}. Thus, we can refer to these
 lemmas in the following.

\eqref{item_UV_covariance_determinant}: This was proved in
 \cite[\mbox{Lemma 6.2}]{K15}.

\eqref{item_UV_covariance_determinant_equal_time}: Let us confirm that
 there exists a constant $c(d,M)\in\R_{>0}$ depending only on $d$ and
 $M$ such that 
\begin{align}
&|C_l^{\delta}(\rho\bx\s s,\eta\by\tau s)|\le c(d,M)(M^{l-N_h}+M^{-l}),\label{eq_UV_equal_time_decay}\\
&(\forall (\rho,\bx,\s),(\eta,\by,\tau)\in\cB\times \G(L)\times \spin,s\in[0,\beta)_h).\notag\end{align}
 By periodicity, for any $j\in\{1,2,\cdots,d\}$,
\begin{align*}
&\frac{L}{2\pi}(e^{-i\frac{2\pi}{L}\<\bx-\by,\be_j\>}-1)C_l^+(\cdot\bx\s
 x,\cdot\by\tau y)\\
&=\frac{\delta_{\s,\tau}}{\beta L^d}\sum_{(\o,\bk)\in\cM_h\times\G(L)^*}
e^{i\<\bx-\by,\bk\>+i(x-y)\o}\chi_{h,l}(\o)\frac{L}{2\pi}\int_{0}^{2\pi/L}dp\\
&\quad \cdot \frac{\partial}{\partial
 p}h^{-1}(I_{2^d}-e^{-i\frac{\o}{h}I_{2^d}+\frac{1}{h}\cE(\bk+p\be_j)})^{-1}\\
&=\frac{\delta_{\s,\tau}}{\beta L^d}\sum_{(\o,\bk)\in\cM_h\times\G(L)^*}
e^{i\<\bx-\by,\bk\>+i(x-y)\o}\chi_{h,l}(\o)\frac{L}{2\pi}\int_{0}^{2\pi/L}dp\\
&\quad\cdot
 h^{-1}(I_{2^d}-e^{-i\frac{\o}{h}I_{2^d}+\frac{1}{h}\cE(\bk+p\be_j)})^{-1}
\left(\frac{\partial}{\partial p}
 h(e^{-i\frac{\o}{h}I_{2^d}+\frac{1}{h}\cE(\bk+p\be_j)}-I_{2^d})\right)\\
&\quad\cdot  h^{-1}(I_{2^d}-e^{-i\frac{\o}{h}I_{2^d}+\frac{1}{h}\cE(\bk+p\be_j)})^{-1}.
\end{align*}
Using the inequalities \cite[\mbox{(6.7),\ (6.10),\ (6.14)}]{K15}, we can
 derive from the above equality that
\begin{align*}
\left\|\frac{L}{2\pi}(e^{-i\frac{2\pi}{L}\<\bx-\by,\be_j\>}-1)C_l^+(\cdot\bx\s
 x,\cdot\by\tau y)\right\|_{2^d\times 2^d}\le c(d,M)M^{-l}.
\end{align*}
This inequality implies that
\begin{align}\label{eq_UV_difference_space_decay}
&|C_l^+(\rho\bx\s x,\eta \by \tau y)|\le c(d,M)M^{-l},\\
&(\forall
 (\rho,\bx,\s,x),(\eta,\by,\tau,y)\in I_0\text{ with }\bx\neq \by).\notag
\end{align}
In the final part of the proof of \cite[\mbox{Lemma 6.2}]{K15} we proved
 that
\begin{align}\label{eq_UV_tadpole_decay}
\|C_l^+(\cdot\b0\s 0,\cdot \b0\s 0)\|_{2^d\times 2^d}\le
 c(M)(M^{l-N_h}+M^{-l})
\end{align}
with a constant $c(M)\in\R_{>0}$ depending only on $M$.
The inequalities \eqref{eq_UV_difference_space_decay},
 \eqref{eq_UV_tadpole_decay} imply \eqref{eq_UV_equal_time_decay} for
 $\delta = +$. The proof for $\delta = -$ is parallel. The determinant
 bound \eqref{eq_UV_covariance_determinant_equal_time} can be
 obtained by combining the determinant bound
 \eqref{eq_UV_covariance_determinant} with
 \eqref{eq_UV_equal_time_decay}.

\eqref{item_UV_covariance_decay},\eqref{item_UV_covariance_decay_difference}: 
These were essentially proved in \cite[\mbox{Lemma 6.2, Lemma 6.3}]{K15}. Recall that the weight
 ``$\fw(0)$'' was given by
 $$
c_w(d+1)^{-2}\min\{M_{UV},(E_2+1)^{-1}\}M^{-2}
$$
with a constant $c_w\in(0,1]$ independent of any parameter
in\\
 \cite[\mbox{Lemma 6.2}]{K15}. Since $E_2=1$ in the present case, 
$\min\{M_{UV},(E_2+1)^{-1}\}=1/2$. We can replace $(1/2)c_w$ in the
 weight ``$\fw(0)$'' in \\
\cite[\mbox{Lemma 6.2, Lemma 6.3}]{K15}
by $c_w$ to obtain the
 weight $c_w(d+1)^{-2}M^{-2}$ with some $c_w\in (0,1]$ and thus
 \eqref{eq_UV_covariance_decay} and
 \eqref{eq_UV_covariance_decay_difference} follow.

\eqref{item_UV_covariance_determinant_difference}:
The inequality \cite[\mbox{(6.27)}]{K15} implies that 
\begin{align}
&|C_l^+(\beta_1)(\rho\bx\s r_{\beta_1}(x),\eta\by\tau r_{\beta_1}(y))
-C_l^+(\beta_2)(\rho\bx\s r_{\beta_2}(x),\eta\by\tau
 r_{\beta_2}(y))|\label{eq_one_component_determinant_UV}\\
&\le c(d,M)\beta_1^{-\frac{1}{2}}M^{-\frac{l}{2}},\quad (\forall
 (\rho,\bx,\s,x),(\eta,\by,\tau,y)\in \hat{I}_0).\notag
\end{align}
Take any $(\rho_i,\bx_i,\s_i,x_i),(\eta_i,\by_i,\tau_i,y_i)\in\hat{I}_0$
 and $\bp_i,\bq_i\in\C^m$ satisfying
 $\|\bp_i\|_{\C^m},\|\bq_i\|_{\C^m}\le 1$ $(i=1,2,\cdots,n)$. Define
 $C_1,C_2\in \Mat(n,\C)$ by
\begin{align*}
C_a:=(\<\bp_i,\bq_j\>_{\C^m}C_l^+(\beta_a)(\rho_i\bx_i\s_ir_{\beta_a}(x_i),
\eta_j\by_j\tau_jr_{\beta_a}(y_j)))_{1\le i,j\le n},\quad (a=1,2).
\end{align*}
Since 
$$
C_1-C_2=\left(\begin{array}{cc}C_1 &
	      I_n\end{array}\right)\left(\begin{array}{c} I_n\\
					 -C_2\end{array}\right),
$$
the Cauchy-Binet formula yields that 
\begin{align*}
\det(C_1-C_2)=\sum_{\gamma:\{1,2,\cdots,n\}\to \{1,2,\cdots,2n\}\atop
 \text{with }\gamma(1)<\gamma(2)<\cdots<\gamma(n)}&\det\left(\left(\begin{array}{cc}C_1 &
	      I_n\end{array}\right)(i,\gamma(j))\right)_{1\le i,j\le n}\\
&\cdot\det\left(\left(\begin{array}{c} I_n\\
					 -C_2\end{array}\right)(\gamma(i),j)\right)_{1\le i,j\le n}.
\end{align*}
By using \eqref{eq_UV_covariance_determinant} and assuming that
 $\gamma(0)=n$, $\gamma(n+1)=n+1$ we see that 
\begin{align}
|\det(C_1-C_2)|
&\le \sum_{m=0}^n
\sum_{\gamma:\{1,2,\cdots,n\}\to \{1,2,\cdots,2n\}\atop
 \text{with }\gamma(1)<\gamma(2)<\cdots<\gamma(n)}1_{\gamma(m)\le
 n<\gamma(m+1)}c_0^mc_0^{n-m}\label{eq_UV_determinant_difference_preliminary}\\
&=\sum_{m=0}^n\left(\begin{array}{c} n \\ m\end{array}\right)
              \left(\begin{array}{c} n \\ n-m\end{array}\right)c_0^n\le
 2^{2n}c_0^n.\notag
\end{align}
By expanding along the 1st column and using
 \eqref{eq_one_component_determinant_UV},
 \eqref{eq_UV_determinant_difference_preliminary} we have
\begin{align}
|\det(C_1-C_2)|&\le
 c(d,M)\beta_1^{-\frac{1}{2}}M^{-\frac{l}{2}}\sum_{s=1}^n\left|\det((C_1-C_2)(i,j))_{1\le
 i,j\le n\atop i\neq s,j\neq
 1}\right|\label{eq_UV_determinant_difference_next}\\
&\le
 c(d,M)\beta_1^{-\frac{1}{2}}M^{-\frac{l}{2}}n2^{2(n-1)}c_0^{n-1}.\notag
\end{align}
By applying the Cauchy-Binet formula once more and substituting
 \eqref{eq_UV_covariance_determinant},
 \eqref{eq_UV_determinant_difference_next},
\begin{align*}
&|\det C_1- \det C_2|\\
&=\Bigg| \sum_{\gamma:\{1,2,\cdots,n\}\to \{1,2,\cdots,2n\}\atop
 \text{with }\gamma(1)<\gamma(2)<\cdots<\gamma(n),\ r(1)\le n}\\
&\qquad\cdot\det\left(\left(\begin{array}{cc}C_1-C_2 &
	      I_n\end{array}\right)(i,\gamma(j))\right)_{1\le i,j\le n}
\det\left(\left(\begin{array}{c} I_n\\
					 C_2\end{array}\right)(\gamma(i),j)\right)_{1\le i,j\le n}\Bigg|\\
&\le \sum_{m=1}^n\sum_{\gamma:\{1,2,\cdots,n\}\to \{1,2,\cdots,2n\}\atop
 \text{with }\gamma(1)<\gamma(2)<\cdots<\gamma(n)}1_{\gamma(m)\le
 n<\gamma(m+1)}c(d,M)\beta^{-\frac{1}{2}}M^{-\frac{l}{2}}m2^{2(m-1)}c_0^{m-1}c_0^{n-m}\\
&\le \beta^{-\frac{1}{2}}M^{-\frac{l}{2}}(c(d,M)c_0)^n.
\end{align*}
Thus, we obtained the determinant bound of the form
 \eqref{eq_UV_covariance_determinant_difference} for $\delta=+$.
The bound for $\delta=-$ can be proved in the same way. 
\end{proof}

\subsection{Isothermal bounds}\label{subsec_isothermal}
Our multi-scale analysis at fixed temperature is built on estimation of
kernels of Grassmann polynomials with respect to scale-dependent
(semi-)norms. Let us define the (semi-)norms at this point. Set 
$$\fw(0):=c_w(d+1)^{-2}M^{-2}$$ 
with the constant $c_w\in (0,1]$ appearing in Lemma
\ref{lem_UV_covariance_properties}. For $l\in\Z_{\le 0}$, set
$\fw(l):=\fw(0)M^l$. For an anti-symmetric function $f$ on $I^m$ $(m\ge
2)$ we define $\|f\|_{l,0}$, $\|f\|_{l,1}$ by 
\begin{align}\label{eq_definition_norm_semi_norm}
&\|f\|_{l,0}:=\sup_{X\in I}\left(\frac{1}{h}\right)^{m-1}\sum_{\bY=(Y_1,Y_2,\cdots,Y_{m-1})\in
 I^{m-1}}e^{\sum_{j=0}^d(\fw(l)d_j(X,Y_1))^{1/2}}|f(X,\bY)|,\\
&\|f\|_{l,1}\notag\\
&:=\sup_{j'\in\{0,1,\cdots,d\}}\sup_{q\in\{1,2,\cdots,m-1\}}\sup_{X\in
 I}\notag\\
&\quad\cdot\left(\frac{1}{h}\right)^{m-1}\sum_{\bY=(Y_1,Y_2,\cdots,Y_{m-1})\in
 I^{m-1}}d_{j'}(X,Y_q)e^{\sum_{j=0}^d(\fw(l)d_j(X,Y_1))^{1/2}}|f(X,\bY)|.\notag
\end{align}
In our Matsubara UV integration, anti-symmetric kernels are measured by
$\|\cdot\|_{0,t}$ $(t=0,1)$. The measurement with $\|\cdot\|_{l,t}$
$(l<0,\ t=0,1)$ will be necessary in the infrared integration in
Section \ref{sec_IR}. From now we assume that
$$h\ge e^{4d}$$
so that the results of Lemma \ref{lem_UV_covariance_properties} are
available.
The inequality \eqref{eq_UV_covariance_decay} implies that
\begin{align}
\|\widetilde{C_l^{\delta}}\|_{0,t}\le
 c_0M^{-l},\quad(\forall l\in\{1,2,\cdots,N_h\},\delta\in\{+,-\},
 t\in\{0,1\}).\label{eq_UV_covariance_decay_norm}
\end{align}

Fix $\delta \in\{+,-\}$ and set
\begin{align*}
&F^{N_h}(\psi):=-\sum_{m=1}^{N_v}\frac{1}{h}\sum_{s\in[0,\beta)_h}\sum_{(\rho_j,\bx_j,\s_j),(\eta_j,\by_j,\tau_j)\in\cB\times
 \G(L)\times \spin\atop
 (j=1,2,\cdots,m)}(1_{\delta=+}+1_{\delta=-}(-1)^m)\\
&\cdot V_m^L(\bU)((\nu(\rho_1,\bx_1)\s_1,\cdots,\nu(\rho_m,\bx_m)\s_m),(\nu(\eta_1,\by_1)\tau_1,\cdots,\nu(\eta_m,\by_m)\tau_m))\\
&\qquad\qquad\qquad \cdot \opsi_{\rho_1\bx_1\s_1 s}\cdots\opsi_{\rho_m\bx_m\s_ms}
\psi_{\eta_1\by_1\tau_1s}\cdots\psi_{\eta_m\by_m\tau_m s},\\
&T^{N_h}(\psi):=0,\\
&J^{N_h}(\psi):=F^{N_h}(\psi)
\end{align*}
with $\bU\in \C^{n_v}$.
We input $J^{N_v}(\psi)$ into the Matsubara UV integration process as
the initial data. We define $F^l(\psi)$, $T^l(\psi)$, $J^l(\psi)\in\bigwedge \cV$
$(l=0,1,\cdots,N_h-1)$ inductively as follows. Assume that we have
$J^{l+1}(\psi)\in\bigwedge \cV$ for some $l\in\{0,1,\cdots,N_h-1\}$. Set
\begin{align*}
&F^l(\psi):=\int J^{l+1}(\psi+\psi^1)d\mu_{C^{\delta}_{l+1}}(\psi^1),\\
&T^{l,(n)}(\psi):=\frac{1}{n!}\left.\left(\frac{\partial}{\partial
 z}\right)^n\right|_{z=0}\log\left(\int e^{z
 J^{l+1}(\psi+\psi^1)}d\mu_{C_{l+1}^{\delta}}(\psi^1)\right)
\end{align*}
for $n\in \N_{\ge 2}$. Then, set 
\begin{align*}
T^l(\psi):=\sum_{n=2}^{\infty}T^{l,(n)}(\psi),\quad
 J^l(\psi):=F^l(\psi)+T^l(\psi)
\end{align*}
on the assumption that $\sum_{n=2}^{\infty}T^{l,(n)}(\psi)$ converges. 
See  \cite[\mbox{Subsection 2.2}]{K15} for the notion of convergence and
differentiation of Grassmann polynomials. Note that
$F^{N_h}(\psi)=-V^+(\psi)+\beta V_0^L$ if $\delta =+$, 
$F^{N_h}(\psi)=-V^-(\psi)+\beta V_0^L$ if $\delta =-$.
Also, an inductive argument based on \cite[\mbox{Lemma 3.9 (1)}]{K15},
 parallel to the proof of \cite[\mbox{Lemma 5.1}]{K15} ensures that 
if $\sum_{n=2}^{\infty}T^{l,(n)}(\psi)$ converges for any $l\in
\{0,1,\cdots,N_h-1\}$,
\begin{align*}
&T_m^{l,(n)}(\psi)=F_m^l(\psi)=0,\\
&(\forall l\in\{0,1,\cdots,N_h\},m\in \{0,1,\cdots,N\}\cap
 (2\N+1),n\in\N_{\ge 2}).
\end{align*}

\begin{lemma}\label{lem_UV_initial_bound}
\begin{align*}
&\|F_{2m}^{N_h}\|_{0,t}\le
 \max_{k\in\{1,2,\cdots,n_v\}}|U_k|e^{d\fw(0)^{1/2}}v_m(\fw(0)),\\
&(\forall m\in\{1,2,\cdots,N_v\},t\in\{0,1\}).
\end{align*}
\end{lemma}
\begin{proof}
 By the uniqueness of the anti-symmetric kernel we
 have that for any $(\rho_j,\bx_j,\s_j,s_j,\theta_j)\in I$
 $(j=1,2,\cdots,2m)$,
\begin{align*}
&F_{2m}^{N_h}(\rho_1\bx_1\s_1 s_1\theta_1,\cdots,\rho_{2m}\bx_{2m}\s_{2m}
 s_{2m}\theta_{2m})\\
&=\frac{-1}{(2m)!}(1_{\delta=+}+1_{\delta=-}(-1)^m)\\
&\quad\cdot
\sum_{\xi\in\S_{2m}}\sgn(\xi)V_m^L(\nu(\rho_{\xi(1)},\bx_{\xi(1)})\s_{\xi(1)},\cdots,\nu(\rho_{\xi(2m)},\bx_{\xi(2m)})\s_{\xi(2m)})\\
&\quad\cdot
 h^{2m-1}1_{s_1=\cdots=s_{2m}}1_{(\theta_{\xi(1)},\cdots,\theta_{\xi(2m)})=(1,\cdots,1,-1,\cdots,-1)}.
\end{align*}
If $\hat{\bx},\hat{\by}\in\G(2L)$,
 $(\rho,\bx),(\eta,\by)\in\cB\times\G(L)$ satisfy
 $\hat{\bx}=\nu(\rho,\bx)$, $\hat{\by}=\nu(\eta,\by)$, then
\begin{align*}
&\frac{L}{2\pi}|e^{i\frac{2\pi}{L}\<\bx-\by,\be_j\>}-1|\le 
\frac{L}{\pi}|e^{i\frac{\pi}{L}\<\hat{\bx}-\hat{\by},\be_j\>}-1|+1,\quad
 (\forall j\in \{1,2,\cdots,d\}).
\end{align*}
Using this inequality and the invariances \eqref{eq_bi_anti_symmetric}, \eqref{eq_hermiticity}, we observe that for $t\in\{0,1\}$,
\begin{align*}
&\|F_{2m}^{N_h}\|_{0,t}\\
&\le \max_{k\in
 \{1,2,\cdots,n_v\}}|U_k|\sup_{\bU\in\overline{D(1)}^{n_v}}\sup_{p,q\in\{1,2,\cdots,2m-1\}\atop j'\in \{1,2,\cdots,d\}}\sup_{(\rho,\bx,\s)\in\cB\times \G(L)\times \spin}\\
&\quad\cdot\sum_{(\rho_j,\bx_j,\s_j)\in\cB\times\G(L)\times\spin\atop
 (j=1,2,\cdots,2m-1)}\left(\frac{L}{2\pi}|e^{i\frac{2\pi}{L}\<\bx-\bx_q,\be_{j'}\>}-1|\right)^t
e^{\sum_{j=1}^d(\fw(0)\frac{L}{2\pi}|e^{i\frac{2\pi}{L}\<\bx-\bx_p,\be_j\>}-1|)^{1/2}}\\
&\quad\cdot
 |V_m^L(\bU)((\nu(\rho,\bx)\s,\nu(\rho_1,\bx_1)\s_1,\cdots,\nu(\rho_{m-1},\bx_{m-1})\s_{m-1}), \\
&\qquad\qquad (\nu(\rho_m,\bx_m)\s_m,\nu(\rho_{m+1},\bx_{m+1})\s_{m+1},\cdots, 
\nu(\rho_{2m-1},\bx_{2m-1})\s_{2m-1}))|\\
&\le \max_{k\in\{1,2,\cdots,n_v\}}|U_k|e^{d\fw(0)^{1/2}}v_m(\fw(0)).
\end{align*}
\end{proof}

The main purpose of this subsection is to prove the following lemma. We
will refer to \cite[\mbox{Lemma 3.8}]{K15} as the main tool in the
proof.

\begin{lemma}\label{lem_UV_integration}
Let $\alpha\in\R_{\ge 1}$ and let $c_0$ be the constant appearing in
 Lemma \ref{lem_UV_covariance_properties}. 
There exists a constant $c\in\R_{>0}$
 independent of any parameter such that if 
\begin{align}
M\ge c^{N_v^2},\quad\alpha\ge c
 M^{\frac{1}{2}}\label{eq_UV_parameter_conditions}
\end{align}
and
\begin{align}
\max_{j\in\{1,2,\cdots,n_v\}}|U_j|e^{d\fw(0)^{1/2}}\sum_{m=1}^{N_v}c_0^m\alpha^{2m}v_m(\fw(0))\le
 \frac{1}{2},\label{eq_interaction_amplitude_UV}
\end{align}
the following inequalities hold for any $l\in\{0,1,\cdots,N_h\}$,
 $t\in\{0,1\}$. 
\begin{align}
&\frac{h}{N}(|F_0^l|+|T_0^l|)\le \alpha^{-1},\label{eq_UV_0th_bound}\\
&\sum_{m=1}^{2N_v}c_0^{\frac{m}{2}}\alpha^m(\|F_m^l\|_{0,t}+\|T_m^l\|_{0,t})\le
 1,\label{eq_UV_bound}\\
&M^{-\frac{N_v}{N_v-1}l}\sum_{m=1}^Nc_0^{\frac{m}{2}}\alpha^mM^{\frac{l}{2N_v-2}m}
(\|F_m^l\|_{0,t}+\|T_m^l\|_{0,t})\le 1.\label{eq_UV_all_bound}
\end{align}
Moreover, for any $l\in \{0,1,\cdots,N_h-1\}$, $m\in \{1,2,\cdots,N\}$,
\begin{align}
&\sum_{n=2}^{\infty}\sup_{\bU\in\C^{n_v}\text{ with }\atop|U_j|\le
 U_{max}(\alpha,M)\ (j=1,\cdots,n_v)}|T_0^{l,(n)}(\bU)|<\infty,
\label{eq_UV_tree_part_uniform_convergence_0}\\
&\sum_{n=2}^{\infty}\sup_{\bU\in\C^{n_v}\text{ with }\atop |U_j|\le
 U_{max}(\alpha,M)\ 
 (j=1,\cdots,n_v)}\|T_m^{l,(n)}(\bU)\|_{0,0}<\infty,\label{eq_UV_tree_part_uniform_convergence}
\end{align}
where
\begin{align*}
U_{max}(\alpha,M):=\left(2e^{d\fw(0)^{1/2}}\sum_{m=1}^{N_v}c_0^m\alpha^{2m}v_m(\fw(0))\right)^{-1}.
\end{align*}
\end{lemma}

\begin{remark}\label{rem_uniform_convergence_correction_UV}
We claim \eqref{eq_UV_tree_part_uniform_convergence_0},
 \eqref{eq_UV_tree_part_uniform_convergence} in order to emphasize the
 uniform convergent property of $\sum_{n=2}^{\infty}T^{l,(n)}(\psi)$ with
 respect to the coupling constants. We should have explicitly claimed
 the uniform convergent properties of the infinite series of the Grassmann
 polynomials produced by the tree expansions in \cite[\mbox{Proposition
 5.2, Proposition 5.6, Proposition 6.4}]{K15}, 
though these properties
 are obvious from the proofs. Strictly speaking, the previous deduction
 of the regularity with the coupling constants \cite[\mbox{Proposition
 6.4 (2)}]{K15} from the point-wise convergent properties
 \cite[\mbox{Proposition 6.4 (1)}]{K15} is incomplete. The claim  
\cite[\mbox{Proposition 6.4 (2)}]{K15} is rigorously proved by
 additionally remarking the uniform convergent properties such as  
\eqref{eq_UV_tree_part_uniform_convergence_0},
\eqref{eq_UV_tree_part_uniform_convergence} in \cite[\mbox{Proposition 6.4
 (1)}]{K15}. With the aim of convincing the readers of the validity of
 the construction, in this paper we intend to make clear the deduction 
of the regularity with the coupling constants from the uniform
 convergent properties. The clarification will be specifically made in
 the proof of Lemma \ref{lem_IR_recursive_structure}
 \eqref{item_IR_recursive_structure} and Lemma
 \ref{lem_input_to_IR_integration}.
\end{remark}

\begin{proof}[Proof of Lemma \ref{lem_UV_integration}]
During the proof the symbol $c$ denotes a generic constant independent of any
 parameter. We replace $c$ by a larger generic constant denoted by
 the same symbol from time to time without any comment. However, such replacements
 do not affect the conclusions of the proof. We prove the claimed
 inequalities by induction with $l\in\{0,1,\cdots,N_h\}$. By 
 assumption and Lemma \ref{lem_UV_initial_bound},
\begin{align}
&\frac{h}{N}(|F_0^{N_h}|+|T_0^{N_h}|)=0,\label{eq_initial_UV_bound}\\
&\sum_{m=1}^{2N_v}c_0^{\frac{m}{2}}\alpha^m(\|F_m^{N_h}\|_{0,t}+\|T_m^{N_h}\|_{0,t})\le
 \frac{1}{2},\notag\\
&M^{-\frac{N_v}{N_v-1}N_h}\sum_{m=1}^Nc_0^{\frac{m}{2}}\alpha^mM^{\frac{N_h}{2N_v-2}m}
(\|F_m^{N_h}\|_{0,t}+\|T_m^{N_h}\|_{0,t})\notag\\
&\le
 \sum_{m=1}^{2N_v}c_0^{\frac{m}{2}}\alpha^m\|F_m^{N_h}\|_{0,t}\le
 \frac{1}{2},\quad (\forall t\in \{0,1\}).\notag
\end{align}
Thus, the inequalities \eqref{eq_UV_0th_bound}, \eqref{eq_UV_bound},
 \eqref{eq_UV_all_bound} hold for $l=N_h$.

Assume that $l\in\{0,1,\cdots,N_h-1\}$ and for any
 $j\in\{l+1,l+2,\cdots,N_h\}$, $t\in\{0,1\}$ the inequalities \eqref{eq_UV_bound},
 \eqref{eq_UV_all_bound} hold.

Let us prepare a couple of inequalities. By
 the hypothesis of induction, for any $t\in \{0,1\}$,
\begin{align}
&\sum_{m=2}^N2^{3m}c_0^{\frac{m}{2}}\alpha^m\|J_m^{l+1}\|_{0,t}\label{eq_bound_from_induction}\\
&\le
 c^{N_v}\sum_{m=2}^{2N_v}c_0^{\frac{m}{2}}\alpha^m\|J_m^{l+1}\|_{0,t}+c^{N_v}M^{-\frac{N_v+1}{N_v-1}(l+1)}\sum_{m=2N_v+2}^Nc_0^{\frac{m}{2}}\alpha^mM^{\frac{l+1}{2N_v-2}m}\|J_m^{l+1}\|_{0,t}\notag\\
&\le c^{N_v},\notag\\
&\sum_{m=2}^N2^{2m}c_0^{\frac{m}{2}}\alpha^mM^{\frac{l}{2N_v-2}m}\|J_m^{l+1}\|_{0,t}\label{eq_bound_from_induction_weight}\\
&\le c
 M^{-\frac{1}{N_v-1}}\sum_{m=2}^Nc_0^{\frac{m}{2}}\alpha^mM^{\frac{l+1}{2N_v-2}m}\|J_m^{l+1}\|_{0,t}\le c M^{-\frac{1}{N_v-1}+\frac{N_v}{N_v-1}(l+1)},\notag
\end{align}
where we especially used the condition that $M\ge c^{N_v}$. 

By combining
 \eqref{eq_UV_covariance_determinant},
 \eqref{eq_UV_covariance_decay_norm}, \eqref{eq_bound_from_induction}
 with \cite[\mbox{Lemma 3.8 (1)}]{K15} we obtain that for any $n\in \N_{\ge
 2}$, 
\begin{align*}
|T_0^{l,(n)}|&\le
 \frac{N}{h}c_0^{-n+1}(c_0M^{-l-1})^{n-1}\left(\sum_{m=2}^N2^{2m}c_0^{\frac{m}{2}}\|J_m^{l+1}\|_{0,0}\right)^n\\
&\le \frac{N}{h}M^{l+1}(c^{N_v}M^{-l-1}\alpha^{-2})^n.
\end{align*}
Thus, on the assumption $M\ge c^{N_v}$,
 \eqref{eq_UV_tree_part_uniform_convergence_0} holds and
\begin{align}
\frac{h}{N}|T_0^l|\le c^{N_v}M^{-l-1}\alpha^{-4}.
\label{eq_tree_0th}
\end{align}
By \eqref{eq_UV_covariance_determinant},
 \eqref{eq_UV_covariance_decay_norm} and \cite[\mbox{Lemma 3.8
 (2)}]{K15}, for any $m\in\{2,3,\cdots,N\}$, $t\in\{0,1\}$, $n\in\N_{\ge
 2}$,
\begin{align}
&\|T_m^{l,(n)}\|_{0,t}\le
 2^{-2m}c_0^{-\frac{m}{2}-n+1}\prod_{i=1}^n\left(\sum_{q_i=0}^1\right)\prod_{j=2}^n\left(\sum_{r_j=0}^1\right)1_{\sum_{i=1}^nq_i+\sum_{j=2}^nr_j=t}\label{eq_tree_higher_preliminary}\\
&\qquad\qquad\quad\cdot (c_0M^{-l-1})^{n-1}\prod_{k=1}^n\left(\sum_{m_k=2}^N2^{3m_k}c_0^{\frac{m_k}{2}}\|J_{m_k}^{l+1}\|_{0,q_k}\right)1_{\sum_{j=1}^nm_j-2n+2\ge m}.\notag
\end{align}
Moreover, by substituting \eqref{eq_bound_from_induction}, 
\begin{align*}
&\sum_{m=2}^Nc_0^{\frac{m}{2}}\alpha^m\|T_m^{l,(n)}\|_{0,t}\\
&\le
 2^{4n-4}M^{-(l+1)(n-1)}\alpha^{-2n+2}\prod_{i=1}^n\left(\sum_{q_i=0}^1\right)\prod_{j=2}^n\left(\sum_{r_j=0}^1\right)1_{\sum_{i=1}^nq_i+\sum_{j=2}^nr_j=t}\\
&\quad\cdot 
\prod_{k=1}^n\left(\sum_{m_k=2}^N2^{m_k}c_0^{\frac{m_k}{2}}\alpha^{m_k}\|J_{m_k}^{l+1}\|_{0,q_k}\right)\\
&\le M^{l+1}\alpha^2(c^{N_v}M^{-l-1}\alpha^{-2})^n,
\end{align*}
which implies on the assumption $M\ge c^{N_v}$ that
 \eqref{eq_UV_tree_part_uniform_convergence} holds and  
\begin{align}
\sum_{m=2}^Nc_0^{\frac{m}{2}}\alpha^m\|T_m^l\|_{0,t}\le
 c^{N_v}M^{-l-1}\alpha^{-2},\quad (\forall
 t\in\{0,1\}).\label{eq_tree_sum}
\end{align}
Also, by \eqref{eq_tree_higher_preliminary} and
 \eqref{eq_bound_from_induction_weight},
\begin{align*}
&M^{-\frac{N_v}{N_v-1}l}
\sum_{m=2}^Nc_0^{\frac{m}{2}}\alpha^mM^{\frac{l}{2N_v-2}m}\|T_m^{l,(n)}\|_{0,t}\\
&\le c^n M^{-\frac{N_v}{N_v-1}l-(l+1+\frac{l}{N_v-1})(n-1)}\alpha^{-2n+2}\prod_{i=1}^n\left(\sum_{q_i=0}^1\right)\prod_{j=2}^n\left(\sum_{r_j=0}^1\right)1_{\sum_{i=1}^nq_i+\sum_{j=2}^nr_j=t}\\
&\quad\cdot 
\prod_{k=1}^n\left(\sum_{m_k=2}^N2^{m_k}c_0^{\frac{m_k}{2}}\alpha^{m_k}M^{\frac{l}{2N_v-2}m_k}
\|J_{m_k}^{l+1}\|_{0,q_k}\right)\\
&\le M^{1-(l+1+\frac{l}{N_v-1})n}\alpha^{-2n+2}(c
 M^{-\frac{1}{N_v-1}+\frac{N_v}{N_v-1}(l+1)})^n\\
&=M\alpha^2(c\alpha^{-2})^n.
\end{align*}
Thus, on the assumption $\alpha\ge c$, 
\begin{align}
M^{-\frac{N_v}{N_v-1}l}
\sum_{m=2}^Nc_0^{\frac{m}{2}}\alpha^mM^{\frac{l}{2N_v-2}m}\|T_m^{l}\|_{0,t}\le
 c M\alpha^{-2},\quad (\forall t\in \{0,1\}).\label{eq_tree_sum_weight}
\end{align}

To establish upper bounds on the free part $F^l(\psi)$, we introduce the
 Grassmann polynomials $\hat{F}^j(\psi)$ $(j=l,l+1,\cdots,N_h)$
 inductively as follows.
Set $\hat{F}^{N_h}(\psi):=0$. Assume that $l'\in\{l,l+1,\cdots,N_h-1\}$
 and we have $\hat{F}^j(\psi)$ $(j=l'+1,l'+2,\cdots,N_h)$. For any $m\in
 \{0,2N_v+1,2N_v+2,\cdots,N\}$, $\hat{F}_m^{l'}(\psi):=0$. 
For any $m\in\{1,2,\cdots,2N_v\}$,
\begin{align}
\hat{F}_m^{l'}(\psi)
:=&F_m^{l'}(\psi)-F_m^{N_h}(\psi)\label{eq_artificial_free_definition}\\
&-\sum_{j=l'+1}^{N_h}\cP_m\sum_{n=m+2}^{2N_v}
\int(F_n^j(\psi+\psi^1)-\hat{F}_n^j(\psi+\psi^1))d\mu_{C_j^{\delta}}(\psi^1),\notag
\end{align}
where $\cP_m:\bigwedge \cV\to \bigwedge^m\cV$ is the standard
 projection. It follows that for any $l'\in\{l,l+1,\cdots,N_h\}$,
 $m\in\{1,2,\cdots,2N_v\}$, $(\rho_j,\bx_j,\s_j,s_j,\theta_j)\in I$
 $(j=1,2,\cdots,m)$,
\begin{align}
1_{\exists j\exists k\in \{1,2,\cdots,m\}(j\neq k\wedge s_j\neq
 s_k)}
 (&F_m^{l'}(\rho_1\bx_1\s_1s_1\theta_1,\cdots,\rho_m\bx_m\s_ms_m\theta_m)\label{eq_artificial_free_part}\\
& -\hat{F}_m^{l'}(\rho_1\bx_1\s_1s_1\theta_1,\cdots,\rho_m\bx_m\s_ms_m\theta_m))=0.\notag
\end{align}
In fact, the equality \eqref{eq_artificial_free_part} is true for
 $l'=N_h$ by definition. Assume that it holds true for any $j\in
 \{l'+1,l'+2,\cdots,N_h\}$. Since 
\begin{align}
&F_m^{l'}(\bX)-\hat{F}_m^{l'}(\bX)
=F_m^{N_h}(\bX)+\sum_{j=l'+1}^{N_h}\sum_{n=m+2}^{2N_v}\left(\begin{array}{c}n
							\\
							     m\end{array}\right)
\left(\frac{1}{h}\right)^{n-m}\sum_{\bY\in
 I^{n-m}}\label{eq_artificial_free_kernel}\\
&\quad\qquad\qquad\qquad\qquad\qquad\qquad\cdot(F_n^j(\bX,\bY)-\hat{F}^j_n(\bX,\bY))\int\psi_{\bY}^1d\mu_{C_j^{\delta}}(\psi^1),\notag\\
&(\forall m\in \{1,2,\cdots,2N_v\},\bX\in I^m),\notag
\end{align}
the equality \eqref{eq_artificial_free_part} holds for $l'$ as
 well. Thus, by induction the equality \eqref{eq_artificial_free_part}
 is true for any $l'\in\{l,l+1,\cdots,N_h\}$. 

Let us prove that for any $l'\in\{l,l+1,\cdots,N_h\}$, $t\in \{0,1\}$,
\begin{align}
\sum_{m=2}^{2N_v}c_0^{\frac{m}{2}}\alpha^m\|\hat{F}_m^{l'}\|_{0,t}\le
 \alpha^{-2}M^{-\frac{l'}{N_v-1}}.\label{eq_artificial_free_part_bound}
\end{align}
This inequality is true for $l'=N_h$ by definition. 
Assume that $l'\in \{l,l+1,\cdots,N_h-1\}$ and 
$\eqref{eq_artificial_free_part_bound}$ holds for any $j\in
 \{l'+1,l'+2,\cdots,N_h\}$. Note that for any $m\in\{2,3,\cdots,2N_v\}$,
\begin{align*}
F_m^{l'}(\psi)
&=F_m^{l'+1}(\psi)+T_m^{l'+1}(\psi)+\cP_m\sum_{n=m+2}^{2N_v}\int
 F_n^{l'+1}(\psi+\psi^1)d\mu_{C_{l'+1}^{\delta}}(\psi^1)\\
&\quad+\cP_m\sum_{n=m+2}^{2N_v}\int
 T_n^{l'+1}(\psi+\psi^1)d\mu_{C_{l'+1}^{\delta}}(\psi^1)\\
&\quad+\cP_m\sum_{n=2N_v+2}^{N}\int
 J_n^{l'+1}(\psi+\psi^1)d\mu_{C_{l'+1}^{\delta}}(\psi^1),
\end{align*}
and thus,
\begin{align}
\hat{F}_m^{l'}(\psi)
&=\hat{F}_m^{l'+1}(\psi)+T_m^{l'+1}(\psi)+\cP_m\sum_{n=m+2}^{2N_v}\int
 \hat{F}_n^{l'+1}(\psi+\psi^1)d\mu_{C_{l'+1}^{\delta}}(\psi^1)\label{eq_artificial_free_decomposition}\\
&\quad+\cP_m\sum_{n=m+2}^{2N_v}\int
 T_n^{l'+1}(\psi+\psi^1)d\mu_{C_{l'+1}^{\delta}}(\psi^1)\notag\\
&\quad+\cP_m\sum_{n=2N_v+2}^{N}\int
 J_n^{l'+1}(\psi+\psi^1)d\mu_{C_{l'+1}^{\delta}}(\psi^1).\notag
\end{align}
It follows from this equality and an estimation similar to 
\cite[\mbox{Lemma 3.1}]{K15} that 
\begin{align*}
\|\hat{F}_m^{l'}\|_{0,t}&\le
 \sum_{n=m}^{2N_v}2^nc_0^{\frac{n-m}{2}}\|\hat{F}_n^{l'+1}\|_{0,t}\\
&\quad
+\sum_{n=m}^{2N_v}2^nc_0^{\frac{n-m}{2}}\|T_n^{l'+1}\|_{0,t}
+\sum_{n=2N_v+2}^N2^nc_0^{\frac{n-m}{2}}\|J_n^{l'+1}\|_{0,t}.
\end{align*}
Using \eqref{eq_UV_all_bound}, \eqref{eq_tree_sum},
 \eqref{eq_artificial_free_part_bound} for $l'+1$ and the conditions
 $\alpha\ge c$, $M\ge c^{N_v^2}$, we have that
\begin{align*}
&\sum_{m=2}^{2N_v}c_0^{\frac{m}{2}}\alpha^m\|\hat{F}_m^{l'}\|_{0,t}\\
&\le
 \sum_{m=2}^{2N_v}2^m\sum_{n=m}^{2N_v}c_0^{\frac{n}{2}}\alpha^n\|\hat{F}_n^{l'+1}\|_{0,t}
+\sum_{m=2}^{2N_v}2^m\sum_{n=m}^{2N_v}c_0^{\frac{n}{2}}\alpha^n\|T_n^{l'+1}\|_{0,t}
\\
&\quad+\sum_{m=2}^{2N_v}2^{2N_v+2}\alpha^{m-2N_v-2}M^{-\frac{l'+1}{N_v-1}(N_v+1)}
\sum_{n=2N_v+2}^Nc_0^{\frac{n}{2}}\alpha^nM^{\frac{l'+1}{2N_v-2}n}
\|J_n^{l'+1}\|_{0,t}\\
&\le c^{N_v}\alpha^{-2}M^{-\frac{l'+1}{N_v-1}}\le
 \alpha^{-2}M^{-\frac{l'}{N_v-1}}.
\end{align*}
Thus, the inequality \eqref{eq_artificial_free_part_bound} for $l'$
 holds. By induction, \eqref{eq_artificial_free_part_bound} holds for all
 $l'\in \{l,l+1,\cdots,N_h\}$.

By \eqref{eq_UV_covariance_determinant_equal_time},
 \eqref{eq_artificial_free_part} and \eqref{eq_artificial_free_kernel}, for any $m\in\{2,3,\cdots,2N_v\}$,
 $t\in \{0,1\}$,
\begin{align*}
\|F_m^{l}\|_{0,t}\le&
 \|\hat{F}_m^l\|_{0,t}+\|F_m^{N_h}\|_{0,t}\\
&+\sum_{j=l+1}^{N_h}\sum_{n=m+2}^{2N_v}2^n(M^{-j}+M^{j-N_h})c_0^{\frac{n-m}{2}}(\|F_n^j\|_{0,t}+\|\hat{F}_n^j\|_{0,t}).
\end{align*}
Moreover, by \eqref{eq_UV_bound} for $l'\in\{l+1,l+2,\cdots,N_h\}$,
 \eqref{eq_initial_UV_bound}, \eqref{eq_artificial_free_part_bound} for
 $l'\in\{l,l+1,\cdots,N_h\}$ and the assumption that $\alpha\ge 2$,
 $M\ge 2$,
\begin{align}
&\sum_{m=2}^{2N_v}c_0^{\frac{m}{2}}\alpha^m\|F_m^{l}\|_{0,t}\label{eq_UV_bound_free_preliminary}\\
&\le
 \sum_{m=2}^{2N_v}c_0^{\frac{m}{2}}\alpha^m\|\hat{F}_m^l\|_{0,t}+\sum_{m=2}^{2N_v}c_0^{\frac{m}{2}}\alpha^m\|F_m^{N_h}\|_{0,t}\notag\\
&\quad+
\sum_{m=2}^{2N_v}\sum_{j=l+1}^{N_h}2^{m+2}\alpha^{-2}(M^{-j}+M^{j-N_h})
\sum_{n=m+2}^{2N_v}c_0^{\frac{n}{2}}\alpha^n(\|F_n^j\|_{0,t}+\|\hat{F}_n^j\|_{0,t})\notag\\
&\le\frac{1}{2}+c^{N_v}\alpha^{-2}.\notag
\end{align}
This also yields that
\begin{align}
M^{-\frac{N_v}{N_v-1}l}\sum_{m=2}^{2N_v}c_0^{\frac{m}{2}}\alpha^mM^{\frac{l}{2N_v-2}m}\|F_m^l\|_{0,t}\le
 \frac{1}{2}+c^{N_v}\alpha^{-2}.\label{eq_free_lower_half_bound}
\end{align}

On the other hand, for $m\in \{2N_v+2,2N_v+3,\cdots,N\}$,
\begin{align*}
F_m^{l}(\psi)=J_m^{l+1}(\psi)+\cP_m\sum_{n=m+2}^N\int
 J_n^{l+1}(\psi+\psi^1)d\mu_{C_{l+1}^{\delta}}(\psi^1)
\end{align*}
and thus by \eqref{eq_UV_covariance_determinant},
\begin{align*}
\|F_m^l\|_{0,t}\le
 \sum_{n=m}^N2^nc_0^{\frac{n-m}{2}}\|J_n^{l+1}\|_{0,t}.
\end{align*}
Moreover, by \eqref{eq_UV_all_bound} for $l+1$ and the condition $M\ge
 c^{N_v}$, 
\begin{align*}
&\sum_{m=2N_v+2}^Nc_0^{\frac{m}{2}}\alpha^mM^{\frac{l}{2N_v-2}m}\|F_m^l\|_{0,t}\\
&\le
 \sum_{n=2N_v+2}^N\sum_{m=2N_v+2}^n\alpha^mM^{\frac{l}{2N_v-2}m}2^nc_0^{\frac{n}{2}}\|J_n^{l+1}\|_{0,t}\\
&\le c
 \sum_{n=2N_v+2}^N2^nc_0^{\frac{n}{2}}\alpha^nM^{\frac{l}{2N_v-2}n}\|J_n^{l+1}\|_{0,t}\\
&\le
 c^{N_v}M^{-\frac{N_v+1}{N_v-1}}\sum_{n=2N_v+2}^Nc_0^{\frac{n}{2}}\alpha^nM^{\frac{l+1}{2N_v-2}n}\|J_n^{l+1}\|_{0,t}\\
&\le
 c^{N_v}M^{-\frac{N_v+1}{N_v-1}+\frac{N_v}{N_v-1}(l+1)},
\end{align*}
or
\begin{align}
M^{-\frac{N_v}{N_v-1}l}\sum_{m=2N_v+2}^{N}c_0^{\frac{m}{2}}\alpha^mM^{\frac{l}{2N_v-2}m}\|F_m^l\|_{0,t}\le
c^{N_v}M^{-\frac{1}{N_v-1}}.
\label{eq_free_upper_half_bound}
\end{align}

It remains to deal with $F_0^l$. Note that
\begin{align*}
F_0^l
&=F_0^{l+1}+T_0^{l+1}+\sum_{m=2}^{2N_v}\int
 \hat{F}_m^{l+1}(\psi)d\mu_{C_{l+1}^{\delta}}(\psi)\\
&\quad+\sum_{m=2}^{2N_v}\int
 (F_m^{l+1}(\psi)-\hat{F}_m^{l+1}(\psi))d\mu_{C_{l+1}^{\delta}}(\psi)
+\sum_{m=2}^{2N_v}\int
 T_m^{l+1}(\psi)d\mu_{C_{l+1}^{\delta}}(\psi)\\
&\quad +\sum_{m=2N_v+2}^{N}\int
 J_m^{l+1}(\psi)d\mu_{C_{l+1}^{\delta}}(\psi).
\end{align*}
Then, by \eqref{eq_UV_covariance_determinant}, 
\eqref{eq_UV_covariance_determinant_equal_time},
\eqref{eq_UV_bound},
\eqref{eq_UV_all_bound},
\eqref{eq_tree_0th},
\eqref{eq_tree_sum},
\eqref{eq_artificial_free_part_bound} for $l'\in\{l+1,l+2,\cdots,N_h\}$
and \eqref{eq_artificial_free_part},
\begin{align}
|F_0^l|
&\le
 |F_0^{l+1}|+|T_0^{l+1}|+\frac{N}{h}\sum_{m=2}^{2N_v}c_0^{\frac{m}{2}}\|\hat{F}_m^{l+1}\|_{0,0}\label{eq_free_0th}\\
&\quad +\frac{N}{h}\sum_{m=2}^{2N_v}(M^{-l-1}+M^{l+1-N_h})c_0^{\frac{m}{2}}(\|\hat{F}_m^{l+1}\|_{0,0}+\|{F}_m^{l+1}\|_{0,0})\notag\\
&\quad +\frac{N}{h}\sum_{m=2}^{2N_v}c_0^{\frac{m}{2}}\|T_m^{l+1}\|_{0,0}
+\frac{N}{h}\sum_{m=2N_v+2}^{N}c_0^{\frac{m}{2}}\|J_m^{l+1}\|_{0,0}\notag\\
&\le
 |F_0^{l+1}|+\frac{N}{h}c^{N_v}\alpha^{-2}(M^{-l-1}+M^{l+1-N_h}+M^{-\frac{l+1}{N_v-1}})\notag\\
&\le
 \frac{N}{h}c^{N_v}\alpha^{-2}\sum_{j=l}^{N_h-1}(M^{-j-1}+M^{j+1-N_h}+M^{-\frac{j+1}{N_v-1}})\notag\\
&\le \frac{N}{h}c^{N_v}\alpha^{-2}.\notag
\end{align}

Finally we sum up \eqref{eq_tree_0th}, \eqref{eq_tree_sum},
 \eqref{eq_tree_sum_weight},
\eqref{eq_UV_bound_free_preliminary},
\eqref{eq_free_lower_half_bound},
\eqref{eq_free_upper_half_bound} and \eqref{eq_free_0th} to deduce that
 for any $t\in \{0,1\}$,
\begin{align*}
&\frac{h}{N}(|F_0^l|+|T_0^l|)\le c^{N_v}\alpha^{-2},\\
&\sum_{m=1}^{2N_v}c_0^{\frac{m}{2}}\alpha^m(\|F_m^l\|_{0,t}+\|T_m^l\|_{0,t})\le
 \frac{1}{2}+c^{N_v}\alpha^{-2},\\
&M^{-\frac{N_v}{N_v-1}l}\sum_{m=1}^Nc_0^{\frac{m}{2}}\alpha^mM^{\frac{l}{2N_v-2}m}
(\|F_m^l\|_{0,t}+\|T_m^l\|_{0,t})\\
&\le \frac{1}{2}+c^{N_v}\alpha^{-2}+
cM\alpha^{-2}+c^{N_v}M^{-\frac{1}{N_v-1}}.
\end{align*}
Recall that so far we have used the conditions $\alpha\ge c$, $M\ge
 c^{N_v^2}$ and \eqref{eq_interaction_amplitude_UV}. Now we can see that under the conditions
 \eqref{eq_UV_parameter_conditions} with a sufficiently large generic
 constant $c$ the inequalities above imply
 \eqref{eq_UV_0th_bound}, \eqref{eq_UV_bound}, \eqref{eq_UV_all_bound}
 for $l$. Therefore, by induction these inequalities hold true for all
 $l\in \{0,1,\cdots,N_h\}$, $t\in \{0,1\}$ on the conditions
 \eqref{eq_UV_parameter_conditions}, \eqref{eq_interaction_amplitude_UV}.
\end{proof}

\subsection{Anisothermal bounds}\label{subsec_anisothermal}
Our result concerning the existence of the zero-temperature limit of the
free energy density is made out of a series of estimates on the differences
between Grassmann polynomials defined at 2 different temperatures. As one
part of these analysis, here we focus on establishing
temperature-dependent upper bounds on Grassmann polynomials produced by
the Matsubara UV integration. In addition to the notations already
introduced in Subsection \ref{subsec_UV_covariance} and Subsection
\ref{subsec_isothermal}, let us define some notations necessary for our
anisothermal measurements. These notations are essentially same as those
introduced in the beginning of \cite[\mbox{Section 4}]{K15}. 

For any
$\bX=((\rho_1,\bx_1,\s_1,s_1),(\rho_2,\bx_2,\s_2,s_2),\cdots,(\rho_m,\bx_m,\s_m,s_m))$\\
$\in (\cB\times\G(L)\times\spin\times(1/h)\Z)^m$, we define
$R_{\beta}(\bX)\in I_0^m$, $N_{\beta}(\bX)\in\Z$ by
\begin{align*}
&R_{\beta}(\bX):=((\rho_1,\bx_1,\s_1,r_{\beta}(s_1)),\cdots,(\rho_m,\bx_m,\s_m,r_{\beta}(s_m))),\\
&N_{\beta}(\bX):=\sum_{j=1}^mn_{\beta}(s_j).
\end{align*}
Though this is admittedly abuse of notation,
we let $R_{\beta}(\bX)$, $N_{\beta}(\bX)$ denote
\begin{align*}
&((\rho_1,\bx_1,\s_1,r_{\beta}(s_1),\theta_1),\cdots,(\rho_m,\bx_m,\s_m,r_{\beta}(s_m),\theta_m))\
 (\in I^m),\\
&\sum_{j=1}^mn_{\beta}(s_j)\ (\in\Z)
\end{align*}
respectively, for $\bX=((\rho_1,\bx_1,\s_1,s_1,\theta_1),\cdots,(\rho_m,\bx_m,\s_m,s_m,\theta_m))$\\
$\in (\cB\times\G(L)\times\spin\times(1/h)\Z\times\{1,-1\})^m$ as well. 

For $\bX=((\rho_1,\bx_1,\s_1,s_1,\theta_1),\cdots,(\rho_m,\bx_m,\s_m,s_m,\theta_m))$\\
$\in (\cB\times\G(L)\times\spin\times(1/h)\Z\times\{1,-1\})^m$, $s\in(1/h)\Z$, 
we let $\bX+s$ denote 
$$
((\rho_1,\bx_1,\s_1,s_1+s,\theta_1),\cdots,(\rho_m,\bx_m,\s_m,s_m+s,\theta_m))
$$
in order to shorten formulas.

In the rest of this subsection we always assume
\eqref{eq_basic_beta_h_assumption}. Set 
\begin{align*}
\left[\frac{\beta_1}{4},\beta_a-\frac{\beta_1}{4}\right)_h:=\left\{\frac{\beta_1}{4},\frac{\beta_1}{4}+\frac{1}{h},\cdots,\beta_a-\frac{\beta_1}{4}-\frac{1}{h}\right\},\quad
 (a=1,2).
\end{align*}
We measure the difference between anti-symmetric functions 
$f(\beta_a):I(\beta_a)^m\to\C$ $(a=1,2,m\in\N_{\ge 2})$ by the quantify
$|f(\beta_1)-f(\beta_2)|_l$, which is defined by 
\begin{align*}
&|f(\beta_1)-f(\beta_2)|_l\\
&:=\sup_{X\in
 I^0}\left(\frac{1}{h}\right)^{m-1}\sum_{(Y_1,Y_2,\cdots,Y_{m-1})\in\hat{I}^{m-1}}e^{\sum_{j=0}^d(\frac{1}{\pi}\fw(l)\hat{d}_j(X,Y_1))^{1/2}}\\
&\quad\cdot |f(\beta_1)(R_{\beta_1}(X,Y_1,\cdots,Y_{m-1}))
-f(\beta_2)(R_{\beta_2}(X,Y_1,\cdots,Y_{m-1}))|.
\end{align*}
In this subsection we estimate Grassmann polynomials by using $|\cdot -
\cdot|_0$. The infrared analysis in Section \ref{sec_IR} will largely
use $|\cdot -
\cdot|_l$ with $l\in\Z_{<0}$. 

With these notations we have for any $l\in \{1,2,\cdots,N_h\}$,
$\delta\in\{+,-\}$ that
\begin{align}
&\widetilde{C_l^{\delta}}(\beta_a)(\bX)=(-1)^{N_{\beta_a}(\bX+s)}\widetilde{C_l^{\delta}}(\beta_a)(R_{\beta_a}(\bX+s)),\label{eq_UV_covariance_temperature_invariance}\\
&(\forall \bX\in I(\beta_a)^2,s\in (1/h)\Z,a\in \{1,2\}),\notag\\
&|\widetilde{C_l^{\delta}}(\beta_1)-\widetilde{C_l^{\delta}}(\beta_2)|_0\le
 \beta_1^{-\frac{1}{2}}M^{-l}c_0.\label{eq_UV_covariance_decay_difference_norm}
\end{align}
The inequality \eqref{eq_UV_covariance_decay_difference_norm} is due to
\eqref{eq_UV_covariance_decay_difference}. 

Fix $\delta\in \{+,-\}$ and for $a=1,2$ let $F^l(\beta_a)(\psi)$,
$T^l(\beta_a)(\psi)$, $J^l(\beta_a)(\psi)\in\bigwedge \cV(\beta_a)$
$(l=0,1,\cdots,N_h)$ be the Grassmann polynomials defined in the
beginning of the previous subsection at the inverse temperature $\beta_a$. 

By anti-symmetry, for any $f(\beta_1)(\psi)\in\bigwedge\cV(\beta_1)$ and
$m\in\{N(\beta_1)+1,N(\beta_1)+2,\cdots,N(\beta_2)\}$,
$f_m(\beta_1)(\psi)=0$.
Keeping this fact in mind, we can write that 
$f(\beta_1)(\psi)=\sum_{m=0}^{N(\beta_2)}f_m(\beta_1)(\psi)$.

\begin{lemma}\label{lem_UV_polynomial_temperature_invariance}
\begin{align*}
&F_m^{l}(\beta_a)(\bX)=(-1)^{N_{\beta_a}(\bX+s)}F_m^l(\beta_a)(R_{\beta_a}(\bX+s)),\\
&T_m^{l,(n)}(\beta_a)(\bX)=(-1)^{N_{\beta_a}(\bX+s)}T_m^{l,(n)}(\beta_a)(R_{\beta_a}(\bX+s)),\\
&(\forall l\in\{0,1,\cdots,N_h\},n\in \N_{\ge 2},a\in\{1,2\},
 m\in\{1,2,\cdots,N(\beta_2)\},\\
&\quad \bX\in I(\beta_a)^m,s\in (1/h)\Z).
\end{align*}
\end{lemma}
\begin{proof}
We can see from the definition that the claimed equalities hold for
 $l=N_h$. Then, by \eqref{eq_UV_covariance_temperature_invariance} the
 same inductive argument based on \cite[\mbox{Lemma 3.9 (1)}]{K15}
as in the proof of \cite[\mbox{Lemma 5.3}]{K15} ensures the results.
\end{proof}

The invariant property summarized in Lemma
\ref{lem_UV_polynomial_temperature_invariance} is one of the basic
assumptions in the general theory \cite[\mbox{Section 4}]{K15}. 
The rest of the assumptions in \cite[\mbox{Section 4}]{K15} are the
bound properties of the covariances which we prepared in Lemma
\ref{lem_UV_covariance_properties}. Thus, 
we can apply  \cite[\mbox{Lemma 4.1,\ Lemma 4.6}]{K15} in the
proof of the following lemma.

\begin{lemma}\label{lem_UV_integration_difference}
Let $\alpha\in\R_{\ge 1}$ and let $c_0$ be the constant appearing in
 Lemma \ref{lem_UV_covariance_properties}. There exists a constant $c\in\R_{>0}$
 independent of any parameter such that if \eqref{eq_UV_parameter_conditions}
 holds with $c$ and \eqref{eq_interaction_amplitude_UV} holds, 
the following inequalities hold for any $l\in\{0,1,\cdots,N_h\}$.
\begin{align}
&\left|\frac{h}{N(\beta_1)}F^l_0(\beta_1)-\frac{h}{N(\beta_2)}F^l_0(\beta_2)\right|
+\left|\frac{h}{N(\beta_1)}T^l_0(\beta_1)-\frac{h}{N(\beta_2)}T^l_0(\beta_2)\right|
\label{eq_UV_0th_bound_difference}\\
&\le \beta_1^{-\frac{1}{2}}\alpha^{-1},\notag\\
&\sum_{m=2}^{2N_v}c_0^{\frac{m}{2}}\alpha^m(|F_m^l(\beta_1)-F_m^l(\beta_2)|_0+|T_m^l(\beta_1)-T_m^l(\beta_2)|_0)\le \beta_1^{-\frac{1}{2}},\label{eq_UV_bound_difference}\\
&M^{-\frac{N_v}{N_v-1}l}\sum_{m=2}^{N(\beta_2)}c_0^{\frac{m}{2}}\alpha^mM^{\frac{l}{2N_v-2}m}
(|F_m^l(\beta_1)-F_m^l(\beta_2)|_{0}+|T_m^l(\beta_1)-T_m^l(\beta_2)|_{0})\label{eq_UV_all_bound_difference}\\
&\le \beta_1^{-\frac{1}{2}}.\notag
\end{align}
\end{lemma} 

\begin{proof}
We assume the conditions \eqref{eq_UV_parameter_conditions}, \eqref{eq_interaction_amplitude_UV} with a
 constant $c'\in\R_{>0}$ so that the results of Lemma
 \ref{lem_UV_integration} hold for $\beta_1$ and $\beta_2$. 
Let us make clear the logic. During the proof we do not touch the initial
 constant $c'$. In the end of the proof we will see that all the
 estimations are justified if the initial constant $c'$ is sufficiently large.
In the
 following we use the symbol $c$ to express a generic positive constant independent of any
 parameter and will replace it by a larger constant denoted by
the same symbol from time to time. This notational convention helps to
 simplify the arguments. Not to confuse, we should stress that a
 constant denoted by $c$ does not depend on $c'$, either.

We prove the claims by induction with $l\in\{0,1,\cdots,N_h\}$. 
By definition the left-hand sides of the claimed inequalities for
 $l=N_h$ vanish. Thus, the results hold for $l=N_h$.

Assume that $l\in\{0,1,\cdots,N_h-1\}$ and for any
 $j\in\{l+1,l+2,\cdots,N_h\}$ the inequalities \eqref{eq_UV_bound_difference},
 \eqref{eq_UV_all_bound_difference} hold.
In the same way as in the derivation of \eqref{eq_bound_from_induction},
 \eqref{eq_bound_from_induction_weight} we can derive from 
the hypothesis of induction that
\begin{align}
&\sum_{m=2}^{N(\beta_2)}2^{3m}c_0^{\frac{m}{2}}\alpha^m|J_m^{l+1}(\beta_1)-J_m^{l+1}(\beta_2)|_{0}
\le c^{N_v}\beta_1^{-\frac{1}{2}},\label{eq_bound_from_induction_difference}\\
&\sum_{m=2}^{N(\beta_2)}2^{2m}c_0^{\frac{m}{2}}\alpha^mM^{\frac{l}{2N_v-2}m}|J_m^{l+1}(\beta_1)-J_m^{l+1}(\beta_2)|_{0}
\le c M^{-\frac{1}{N_v-1}+\frac{N_v}{N_v-1}(l+1)}\beta_1^{-\frac{1}{2}},
\label{eq_bound_from_induction_weight_difference}
\end{align}
on the assumption that $M\ge c^{N_v}$.

Substitution of 
 \eqref{eq_UV_covariance_determinant},
\eqref{eq_UV_covariance_determinant_difference},
 \eqref{eq_UV_covariance_decay_norm}, 
\eqref{eq_bound_from_induction},
 \eqref{eq_UV_covariance_decay_difference_norm},
\eqref{eq_bound_from_induction_difference} into the inequality in
\cite[\mbox{Lemma 4.6 (1)}]{K15} yields that for any $n\in \N_{\ge 2}$,
\begin{align*}
&\left|\frac{h}{N(\beta_1)}T_0^{l,(n)}(\beta_1)
- \frac{h}{N(\beta_2)}T_0^{l,(n)}(\beta_2)\right|\\
&\le
c^nM^{-(l+1)(n-1)}\left(\sum_{m=2}^{N(\beta_2)}2^{3m}c_0^{\frac{m}{2}}\sum_{a=1}^2\|J_m^{l+1}(\beta_a)\|_{0,0}\right)^{n-1}\\
&\quad\cdot \sum_{m=2}^{N(\beta_2)}2^{3m}c_0^{\frac{m}{2}}\left(
\beta_1^{-\frac{1}{2}}\sum_{b=1}^2\sum_{t=0}^1
\|J_m^{l+1}(\beta_b)\|_{0,t}+|J_m^{l+1}(\beta_1)-J_m^{l+1}(\beta_2)|_0\right)\\
&\le \beta_1^{-\frac{1}{2}}M^{l+1}(c^{N_v}M^{-l-1}\alpha^{-2})^n. 
\end{align*}
Moreover, on the assumption that $M\ge c^{N_v}$,
\begin{align}
\left|\frac{h}{N(\beta_1)}T_0^l(\beta_1)-\frac{h}{N(\beta_2)}T_0^l(\beta_2)\right|
\le c^{N_v}\beta_1^{-\frac{1}{2}}M^{-l-1}\alpha^{-4}.
\label{eq_tree_0th_difference}
\end{align}
By \eqref{eq_UV_covariance_determinant},
\eqref{eq_UV_covariance_determinant_difference},
 \eqref{eq_UV_covariance_decay_norm}, 
 \eqref{eq_UV_covariance_decay_difference_norm}
and \cite[\mbox{Lemma 4.6
 (2)}]{K15} we have for any $m\in\{2,3,\cdots,N(\beta_2)\}$, $n\in\N_{\ge
 2}$ that
\begin{align}
&|T_m^{l,(n)}(\beta_1)-T_m^{l,(n)}(\beta_2)|_{0}\label{eq_tree_higher_preliminary_difference}\\
&\le c^n\cdot 2^{-2m}c_0^{-\frac{m}{2}}M^{-(l+1)(n-1)}
\prod_{j=2}^n\left(\sum_{m_j=2}^{N(\beta_2)}2^{4m_j}c_0^{\frac{m_j}{2}}\sum_{a=1}^2\|J_{m_j}^{l+1}(\beta_a)\|_{0,0}\right)\notag\\
&\quad\cdot
 \sum_{m_1=2}^{N(\beta_2)}2^{4m_1}c_0^{\frac{m_1}{2}}
\left(\beta_1^{-\frac{1}{2}}\sum_{b=1}^2\sum_{t=0}^1\|J_{m_1}^{l+1}(\beta_b)\|_{0,t}
+|J_{m_1}^{l+1}(\beta_1)-J_{m_1}^{l+1}(\beta_2)|_0\right)\notag\\
&\quad\cdot 1_{\sum_{j=1}^nm_j-2n+2\ge m}.\notag
\end{align}
Then, by \eqref{eq_bound_from_induction} and
 \eqref{eq_bound_from_induction_difference}, 
\begin{align*}
&\sum_{m=2}^{2N_v}c_0^{\frac{m}{2}}\alpha^m|T_m^{l,(n)}(\beta_1)-T_m^{l,(n)}(\beta_2)
|_{0}\\
&\le c^nM^{-(l+1)(n-1)}\alpha^{-2n+2}
\left(\sum_{m=2}^{N(\beta_2)}2^{2m}c_0^{\frac{m}{2}}\alpha^{m}\sum_{a=1}^2
\|J_{m}^{l+1}(\beta_a)\|_{0,0}\right)^{n-1}\\
&\quad\cdot \sum_{m=2}^{N(\beta_2)}2^{2m}c_0^{\frac{m}{2}}\alpha^{m}
\left(\beta_1^{-\frac{1}{2}}\sum_{b=1}^2
\sum_{t=0}^1\|J_{m}^{l+1}(\beta_b)\|_{0,t} +
|J_m^{l+1}(\beta_1)-J_m^{l+1}(\beta_2)|_0\right)\\
&\le \beta_1^{-\frac{1}{2}}M^{l+1}\alpha^2(c^{N_v}M^{-l-1}\alpha^{-2})^n.
\end{align*}
Thus, on the assumption $M\ge c^{N_v}$ we have that 
\begin{align}
\sum_{m=2}^{2N_v}c_0^{\frac{m}{2}}\alpha^m|T_m^l(\beta_1)-T_m^l(\beta_2)|_{0}\le
 \beta_1^{-\frac{1}{2}}c^{N_v}M^{-l-1}\alpha^{-2}.\label{eq_tree_sum_difference}
\end{align}
Also, it follows from \eqref{eq_tree_higher_preliminary_difference} and
 \eqref{eq_bound_from_induction_weight},
 \eqref{eq_bound_from_induction_weight_difference} that
\begin{align*}
&M^{-\frac{N_v}{N_v-1}l}
\sum_{m=2}^{N(\beta_2)}c_0^{\frac{m}{2}}\alpha^mM^{\frac{l}{2N_v-2}m}|T_m^{l,(n)}(\beta_1)-T_m^{l,(n)}(\beta_2)|_{0}\\
&\le c^n
 M^{-\frac{N_v}{N_v-1}l-(l+1+\frac{l}{N_v-1})(n-1)}\alpha^{-2n+2}\\
&\quad\cdot\left(\sum_{m=2}^{N(\beta_2)}2^{2m}\alpha^{m}c_0^{\frac{m}{2}}M^{\frac{l}{2N_v-2}m}
\sum_{a=1}^2\|J_{m}^{l+1}(\beta_a)\|_{0,0}\right)^{n-1}\\
&\quad \cdot
 \sum_{m=2}^{N(\beta_2)}2^{2m}c_0^{\frac{m}{2}}\alpha^{m}M^{\frac{l}{2N_v-2}m}\\
&\quad\cdot \left(
\beta_1^{-\frac{1}{2}}\sum_{b=1}^2\sum_{t=0}^1\|J_{m}^{l+1}(\beta_b)\|_{0,t}
+|J_m^{l+1}(\beta_1)-J_m^{l+1}(\beta_2)|_0\right)\\
&\le \beta_1^{-\frac{1}{2}}M^{1-(l+1+\frac{l}{N_v-1})n}\alpha^{-2n+2}(c
 M^{-\frac{1}{N_v-1}+\frac{N_v}{N_v-1}(l+1)})^n\\
&=\beta_1^{-\frac{1}{2}} M\alpha^2(c\alpha^{-2})^n.
\end{align*}
Thus, by the assumption $\alpha\ge c$, 
\begin{align}
M^{-\frac{N_v}{N_v-1}l}
\sum_{m=2}^{N(\beta_2)}c_0^{\frac{m}{2}}\alpha^mM^{\frac{l}{2N_v-2}m}|T_m^{l}(\beta_1)-T_m^{l}(\beta_2)|_{0}
\le
 c \beta_1^{-\frac{1}{2}} M\alpha^{-2}.\label{eq_tree_sum_weight_difference}
\end{align}

In order to find upper bounds on the difference between
 $F^l(\beta_1)(\psi)$ and $F^l(\beta_2)(\psi)$, we need to establish upper
 bounds on the difference between $\hat{F}^l(\beta_1)(\psi)$ and
 $\hat{F}^l(\beta_2)(\psi)$. To this end, first
we need to confirm that
\begin{align}
&\hat{F}_m^{l'}(\beta_a)(\bX)=(-1)^{N_{\beta_a}(\bX+s)}\hat{F}_m^{l'}(\beta_a)(R_{\beta_a}(\bX+s)),\label{eq_artificial_free_invariance}\\
&(\forall l'\in\{l,l+1,\cdots,N_h\},a\in\{1,2\},
 m\in\{1,2,\cdots,2N_v\},\notag\\
&\quad \bX\in I(\beta_a)^m,s\in (1/h)\Z
 ),\notag
\end{align}
where $\hat{F}_m^{l'}(\beta_a)(\bX)$ $(m=1,2,\cdots,2N_v)$ are
 the kernels of $\hat{F}^{l'}(\beta_a)(\psi)\in \bigwedge \cV(\beta_a)$
 defined in \eqref{eq_artificial_free_definition}. By definition,
 \eqref{eq_artificial_free_invariance} holds for $l'=N_h$. Assume that
$l'\in \{l,l+1,\cdots,N_h-1\}$ and 
 \eqref{eq_artificial_free_invariance} is true for $j\in
 \{l'+1,l'+2,\cdots,N_h\}$. Take any $s\in (1/h)\Z$. It follows from
 \eqref{eq_UV_covariance_temperature_invariance} that for any
 $a\in\{1,2\}$, $n\in \N$, $\bY\in I(\beta_a)^n$,  $j\in
 \{l'+1,l'+2,\cdots,N_h\}$,
\begin{align}
\int\psi_{R_{\beta_a}(\bY+s)}d\mu_{C_j^{\delta}(\beta_a)}(\psi)=(-1)^{N_{\beta_a}(\bY+s)}\int
 \psi_{\bY}d\mu_{C_j^{\delta}(\beta_a)}(\psi).\label{eq_Grassmann_integral_invariance_specific}
\end{align}
By using this equality, Lemma
 \ref{lem_UV_polynomial_temperature_invariance},
 \eqref{eq_artificial_free_kernel} and the induction hypothesis we have
 that for any $\bX\in
 I(\beta_a)^m$, 
\begin{align*}
&(-1)^{N_{\beta_a}(\bX+s)}\hat{F}_m^{l'}(\beta_a)(R_{\beta_a}(\bX+s))\\
&=F_m^{l'}(\beta_a)(\bX)-F_m^{N_h}(\beta_a)(\bX)\\
&\quad -\sum_{j=l'+1}^{N_h}\sum_{n=m+2}^{2N_v}\left(\begin{array}{c}n \\
					      m\end{array}\right)
\left(\frac{1}{h}\right)^{n-m}\sum_{\bY\in I(\beta_a)^{n-m}}(-1)^{N_{\beta_a}(\bX+s)}\\
&\qquad\qquad\cdot (F_n^j(\beta_a)(R_{\beta_a}(\bX+s),\bY)-
 \hat{F}_n^j(\beta_a)(R_{\beta_a}(\bX+s),\bY))\\
&\qquad\qquad\cdot \int \psi_{\bY}^1d\mu_{C_j^{\delta}(\beta_a)}(\psi^1)\\
&=F_m^{l'}(\beta_a)(\bX)-F_m^{N_h}(\beta_a)(\bX)\\
&\quad -\sum_{j=l'+1}^{N_h}\sum_{n=m+2}^{2N_v}\left(\begin{array}{c}n \\
					      m\end{array}\right)
\left(\frac{1}{h}\right)^{n-m}\sum_{\bY\in I(\beta_a)^{n-m}}(-1)^{N_{\beta_a}(\bX+s)+N_{\beta_a}(\bY+s)}\\
&\qquad\cdot (F_n^j(\beta_a)(R_{\beta_a}(\bX+s),R_{\beta_a}(\bY+s))\\
&\qquad\quad -
 \hat{F}_n^j(\beta_a)(R_{\beta_a}(\bX+s),R_{\beta_a}(\bY+s)))\\
&\qquad\cdot\int \psi_{\bY}^1d\mu_{C_j^{\delta}(\beta_a)}(\psi^1)\\
&=\hat{F}_m^{l'}(\beta_a)(\bX).
\end{align*}
Thus, by induction the equality \eqref{eq_artificial_free_invariance}
 holds for all $l'\in \{l,l+1,\cdots,N_h\}$.

Let us prove that for any $l'\in\{l,l+1,\cdots,N_h\}$, 
\begin{align}
\sum_{m=2}^{2N_v}c_0^{\frac{m}{2}}\alpha^m|\hat{F}_m^{l'}(\beta_1)-
\hat{F}_m^{l'}(\beta_2)|_{0}\le\beta_1^{-\frac{1}{2}}
 \alpha^{-2}M^{-\frac{l'}{N_v-1}}.\label{eq_artificial_free_part_bound_difference}
\end{align}
For $l'=N_h$ the inequality
 \eqref{eq_artificial_free_part_bound_difference} holds since its
 left-hand side is zero. 
Assume that $l'\in \{l,l+1,\cdots,N_h-1\}$ and 
$\eqref{eq_artificial_free_part_bound_difference}$ holds for all $j\in
 \{l'+1,l'+2,\cdots,N_h\}$. 
By \eqref{eq_UV_covariance_determinant},
\eqref{eq_UV_covariance_determinant_difference},
\eqref{eq_artificial_free_decomposition} and the estimation parallel to 
\cite[\mbox{Lemma 4.1 (2)}]{K15},
\begin{align*}
&|\hat{F}_m^{l'}(\beta_1)-\hat{F}_m^{l'}(\beta_2)|_{0}\\
&\le|\hat{F}_m^{l'+1}(\beta_1)-\hat{F}_m^{l'+1}(\beta_2)|_{0}
+|T_m^{l'+1}(\beta_1)-T_m^{l'+1}(\beta_2)|_{0}\\
&\quad+
c \sum_{n=m+2}^{2N_v}2^{2n}c_0^{\frac{n-m}{2}}\Bigg(
|\hat{F}_n^{l'+1}(\beta_1)-\hat{F}_n^{l'+1}(\beta_2)|_{0}\\
&\qquad\qquad+\beta_1^{-\frac{1}{2}}\sum_{t=0}^1\sum_{a=1}^2\|\hat{F}_n^{l'+1}(\beta_a)\|_{0,t}
+|T_n^{l'+1}(\beta_1)-T_n^{l'+1}(\beta_2)|_0\\
&\qquad\qquad+\beta_1^{-\frac{1}{2}}\sum_{t=0}^1\sum_{a=1}^2\|T_n^{l'+1}(\beta_a)\|_{0,t}\Bigg)\\
&
\quad+c\sum_{n=2N_v+2}^{N(\beta_2)}2^{2n}c_0^{\frac{n-m}{2}}\Bigg(|J_n^{l'+1}(\beta_1)-J_n^{l'+1}(\beta_2)|_0\\
&\qquad\qquad+\beta_1^{-\frac{1}{2}}\sum_{t=0}^1\sum_{a=1}^2
\|J_n^{l'+1}(\beta_a)\|_{0,t}\Bigg).
\end{align*} 
On the assumption  $\alpha\ge c$, $M\ge c^{N_v^2}$, insertion of 
\eqref{eq_UV_all_bound}, 
\eqref{eq_tree_sum},
\eqref{eq_artificial_free_part_bound},
\eqref{eq_UV_all_bound_difference},
\eqref{eq_tree_sum_difference},
\eqref{eq_artificial_free_part_bound_difference}
for $l'+1$ yields that
\begin{align*}
&\sum_{m=2}^{2N_v}c_0^{\frac{m}{2}}\alpha^m|\hat{F}_m^{l'}(\beta_1)-\hat{F}_m^{l'}(\beta_2)|_{0}\\
&\le c
 \sum_{m=2}^{2N_v}2^{2m}\sum_{n=m}^{2N_v}c_0^{\frac{n}{2}}\alpha^n
\Bigg(|\hat{F}_n^{l'+1}(\beta_1)-\hat{F}_n^{l'+1}(\beta_2)|_{0}
+|T_n^{l'+1}(\beta_1)-T_n^{l'+1}(\beta_2)|_{0}\\
&\qquad +\beta_1^{-\frac{1}{2}}\sum_{t=0}^1\sum_{a=1}^2\|\hat{F}_n^{l'+1}(\beta_a)\|_{0,t}
+\beta_1^{-\frac{1}{2}}\sum_{t=0}^1\sum_{a=1}^2\|T_n^{l'+1}(\beta_a)\|_{0,t}\Bigg)\\
&\quad +c\sum_{m=2}^{2N_v}2^{4N_v+4}\alpha^{m-2N_v-2}M^{-\frac{l'+1}{N_v-1}(N_v+1)}
\sum_{n=2N_v+2}^{N(\beta_2)}c_0^{\frac{n}{2}}\alpha^nM^{\frac{l'+1}{2N_v-2}n}\\
&\quad\qquad\cdot\left(
|J_n^{l'+1}(\beta_1)-J_n^{l'+1}(\beta_2)|_{0}
+\beta_1^{-\frac{1}{2}}
\sum_{t=0}^1\sum_{a=1}^2\|J_n^{l'+1}(\beta_a)\|_{0,t}
\right)\\
&\le  c^{N_v}\beta_1^{-\frac{1}{2}}\alpha^{-2}M^{-\frac{l'+1}{N_v-1}}\le\beta_1^{-\frac{1}{2}} \alpha^{-2}M^{-\frac{l'}{N_v-1}}.
\end{align*}
Therefore, the induction concludes that
 \eqref{eq_artificial_free_part_bound_difference}  holds for all
 $l'\in \{l,l+1,\cdots,N_h\}$.

We can see from \eqref{eq_artificial_free_part} and
 \eqref{eq_artificial_free_kernel} that for any $X_0\in I^0$, $\bX\in
 \hat{I}^{m-1}$, $a\in \{1,2\}$, $m\in\{2,3,\cdots,2N_v\}$,
\begin{align*}
&F_m^{l}(\beta_a)(X_0,R_{\beta_a}(\bX))\\
&=\hat{F}_m^{l}(\beta_a)(X_0,R_{\beta_a}(\bX))+
F_m^{N_h}(\beta_a)(X_0,R_{\beta_a}(\bX))\\
&\quad+\sum_{j=l+1}^{N_h}\sum_{n=m+2}^{2N_v}\left(\begin{array}{c}n
							\\
							     m\end{array}\right)
\left(\frac{1}{h}\right)^{n-m}\sum_{\bY\in
 (I^0)^{n-m}}\\
&\qquad\cdot(F_n^j(\beta_a)(X_0,R_{\beta_a}(\bX,\bY))-\hat{F}^j_n(\beta_a)(X_0,R_{\beta_a}(\bX,\bY)))\\
&\qquad\cdot \int\psi_{\bY}^1d\mu_{C_j^{\delta}(\beta_a)}(\psi^1).
\end{align*}
Then, by \eqref{eq_UV_covariance_determinant_equal_time},
 \eqref{eq_UV_covariance_determinant_difference}, 
\begin{align*}
&|F_m^{l}(\beta_1)(X_0,R_{\beta_1}(\bX))-F_m^{l}(\beta_2)(X_0,R_{\beta_2}(\bX))|\\
&\le |\hat{F}_m^{l}(\beta_1)(X_0,R_{\beta_1}(\bX))-\hat{F}_m^{l}(\beta_2)(X_0,R_{\beta_2}(\bX))|\\
&\quad+\sum_{j=l+1}^{N_h}\sum_{n=m+2}^{2N_v}\left(\begin{array}{c}n
							\\
							     m\end{array}\right)
\left(\frac{1}{h}\right)^{n-m}\sum_{\bY\in
 (I^0)^{n-m}}\\
&\qquad\cdot(|F_n^j(\beta_1)(X_0,R_{\beta_1}(\bX,\bY))-F_n^j(\beta_2)(X_0,R_{\beta_2}(\bX,\bY))|\\
&\qquad\quad +
|\hat{F}^j_n(\beta_1)(X_0,R_{\beta_1}(\bX,\bY))
-\hat{F}^j_n(\beta_2)(X_0,R_{\beta_2}(\bX,\bY))|)\\
&\qquad\cdot
(M^{-j}+M^{j-N_h})c_0^{\frac{n-m}{2}}\\
&\quad+\sum_{j=l+1}^{N_h}\sum_{n=m+2}^{2N_v}\left(\begin{array}{c}n
							\\
							     m\end{array}\right)
\left(\frac{1}{h}\right)^{n-m}\sum_{\bY\in
 (I^0)^{n-m}}\\
&\qquad\cdot(|F_n^j(\beta_2)(X_0,R_{\beta_2}(\bX,\bY))|
+|\hat{F}^j_n(\beta_2)(X_0,R_{\beta_2}(\bX,\bY))|)\\
&\qquad\cdot\beta_{1}^{-\frac{1}{2}}M^{-\frac{j}{2}}c_0^{\frac{n-m}{2}}.
\end{align*}
Thus,
\begin{align*}
&|F_m^{l}(\beta_1)-F_m^{l}(\beta_2)|_{0}\\
&\le |\hat{F}_m^l(\beta_1)-\hat{F}_m^l(\beta_2)|_{0}\\
&\quad +\sum_{j=l+1}^{N_h}(M^{-j}+M^{j-N_h})\\
&\qquad\cdot \sum_{n=m+2}^{2N_v}2^nc_0^{\frac{n-m}{2}}
(|F_n^j(\beta_1)-F_n^j(\beta_2)|_{0}+|\hat{F}_n^j(\beta_1)-\hat{F}_n^j(\beta_2)|_{0})\\
&\quad +\beta_1^{-\frac{1}{2}}\sum_{j=l+1}^{N_h}M^{-\frac{j}{2}}\sum_{n=m+2}^{2N_v}2^nc_0^{\frac{n-m}{2}}
(\|F_n^j(\beta_2)\|_{0,0}+\|\hat{F}_n^j(\beta_2)\|_{0,0}).
\end{align*}
Moreover, by substituting \eqref{eq_UV_bound}, 
\eqref{eq_artificial_free_part_bound},
\eqref{eq_UV_bound_difference} for $j\in
 \{l+1,l+2,\cdots,N_h\}$, 
\eqref{eq_artificial_free_part_bound_difference} for $j\in\{l,l+1,\cdots,N_h\}$ and using the condition $\alpha\ge c$
 we deduce that
\begin{align}
&\sum_{m=2}^{2N_v}c_0^{\frac{m}{2}}\alpha^m|F_m^{l}(\beta_1)-F_m^{l}(\beta_2)|_{0}\label{eq_UV_bound_free_preliminary_difference}\\
&\le
 \sum_{m=2}^{2N_v}c_0^{\frac{m}{2}}\alpha^m|\hat{F}_m^{l}(\beta_1)-\hat{F}_m^{l}(\beta_2)|_{0}
\notag\\
&\quad + \sum_{m=2}^{2N_v}\sum_{j=l+1}^{N_h}2^{m+2}\alpha^{-2}(M^{-j}+M^{j-N_h})\notag\\
&\qquad\cdot \sum_{n=m+2}^{2N_v}c_0^{\frac{n}{2}}\alpha^n
(|F_n^j(\beta_1)-F_n^j(\beta_2)|_0+|\hat{F}_n^j(\beta_1)-\hat{F}_n^j(\beta_2)|_0)\notag\\&\quad +\beta_1^{-\frac{1}{2}}
\sum_{m=2}^{2N_v}\sum_{j=l+1}^{N_h}2^{m+2}\alpha^{-2}M^{-\frac{j}{2}}\sum_{n=m+2}^{2N_v}c_0^{\frac{n}{2}}\alpha^n
(\|F_n^j(\beta_2)\|_{0,0}+\|\hat{F}_n^j(\beta_2)\|_{0,0})\notag\\
&\le \beta_1^{-\frac{1}{2}}c^{N_v}\alpha^{-2},\notag
\end{align}
which also implies that
\begin{align}
M^{-\frac{N_v}{N_v-1}l}\sum_{m=2}^{2N_v}c_0^{\frac{m}{2}}\alpha^mM^{\frac{l}{2N_v-2}m}
|F_m^{l}(\beta_1)-F_m^{l}(\beta_2)|_{0}\le\beta_1^{-\frac{1}{2}}c^{N_v}\alpha^{-2}.\label{eq_free_lower_half_bound_difference}
\end{align}

It follows from \eqref{eq_UV_covariance_determinant},
 \eqref{eq_UV_covariance_determinant_difference} and 
\cite[\mbox{Lemma 4.1 (2)}]{K15} that
for $m\in\{2N_v+2,2N_v+3,\cdots,N(\beta_2) \}$,
\begin{align*}
&|F_m^l(\beta_1)-F_m^l(\beta_2)|_{0}\\
&\le |J_m^{l+1}(\beta_1)-J_m^{l+1}(\beta_2)|_0\\
&\quad +c \sum_{n=m+2}^{N(\beta_2)}
2^{2n}c_0^{\frac{n-m}{2}}\Bigg(|J_n^{l+1}(\beta_1)-J_n^{l+1}(\beta_2)|_{0}
+\beta_1^{-\frac{1}{2}}\sum_{a=1}^2\sum_{t=0}^1\|J_n^{l+1}(\beta_a)\|_{0,t}\Bigg).
\end{align*}
Then, by \eqref{eq_UV_all_bound}, \eqref{eq_UV_all_bound_difference} for
 $l+1$ and the conditions $M\ge c^{N_v}$, $\alpha\ge c$, 
\begin{align*}
&\sum_{m=2N_v+2}^{N(\beta_2)}c_0^{\frac{m}{2}}\alpha^mM^{\frac{l}{2N_v-2}m}
|F_m^l(\beta_1)-F_m^l(\beta_2)|_{0}\\
&\le c 
 \sum_{n=2N_v+2}^{N(\beta_2)}\sum_{m=2N_v+2}^n2^{2n}c_0^{\frac{n}{2}}\alpha^mM^{\frac{l}{2N_v-2}m}\\
&\quad\cdot \left(|J_n^{l+1}(\beta_1)-J_n^{l+1}(\beta_2)|_{0}
+\beta_1^{-\frac{1}{2}}\sum_{a=1}^2\sum_{t=0}^1\|J_n^{l+1}(\beta_a)\|_{0,t}\right)\\
&\le c 
 \sum_{n=2N_v+2}^{N(\beta_2)}2^{2n}c_0^{\frac{n}{2}}\alpha^nM^{\frac{l}{2N_v-2}n}
\\
&\quad\cdot \left(|J_n^{l+1}(\beta_1)-J_n^{l+1}(\beta_2)|_{0}
+\beta_1^{-\frac{1}{2}}\sum_{a=1}^2\sum_{t=0}^1\|J_n^{l+1}(\beta_a)\|_{0,t}\right)\\
&\le
 c^{N_v}M^{-\frac{N_v+1}{N_v-1}+\frac{N_v}{N_v-1}(l+1)}\beta_1^{-\frac{1}{2}},
\end{align*}
and thus
\begin{align}
M^{-\frac{N_v}{N_v-1}l}
\sum_{m=2N_v+2}^{N(\beta_2)}c_0^{\frac{m}{2}}\alpha^mM^{\frac{l}{2N_v-2}m}
|F_m^l(\beta_1)-F_m^l(\beta_2)|_{0}
\le
 c^{N_v}M^{-\frac{1}{N_v-1}}\beta_1^{-\frac{1}{2}}.\label{eq_free_upper_half_bound_difference}
\end{align}

Finally, let us estimate the difference between $F_0^l(\beta_1)$ and
 $F_0^l(\beta_2)$. Note that
\begin{align}
&F_0^l(\beta_a)\label{eq_free_0th_temperature_dependent_decomposition}\\
&=F_0^{l+1}(\beta_a)+T_0^{l+1}(\beta_a)+\sum_{m=2}^{2N_v}\int
 (F_m^{l+1}(\beta_a)(\psi)-\hat{F}_m^{l+1}(\beta_a)(\psi))d\mu_{C_{l+1}^{\delta}(\beta_a)}(\psi)\notag
\\
&\quad+\sum_{m=2}^{2N_v}\int
 \hat{F}_m^{l+1}(\beta_a)(\psi)d\mu_{C_{l+1}^{\delta}(\beta_a)}(\psi)
+\sum_{m=2}^{2N_v}\int
 T_m^{l+1}(\beta_a)(\psi)d\mu_{C_{l+1}^{\delta}(\beta_a)}(\psi)\notag\\
&\quad+\sum_{m=2N_v+2}^{N(\beta_2)}\int
 J_m^{l+1}(\beta_a)(\psi)d\mu_{C_{l+1}^{\delta}(\beta_a)}(\psi).\notag
\end{align}
By \eqref{eq_artificial_free_part}, 
Lemma \ref{lem_UV_polynomial_temperature_invariance},
\eqref{eq_artificial_free_invariance},
\eqref{eq_Grassmann_integral_invariance_specific},
\begin{align*}
&\int
 (F_m^{l+1}(\beta_a)(\psi)-\hat{F}_m^{l+1}(\beta_a)(\psi))d\mu_{C_{l+1}^{\delta}(\beta_a)}(\psi)\\
&=\left(\frac{1}{h}\right)^m\sum_{s\in [0,\beta_a)_h}\sum_{X\in
 I^0}\sum_{\bY\in
 I(\beta_a)^{m-1}}(F_m^{l+1}(\beta_a)(X,R_{\beta_a}(\bY-s))\\
&\qquad\qquad\qquad -\hat{F}_m^{l+1}(\beta_a)(X,R_{\beta_a}(\bY-s)))\int
 \psi_{X}\psi_{R_{\beta_a}(\bY-s)}d\mu_{C_{l+1}^{\delta}(\beta_a)}(\psi)\\
&=\beta_a\left(\frac{1}{h}\right)^{m-1}\sum_{X\in
 I^0}\sum_{\bY\in
 (I^0)^{m-1}}(F_m^{l+1}(\beta_a)(X,\bY)-\hat{F}_m^{l+1}(\beta_a)(X,\bY))\\
&\qquad\cdot \int
 \psi_{X}\psi_{\bY}d\mu_{C_{l+1}^{\delta}(\beta_a)}(\psi).
\end{align*}
Combined with \eqref{eq_UV_covariance_determinant_equal_time} and
 \eqref{eq_UV_covariance_determinant_difference}, this equality implies
 that
\begin{align}
&\Bigg|\frac{h}{N(\beta_1)}\int
 (F_m^{l+1}(\beta_1)(\psi)-\hat{F}_m^{l+1}(\beta_1)(\psi))d\mu_{C_{l+1}^{\delta}(\beta_1)}(\psi)\label{eq_free_0th_special_part_difference}\\
&-\frac{h}{N(\beta_2)}\int
 (F_m^{l+1}(\beta_2)(\psi)-\hat{F}_m^{l+1}(\beta_2)(\psi))d\mu_{C_{l+1}^{\delta}(\beta_2)}(\psi)\Bigg|\notag\\
&\le
 (|F_m^{l+1}(\beta_1)-F_m^{l+1}(\beta_2)|_0+|\hat{F}_m^{l+1}(\beta_1)-\hat{F}_m^{l+1}(\beta_2)|_0)\notag\\
&\quad\cdot(M^{-l-1}+M^{l+1-N_h})c_0^{\frac{m}{2}}\notag\\
&\quad +
 (\|F_m^{l+1}(\beta_2)\|_{0,0}+\|\hat{F}_m^{l+1}(\beta_2)\|_{0,0})M^{-\frac{l+1}{2}}\beta_1^{-\frac{1}{2}}c_0^{\frac{m}{2}}.\notag
\end{align}
By applying the same estimation as  
\cite[\mbox{Lemma 4.1 (1)}]{K15} to
 \eqref{eq_free_0th_temperature_dependent_decomposition} and using  
\eqref{eq_UV_covariance_determinant}, 
\eqref{eq_UV_covariance_determinant_difference},
\eqref{eq_UV_bound},
\eqref{eq_UV_all_bound},
\eqref{eq_tree_sum},
\eqref{eq_artificial_free_part_bound},
\eqref{eq_UV_bound_difference},
\eqref{eq_UV_all_bound_difference},
\eqref{eq_tree_0th_difference},
\eqref{eq_tree_sum_difference},
\eqref{eq_artificial_free_part_bound_difference},  
\eqref{eq_free_0th_special_part_difference}
for $l'\in\{l+1,l+2,\cdots,N_h\}$ we observe that
\begin{align}
&\left|\frac{h}{N(\beta_1)}F_0^l(\beta_1)-\frac{h}{N(\beta_2)}F_0^l(\beta_2)\right|\label{eq_free_0th_difference}\\
&\le
 \left|\frac{h}{N(\beta_1)}F_0^{l+1}(\beta_1)-\frac{h}{N(\beta_2)}F_0^{l+1}(\beta_2)\right|\notag\\
&\quad + \left|\frac{h}{N(\beta_1)}T_0^{l+1}(\beta_1)-\frac{h}{N(\beta_2)}T_0^{l+1}(\beta_2)\right|\notag\\
&\quad
 +\sum_{m=2}^{2N_v}(M^{-l-1}+M^{l+1-N_h})c_0^{\frac{m}{2}}\notag\\
&\qquad\qquad\cdot (|F_m^{l+1}(\beta_1)-F_m^{l+1}(\beta_2)|_0+|\hat{F}_m^{l+1}(\beta_1)-\hat{F}_m^{l+1}(\beta_2)|_0)\notag\\
&\quad +\sum_{m=2}^{2N_v}M^{-\frac{l+1}{2}}\beta_1^{-\frac{1}{2}}c_0^{\frac{m}{2}}(\|\hat{F}_m^{l+1}(\beta_2)\|_{0,0}+\|{F}_m^{l+1}(\beta_2)\|_{0,0})\notag\\
&\quad + c\sum_{m=2}^{2N_v}2^m
 c_0^{\frac{m}{2}}\Bigg(|\hat{F}_m^{l+1}(\beta_1)-\hat{F}_m^{l+1}(\beta_2)|_0
+\beta_1^{-\frac{1}{2}}\sum_{a=1}^2\sum_{t=0}^1\|\hat{F}_m^{l+1}(\beta_a)\|_{0,t}\notag\\
&\qquad\qquad\qquad\qquad+|T_m^{l+1}(\beta_1)-T_m^{l+1}(\beta_2)|_0
+\beta_1^{-\frac{1}{2}}\sum_{a=1}^2\sum_{t=0}^1\|{T}_m^{l+1}(\beta_a)\|_{0,t}\Bigg)\notag\\
&\quad + c\sum_{m=2N_v+2}^{N(\beta_2)}2^m
 c_0^{\frac{m}{2}}\Bigg(|{J}_m^{l+1}(\beta_1)-{J}_m^{l+1}(\beta_2)|_0
+\beta_1^{-\frac{1}{2}}\sum_{a=1}^2\sum_{t=0}^1\|{J}_m^{l+1}(\beta_a)\|_{0,t}
\Bigg)\notag\\
&\le
 \left|\frac{h}{N(\beta_1)}F_0^{l+1}(\beta_1)-\frac{h}{N(\beta_2)}F_0^{l+1}(\beta_2)\right|\notag\\
&\quad + c^{N_v}\beta_1^{-\frac{1}{2}}\alpha^{-2}(M^{-\frac{l+1}{2}}+M^{-\frac{l+1}{N_v-1}}+M^{l+1-N_h})\notag\\
&\le  c^{N_v}\beta_1^{-\frac{1}{2}}\alpha^{-2}\sum_{j=l}^{N_h-1}(M^{-\frac{j+1}{2}}+M^{-\frac{j+1}{N_v-1}}+M^{j+1-N_h})\notag\\
&\le c^{N_v}\beta_1^{-\frac{1}{2}}\alpha^{-2}.\notag
\end{align}

By putting 
\eqref{eq_tree_0th_difference}, 
\eqref{eq_tree_sum_difference},
\eqref{eq_tree_sum_weight_difference},
\eqref{eq_UV_bound_free_preliminary_difference},
\eqref{eq_free_lower_half_bound_difference},
\eqref{eq_free_upper_half_bound_difference},
\eqref{eq_free_0th_difference} together we obtain that
\begin{align}
&\left|\frac{h}{N(\beta_1)}F^l_0(\beta_1)-\frac{h}{N(\beta_2)}F^l_0(\beta_2)\right|
+\left|\frac{h}{N(\beta_1)}T^l_0(\beta_1)-\frac{h}{N(\beta_2)}T^l_0(\beta_2)\right|\label{eq_UV_0th_bound_difference_induction}\\
&\le c^{N_v}\alpha^{-2}\beta_1^{-\frac{1}{2}},\notag\\
&\sum_{m=2}^{2N_v}c_0^{\frac{m}{2}}\alpha^m(|F_m^l(\beta_1)-F_m^l(\beta_2)|_0+|T_m^l(\beta_1)-T_m^l(\beta_2)|_0)\le
 c^{N_v}\alpha^{-2} \beta_1^{-\frac{1}{2}},\label{eq_UV_bound_difference_induction}\\
&M^{-\frac{N_v}{N_v-1}l}\sum_{m=2}^{N(\beta_2)}c_0^{\frac{m}{2}}\alpha^mM^{\frac{l}{2N_v-2}m}
(|F_m^l(\beta_1)-F_m^l(\beta_2)|_{0}+|T_m^l(\beta_1)-T_m^l(\beta_2)|_{0})\label{eq_UV_all_bound_difference_induction}\\
&\le (c^{N_v}\alpha^{-2}+cM\alpha^{-2}+c^{N_v}M^{-\frac{1}{N_v-1}})
\beta_1^{-\frac{1}{2}}.\notag
\end{align}
Recall that in the derivation of the above inequalities we assumed the
 conditions $\alpha\ge c$, $M\ge c^{N_v^2}$ with a constant $c\in
 \R_{>0}$ independent of any parameter including $c'$. If we start by the
 conditions \eqref{eq_UV_parameter_conditions} and
 \eqref{eq_interaction_amplitude_UV} with a sufficiently large constant
 $c'$, then the conditions $\alpha\ge c$, $M\ge c^{N_v^2}$ are
 satisfied,
the right-hand side of \eqref{eq_UV_0th_bound_difference_induction} is
 less than $\beta_1^{-1/2}\alpha^{-1}$ and the right-hand sides of
 \eqref{eq_UV_bound_difference_induction},
 \eqref{eq_UV_all_bound_difference_induction} are less than $\beta_1^{-1/2}$.
Therefore, we obtain
 \eqref{eq_UV_0th_bound_difference}, \eqref{eq_UV_bound_difference},
 \eqref{eq_UV_all_bound_difference} for $l$ on the assumptions
 \eqref{eq_UV_parameter_conditions}, \eqref{eq_interaction_amplitude_UV} with a generic constant $c'$, which
 does not depend on any parameter. The induction with $l\in
 \{0,1,\cdots,N_h\}$ now concludes the proof.
\end{proof}

\section{The infrared integration}\label{sec_IR}

In this section we perform the multi-scale analysis around the singular
point of the covariance in momentum space, namely the infrared
analysis. The output of the Matsubara ultra-violet integration is
substituted into the infrared integration as the initial data. So the
infrared integration is the second step of the whole multi-scale
integration process. Conservation of symmetries is essential to validate
the iteration of the integration. We have to keep track of the preserved
 symmetries as well as analyticity and scale-dependent bound
properties of Grassmann polynomials during the iteration. For this
purpose it is convenient to organize sets of Grassmann polynomials
having the relevant properties and define maps between these sets
resembling the real renormalization group maps in advance. We plan to do
so in the first subsection. In the second subsection we will complete
the proof of Theorem \ref{thm_main_theorem} by making use of the tools
developed in the preceding subsection. We should remark that in
principle one can reach our main result by combining the materials
prepared so far in this paper with calculations parallel to those
presented in \cite[\mbox{Section 7}]{K15}. Apart from proving the
theorem itself, this section is aimed at providing a more organized
construction of the infrared integration than the previous version
\cite[\mbox{Section 7}]{K15} so that the readers can confirm the validity of
the infrared integration more clearly.

Throughout this section we assume that 
$$
M\ge 2,\quad h\ge e^{4d},\quad L\ge \beta,
$$
unless stated otherwise.

\subsection{General lemmas}\label{subsec_general_lemma}
Let $n\in\N$ and let $D$ be a bounded domain of $\C^n$ satisfying that
$\overline{\bz}\in \overline{D}$ for any $\bz\in\overline{D}$, where
$\overline{D}$ denotes the closure of $D$. Set
\begin{align*}
&C\left(\overline{D};\bigwedge\cV\right)\\
&:=\left\{J\in\Map\left(\overline{D},\bigwedge\cV\right)\
 \big|\ \bU\mapsto J(\bU)(\psi)\text{ is continuous in }\overline{D}\right\},\\
&C^{\o}\left(D;\bigwedge\cV\right)\\
&:=\left\{J\in\Map\left(D,\bigwedge\cV\right)\
 \big|\ \bU\mapsto J(\bU)(\psi)\text{ is analytic in }D\right\}.
\end{align*}
See \cite[\mbox{Subsection 2.2}]{K15} for the meaning of continuity and
analyticity of Grassmann polynomials. We are going to define a subset of 
$C(\overline{D};\bigwedge\cV)\cap
C^{\o}(D;\bigwedge\cV)$ to which Grassmann polynomials dealt
in our infrared analysis belong. To describe symmetric properties of 
Grassmann polynomials, let us fix some
notational conventions. Let $S$ be a bijective map from $I$ to $I$ and $Q$ be a map
from $I$ to $\R$. The maps $S_m:I^m\to I^m$, $Q_m:I^m\to\R$ $(m\in\N)$
are defined by 
\begin{align*}
&S_m(X_1,X_2,\cdots,X_m):=(S(X_1),S(X_2),\cdots,S(X_m)), \\
&Q_m(X_1,X_2,\cdots,X_m):=\sum_{j=1}^mQ(X_j).
\end{align*}
For $f(\psi)=\sum_{m=0}^N(\frac{1}{h})^m\sum_{\bX\in
I^m}f_m(\bX)\psi_{\bX}\in \bigwedge\cV$, define $f(\cR \psi),\overline{f}(\psi)\in\bigwedge\cV$ by
\begin{align*}
&f(\cR\psi):=\sum_{m=0}^{N}\left(\frac{1}{h}\right)^m\sum_{\bX\in
 I^m}f_m(\bX)e^{iQ_m(S_m(\bX))}\psi_{S_m(\bX)},\\
&\overline{f}(\psi):=\sum_{m=0}^{N}\left(\frac{1}{h}\right)^m\sum_{\bX\in
 I^m}\overline{f_m(\bX)}\psi_{\bX}.
\end{align*}
In fact these notational rules have been introduced in \cite[\mbox{Subsection 3.3}]{K15}.
In addition, for $\bx\in\Z^d$ we let $r_L(\bx)$ denote a site of $\G(L)$
satisfying $\bx=r_L(\bx)$ in $(\Z/L\Z)^d$. 

Now, for parameters $c_0$, $\alpha\in\R_{\ge 1}$, $M\in\R_{\ge 2}$, $l\in\Z_{\le 0}$
we define the subset $\cS(D,c_0,\alpha,M)(l)$ of $C(\overline{D};\bigwedge\cV)\cap
C^{\o}(D;\bigwedge\cV)$ as follows. $J\in
C(\overline{D};\bigwedge\cV)$ $\cap
C^{\o}(D;\bigwedge\cV)$ belongs to $\cS(D,c_0,\alpha,M)(l)$
if and only if $J$ satisfies the following properties.
\begin{enumerate}[(i)]
\item\label{item_IR_boundedness}
\begin{align}
&\frac{h}{N}|J_0(\bU)|\le
 M^{(d+\frac{3}{2})l}\alpha^{-1},\label{eq_IR_boundedness_0th}\\
&M^{-(d+\frac{3}{2})l+tl}\sum_{m=2}^Nc_{0}^{\frac{m}{2}}\alpha^mM^{\frac{d}{2}lm}\|J_m(\bU)\|_{l,t}\le 1,\ (\forall \bU\in\overline{D},t\in \{0,1\}).\label{eq_IR_boundedness}
\end{align}
\item\label{item_IR_invariances}
$$
J(\bU)(\psi)=J(\bU)(\cR\psi),\ (\forall \bU\in \overline{D}),
$$
for each $S:I\to I$ and $Q:I\to \R$ defined as follows.
\begin{align}
&S((\rho,\bx,\s,x,\theta)):=(\rho,\bx,\s,x,\theta),\label{eq_IR_particle_hole}\\
&Q((\rho,\bx,\s,x,\theta)):=\frac{\pi}{2}\theta,\ (\forall
 (\rho,\bx,\s,x,\theta)\in I).\notag
\end{align}
\begin{align}
&S((\rho,\bx,\s,x,\theta)):=(\rho,\bx,\s,x,\theta),\label{eq_IR_spin_up}\\
&Q((\rho,\bx,\s,x,\theta)):= \pi 1_{\s=\ua },\ (\forall
 (\rho,\bx,\s,x,\theta)\in I).\notag
\end{align}
\begin{align}
&S((\rho,\bx,\s,x,\theta)):=(\rho,\bx,-\s,x,\theta),\label{eq_IR_spin_inversion}\\
&Q((\rho,\bx,\s,x,\theta)):= 0,\ (\forall
 (\rho,\bx,\s,x,\theta)\in I).\notag
\end{align}
\begin{align}
&S((\rho,\bx,\s,x,\theta)):=(\rho,r_L(\bx+\bz),\s,r_{\beta}(x+s),\theta),\label{eq_IR_translation}\\
&Q((\rho,\bx,\s,x,\theta)):= \pi n_{\beta}(r_{\beta}(x-s)+s),\ (\forall
 (\rho,\bx,\s,x,\theta)\in I),\notag
\end{align}
where $\bz\in \Z^d$ and $s\in (1/h)\Z$ are arbitrarily taken and fixed.
\begin{align}
&S((\rho,\bx,\s,x,\theta)):=(\rho,r_L(-\bx-b(\rho)),\s,x,\theta),\label{eq_IR_momentum_inversion}\\
&Q((\rho,\bx,\s,x,\theta)):=\theta\<\bx,(2\pi/L)\beps^L\>+\theta\<b(\rho),(\pi/L)\beps^L\>,\notag\\
&(\forall
 (\rho,\bx,\s,x,\theta)\in I).\notag
\end{align}
\item\label{item_IR_invariances_complex}
$$
J(\bU)(\psi)=\overline{J(\overline{\bU})}(\cR\psi),\ (\forall \bU\in \overline{D}),
$$
for each $S:I\to I$ and $Q:I\to \R$ defined as follows.
\begin{align}
&S((\rho,\bx,\s,x,\theta)):=(\rho,\bx,\s,r_{\beta}(-x),-\theta),\label{eq_IR_adjoint}\\
&Q((\rho,\bx,\s,x,\theta)):= \pi (1_{\theta=1}+1_{x\neq 0}),\ (\forall
 (\rho,\bx,\s,x,\theta)\in I).\notag
\end{align}
\begin{align}
&S((\rho,\bx,\s,x,\theta)):=(\rho,\bx,\s,x,-\theta),\label{eq_IR_inversion}\\
&Q((\rho,\bx,\s,x,\theta)):= \<b(\rho),\bpi\>,\ (\forall
 (\rho,\bx,\s,x,\theta)\in I).\notag
\end{align}
\end{enumerate}

Moreover, on the assumption \eqref{eq_basic_beta_h_assumption} we define
the subset $\hat{\cS}(D,c_0,\alpha,M)(l)$ of
$\cS(D,c_0,\alpha,M)(l)(\beta_1)\times \cS(D,c_0,\alpha,M)(l)(\beta_2)$
as follows. $(J(\beta_1),J(\beta_2))$ $\in
\cS(D,c_0,\alpha,M)(l)(\beta_1)\times \cS(D,c_0,\alpha,M)(l)(\beta_2)$
belongs to $\hat{\cS}(D,c_0,\alpha,M)(l)$ if and only if 
\begin{align}
&\left|\frac{h}{N(\beta_1)}J_0(\beta_1)(\bU)-\frac{h}{N(\beta_2)}J_0(\beta_2)(\bU)
\right|\le \beta_1^{-\frac{1}{2}}
 M^{(d+\frac{1}{2})l}\alpha^{-1},\label{eq_IR_boundedness_difference_0th}\\
&M^{-(d+\frac{1}{2})l}\sum_{m=2}^{N(\beta_2)}c_{0}^{\frac{m}{2}}\alpha^mM^{\frac{d}{2}lm}|J_m(\beta_1)(\bU)-J_m(\beta_2)(\bU)|_{l}\le \beta_{1}^{-\frac{1}{2}},\ (\forall \bU\in\overline{D}).\label{eq_IR_boundedness_difference}
\end{align}

We will later define a set designed to contain kernels of quadratic
Grassmann polynomials belonging to $\cS(D,c_0,\alpha,M)(l)$. Since one
criterion to be an element of the set involves cut-off functions for
the infrared integration, let us define the cut-off functions at this stage.
Set   
\begin{align*}
&f_{\bt}:=\\
&\frac{1}{4}\min_{p\in\{1,2,\cdots,d\}}t_p^2\left(1- 
\frac{1}{2}\max_{m\in\{1,2,\cdots,d\}}\left(\sum_{j=1}^{m-1}|1+e^{i\theta_{j,m}}|+\sum_{j=m+1}^d|1+e^{i\theta_{m,j}}|\right)\right),\\
&M_{IR}:=\frac{\sqrt{6}}{\pi}\left(\frac{\pi^2}{3}M_{UV}^2+d\right)^{\frac{1}{2}},\\
&N_{\beta}:=\min\Bigg\{\Bigg\lfloor\frac{\log\big(\frac{\pi}{\beta}\big(\frac{\pi}{\sqrt{3}}M_{IR}\big)^{-1}\big)}{\log
M}\Bigg\rfloor,0\Bigg\}.
\end{align*}
By the assumption \eqref{eq_hopping_amplitude_condition}, $f_{\bt}\le 1/4$.
Using the function $\phi$ introduced in Subsection
\ref{subsec_UV_covariance}, we define the functions $\chi_l:\R^{d+1}\to\R$
$(l\in \Z)$ by
\begin{align*}
&\chi_l(\o,\bk)\\
&:=\phi(M_{UV}^{-2}\o^2)\Bigg(\phi\left(M_{IR}^{-2}M^{-2(l+1)}\left(\o^2+f_{\bt}\sum_{j=1}^d|1+e^{i\frac{\pi}{L}\eps_j^L+ik_j}|^2\right)\right)\\
&\qquad\qquad\qquad\quad -\phi\left(M_{IR}^{-2}M^{-2l}\left(\o^2+f_{\bt}\sum_{j=1}^d|1+e^{i\frac{\pi}{L}\eps_j^L+ik_j}|^2\right)\right)\Bigg).
\end{align*}
We can check that 
\begin{align}
&\phi\left(M_{IR}^{-2}M^{-2}\left(\o^2+f_{\bt}\sum_{j=1}^d
|1+e^{i\frac{\pi}{L}\eps_j^L+ik_j}|^2\right)\right)=1,\notag\\
&(\forall \o\in
 \R\text{ with }\phi(M_{UV}^{-2}\o^2)\neq
 0,\bk\in\R^d),\notag\\
&\phi\left(M_{IR}^{-2}M^{-2N_{\beta}}\left(\o^2+f_{\bt}\sum_{j=1}^d
|1+e^{i\frac{\pi}{L}\eps_j^L+ik_j}|^2\right)\right)=0,\label{eq_IR_cut_off_vanish}\\
&(\forall \o\in
 \R\text{ with }|\o|\ge \pi/\beta,\bk\in\R^d).\notag
\end{align}
These equalities imply that
\begin{align*}
\sum_{l=0}^{N_{\beta}}\chi_l(\o,\bk)=\phi(M_{UV}^{-2}\o^2),\ (\forall
\o\in \cM,\bk\in \R^d).
\end{align*}
The support property of $\chi_l$ is described as follows.
\begin{align}
&\chi_l(\o,\bk)\left\{\begin{array}{ll}=0,&\text{if
	       }\left(\o^2+f_{\bt}\sum_{j=1}^d|1+e^{i\frac{\pi}{L}\eps_j^L+ik_j}|^2\right)^{\frac{1}{2}}\le
	       \frac{\pi}{\sqrt{6}}M_{IR}M^l,\\
\in [0,1],&\text{if
	       } \frac{\pi}{\sqrt{6}}M_{IR}M^l<
	       \left(\o^2+f_{\bt}\sum_{j=1}^d|1+e^{i\frac{\pi}{L}\eps_j^L+ik_j}|^2\right)^{\frac{1}{2}}\\
&\qquad\qquad\quad\ < \frac{\pi}{\sqrt{3}}M_{IR}M^{l+1},\\
=0,&\text{if
	       }\left(\o^2+f_{\bt}\sum_{j=1}^d|1+e^{i\frac{\pi}{L}\eps_j^L+ik_j}|^2\right)^{\frac{1}{2}}\ge
		     \frac{\pi}{\sqrt{3}}M_{IR}M^{l+1}.\end{array}
\right.\label{eq_IR_cut_off_description}
\end{align}
We define the functions $\chi_{\le l}:\R^{d+1}\to\R$ ($l\in\Z$ with
$l\ge N_{\beta}$), $\hat{\chi}_{\le m}:\R^{d+1}\to\R$ ($m\in\Z$) by 
\begin{align*}
&\chi_{\le l}(\o,\bk):=\sum_{j=l}^{N_{\beta}}\chi_j(\o,\bk),\\
&\hat{\chi}_{\le
 m}(\o,\bk)\\
&:=\phi(M_{UV}^{-2}\o^2)\phi\left(M_{IR}^{-2}M^{-2(m+1)}\left(\o^2+f_{\bt}\sum_{j=1}^d|1+e^{i\frac{\pi}{L}\eps_j^L+ik_j}|^2
\right)\right),\\
&(\forall (\o,\bk)\in\R^{d+1}).
\end{align*}
Here let us list properties of these cut-off functions for later
use. For simplicity we write $\partial/\partial k_0$ in place of the
differential operator $\partial/\partial \o$ in the following.
Note that the condition $L\ge \beta$ is necessary in the proof of the
item \eqref{item_IR_cut_off_support} of the next lemma.

\begin{lemma}\label{lem_IR_cut_off_properties}
\begin{enumerate}
\item\label{item_IR_cut_off_equivalence}
Assume that $0<\beta_1\le \beta_2$. Then,
\begin{align*}
&\chi_{\le l}(\beta_1)(\o,\bk)=\chi_{\le
 l}(\beta_2)(\o,\bk)=\hat{\chi}_{\le l}(\o,\bk),\\
&(\forall (\o,\bk)\in\R^{d+1}\text{ with }|\o|\ge \pi/\beta_1,\
 l\in\Z\text{ with }l\ge N_{\beta_1}).
\end{align*}
\item\label{item_IR_cut_off_large_matsubara}
\begin{align*}
\hat{\chi}_{\le l}(\o,\bk)=0,\ (\forall (\o,\bk)\in\R^{d+1}\text{ with
 }|\o|\ge \pi h,\ l\in\Z).
\end{align*}
\item\label{item_IR_cut_off_derivative}
There exists a constant $c_{\chi}\in\R_{>0}$ independent of any
     parameter such that
\begin{align*}
&\left|\left(\frac{\partial}{\partial k_j}\right)^n\hat{\chi}_{\le
 l}(\o,\bk)\right|,\ \left|\left(\frac{\partial}{\partial
 k_j}\right)^n\chi_{ l}(\o,\bk)\right|\le (c_{\chi}
 \fw(l)^{-1})^n(n!)^2,\\
&(\forall (\o,\bk)\in\R^{d+1},n\in \N\cup
 \{0\},j\in\{0,1,\cdots,d\},l\in\{0,-1,\cdots,N_{\beta}\}).
\end{align*}
\item\label{item_IR_cut_off_implication}
If there exists $(\o,\bk)\in\R^{d+1}$ with $|\o|\ge \pi/\beta$ such that 
$\hat{\chi}_{\le 0}(\o,\bk)\neq 0$, then 
$$
\frac{1}{\beta}\le
     M_{IR}M^{N_{\beta}+1}.$$
\item\label{item_IR_cut_off_support}
Assume that $1/\beta\le M_{IR}M^{N_{\beta}+1}$. Then, there exists a
     constant $c(M,d)\in\R_{>0}$ depending only on $M$, $d$ such that
the following inequalities hold for any $l\in\Z$ with $l\ge N_{\beta}$,
     $(\o',\bk')\in\R^{d+1}$.
\begin{align*}
&\frac{1}{\beta L^d}\sum_{(\o,\bk)\in \cM\times\G(L)^*}1_{\hat{\chi}_{\le
     l}(\o+\o',\bk+\bk')\neq 0}\le c(M,d)
 f_{\bt}^{-\frac{d}{2}}M^{(d+1)l},\\
&\int_{-\infty}^{\infty}d\o
\frac{1}{L^d}\sum_{\bk\in\G(L)^*}1_{\hat{\chi}_{\le
     l}(\o,\bk+\bk')\neq 0}\le c(M,d)
 f_{\bt}^{-\frac{d}{2}}M^{(d+1)l},\\
&\frac{1}{L^d}\sum_{\bk\in\G(L)^*}1_{\hat{\chi}_{\le
     l}(\o',\bk+\bk')\neq 0}\le c(M,d) f_{\bt}^{-\frac{d}{2}}M^{dl}.
\end{align*}
\end{enumerate}
\end{lemma}
\begin{proof}
\eqref{item_IR_cut_off_equivalence}: The claim follows from
 \eqref{eq_IR_cut_off_vanish}.

\eqref{item_IR_cut_off_large_matsubara}: The assumption $h\ge e^{4d}$
 and the support property of $\phi(M_{UV}^{-2}\o^2)$ imply that
 result.

\eqref{item_IR_cut_off_derivative}: The proof for this claim is
 essentially same as the proof for \cite[\mbox{Lemma 7.4}]{K15}.
Here we especially need to use the fact $f_{\bt}\le 1$.

\eqref{item_IR_cut_off_implication}: This was essentially proved in \cite[\mbox{Lemma 7.5}]{K15}.

\eqref{item_IR_cut_off_support}: We can derive the claimed inequalities
 from the support property of $\hat{\chi}_{\le l}$
and the assumptions $1/\beta\le M_{IR}M^{N_{\beta}+1}$, $L\ge \beta$.
\end{proof}

Set 
\begin{align*}
&C^{\infty}(\R^{d+1};\Mat(2^d,\C))\\
&:=\{f:\R^{d+1}\to \Mat(2^d,\C)\ |\ f(\cdot)(\rho,\eta)\in
 C^{\infty}(\R^{d+1};\C)\ (\forall \rho,\eta\in\cB)\}.
\end{align*}
Here we define the subset $\cK(D,\alpha,M)(l)$ of
$\Map(\overline{D},C^{\infty}(\R^{d+1};\Mat(2^d,\C)))$ which is designed to
contain kernels of relevant quadratic Grassmann polynomials.
Let $l\in \Z_{\le 0}$ and $c_{\chi}$ be the constant appearing in Lemma
\ref{lem_IR_cut_off_properties} \eqref{item_IR_cut_off_derivative}. 
$W\in \Map(\overline{D},C^{\infty}(\R^{d+1};\Mat(2^d,\C)))$ belongs to
$\cK(D,\alpha,M)(l)$ if and only if $W$ satisfies the following
conditions.
\begin{enumerate}[(i)]
\item\label{item_kernel_analyticity}
$\bU\mapsto W(\bU)(\o,\bk)(\rho,\eta)$ is continuous in $\overline{D}$,
     analytic in $D$ for any $(\o,\bk)\in\R^{d+1}$, $\rho,\eta\in\cB$.
\item\label{item_kernel_periodicity}
\begin{align*}
&W(\bU)(\o,\bk)=W(\bU)(\o,\bp),\\
&(\forall \bU\in \overline{D},\o\in\R,\bk,\bp\in\R^d\text{ with
 }\bk=\bp\text{ in }(\R/2\pi \Z)^d).
\end{align*}
\item\label{item_kernel_hermiticity}
\begin{align*}
&W(\bU)(\o,\bk)=W(\overline{\bU})(-\o,\bk)^*,\\
&(\forall \bU\in \overline{D},(\o,\bk)\in\cM\times((2\pi/L)\Z)^d).
\end{align*}
\item\label{item_kernel_inversion}
\begin{align*}
&U_d(\bpi)W(\bU)(\o,\bk)U_d(\bpi)^*=-W(\overline{\bU})(\o,\bk)^*,\\
&(\forall \bU\in \overline{D}, (\o,\bk)\in\cM\times((2\pi/L)\Z)^d).
\end{align*}
\item\label{item_kernel_momentum_inversion}
\begin{align*}
&U_d\left(\frac{\pi}{L}\beps^L\right)
U_d(\bk)W(\bU)\left(\o,-\bk-\frac{2\pi}{L}\beps^L\right)U_d(\bk)^*
U_d\left(\frac{\pi}{L}\beps^L\right)^*\\
&=W(\bU)(\o,\bk),\ (\forall \bU\in \overline{D},(\o,\bk)\in\cM\times((2\pi/L)\Z)^d).
\end{align*}
\item\label{item_kernel_derivative}
\begin{align}
&\prod_{j=0}^d\left(\sum_{n_j=0}^{\infty}\left(\frac{1}{2c_{\chi}+\pi}\right)^{n_j}
\frac{\fw(l)^{n_j}}{(2n_j)!}\right)
\left|\prod_{p=0}^d\left(\frac{\partial}{\partial
 k_p}\right)^{n_p}W(\bU)(\o,\bk)(\rho,\eta)\right|\label{eq_kernel_bound_derivative}\\
&\quad\cdot 1_{\sum_{q=0}^dn_q>0}\notag\\
&\le \alpha^{-2}M^l,\ (\forall \bU\in \overline{D},(\o,\bk)\in\R^{d+1},
 \rho,\eta\in\cB).\notag
\end{align}
\item\label{item_kernel_bound}
\begin{align}
&| 1_{\hat{\chi}_{\le j}(\o,\bk)\neq 0}W(\bU)(\o,\bk)(\rho,\eta) |\le
 \alpha^{-2}M^j,\label{eq_kernel_bound}\\
&(\forall \bU\in \overline{D},(\o,\bk)\in\R^{d+1},\rho,\eta\in\cB,
 j\in\Z\text{ with }j\ge N_{\beta}).\notag
\end{align}
\end{enumerate}

On the assumption \eqref{eq_basic_beta_h_assumption} we define the
subset $\hat{\cK}(D,\alpha,M)(l)$ of \\
$\cK(D,\alpha,M)(l)(\beta_1)\times\cK(D,\alpha,M)(l)(\beta_2)$ as
follows.
$(W(\beta_1),W(\beta_2))\in
\cK(D,\alpha,M)(l)(\beta_1)\times\cK(D,\alpha,M)(l)(\beta_2)$ belongs to
$\hat{\cK}(D,\alpha,M)(l)$ if and only if 
\begin{align}
&\prod_{j=0}^d\left(\sum_{n_j=0}^{\infty}\left(\frac{1}{2c_{\chi}+\pi^2}\right)^{n_j}
\frac{\fw(l)^{n_j}}{(2n_j)!}\right)\label{eq_kernel_bound_difference}\\
&\quad \cdot\left|\prod_{p=0}^d\left(\frac{\partial}{\partial
 k_p}\right)^{n_p}(W(\beta_1)(\bU)(\o,\bk)(\rho,\eta)-W(\beta_2)(\bU)(\o,\bk)(\rho,\eta))\right|\notag\\
&\le \beta_1^{-\frac{1}{2}}\alpha^{-2},\ (\forall \bU\in \overline{D},(\o,\bk)\in\R^{d+1},
 \rho,\eta\in\cB).\notag
\end{align}

Let us make an inequality which will enable us to substitute an
element of $\cK(D,\alpha,M)(l)$ into the denominator of the free
covariance.
\begin{lemma}\label{lem_IR_inverse_bound} There exists
a constant $c(d)\in\R_{>0}$ depending only on $d$ such that if
$\alpha\ge c(d)$, the following inequality holds for any $W\in \cK(D,\alpha,M)(l)$.
\begin{align*}
&\|(i\o I_{2^d}-\cE(\bk)-W(\bU)(\o,\bk))^{-1}\|_{2^d\times 2^d}\le
 M^{-l'},\\
&(\forall l'\in\Z\text{ with }l'\ge N_{\beta},(\o,\bk)\in\R^{d+1}\text{ satisfying }\chi_{l'}(\o,\bk)\neq
 0,\bU\in \overline{D}).
\end{align*}
\end{lemma}
\begin{proof}
By Lemma \ref{lem_hopping_properties} \eqref{item_hopping_lower} and
 \eqref{eq_IR_cut_off_description}, 
\begin{align*}
\|(i\o I_{2^d}-\cE(\bk))^{-1}\|_{2^d\times 2^d}&\le
 \left(\o^2+f_{\bt}\sum_{j=1}^d|1+e^{i\frac{\pi}{L}\eps_j^L+ik_j}|^2\right)^{-\frac{1}{2}}\\
&\le
 \frac{\sqrt{6}}{\pi}M_{IR}^{-1}M^{-l'}.
\end{align*}
Thus, by \eqref{eq_kernel_bound},
\begin{align*}
&\|(i\o I_{2^d}-\cE(\bk)-W(\bU)(\o,\bk))^{-1}\|_{2^d\times 2^d}\\
&\le \|(i\o I_{2^d}-\cE(\bk))^{-1}\|_{2^d\times 2^d}
\sum_{n=0}^{\infty}\|(i\o
 I_{2^d}-\cE(\bk))^{-1}W(\bU)(\o,\bk)\|_{2^d\times 2^d}^n\\
&\le
 \frac{\sqrt{6}}{\pi}M_{IR}^{-1}M^{-l'}\sum_{n=0}^{\infty}(c(d)\alpha^{-2})^n
\le M^{-l'}.
\end{align*}
In order to derive the last inequality, we used the condition $\alpha\ge
 c(d)$.
\end{proof}

At every step of the iterative IR integration we receive a Grassmann
polynomial from the preceding IR integration and substitute the
kernel of its quadratic term into the covariance. Our aim here is to
construct lemmas which justify this process.
Let us define the maps $r_{\beta}':[0,\beta)\to
[-\beta/2,\beta/2)$, $s_L:[0,L)\to [-L/2,L/2)$, $r_L':[0,L)^d\to
[-L/2,L/2)^d$ by
\begin{align*}
&r_{\beta}'(x):=\left\{\begin{array}{ll} x & \text{if }x\in
			\left[0,\frac{\beta}{2}\right),\\
                      x-\beta & \text{if }x\in
			\left[\frac{\beta}{2},\beta\right),\end{array}\right.\quad
 s_{L}(x):=\left\{\begin{array}{ll} x & \text{if }x\in
			\left[0,\frac{L}{2}\right),\\
                      x-L & \text{if }x\in
		  \left[\frac{L}{2},L\right),\end{array}\right.\\
&r_L'(\bx):=(s_L(x_1),s_L(x_2),\cdots,s_L(x_d)).
\end{align*}
Let $l\in\{0,-1,\cdots,N_{\beta}\}$ and $(J^0,J^{-1},\cdots,J^l)\in
\prod_{j=0}^l\cS(D,c_0,\alpha,M)(j)$. For $\bU\in \overline{D}$,
$(\o,\bk)\in\R^{d+1}$, $\rho,\eta\in \cB$, set  
\begin{align}
&W^j(\bU)(\o,\bk)(\rho,\eta)
:=\frac{2}{h}\sum_{(\bx,s)\in
 \G(L)\times[0,\beta)_h}e^{-i\<\bk,r_L'(\bx)\>-i\o
 r_{\beta}'(s)}(-1)^{n_{\beta}(r_{\beta}'(s))}\label{eq_effective_kernel_definition}\\
&\qquad\qquad\qquad\qquad\qquad\qquad\cdot J_2^j(\bU)((\rho,\bx,\ua,s,-1),(\eta,\b0,\ua,0,1)),\notag\\
&(j=0,-1,\cdots,l),\notag\\
&E_l(\bU)(\o,\bk)(\rho,\eta):=1_{\frac{1}{\beta}\le
 M_{IR}M^{N_{\beta}+1}}\sum_{j=0}^l\hat{\chi}_{\le
 j}(\o,\bk)W^j(\bU)(\o,\bk)(\rho,\eta).\label{eq_effective_self_energy_definition}
\end{align}
In fact $J^j$ mimics the output of the infrared integration at one
scale. The kernel of its quadratic part is characterized as in
\eqref{eq_effective_kernel_definition} and substituted into the
covariance. The covariance at scale $l$ contains a collection of the
kernels of the form \eqref{eq_effective_self_energy_definition}.
We are going to prove that $E_l\in \cK(D,\alpha,M)(l)$ and
$(E_l(\beta_1),E_l(\beta_2))$ $\in\hat{\cK}(D,\alpha,M)(l)$ on the
assumption \eqref{eq_basic_beta_h_assumption}, which is important information
for the validity of the process. We need the next lemma.
\begin{lemma}\label{lem_kernel_one_scale_bound}
\begin{enumerate}
\item\label{item_kernel_one_scale_bound}
\begin{align*}
&\prod_{j=0}^d\left(\sum_{n_j=0}^{\infty}\left(\frac{2}{\pi}\right)^{n_j}
\frac{\fw(l)^{n_j}}{(2n_j)!}\right)
\left|\prod_{p=0}^d\left(\frac{\partial}{\partial
 k_p}\right)^{n_p}W^l(\bU)(\o,\bk)(\rho,\eta)\right|\\
&\le
 2\|J_2^l(\bU)\|_{l,0},\quad (\forall \bU\in \overline{D},(\o,\bk)\in\R^{d+1},
 \rho,\eta\in\cB).
\end{align*}
\item\label{item_kernel_one_scale_bound_difference}
Assume that \eqref{eq_basic_beta_h_assumption} holds. Then,
\begin{align*}
&\prod_{j=0}^d\left(\sum_{n_j=0}^{\infty}\left(\frac{2}{\pi^2}\right)^{n_j}
\frac{\fw(l)^{n_j}}{(2n_j)!}\right)\\
&\quad\cdot\left|\prod_{p=0}^d\left(\frac{\partial}{\partial
 k_p}\right)^{n_p}(W^l(\beta_1)(\bU)(\o,\bk)(\rho,\eta)-W^l(\beta_2)(\bU)(\o,\bk)(\rho,\eta))\right|\\
&\le
 2|J_2^l(\beta_1)(\bU)-J_2^l(\beta_2)(\bU)|_{l}+\frac{4\pi}{\beta_1}\sum_{a=1}^2 \|J_2^l(\beta_a)(\bU)\|_{l,1},\\
&\ (\forall \bU\in \overline{D},(\o,\bk)\in\R^{d+1},
 \rho,\eta\in\cB).
\end{align*}
\end{enumerate}
\end{lemma}
\begin{proof}
\eqref{item_kernel_one_scale_bound}: Note that
\begin{align*}
&\prod_{j=0}^d\left(\left(\frac{2}{\pi}\right)^{n_j}
\frac{\fw(l)^{n_j}}{(2n_j)!}\right)
\left|\prod_{p=0}^d\left(\frac{\partial}{\partial
 k_p}\right)^{n_p}W^l(\bU)(\o,\bk)(\rho,\eta)\right|\\
&\le \frac{2}{h}\sum_{(\bx,s)\in\G(L)\times[0,\beta)_h}
\prod_{j=0}^d\left(\frac{\fw(l)^{n_j}}{(2n_j)!}d_j((\rho,\bx,\ua,s,-1),(\eta,\b0,\ua,0,1))^{n_j}\right)\\
&\quad\cdot
|J_2^{l}(\bU)((\rho,\bx,\ua,s,-1),(\eta,\b0,\ua,0,1))|.
\end{align*}
This inequality leads to the result.

\eqref{item_kernel_one_scale_bound_difference}:
\begin{align*}
&\prod_{j=0}^d\left(\left(\frac{2}{\pi^2}\right)^{n_j}
\frac{\fw(l)^{n_j}}{(2n_j)!}\right)\\
&\cdot\left|\prod_{p=0}^d\left(\frac{\partial}{\partial
 k_p}\right)^{n_p}(W^l(\beta_1)(\bU)(\o,\bk)(\rho,\eta)-W^l(\beta_2)(\bU)(\o,\bk)(\rho,\eta))\right|\\
&\le  \frac{2}{h}\sum_{(\bx,s)\in\G(L)\times[-\beta_1/4,\beta_1/4)_h}\\
&\quad\cdot \prod_{j=0}^d\left(\frac{1}{(2n_j)!}\left(\frac{\fw(l)}{\pi}\right)^{n_j}
\hat{d}_j((\rho,\bx,\ua,s,-1),(\eta,\b0,\ua,0,1))^{n_j}\right)\\
&\quad\cdot|J_2^{l}(\beta_1)(\bU)(R_{\beta_1}((\rho,\bx,\ua,s,-1),(\eta,\b0,\ua,0,1)))\\
&\qquad -J_2^{l}(\beta_2)(\bU)(R_{\beta_2}((\rho,\bx,\ua,s,-1),(\eta,\b0,\ua,0,1)))|\\
&\quad + \frac{2}{h}\sum_{a=1}^2\sum_{(\bx,s)\in\G(L)\times
 [0,\beta_a)_h}\left(\frac{2}{\pi}\right)^{n_0}\\
&\qquad \cdot\left(1_{s\in
 [\frac{\beta_1}{4},\frac{\beta_a}{2})}\frac{|s|^{n_0}}{(\frac{\beta_a}{2\pi}|e^{i\frac{2\pi}{\beta_a}s}-1|)^{n_0+1}}
+ 1_{s\in
 [\frac{\beta_a}{2},\beta_a-\frac{\beta_1}{4})}\frac{|s-\beta_a|^{n_0}}{(\frac{\beta_a}{2\pi}|e^{i\frac{2\pi}{\beta_a}s}-1|)^{n_0+1}}\right)\\
&\qquad\cdot |d_0(\beta_a)((\rho,\bx,\ua,s,-1),(\eta,\b0,\ua,0,1))|\\
&\qquad\cdot\prod_{j=0}^d\left(\frac{\fw(l)^{n_j}}{(2n_j)!}d_j(\beta_a)((\rho,\bx,\ua,s,-1),(\eta,\b0,\ua,0,1))^{n_j}\right)\\
&\qquad\cdot|J_2^{l}(\beta_a)(\bU)((\rho,\bx,\ua,s,-1),(\eta,\b0,\ua,0,1))|.
\end{align*}
The result follows from the above inequality.
\end{proof}

In the following $c_{\chi}$ is the constant appearing in Lemma
\ref{lem_IR_cut_off_properties} \eqref{item_IR_cut_off_derivative}.

\begin{lemma}\label{lem_self_energy_properties}
There exists a constant $c(d,M,c_w,c_{\chi})\in\R_{>0}$ depending only
 on $d,M,c_w,c_{\chi}$ such that if $c_0\ge
 c(d,M,c_w,c_{\chi}) f_{\bt}^{-\frac{1}{2}}$, the following statements
 hold.
\begin{enumerate}
\item\label{item_self_energy_property}
$$E_l\in \cK(D,\alpha,M)(l).$$ 
\item\label{item_self_energy_property_difference}
Assume that \eqref{eq_basic_beta_h_assumption} holds, $l\in
     \{0,-1,\cdots,N_{\beta_1}\}$ and $(J^j(\beta_1),J^j(\beta_2))$ $\in
     \hat{\cS}(D,c_0,\alpha,M)(j)$ $(j=0,-1,\cdots,l)$. Then,
$$(E_l(\beta_1),E_l(\beta_2))\in \hat{\cK}(D,\alpha,M)(l).$$
\end{enumerate}
\end{lemma}
\begin{proof}
\eqref{item_self_energy_property}: It suffices to consider the case that
 $1/\beta\le M_{IR}M^{N_{\beta}+1}$. The continuity and analyticity with
 $\bU$ is clear. Let us prove the invariant properties. The
 periodicity claimed in \eqref{item_kernel_periodicity} follows from the
 definition. Since 
\begin{align*}
\hat{\chi}_{\le j}(\o,\bk)=\hat{\chi}_{\le j}(-\o,\bk)=\hat{\chi}_{\le
 j}\left(\o,-\bk-\frac{2\pi}{L}\beps^L\right),\quad (\forall (\o,\bk)\in\R^{d+1}),
\end{align*}
it is sufficient to confirm the invariances of $W^j$
 $(j\in\{0,-1,\cdots,l\})$ to prove the invariances of $E_l$.
The proof for the invariance of $W^j$ claimed in
 \eqref{item_kernel_hermiticity}, \eqref{item_kernel_inversion},
 \eqref{item_kernel_momentum_inversion} is parallel to the proof for\\ 
\cite[\mbox{Lemma 7.6 (2), (7.25), (7.26), (7.27)}]{K15} respectively.
Here we only provide the sketch of the proof. 
The invariance in \eqref{item_kernel_hermiticity} is proved by combining
 the anti-symmetry of $J_2^j(\bU)(\cdot)$, the invariance
 $J_2^j(\bU)(\psi)=J_2^j(\bU)(\cR\psi)$ for $S:I\to I$, $Q:I\to \R$
 defined in \eqref{eq_IR_translation} and the invariance
 $J_2^j(\bU)(\psi)=\overline{J_2^j(\overline{\bU})}(\cR\psi)$ for
 $S:I\to I$, $Q:I\to \R$ defined in \eqref{eq_IR_adjoint}. 
The invariance in \eqref{item_kernel_inversion} follows from the
 anti-symmetry of $J_2^j(\bU)(\cdot)$, the invariance $J_2^j(\bU)(\psi)=J_2^j(\bU)(\cR\psi)$ for $S:I\to I$, $Q:I\to \R$
 defined in \eqref{eq_IR_translation} and the invariance
 $J_2^j(\bU)(\psi)=\overline{J_2^j(\overline{\bU})}(\cR\psi)$ for
 $S:I\to I$, $Q:I\to \R$ defined in \eqref{eq_IR_inversion}.
The invariance in \eqref{item_kernel_momentum_inversion} is due to the
 invariance $J_2^j(\bU)(\psi)=J_2^j(\bU)(\cR\psi)$ for $S:I\to I$, $Q:I\to \R$
 defined in \eqref{eq_IR_translation} and \eqref{eq_IR_momentum_inversion}. 

Next let us show the bound property \eqref{item_kernel_derivative}. Take
 $n_0,n_1,\cdots,n_d\in\N\cup \{0\}$ satisfying $\sum_{j=0}^dn_j>0$. Using
 \eqref{eq_IR_boundedness}, Lemma \ref{lem_IR_cut_off_properties}
 \eqref{item_IR_cut_off_derivative} and Lemma
 \ref{lem_kernel_one_scale_bound} \eqref{item_kernel_one_scale_bound},
 we can derive that for any $(\o,\bk)\in\R^{d+1}$,
\begin{align}
&\left|\prod_{j=0}^d\left(\frac{\fw(l)^{n_j}}{(2n_j)!}\left(\frac{\partial}{\partial
 k_j}\right)^{n_j}\right)E_l(\o,\bk)(\rho,\eta)\right|\label{eq_self_energy_derivative_pre}\\
&\le
 M^{l}\sum_{p=0}^lM^{-p}\prod_{j=0}^d\left(\frac{\fw(p)^{n_j}}{(2n_j)!}\right)\prod_{q=0}^d\left(\sum_{m_q=0}^{n_q}\left(\begin{array}{c}
															 n_q\\ m_q\end{array}\right)\right)
\notag\\
&\quad\cdot\left|\prod_{r=0}^d\left(\frac{\partial}{\partial
 k_r}\right)^{m_r}\hat{\chi}_{\le p}(\o,\bk)\right|
\left|\prod_{s=0}^d\left(\frac{\partial}{\partial
 k_s}\right)^{n_s-m_s}W^p(\o,\bk)(\rho,\eta)\right|\notag\\
&\le  M^{l}\sum_{p=0}^lM^{-p}\prod_{j=0}^d\left(\frac{\fw(p)^{n_j}}{(2n_j)!}\right)\prod_{q=0}^d\left(\sum_{m_q=0}^{n_q}\left(\begin{array}{c}
															 n_q\\ m_q\end{array}\right)\right)\notag\\
&\quad\cdot \left|\prod_{r=0}^d(c_{\chi}\fw(p)^{-1})^{m_r}(m_r!)^2\right|
\left|\prod_{s=0}^d\left(\frac{\pi}{2}\fw(p)^{-1}\right)^{n_s-m_s}
(2(n_s-m_s))!\right|
 2\|J_2^p\|_{p,0}\notag\\
&\le  2M^{l}\sum_{p=0}^lM^{-p}\|J_2^p\|_{p,0}
\prod_{j=0}^d\left(\sum_{m_j=0}^{n_j}\left(\begin{array}{c}								 n_j\\ m_j\end{array}\right)
c_{\chi}^{m_j}\left(\frac{\pi}{2}\right)^{n_j-m_j}
\right)\notag\\
&\le
 \frac{2}{1-M^{-\frac{1}{2}}}\left(c_{\chi}+\frac{\pi}{2}\right)^{\sum_{j=0}^dn_j}c_0^{-1}\alpha^{-2}M^l,\notag
\end{align}
which implies that
\begin{align*}
&\prod_{j=0}^d\left(\sum_{n_j=0}^{\infty}\left(\frac{1}{2c_{\chi}+\pi}\right)^{n_j}
\frac{\fw(l)^{n_j}}{(2n_j)!}\right)
\left|\prod_{p=0}^d\left(\frac{\partial}{\partial
 k_p}\right)^{n_p}E_l(\bU)(\o,\bk)(\rho,\eta)\right|\\
&\quad\cdot 1_{\sum_{q=0}^dn_q>0}\notag\\
&\le \frac{2^{d+1}}{1-M^{-\frac{1}{2}}}c_0^{-1}\alpha^{-2}M^l
\le \frac{2^{d+1}}{1-2^{-\frac{1}{2}}}c_0^{-1}\alpha^{-2}M^l,\\
&(\forall \bU\in\overline{D},(\o,\bk)\in\R^{d+1},\rho,\eta\in\cB),
\end{align*}
where we used the condition $M\ge 2$.
Thus, if $c_0\ge 2^{d+2}/(1-2^{-\frac{1}{2}})$, $E_l$ satisfies the inequality in
 \eqref{item_kernel_derivative}.

It remains to prove the bound property \eqref{item_kernel_bound}. Take
 $\rho,\eta\in \cB$. If $e^{i\<b(\rho)-b(\eta),\bpi\>}=1$, the
 invariances \eqref{item_kernel_hermiticity},
 \eqref{item_kernel_inversion} imply that 
\begin{align*}
&E_l(\bU)(\o,\bk)(\rho,\eta)=-E_l(\bU)(-\o,\bk)(\rho,\eta),\\
&(\forall \bU\in \overline{D},(\o,\bk)\in\cM\times
 ((2\pi/L)\Z)^d).
\end{align*}
Thus, by \eqref{eq_self_energy_derivative_pre},
\begin{align}
&|E_l(\o,\bk)(\rho,\eta)|\label{eq_self_energy_bound_matsubara}\\
&\le
 \frac{1}{2}\left|E_l(\o,\bk)(\rho,\eta)-E_l\left(\frac{\pi}{\beta},\bk\right)(\rho,\eta)\right|\notag\\
&\quad +  \frac{1}{2}\left|E_l(\o,\bk)(\rho,\eta)-E_l\left(-\frac{\pi}{\beta},\bk\right)(\rho,\eta)\right|\notag\\
&\le
 c(d,M,c_{\chi},c_w)\left(|\o|+\frac{1}{\beta}\right)c_0^{-1}\alpha^{-2},\quad (\forall (\o,\bk)\in \cM\times ((2\pi/L)\Z)^d),\notag
\end{align}
where $c(d,M,c_{\chi},c_w)\in\R_{>0}$ is a constant depending only on
 $d,M,c_{\chi},c_w$.
If $e^{i\<b(\rho)-b(\eta),\bpi\>}=-1$, the invariance
 \eqref{item_kernel_momentum_inversion} and the periodicity
 \eqref{item_kernel_periodicity} yield that  
\begin{align*}
&E_l(\bU)(\o,\bpi)(\rho,\eta)=-e^{i\<b(\rho)-b(\eta),\frac{\pi}{L}\beps^L\>}
E_l(\bU)\left(\o,\bpi-\frac{2\pi}{L}\beps^L\right)(\rho,\eta),\\
&(\forall \bU\in \overline{D},\o\in\cM).
\end{align*}
Thus, by using the bound \eqref{eq_IR_boundedness}, Lemma
 \ref{lem_kernel_one_scale_bound} \eqref{item_kernel_one_scale_bound}
 and \eqref{eq_self_energy_derivative_pre} we have that
\begin{align}
|E_l(\o,\bk)(\rho,\eta)|
&\le \frac{1}{2}|E_l(\o,\bk)(\rho,\eta)-E_l(\o,\bpi)(\rho,\eta)|\label{eq_self_energy_bound_momentum}\\
&\quad +
 \frac{1}{2}\left|E_l(\o,\bk)(\rho,\eta)-E_l\left(\o,\bpi-\frac{2\pi}{L}\beps^L\right)(\rho,\eta)\right|\notag\\
&\quad+\frac{c(d)}{L}\left|E_l\left(\o,\bpi-\frac{2\pi}{L}\beps^L\right)(\rho,\eta)\right|\notag\\
&\le
 c(d,M,c_{\chi},c_w)\left(\sum_{j=1}^d|k_j-\pi|+\frac{1}{L}\right)c_0^{-1}\alpha^{-2},\notag\\
&(\forall (\o,\bk)\in\cM\times ((2\pi/L)\Z)^d).\notag
\end{align}
By coupling \eqref{eq_self_energy_bound_matsubara} with
 \eqref{eq_self_energy_bound_momentum} we obtain
\begin{align*}
&|E_l(\o,\bk)(\rho,\eta)|\le c(d,M,c_{\chi},c_w)\left(|\o|+\sum_{j=1}^d|k_j-\pi|+\frac{1}{\beta}+\frac{1}{L}\right)c_0^{-1}\alpha^{-2},\\
&(\forall (\o,\bk)\in\cM\times
 ((2\pi/L)\Z)^d,\rho,\eta\in\cB).
\end{align*}
For any $(\o,\bk)\in\R^{d+1}$ there exists $(\hat{\o},\hat{\bk})\in\cM\times
 ((2\pi/L)\Z)^d$ such that $|\o-\hat{\o}|\le \pi/\beta$,
 $\|\bk-\hat{\bk}\|_{\R^d}\le \sqrt{d}\pi/L$. By taking into account
 this fact,
 we can deduce from the above inequality for $(\hat{\o},\hat{\bk})$ and
 \eqref{eq_self_energy_derivative_pre} that 
\begin{align*}
&|E_l(\o,\bk)(\rho,\eta)|\le c(d,M,c_{\chi},c_w)\left(|\o|+\sum_{j=1}^d|k_j-\pi|+\frac{1}{\beta}+\frac{1}{L}\right)c_0^{-1}\alpha^{-2},\\
&(\forall (\o,\bk)\in\R^{d+1},\rho,\eta\in\cB).
\end{align*}
Then, by using the periodicity of $E_l(\o,\bk)$ with $\bk$, the support
 property of $\hat{\chi}_{\le j}$, the assumption $1/L\le 1/\beta \le
 M_{IR}M^{N_{\beta}+1}$ and the fact $f_{\bt}\le 1$ we reach the
 inequality that
\begin{align*}
&|E_l(\bU)(\o,\bk)(\rho,\eta)|\le
 c(d,M,c_{\chi},c_w)f_{\bt}^{-\frac{1}{2}}c_0^{-1}\alpha^{-2}M^j,\\
&(\forall \bU\in\overline{D},j\in\Z\text{ with }j\ge N_{\beta},\\
&\quad(\o,\bk)\in\R^{d+1}\text{
 satisfying }\hat{\chi}_{\le j}(\o,\bk)\neq 0,\rho,\eta\in\cB).
\end{align*}
Thus, $E_l$ satisfies the property \eqref{item_kernel_bound} under the
 assumption that $c_0 \ge c(d,M,c_{\chi},c_w)f_{\bt}^{-\frac{1}{2}}$.

\eqref{item_self_energy_property_difference}: Take
 $n_0,n_1,\cdots,n_d\in \N\cup\{0\}$. By \eqref{eq_IR_boundedness},
\eqref{eq_IR_boundedness_difference}, Lemma
 \ref{lem_IR_cut_off_properties} \eqref{item_IR_cut_off_derivative}, 
Lemma \ref{lem_kernel_one_scale_bound}
 \eqref{item_kernel_one_scale_bound_difference} and the assumption $M\ge
 2$,
\begin{align*}
&\Bigg|\prod_{j=0}^d\left(\frac{\fw(l)^{n_j}}{(2n_j)!}\left(\frac{\partial}{\partial
 k_j}\right)^{n_j}\right)(E_l(\beta_1)(\o,\bk)(\rho,\eta)-E_l(\beta_2)(\o,\bk)(\rho,\eta))\Bigg|\\
&\le
 \sum_{p=0}^l\prod_{j=0}^d\left(\frac{\fw(l)^{n_j}}{(2n_j)!}\sum_{m_j=0}^{n_j}
\left(\begin{array}{c} n_j\\ m_j\end{array}\right)\right)
\left|\prod_{q=0}^d\left(\frac{\partial}{\partial
 k_q}\right)^{m_q}\hat{\chi}_{\le p}(\o,\bk)\right|\\
&\quad\cdot\left|\prod_{r=0}^d\left(\frac{\partial}{\partial
 k_r}\right)^{n_r-m_r}(W^p(\beta_1)(\o,\bk)(\rho,\eta)-W^p(\beta_2)(\o,\bk)(\rho,\eta))\right|\\
&\le\sum_{p=0}^l\prod_{j=0}^d\left(\sum_{m_j=0}^{n_j}\left(\begin{array}{c} n_j\\ m_j\end{array}\right)
c_{\chi}^{m_j}\left(\frac{\pi^2}{2}\right)^{n_j-m_j}\right)\\
&\quad\cdot
 \left(2|J_2^p(\beta_1)-J_2^p(\beta_2)|_p+\frac{4\pi}{\beta_1}\sum_{a=1}^2\|J_2^p(\beta_a)\|_{p,1}\right)\\
&\le \left(c_{\chi}+\frac{\pi^2}{2}\right)^{\sum_{j=0}^dn_j}\sum_{p=0}^l(2\beta_1^{-\frac{1}{2}}c_0^{-1}\alpha^{-2}M^{\frac{p}{2}}+4\pi
 \beta_1^{-1}c_0^{-1}\alpha^{-2}M^{\frac{p}{2}})\\
&\le \left(c_{\chi}+\frac{\pi^2}{2}\right)^{\sum_{j=0}^dn_j}\frac{2+4\pi}{1-2^{-\frac{1}{2}}}c_0^{-1}\alpha^{-2}\beta_1^{-\frac{1}{2}}.
\end{align*} 
This inequality implies that
\begin{align*}
&\prod_{j=0}^d\left(\sum_{n_j=0}^{\infty}\left(\frac{1}{2c_{\chi}+\pi^2}\right)^{n_j}\frac{\fw(l)^{n_j}}{(2n_j)!}\right)\\
&\quad\cdot\left|\prod_{p=0}^d\left(\frac{\partial}{\partial
 k_j}\right)^{n_p}(E_l(\beta_1)(\bU)(\o,\bk)(\rho,\eta)-E_l(\beta_2)(\bU)(\o,\bk)(\rho,\eta))\right|\\
&\le
 \frac{2^{d+1}(2+4\pi)}{1-2^{-\frac{1}{2}}}c_0^{-1}\alpha^{-2}\beta_1^{-\frac{1}{2}},\quad (\forall \bU\in \overline{D},\rho,\eta\in\cB,
 (\o,\bk)\in\R^{d+1}).
\end{align*}
Thus, the claim holds true if $c_0\ge 2^{d+1}(2+4\pi)/(1-2^{-\frac{1}{2}})$.
\end{proof}

In the next lemma we summarize properties of a function of
$\bU\in\overline{D}$ which resembles the final output of the infrared
integration. In the following $C^{\o}(D;\C)$ denotes the set of
analytic functions in $D$. 

\begin{lemma}\label{lem_IR_final_output}
Assume that $l\in \{0,-1,\cdots,N_{\beta}\}$, $G_l\in
 \cK(D,\alpha,M)(l)$, $G_{l+1}\in
 \cK(D,\alpha,M)(l+1)$ if $l\le -1$, $G_{l+1}=0$ if $l=0$. 
Moreover, assume that 
\begin{align*}
&G_l(\bU)(\o,\bk)-G_{l+1}(\bU)(\o,\bk)=O,\\
&(\forall \bU\in \overline{D},(\o,\bk)\in\R^{d+1}\text{ with
 }\hat{\chi}_{\le l}(\o,\bk)=0).
\end{align*}
Define the function
 $H_l:\overline{D}\to \C$ by 
\begin{align*}
&H_l(\bU)\\
&:=\frac{1}{\beta L^d}\sum_{(\o,\bk)\in\cM\times\G(L)^*}\log(\det(I_{2^d}-(i\o
 I_{2^d}-\cE(\bk)-G_{l+1}(\bU)(\o,\bk))^{-1}\\
&\qquad\qquad\qquad\qquad\qquad\qquad\qquad \cdot (G_{l}(\bU)(\o,\bk)-G_{l+1}(\bU)(\o,\bk)))).
\end{align*}
Then, there exist a constant $c(d)\in\R_{>0}$ depending only on $d$ and
 a constant $c(d,M,c_w,c_{\chi})\in\R_{>0}$ depending only on
 $d,M,c_w,c_{\chi}$ such that the following statements hold true if
 $\alpha \ge c(d)$.
\begin{enumerate}
\item\label{item_IR_final_analyticity}
$$
H_l\in C(\overline{D};\C)\cap C^{\o}(D;\C).
$$ 
\item\label{item_IR_final_bound}
\begin{align*}
|H_l(\bU)|\le
 c(d,M,c_{w},c_{\chi})f_{\bt}^{-\frac{d}{2}}M^{(d+1)l}\alpha^{-2},\quad
 (\forall \bU\in \overline{D}).
\end{align*}
\item\label{item_IR_final_bound_difference}
In addition, assume that \eqref{eq_basic_beta_h_assumption} holds, $l\in
     \{0,-1,\cdots,N_{\beta_1}\}$, \\
$(G_l(\beta_1),G_l(\beta_2))\in
     \hat{\cK}(D,\alpha,M)(l)$ and  \\
$(G_{l+1}(\beta_1),G_{l+1}(\beta_2))\in
     \hat{\cK}(D,\alpha,M)(l+1)$ if $l\le -1$. Then, 
\begin{align*}
|H_l(\beta_1)(\bU)-H_l(\beta_2)(\bU)|\le
 c(d,M,c_{w},c_{\chi})\beta_1^{-\frac{1}{2}}
f_{\bt}^{-\frac{d}{2}}M^{dl}\alpha^{-2},\quad
 (\forall \bU\in \overline{D}).
\end{align*}
\end{enumerate}
\end{lemma}

\begin{proof}
\eqref{item_IR_final_analyticity}, \eqref{item_IR_final_bound}: Take
 $j\in \{l,l-1,\cdots,N_{\beta}\}$. It follows from Lemma
 \ref{lem_IR_inverse_bound} and \eqref{eq_kernel_bound} that for any
 $(\o,\bk)\in\R^{d+1}$ satisfying $\chi_j(\o,\bk)\neq 0$, 
\begin{align}
&\|(i\o I_{2^d}-\cE(\o,\bk)-G_{l+1}(\o,\bk))^{-1}\|_{2^d\times 2^d}\le
 M^{-j},\label{eq_application_inverse_bound}\\
&\|G_l(\o,\bk)-G_{l+1}(\o,\bk)\|_{2^d\times 2^d}\le c(d)\alpha^{-2}M^j,
\label{eq_application_bound}
\end{align}
on the assumption that $\alpha$ is larger than a positive constant
 depending only on $d$. By using the above inequalities and 
 Lemma \ref{lem_IR_cut_off_properties} \eqref{item_IR_cut_off_equivalence},\eqref{item_IR_cut_off_implication},\eqref{item_IR_cut_off_support} we
 see that
\begin{align}
|H_l(\bU)|&\le
\frac{1}{\beta
 L^d}\sum_{(\o,\bk)\in\cM\times\G(L)^*}\sum_{j=l}^{N_{\beta}}1_{\chi_{j}(\o,\bk)\neq 0}\label{eq_estimation_typical_uniform}\\
&\quad\cdot\sum_{n=1}^{\infty}|\det(I_{2^d}-(i\o
  I_{2^d}-\cE(\bk)-G_{l+1}(\bU)(\o,\bk))^{-1}\notag\\
&\qquad\qquad\qquad\qquad \cdot
 (G_{l}(\bU)(\o,\bk)-G_{l+1}(\bU)(\o,\bk)))-1|^n\notag\\
&\le \frac{1}{\beta
 L^d}\sum_{(\o,\bk)\in\cM\times\G(L)^*}\sum_{j=l}^{N_{\beta}}1_{\chi_{j}(\o,\bk)\neq 0}
\sum_{n=1}^{\infty}(c(d)\alpha^{-2})^n\notag\\
&\le c(d,M)f_{\bt}^{-\frac{d}{2}}M^{(d+1)l}\alpha^{-2},\quad (\forall
 \bU\in \overline{D}).\notag
\end{align}
This implies \eqref{item_IR_final_bound}. 
Take $j\in \{l,l-1,\cdots,N_{\beta}\}$ and $(\o,\bk)\in\R^{d+1}$
 satisfying $\chi_j(\o,\bk)\neq 0$. Since 
\begin{align}
&(i\o
 I_{2^d}-\cE(\bk)-G_{l+1}(\bU)(\o,\bk))^{-1}\label{eq_neumann_series_expansion}\\
&=\sum_{n=0}^{\infty}((i\o
 I_{2^d}-\cE(\bk))^{-1}G_{l+1}(\bU)(\o,\bk))^n(i\o
 I_{2^d}-\cE(\bk))^{-1}\notag
\end{align}
and this series converges uniformly with $\bU$,
\begin{align*}
&(i\o
 I_{2^d}-\cE(\bk)-G_{l+1}(\cdot)(\o,\bk))^{-1}(\rho,\eta)\in
 C(\overline{D};\C)\cap C^{\o}(D;\C),\\
&(\forall \rho,\eta\in \cB).
\end{align*}
Thus,
\begin{align*}
&\det(I_{2^d}-(i\o
 I_{2^d}-\cE(\bk)-G_{l+1}(\cdot)(\o,\bk))^{-1}(G_l(\cdot)(\o,\bk)-G_{l+1}(\cdot)(\o,\bk)))\\
&\in C(\overline{D};\C)\cap C^{\o}(D;\C).
\end{align*}
Moreover, an estimation similar to \eqref{eq_estimation_typical_uniform} implies that the series 
\begin{align*}
\sum_{n=1}^{\infty}\frac{(-1)^{n-1}}{n}(\det(I_{2^d}-&(i\o
 I_{2^d}-\cE(\bk)-G_{l+1}(\bU)(\o,\bk))^{-1}\\
&\cdot (G_l(\bU)(\o,\bk)-G_{l+1}(\bU)(\o,\bk)))-1)^n
\end{align*}
converges uniformly with $\bU$ and thus 
\begin{align*}
&\log(\det(I_{2^d}-(i\o
 I_{2^d}-\cE(\bk)-G_{l+1}(\cdot)(\o,\bk))^{-1}\\
&\qquad\qquad\qquad\cdot (G_l(\cdot)(\o,\bk)-G_{l+1}(\cdot)(\o,\bk))))\\
&\in C(\overline{D};\C)\cap C^{\o}(D;\C).
\end{align*}
Therefore, the claim \eqref{item_IR_final_analyticity} holds true. 

\eqref{item_IR_final_bound_difference}: Let us prepare a couple of
 necessary inequalities. By
 \eqref{eq_kernel_bound_derivative}, 
\begin{align}
&\left\|\frac{\partial}{\partial \o}(i\o
 I_{2^d}-\cE(\bk)-G_{l+1}(\beta_a)(\bU)(\o,\bk))\right\|_{2^d\times
 2^d}\le
 c(d,M,c_w,c_{\chi}),\label{eq_application_bound_derivative_mixed}\\
&\left\|\frac{\partial}{\partial \o}(G_l(\beta_a)(\bU)(\o,\bk)
-G_{l+1}(\beta_a)(\bU)(\o,\bk))\right\|_{2^d\times 2^d}
\le c(d,M,c_w,c_{\chi})\alpha^{-2},\label{eq_application_bound_derivative}\\
&(\forall \bU\in \overline{D},(\o,\bk)\in\R^{d+1},a\in
 \{1,2\}).\notag
\end{align}
Then, define the functions $\hat{H}_l(\beta_a):\overline{D}\to \C$
 $(a=1,2)$ by 
\begin{align*}
\hat{H}_l(\beta_a)(\bU):=&\frac{1}{2\pi
 L^d}\sum_{\bk\in\G(L)^*}\left(\int_{-\pi
 h}^{-\frac{\pi}{\beta_a}}d\o
+\int^{\pi h}_{\frac{\pi}{\beta_a}}d\o\right)\\
&\cdot\log(\det(I_{2^d}-(i\o
 I_{2^d}-\cE(\bk)-G_{l+1}(\beta_a)(\bU)(\o,\bk))^{-1}\\
&\qquad\qquad\qquad\cdot (G_{l}(\beta_a)(\bU)(\o,\bk)-G_{l+1}(\beta_a)(\bU)(\o,\bk)))).
\end{align*}
By using Lemma \ref{lem_IR_cut_off_properties} \eqref{item_IR_cut_off_equivalence},\eqref{item_IR_cut_off_large_matsubara},\eqref{item_IR_cut_off_implication},\eqref{item_IR_cut_off_support}, 
\eqref{eq_application_inverse_bound}, \eqref{eq_application_bound}, 
\eqref{eq_application_bound_derivative_mixed} and
\eqref{eq_application_bound_derivative} we deduce that
\begin{align}
&|H_l(\beta_a)(\bU)-\hat{H}_l(\beta_a)(\bU)|\label{eq_IR_final_continuous_discrete}\\
&\le
\frac{1}{2\pi L^d}\sum_{\bk\in\G(L)^*}\sum_{m=0}^{\frac{\beta_a h}{2}-1}\notag\\
&\quad\cdot\left(\int_{\frac{\pi}{\beta_a}+\frac{2\pi}{\beta_a}m}^{\frac{\pi}{\beta_a}+\frac{2\pi}{\beta_a}(m+1)}d\o\int_{\frac{\pi}{\beta_a}+\frac{2\pi}{\beta_a}m}^{\o}d\eta
+\int^{-\frac{\pi}{\beta_a}-\frac{2\pi}{\beta_a}m}_{-\frac{\pi}{\beta_a}-\frac{2\pi}{\beta_a}(m+1)}d\o\int^{-\frac{\pi}{\beta_a}-\frac{2\pi}{\beta_a}m}_{\o}d\eta\right)\notag\\
&\quad\cdot\Bigg|\frac{\partial}{\partial\eta}\log(\det(I_{2^d}-(i\eta
  I_{2^d}-\cE(\bk)-G_{l+1}(\beta_a)(\eta,\bk))^{-1}\notag\\
&\qquad\qquad\qquad\qquad\qquad \cdot
 (G_{l}(\beta_a)(\eta,\bk)-G_{l+1}(\beta_a)(\eta,\bk))))\Bigg|\notag\\
&\le \frac{1}{\beta_1L^d}\sum_{\bk\in \G(L)^*}\sum_{j=l}^{N_{\beta_a}}
\left(\int_{\frac{\pi}{\beta_a}}^{\pi h}d\o+
 \int^{-\frac{\pi}{\beta_a}}_{-\pi h}d\o\right)1_{\chi_j(\o,\bk)\neq 0}
\notag\\
&\quad \cdot \Bigg|\frac{\partial}{\partial\o}
\log(\det(I_{2^d}-(i\o
  I_{2^d}-\cE(\bk)-G_{l+1}(\beta_a)(\o,\bk))^{-1}\notag\\
&\qquad\qquad\qquad\qquad\qquad \cdot
 (G_{l}(\beta_a)(\o,\bk)-G_{l+1}(\beta_a)(\o,\bk))))\Bigg|\notag\\
&\le \frac{1}{\beta_1L^d}\sum_{\bk\in \G(L)^*}\sum_{j=l}^{N_{\beta_a}}
\left(\int_{\frac{\pi}{\beta_a}}^{\pi h}d\o+
 \int^{-\frac{\pi}{\beta_a}}_{-\pi
 h}d\o\right)1_{\chi_j(\o,\bk)\neq 0}\notag\\
&\quad\cdot c(d)\Big\| \frac{\partial}{\partial \o}((i\o
 I_{2^d}-\cE(\bk)-G_{l+1}(\beta_a)(\o,\bk))^{-1}\notag\\
&\qquad\qquad\qquad\cdot(G_l(\beta_a)(\o,\bk)-G_{l+1}(\beta_a)(\o,\bk)))\Big\|_{2^d\times
 2^d}\notag\\
&\le \frac{1}{\beta_1L^d}\sum_{\bk\in \G(L)^*}\sum_{j=l}^{N_{\beta_a}}
\left(\int_{\frac{\pi}{\beta_a}}^{\pi h}d\o+
 \int^{-\frac{\pi}{\beta_a}}_{-\pi
 h}d\o\right)1_{\chi_j(\o,\bk)\neq 0}\notag\\
&\quad\cdot \Bigg(c(d)
\|(i\o I_{2^d}-\cE(\bk)-G_{l+1}(\beta_a)(\o,\bk))^{-1}\|_{2^d\times
 2^d}^2\notag\\
&\qquad\qquad\cdot\Big\| \frac{\partial}{\partial \o}(i\o
 I_{2^d}-\cE(\bk)-G_{l+1}(\beta_a)(\o,\bk))\Big\|_{2^d\times 2^d}\notag\\
&\qquad\qquad\cdot\|G_l(\beta_a)(\o,\bk)-G_{l+1}(\beta_a)(\o,\bk)\|_{2^d\times
 2^d}\notag\\
&\qquad+c(d)\|(i\o I_{2^d}-\cE(\bk)-G_{l+1}(\beta_a)(\o,\bk))^{-1}\|_{2^d\times 2^d}\notag\\
&\qquad\qquad\quad\cdot \Big\| \frac{\partial}{\partial
 \o}(G_l(\beta_a)(\o,\bk)-G_{l+1}(\beta_a)(\o,\bk))\Big\|_{2^d\times 2^d}\Bigg)\notag\\
&\le
 c(d,M,c_w,c_{\chi})\beta_1^{-1}f_{\bt}^{-\frac{d}{2}}\alpha^{-2}\sum_{j=l}^{N_{\beta_a}}M^{dj}\notag\\
&\le
 c(d,M,c_w,c_{\chi})\beta_1^{-1}f_{\bt}^{-\frac{d}{2}}\alpha^{-2}M^{dl}.\notag
\end{align}

Take $j\in \{l,l-1,\cdots,N_{\beta_1}\}$ and $(\o,\bk)\in\R^{d+1}$
 with $\chi_j(\o,\bk)\neq 0$. Note that for any $a,b\in\C\backslash
 \R_{\le 0}$ with $|a-1|\le 1/2$, $|b-1|\le 1/2$, $|\log a-\log b|\le
 2|a-b|$. Using this inequality, \eqref{eq_kernel_bound_difference}, \eqref{eq_application_inverse_bound},
 \eqref{eq_application_bound} as
 well as the assumption $\alpha\ge c(d)$, we can justify the following
 calculation. 
\begin{align*}
&|\log(\det(I_{2^d}-(i\o
  I_{2^d}-\cE(\bk)-G_{l+1}(\beta_1)(\o,\bk))^{-1}\notag\\
&\qquad\qquad\qquad\quad \cdot
 (G_{l}(\beta_1)(\o,\bk)-G_{l+1}(\beta_1)(\o,\bk))))\\
&\quad -\log(\det(I_{2^d}-(i\o
  I_{2^d}-\cE(\bk)-G_{l+1}(\beta_2)(\o,\bk))^{-1}\notag\\
&\qquad\qquad\qquad\qquad\quad \cdot
 (G_{l}(\beta_2)(\o,\bk)-G_{l+1}(\beta_2)(\o,\bk))))|\\
&\le 2|\det(I_{2^d}-(i\o
  I_{2^d}-\cE(\bk)-G_{l+1}(\beta_1)(\o,\bk))^{-1}\notag\\
&\qquad\qquad\qquad\quad \cdot
 (G_{l}(\beta_1)(\o,\bk)-G_{l+1}(\beta_1)(\o,\bk)))\\
&\qquad -\det(I_{2^d}-(i\o
  I_{2^d}-\cE(\bk)-G_{l+1}(\beta_2)(\o,\bk))^{-1}\notag\\
&\qquad\qquad\qquad\qquad \cdot
 (G_{l}(\beta_2)(\o,\bk)-G_{l+1}(\beta_2)(\o,\bk)))|\\
&\le c(d) \|(i\o
  I_{2^d}-\cE(\bk)-G_{l+1}(\beta_1)(\o,\bk))^{-1}\\
&\qquad\qquad\cdot 
 (G_{l}(\beta_1)(\o,\bk)-G_{l+1}(\beta_1)(\o,\bk))\\
&\qquad\qquad -(i\o
  I_{2^d}-\cE(\bk)-G_{l+1}(\beta_2)(\o,\bk))^{-1}\\
&\quad\qquad\qquad\cdot(G_{l}(\beta_2)(\o,\bk)-G_{l+1}(\beta_2)(\o,\bk))\|_{2^d\times
 2^d}\\
&\le c(d) 
 \|(i\o I_{2^d}-\cE(\bk)-G_{l+1}(\beta_1)(\o,\bk))^{-1}\|_{2^d\times
 2^d}\\
&\quad\cdot \|(i\o
 I_{2^d}-\cE(\bk)-G_{l+1}(\beta_2)(\o,\bk))^{-1}\|_{2^d\times 2^d}\\
&\quad\cdot\|G_{l+1}(\beta_1)(\o,\bk)-G_{l+1}(\beta_2)(\o,\bk)\|_{2^d\times
 2^d}\\
&\quad\cdot \|G_{l}(\beta_1)(\o,\bk)-G_{l+1}(\beta_1)(\o,\bk)\|_{2^d\times
 2^d}\\
&\quad +c(d)  \|(i\o
 I_{2^d}-\cE(\bk)-G_{l+1}(\beta_2)(\o,\bk))^{-1}\|_{2^d\times 2^d}\\
&\qquad\cdot (\|G_{l}(\beta_1)(\o,\bk)-G_{l}(\beta_2)(\o,\bk)\|_{2^d\times 2^d}\\
&\qquad\quad+\|G_{l+1}(\beta_1)(\o,\bk)-G_{l+1}(\beta_2)(\o,\bk)\|_{2^d\times
 2^d})\\
&\le c(d)M^{-j}\beta_1^{-\frac{1}{2}}\alpha^{-2}.
\end{align*}
It follows from this inequality, Lemma \ref{lem_IR_cut_off_properties}
 \eqref{item_IR_cut_off_equivalence},\eqref{item_IR_cut_off_implication},\eqref{item_IR_cut_off_support}, 
\eqref{eq_application_inverse_bound} and 
\eqref{eq_application_bound} that
\begin{align}
&|\hat{H}_l(\beta_1)(\bU)-\hat{H}_l(\beta_2)(\bU)|\label{eq_IR_final_continuous_continuous}\\
&\le \frac{1}{2\pi L^d}\sum_{\bk\in \G(L)^*}\left(\int_{-\pi h}^{-\frac{\pi}{\beta_1}}d\o+
\int^{\pi
 h}_{\frac{\pi}{\beta_1}}d\o\right)\sum_{j=l}^{N_{\beta_1}}1_{\chi_j(\o,\bk)\neq 0}\notag\\
&\quad\cdot |\log(\det(I_{2^d}-(i\o
 I_{2^d}-\cE(\bk)-G_{l+1}(\beta_1)(\o,\bk))^{-1}\notag\\
&\qquad\qquad\qquad\qquad\cdot (G_{l}(\beta_1)(\o,\bk)-G_{l+1}(\beta_1)(\o,\bk))))\notag\\
&\qquad-\log(\det(I_{2^d}-(i\o
 I_{2^d}-\cE(\bk)-G_{l+1}(\beta_2)(\o,\bk))^{-1}\notag\\
&\qquad\qquad\qquad\qquad\quad\cdot (G_{l}(\beta_2)(\o,\bk)-G_{l+1}(\beta_2)(\o,\bk))))|\notag\\
&\quad + \frac{1}{2\pi L^d}\sum_{\bk\in
 \G(L)^*}\left(\int_{-\frac{\pi}{\beta_1}}^{-\frac{\pi}{\beta_2}}d\o+
\int^{\frac{\pi}{\beta_1}}_{\frac{\pi}{\beta_2}}d\o\right)
\sum_{j=l}^{N_{\beta_2}}1_{\chi_j(\o,\bk)\neq 0}\notag\\
&\qquad\cdot |\log(\det(I_{2^d}-(i\o
 I_{2^d}-\cE(\bk)-G_{l+1}(\beta_2)(\o,\bk))^{-1}\notag\\
&\qquad\qquad\qquad\qquad\quad\cdot (G_{l}(\beta_2)(\o,\bk)-G_{l+1}(\beta_2)(\o,\bk))))|\notag\\
&\le \frac{1}{2\pi L^d}\sum_{\bk\in \G(L)^*}
\left(\int^{-\frac{\pi}{\beta_1}}_{-\pi h}d\o+
      \int_{\frac{\pi}{\beta_1}}^{\pi
 h}d\o\right)\sum_{j=l}^{N_{\beta_1}}1_{\chi_j(\o,\bk)\neq 0}
c(d)M^{-j}\beta_1^{-\frac{1}{2}}\alpha^{-2}\notag\\
&\quad + \frac{1}{2\pi L^d}\sum_{\bk\in \G(L)^*}
\left(\int^{-\frac{\pi}{\beta_2}}_{-\frac{\pi}{\beta_1}}d\o+
      \int^{\frac{\pi}{\beta_1}}_{\frac{\pi}{\beta_2}}d\o\right)
\sum_{j=l}^{N_{\beta_2}}1_{\chi_j(\o,\bk)\neq 0} c(d)\alpha^{-2}\notag\\
&\le
 c(d,M)\beta_1^{-\frac{1}{2}}f_{\bt}^{-\frac{d}{2}}\alpha^{-2}M^{dl}.\notag
\end{align}
By coupling \eqref{eq_IR_final_continuous_discrete} with
 \eqref{eq_IR_final_continuous_continuous} we obtain the claimed
 inequality.
\end{proof}

Here we introduce sets of covariances. In the next
subsection we will see that the
actual covariances in the infrared integration belong to these sets. For $l\in\Z_{\le 0}$ we define the subset $\cR(D,c_0,M)(l)$
of $\Map(\overline{D},\Map(I_0^{2},\C))$ as follows. $C\in
\Map(\overline{D},\Map(I_0^2,\C))$ belongs to $\cR(D,c_0,M)(l)$ if and
only if the following statements hold.
\begin{enumerate}[(i)]
\item\label{item_IR_covariance_analyticity}
$C(\cdot)(\bX)\in C(\overline{D};\C)\cap C^{\o}(D;\C)$, $(\forall \bX\in
     I_0^2)$. 
\item\label{item_IR_covariance_determinant}
\begin{align}
&|\det(\<\bp_i,\bq_j\>_{\C^m}C(\bU)(X_i,Y_j))_{1\le i,j\le n}|\le
 (c_{0}M^{dl})^n,\label{eq_IR_covariance_determinant}\\
&\ (\forall m,n\in\N,\bp_i,\bq_i\in\C^m\text{ with }
\|\bp_i\|_{\C^m},\|\bq_i\|_{\C^m}\le 1,\notag\\
&\quad X_i,Y_i\in I_0\
 (i=1,2,\cdots,n),\bU\in\overline{D}).\notag
\end{align}
\item\label{item_IR_covariance_decay}
\begin{align}
\|\widetilde{C}(\bU)\|_{l-1,t}\le c_0M^{-l-tl},\quad (\forall
 t\in\{0,1\},\bU\in \overline{D}),\label{eq_IR_covariance_decay}
\end{align}
where $\widetilde{C}(\bU):I^2\to \C$ is the anti-symmetric extension of
     $C(\bU)$ defined as in \eqref{eq_covariance_anti_symmetrization}.
\item\label{item_IR_covariance_invariance}
\begin{align*}
\widetilde{C}(\bU)(\bX)=e^{iQ_2(S_2(\bX))}\widetilde{C}(\bU)(S_2(\bX)),\
 (\forall \bX\in I^2,\bU\in \overline{D}),
\end{align*}
for each $S:I\to I$, $Q:I\to \R$ defined in 
     \eqref{eq_IR_particle_hole}, \eqref{eq_IR_spin_up},
     \eqref{eq_IR_spin_inversion}, \eqref{eq_IR_translation},
     \eqref{eq_IR_momentum_inversion}.
\item\label{item_IR_covariance_invariance_complex}
\begin{align*}
\widetilde{C}(\bU)(\bX)=e^{-iQ_2(S_2(\bX))}\overline{\widetilde{C}(\overline{\bU})(S_2(\bX))},\
 (\forall \bX\in I^2,\bU\in \overline{D}),
\end{align*}
for each $S:I\to I$, $Q:I\to \R$ defined in
     \eqref{eq_IR_adjoint}, \eqref{eq_IR_inversion}.
\end{enumerate}

It will be important to measure the difference between the covariances
defined at different temperatures. For this purpose we introduce the
subset $\hat{\cR}(D,c_0,M)(l)$ of $\cR(\beta_1)(D,c_0,M)(l)\times
\cR(\beta_2)(D,c_0,M)(l)$ on the assumption
\eqref{eq_basic_beta_h_assumption} as
follows. $(C(\beta_1),C(\beta_2))\in \cR(\beta_1)(D,c_0,M)(l)\times
\cR(\beta_2)(D,c_0,M)(l)$ belongs to $\hat{\cR}(D,c_0,M)(l)$ if and only
if the following statements hold true. 
\begin{enumerate}[(i)]
\item\label{item_IR_covariance_determinant_difference}
\begin{align}
&|\det(\<\bp_i,\bq_j\>_{\C^m}C(\beta_1)(\bU)(R_{\beta_1}(X_i,Y_j)))_{1\le
  i,j\le n}\label{eq_IR_covariance_determinant_difference}\\
&\quad-
\det(\<\bp_i,\bq_j\>_{\C^m}C(\beta_2)(\bU)(R_{\beta_2}(X_i,Y_j)))_{1\le
  i,j\le n}|\notag\\
&\le \beta_1^{-\frac{1}{2}}M^{-l}(c_{0} M^{dl})^n,\notag\\
&(\forall m,n\in\N,\bp_i,\bq_i\in\C^m\text{ with }
\|\bp_i\|_{\C^m},\|\bq_i\|_{\C^m}\le 1,\notag\\
&\quad X_i,Y_i\in \hat{I}_0\
 (i=1,2,\cdots,n),\bU\in\overline{D}).\notag
\end{align}
\item\label{item_IR_covariance_decay_difference}
\begin{align}
|\widetilde{C}(\beta_1)(\bU)-\widetilde{C}(\beta_2)(\bU)|_{l-1}\le
 \beta_1^{-\frac{1}{2}}c_{0}M^{-2l},\ (\forall \bU\in\overline{D}).\label{eq_IR_covariance_decay_difference}
\end{align}
\end{enumerate}

Take any $l\in \{0,-1,\cdots,N_{\beta}\}$ and $G_l\in \cK(D,\alpha,M)(l)$.
The same estimation as in Lemma \ref{lem_IR_inverse_bound} ensures that
we can define $C_l\in \Map(\overline{D},$ \\
$\Map(I_0^2,\C))$ by
\begin{align}
C_l(\bU)(\rho\bx\s x,\eta\by\tau y):=&\frac{\delta_{\s,\tau}}{\beta
 L^d}\sum_{(\o,\bk)\in \cM\times
 \G(L)^*}e^{i\<\bx-\by,\bk\>+i(x-y)\o}\chi_l(\o,\bk)\label{eq_effective_covariance_general_definition}\\
&\cdot (i\o
 I_{2^d}-\cE(\bk)-G_l(\bU)(\o,\bk))^{-1}(\rho,\eta).\notag
\end{align}
In fact $C_l$ is intended to be a generalization of the actual
covariance appearing in the infrared integration process which we
perform in the next subsection.
Let us summarize properties of $C_l$. 

\begin{lemma}\label{lem_IR_covariance_inclusion}
Assume that 
\begin{align}
M\ge
 8(d+1)^2(c_{\chi}+(1+\sqrt{2})^2(8c_{\chi}+4\pi)).\label{eq_crucial_condition_M}
\end{align}
Then, there exist a constant $c(d,M,c_w,c_{\chi})\in\R_{>0}$ depending
 only on $d$, $M$, $c_w$, $c_{\chi}$ and a constant $c(d)\in \R_{>0}$
 depending only on $d$ such that the following statements hold if 
$c_0\ge c(d,M,c_w,c_{\chi})f_{\bt}^{-\frac{d}{2}}$ and $\alpha\ge
 c(d)$. 
\begin{enumerate}
\item\label{item_IR_covariance_inclusion}
$$
C_l\in \cR(D,c_0,M)(l).
$$
\item\label{item_IR_covariance_inclusion_difference}
Assume in addition that \eqref{eq_basic_beta_h_assumption} holds,
     $l\in\{0,-1,\cdots,N_{\beta_1}\}$ and \\$(G_l(\beta_1),G_l(\beta_2))\in
     \hat{\cK}(D,\alpha,M)(l)$. Then,
$$
(C_l(\beta_1),C_l(\beta_2))\in \hat{\cR}(D,c_0,M)(l).
$$
\end{enumerate}
\end{lemma}

\begin{remark} To guarantee that $C_l$, $(C_l(\beta_1),C_l(\beta_2))$
 satisfy \eqref{eq_IR_covariance_decay},
 \eqref{eq_IR_covariance_decay_difference} respectively, we use the
 condition \eqref{eq_crucial_condition_M}. 
The bound properties \eqref{eq_IR_covariance_decay},
 \eqref{eq_IR_covariance_decay_difference} are crucially important for
 changing the measurement with the scale $l$ to that with the next scale
 $l-1$ at every step of the infrared integration. We prefer to make
 explicit a sufficient condition for $M$ to justify the crux of our RG
 regime. Also, recall that the only condition of $M$ apart from the
 basic condition $M\ge 2$ so far is $M\ge
 c^{N_v^2}$ for some generic constant $c\in\R_{>0}$ in
 \eqref{eq_UV_parameter_conditions}. The inequality
 \eqref{eq_crucial_condition_M} is the second nontrivial condition imposed on $M$. 
\end{remark}

The next lemma will be useful in the proof of Lemma
\ref{lem_IR_covariance_inclusion}.

\begin{lemma}\label{lem_tool_for_decay_bound}
Let $A_j$,$B$,$C\in\R_{\ge 0}$ $(j=0,1,\cdots,d)$, $D\in \R_{>0}$  and
 assume that $A_j^nB\le CD^n(n!)^2$, $(\forall j\in \{0,1,\cdots,d\},
 n\in \N\cup\{0\})$. Then,
$$
B\le 4C e^{-\sum_{j=0}^d\left(\frac{A_j}{(d+1)^2D}\right)^{1/2}}.
$$
\end{lemma}
\begin{proof}
By assumption, 
\begin{align*}
\frac{1}{n!}\left(\frac{A_j}{4D}\right)^{\frac{n}{2}}B^{\frac{1}{2}}\le
 C^{\frac{1}{2}}\left(\frac{1}{2}\right)^n,\quad (\forall j\in \{0,1,\cdots,d\},
 n\in \N\cup\{0\}).
\end{align*}
By summing both sides over $n\in \N\cup\{0\}$ and squaring them we obtain
$$
e^{\left(\frac{A_j}{D}\right)^{1/2}}B\le 4C,\quad (\forall j\in \{0,1,\cdots,d\}),
$$
which leads to the result.
\end{proof}

\begin{proof}[Proof of Lemma \ref{lem_IR_covariance_inclusion}]
\eqref{item_IR_covariance_inclusion}: The expansion of the integrand of $C_l$
 as in \eqref{eq_neumann_series_expansion} converges uniformly with
 respect to $\bU\in \overline{D}$. This implies that $C_l$ satisfies the property
 \eqref{item_IR_covariance_analyticity} of $\cR(D,c_0,M)(l)$. Let us
 check that $C_l$ satisfies the invariant properties. The invariance
 with $S:I\to I$, $Q:I\to \R$ defined in \eqref{eq_IR_particle_hole},
 \eqref{eq_IR_spin_up}, \eqref{eq_IR_spin_inversion} is clear. The
 invariance with $S$, $Q$ defined in \eqref{eq_IR_translation} 
straightforwardly follows from the definitions. We can refer to the
 proof of the same invariance in \cite[\mbox{Lemma 7.13
 (3)}]{K15}. For $(\rho,\bx,\s,x),(\eta,\by,\tau,y)\in I_0$,
 $\bU\in\overline{D}$, 
\begin{align*}
&C_l(\bU)((\rho,r_L(-\bx-b(\rho)),\s,x),(\eta,r_L(-\by-b(\eta)),\tau,y))\\
&\quad\cdot e^{i\<-\bx-b(\rho),\frac{2\pi}{L}\beps^L\>+i\<b(\rho),\frac{\pi}{L}\beps^L\>}
e^{-i\<-\by-b(\eta),\frac{2\pi}{L}\beps^L\>-i\<b(\eta),\frac{\pi}{L}\beps^L\>}\\&=\frac{\delta_{\s,\tau}}{\beta
 L^d}\sum_{(\o,\bk)\in \cM\times
 \G(L)^*}e^{i\<\bx-\by,\bk\>+i(x-y)\o}\chi_l\left(\o,-\bk-\frac{2\pi}{L}\beps^L\right) U_d\left(\frac{\pi}{L}\beps^L\right)U_d(\bk)\\
&\quad\cdot \left(i\o
 I_{2^d}-\cE\left(-\bk-\frac{2\pi}{L}\beps^L\right)-G_l(\bU)\left(\o,-\bk-\frac{2\pi}{L}\beps^L\right)\right)^{-1}\\
&\quad\cdot U_d(\bk)^* U_d\left(\frac{\pi}{L}\beps^L\right)^*(\rho,\eta)\\
&=C_l(\bU)((\rho,\bx,\s,x),(\eta,\by,\tau,y)),
\end{align*}
where we used the facts
 $\chi_l(\o,-\bk-(2\pi/L)\beps^L)=\chi_l(\o,\bk)$,
 $\cE(\bk)=E(-\beps^L,-\btheta)(\bk)$, Lemma
 \ref{lem_hopping_properties} \eqref{item_unitary_hopping_invariance}
 and the invariance \eqref{item_kernel_momentum_inversion} of
 $\cK(D,\alpha,M)(l)$. The above equality implies the 
invariance with $S:I\to I$, $Q:I\to \R$
 defined in \eqref{eq_IR_momentum_inversion}. 
Thus, we have checked that all the invariances in the item
 \eqref{item_IR_covariance_invariance} of $\cR(D,c_0,M)(l)$ hold.

An argument based on the invariances $\cE(\bk)=\cE(\bk)^*$,
 $G_l(\bU)(\o,\bk)=G_l(\overline{\bU})(-\o,\bk)^*$ $(\forall \bU\in
 \overline{D},(\o,\bk)\in\cM\times ((2\pi/L)\Z)^d)$, parallel to the
 proof of  \cite[\mbox{Lemma 7.13 (4)}]{K15} shows the invariance
 $$\widetilde{C_l}(\bU)(\bX)=e^{-i Q_2(S_2(\bX))}\overline{\widetilde{C_l}(\overline{\bU})(S_2(\bX))}$$
 with $S$, $Q$ defined in \eqref{eq_IR_adjoint}. 

For $(\rho,\bx,\s,x),(\eta,\by,\tau,y)\in I_0$, 
\begin{align*}
&e^{i\<b(\rho),\bpi\>+i\<b(\eta),\bpi\>}\overline{C_l(\overline{\bU})(\rho\bx\s
 x,\eta\by \tau y)}\\
&=\frac{\delta_{\s,\tau}}{\beta L^d}\sum_{(\o,\bk)\in\cM\times \G(L)^*}
e^{i\<\by-\bx,\bk\>+i(y-x)\o}\chi_l(\o,\bk)\\
&\quad\cdot e^{i\<b(\rho),\bpi\>+i\<b(\eta),\bpi\>}(-i\o
 I_{2^d}-\cE(\bk)-G_l(\overline{\bU})(\o,\bk)^*)^{-1}(\eta,\rho)\\
&=-C_l(\bU)(\eta\by\tau y, \rho\bx\s x),
\end{align*}
where we used Lemma \ref{lem_hopping_properties}
 \eqref{item_unitary_hopping_minus} and the invariance in \eqref{item_kernel_inversion} of $\cK(D,\alpha,M)(l)$. This equality
 leads to the invariance with $S:I\to I$, $Q:I\to \R$ 
defined in \eqref{eq_IR_inversion}. Thus, $C_l$ satisfies the
 invariances stated in the item \eqref{item_IR_covariance_invariance_complex} of
 $\cR(D,c_0,M)(l)$. 

By combining Lemma
 \ref{lem_IR_cut_off_properties}
 \eqref{item_IR_cut_off_implication},\eqref{item_IR_cut_off_support},
 Lemma \ref{lem_IR_inverse_bound} with the standard application of Gram's
 inequality we can show that 
\begin{align}
&|\det(\<\bp_i,\bq_j\>_{\C^m}C_l(\bU)(X_i,Y_j))_{1\le i,j\le n}|\le
 (c(d,M)f_{\bt}^{-\frac{d}{2}}M^{dl})^n,\label{eq_IR_covariance_determinant_pre}\\
&\ (\forall m,n\in\N,\bp_i,\bq_i\in\C^m\text{ with }
\|\bp_i\|_{\C^m},\|\bq_i\|_{\C^m}\le 1,\notag\\
&\quad X_i,Y_i\in I_0\
 (i=1,2,\cdots,n),\bU\in\overline{D}).\notag
\end{align}
This means that $C_l$ satisfies the determinant bound
 \eqref{eq_IR_covariance_determinant} for any $c_0\ge
 c(d,M)f_{\bt}^{-\frac{d}{2}}$.

It remains to prove \eqref{eq_IR_covariance_decay}. Take $j\in
 \{0,1,\cdots,d\}$. By Lemma \ref{lem_hopping_properties}
 \eqref{item_derivative_hopping_upper},
 \eqref{eq_kernel_bound_derivative}, the facts that $\|A\|_{2^d\times
 2^d}\le 2^d\max_{\rho,\eta\in \cB}|A(\rho,\eta)|$ $(\forall A\in
 \Mat(2^d,\C))$, $(2n)!\le 2^{2n}(n!)^2$ $(\forall n\in\N)$, $\fw(0)\le
 1$ and the condition $\alpha^2\ge 2^d$ we have
\begin{align}
&\left\|\left(\frac{\partial}{\partial k_j}\right)^n(i\o I_{2^d}-\cE(\bk)-G_l(\o,\bk))
\right\|_{2^d\times 2^d}\label{eq_derivative_bound_for_decay}\\
&\le 1 + 2^d\alpha^{-2}M^l (2c_{\chi}+\pi)^n \fw(l)^{-n}(2n)!\notag\\
&\le 2M^{l}(8c_{\chi}+4\pi)^n\fw(l)^{-n}(n!)^2,\quad (\forall n\in \N_{\ge 1},(\o,\bk)\in\R^{d+1}).\notag
\end{align}
Take any $(\o,\bk)\in\R^{d+1}$ satisfying $\chi_l(\o,\bk)\neq 0$. 
By \eqref{eq_derivative_bound_for_decay} and Lemma
 \ref{lem_IR_inverse_bound} we can apply \cite[\mbox{Lemma C.3 (2)}]{K15}
 with $s=M^{-l}$, $q=2M^l$, $r=(8c_{\chi}+4\pi)\fw(l)^{-1}$, $t=2$ to
 deduce that 
\begin{align}
&\left\|\left(\frac{\partial}{\partial k_j}\right)^n(i\o I_{2^d}-\cE(\bk)-G_l(\o,\bk))^{-1}
\right\|_{2^d\times 2^d}\label{eq_mono_integrant_derivative}\\
&\le \frac{M^{-2l}\cdot 2M^l}{(1+(M^{-l}\cdot 2M^l)^{\frac{1}{2}})^2}
\big((8c_{\chi}+4\pi)\fw(l)^{-1}(1+(M^{-l}\cdot
 2M^l)^{\frac{1}{2}})^2\big)^n(n!)^2\notag\\
&\le M^{-l}\big((1+\sqrt{2})^2(8c_{\chi}+4\pi)\fw(l)^{-1}\big)^n(n!)^2,\quad
 (\forall n\in \N\cup\{0\}).\notag
\end{align}
Moreover, by Lemma \ref{lem_IR_cut_off_properties}
 \eqref{item_IR_cut_off_derivative},
\begin{align}
&\left\|\left(\frac{\partial}{\partial k_j}\right)^n\chi_l(\o,\bk)(i\o I_{2^d}-\cE(\bk)-G_l(\o,\bk))^{-1}
\right\|_{2^d\times 2^d}\label{eq_total_integrant_derivative}\\
&\le
 M^{-l}\big((c_{\chi}+(1+\sqrt{2})^2(8c_{\chi}+4\pi))\fw(l)^{-1}\big)^n(n!)^2,\notag\\
&(\forall n\in \N\cup\{0\},(\o,\bk)\in\R^{d+1}).\notag
\end{align}
By using the above inequality and Lemma
 \ref{lem_IR_cut_off_properties} \eqref{item_IR_cut_off_implication},\eqref{item_IR_cut_off_support} we can estimate as follows.
\begin{align*}
&\left\|\left(\frac{\beta}{2\pi}\right)^n(e^{-i(x-y)\frac{2\pi}{\beta}}-1)^nC_l(\cdot\bx\s
 x,\cdot \by\tau y)\right\|_{2^d\times 2^d}\\
&\le \frac{1}{\beta L^d}\sum_{(\o,\bk)\in\cM\times
 \G(L)^*}1_{\frac{1}{\beta}\le
 M_{IR}M^{N_{\beta}+1}}\prod_{j=1}^n\left(\frac{\beta}{2\pi}\int_{0}^{\frac{2\pi}{\beta}}d\o_j\right)\\
&\qquad\cdot\left\|\left(\frac{\partial}{\partial \o'}\right)^n\chi_l(\o',\bk)(i\o' I_{2^d}-\cE(\bk)-G_l(\o',\bk))^{-1}
\right\|_{2^d\times 2^d}\Bigg|_{\o'=\o+\sum_{j=1}^n\o_j}\\
&\le
 M^{-l}\big((c_{\chi}+(1+\sqrt{2})^2(8c_{\chi}+4\pi))\fw(l)^{-1}\big)^n(n!)^2\\
&\quad\cdot
 \prod_{j=1}^n\left(\frac{\beta}{2\pi}\int_{0}^{\frac{2\pi}{\beta}}dw_j\right)
\frac{1}{\beta L^d}\sum_{(\o,\bk)\in\cM\times
 \G(L)^*}1_{\frac{1}{\beta}\le
 M_{IR}M^{N_{\beta}+1}}1_{\chi_l(\o+\sum_{j=1}^n\o_j,\bk)\neq 0}\\
&\le
 c(M,d)f_{\bt}^{-\frac{d}{2}}M^{dl}\big((c_{\chi}+(1+\sqrt{2})^2(8c_{\chi}+4\pi))\fw(l)^{-1}\big)^n(n!)^2.
\end{align*}
By repeating the same procedure as above we have that
\begin{align*}
&|d_j(\bX)^n\widetilde{C_l}(\bX)|\\
&\le
 c(M,d)f_{\bt}^{-\frac{d}{2}}M^{dl}\big((c_{\chi}+(1+\sqrt{2})^2(8c_{\chi}+4\pi))\fw(l)^{-1}\big)^n(n!)^2,\\
&(\forall j\in \{0,1,\cdots,d\},n\in \N\cup\{0\},\bX\in I^2).
\end{align*}
Here we can apply Lemma \ref{lem_tool_for_decay_bound} to derive that
\begin{align}
&|\widetilde{C_l}(\bX)|\le
 c(M,d)f_{\bt}^{-\frac{d}{2}}M^{dl}e^{-\sum_{j=0}^d\left(\frac{\fw(l)d_j(\bX)}{(d+1)^2(c_{\chi}+(1+\sqrt{2})^2(8c_{\chi}+4\pi))}\right)^{1/2}},\quad (\forall \bX\in I^2).\label{eq_IR_covariance_exponential_decay}
\end{align}
Moreover, on the assumption \eqref{eq_crucial_condition_M},
\begin{align*}
&|\widetilde{C_l}(\bX)|\le
 c(M,d)f_{\bt}^{-\frac{d}{2}}M^{dl}e^{-2\sqrt{2}\sum_{j=0}^d(\fw(l-1)d_j(\bX))^{1/2}},\quad
 (\forall \bX\in I^2).
\end{align*}
which implies that
\begin{align*}
\|\widetilde{C_l}(\bU)\|_{l-1,t}\le
 c(M,d,c_w)f_{\bt}^{-\frac{d}{2}}M^{-l-tl},\quad (\forall t\in \{0,1\},
 \bU\in\overline{D}).
\end{align*}
Thus, if $c_0\ge c(M,d,c_w)f_{\bt}^{-\frac{d}{2}}$, the covariance $C_l$
satisfies the inequality
 \eqref{eq_IR_covariance_decay}.

\eqref{item_IR_covariance_inclusion_difference}: First note that the
 assumption $\beta_a\ge 1$ implies that $1/\beta_a\le
 M_{IR}M^{N_{\beta_a}+1}$ $(a=1,2)$ and thus the results of Lemma
 \ref{lem_IR_cut_off_properties} \eqref{item_IR_cut_off_support} for
 $\beta_1$, $\beta_2$ are available.
For $l\in \{0,-1,\cdots,N_{\beta_a}\}$, $a\in \{1,2\}$, define
 $C_{ont,l}(\beta_a)\in \Map(\overline{D},\Map(\hat{I}_0^2,\C))$ by 
\begin{align*}
&C_{ont,l}(\beta_a)(\bU)(\rho\bx\s x,\eta\by\tau y)\\
&:=(-1)^{n_{\beta_a}(x)+n_{\beta_a}(y)}\frac{\delta_{\s,\tau}}{2\pi
 L^d}\sum_{\bk\in\G(L)^*}\int_{-\pi h}^{\pi h}d\o
 e^{i\<\bx-\by,\bk\>+i(x-y)\o}\chi_{l}(\o,\bk)\\
&\qquad\cdot (i\o
 I_{2^d}-\cE(\bk)-G_l(\beta_a)(\bU)(\o,\bk))^{-1}(\rho,\eta).
\end{align*}
By taking into account Lemma \ref{lem_IR_cut_off_properties}
 \eqref{item_IR_cut_off_large_matsubara} we can justify the following
 transformation. For any $(\bx,\s,x),(\by,\tau,y)\in
 \G(L)\times\spin\times[-\beta_1/4,\beta_1/4)_h$,
\begin{align*}
&C_{ont,l}(\beta_a)(\cdot\bx\s x,\cdot\by\tau
 y)-C_{l}(\beta_a)(\cdot\bx\s r_{\beta_a}(x),\cdot\by\tau r_{\beta_a}(y))\\
&=(-1)^{n_{\beta_a}(x)+n_{\beta_a}(y)}\frac{\delta_{\s,\tau}}{2\pi
 L^d}\sum_{\bk\in \G(L)^*}e^{i\<\bx-\by,\bk\>}\sum_{m=0}^{\beta_a
 h}\int_{-\pi h -\frac{\pi}{\beta_a}+\frac{2\pi}{\beta_a}m}^{-\pi h
 -\frac{\pi}{\beta_a}+\frac{2\pi}{\beta_a}(m+1)}d\o\\
&\quad\cdot\int_{-\pi h-\frac{\pi}{\beta_a}+\frac{2\pi}{\beta_a}m}^{\o}du 
\frac{\partial}{\partial u}(e^{i(x-y)u}\chi_l(u,\bk)(iu
 I_{2^d}-\cE(\bk)-G_l(\beta_a)(u,\bk))^{-1}).
\end{align*}
Then, by Lemma \ref{lem_IR_cut_off_properties}
 \eqref{item_IR_cut_off_support} and
 \eqref{eq_total_integrant_derivative},
\begin{align}
&\|C_{ont,l}(\beta_a)(\cdot\bx\s x,\cdot\by \tau
 y)-C_l(\beta_a)(\cdot\bx\s r_{\beta_a}(x),\cdot\by \tau
 r_{\beta_a}(y))\|_{2^d\times
 2^d}\label{eq_covariance_continuous_normal}\\
&\le
 c(d,M,c_w,c_{\chi})\beta_1^{-1}f_{\bt}^{-\frac{d}{2}}M^{dl}(|x-y|+M^{-l}).
\notag
\end{align}
Calculation parallel to that leading to
 \eqref{eq_IR_covariance_exponential_decay} yields that
\begin{align}
&|\widetilde{C_{ont,l}}(\beta_a)(\bX)|\le
 c(M,d)f_{\bt}^{-\frac{d}{2}}M^{dl}e^{-\sum_{j=0}^d\left(\frac{\fw(l)\hat{d}_j(\bX)}{(d+1)^2(c_{\chi}+(1+\sqrt{2})^2(8c_{\chi}+4\pi))}\right)^{1/2}},\label{eq_continuous_covariance_decay}\\
&(\forall \bX\in \hat{I}^2).\notag
\end{align}
By using the inequality $d_j(R_{\beta_a}(\bX))\ge (2/\pi)\hat{d}_j(\bX)$
 $(\forall \bX\in \hat{I}^2)$ we can derive from
 \eqref{eq_IR_covariance_exponential_decay} that
\begin{align}
&|\widetilde{C_l}(\beta_a)(R_{\beta_a}(\bX))|\label{eq_IR_covariance_decay_another_metric}\\
&\le
 c(M,d)f_{\bt}^{-\frac{d}{2}}M^{dl}e^{-\sum_{j=0}^d\left(\frac{2\fw(l)\hat{d}_j(\bX)}{\pi
 (d+1)^2(c_{\chi}+(1+\sqrt{2})^2(8c_{\chi}+4\pi))}\right)^{1/2}},\quad (\forall \bX\in \hat{I}^2).\notag
\end{align}
By putting \eqref{eq_covariance_continuous_normal},
 \eqref{eq_continuous_covariance_decay},
 \eqref{eq_IR_covariance_decay_another_metric} together,
\begin{align}
&|\widetilde{C_{ont,l}}(\beta_a)(\bX)-\widetilde{C_l}(\beta_a)(R_{\beta_a}(\bX))|\label{eq_covariance_continuous_normal_decay}\\
&\le
 c(d,M,c_w,c_{\chi})\beta_1^{-\frac{1}{2}}f_{\bt}^{-\frac{d}{2}}M^{dl}(\hat{d}_0(\bX)^{\frac{1}{2}}+M^{-\frac{l}{2}})\notag\\
&\quad\cdot e^{-\sum_{j=0}^d\left(\frac{\fw(l)\hat{d}_j(\bX)}{2\pi(d+1)^2(c_{\chi}+(1+\sqrt{2})^2(8c_{\chi}+4\pi))}\right)^{1/2}}\notag\\
&\le
 c(d,M,c_w,c_{\chi})\beta_1^{-\frac{1}{2}}f_{\bt}^{-\frac{d}{2}}M^{(d-\frac{1}{2})l}
 e^{-\sum_{j=0}^d\left(\frac{\fw(l)\hat{d}_j(\bX)}{4\pi(d+1)^2(c_{\chi}+(1+\sqrt{2})^2(8c_{\chi}+4\pi))}\right)^{1/2}},\notag\\
&(\forall \bX\in \hat{I}^2,a\in\{1,2\}).\notag
\end{align}

We need to establish a decay bound on
 $C_{ont,l}(\beta_1)-C_{ont,l}(\beta_2)$. Remark that 
\begin{align*}
&C_{ont,l}(\beta_1)(\cdot\bx\s x,\cdot\by\tau
 y)-C_{ont,l}(\beta_2)(\cdot\bx\s x,\cdot\by\tau y)\\
&=(-1)^{1_{x<0}+1_{y<0}}\frac{\delta_{\s,\tau}}{2\pi
 L^d}\sum_{\bk\in \G(L)^*}\int_{-\pi h}^{\pi h}d\o 
e^{i\<\bx-\by,\bk\>+i(x-y)\o}\chi_{l}(\o,\bk)\\
&\quad\cdot(i\o I_{2^d}-\cE(\bk)-G_l(\beta_1)(\o,\bk))^{-1}
(G_l(\beta_1)(\o,\bk)-G_l(\beta_2)(\o,\bk))\\
&\quad\cdot(i\o I_{2^d}-\cE(\bk)-G_l(\beta_2)(\o,\bk))^{-1}.
\end{align*}
Then, by Lemma \ref{lem_IR_cut_off_properties}
 \eqref{item_IR_cut_off_support}, \eqref{eq_kernel_bound_difference},
 \eqref{eq_mono_integrant_derivative},
 \eqref{eq_total_integrant_derivative} and the fact $(n!)^2\le (2n)!\le 2^{2n}(n!)^2$ we have that for $j\in$ $\{0,1,\cdots,d\}$, $\bX\in \hat{I}^2$,
\begin{align*}
&\hat{d}_j(\bX)^n|\widetilde{C_{ont,l}}(\beta_1)(\bX)-\widetilde{C_{ont,l}}(\beta_2)(\bX)|\\
&\le
 c(M,d)f_{\bt}^{-\frac{d}{2}}M^{(d-1)l}\beta_1^{-\frac{1}{2}}\alpha^{-2}\\
&\quad\cdot \sum_{m_1=0}^n\left(\begin{array}{c} n \\ m_1\end{array}\right)
\big((c_{\chi}+(1+\sqrt{2})^2(8c_{\chi}+4\pi))\fw(l)^{-1}\big)^{m_1}(m_1!)^2\\
&\quad\cdot \sum_{m_2=0}^{n-m_1}\left(\begin{array}{c} n-m_1 \\ m_2\end{array}\right)
((2c_{\chi}+\pi^2)\fw(l)^{-1})^{m_2}(2m_2)!\\
&\quad\cdot ((1+\sqrt{2})^2(8c_{\chi}+4\pi)\fw(l)^{-1})^{n-m_1-m_2}((n-m_1-m_2)!)^2\\
&\le
 c(M,d)f_{\bt}^{-\frac{d}{2}}M^{(d-1)l}\beta_1^{-\frac{1}{2}}\alpha^{-2}(2n)!\\
&\quad\cdot\big((c_{\chi}+(1+\sqrt{2})^2(8c_{\chi}+4\pi))\fw(l)^{-1}
+(2c_{\chi}+\pi^2)\fw(l)^{-1}\\
&\qquad+(1+\sqrt{2})^2(8c_{\chi}+4\pi)\fw(l)^{-1}\big)^n\\
&\le
 c(M,d)f_{\bt}^{-\frac{d}{2}}M^{(d-1)l}\beta_1^{-\frac{1}{2}}\alpha^{-2}(n!)^2\\&\quad\cdot (4\pi (c_{\chi}+(1+\sqrt{2})^2(8c_{\chi}+4\pi))\fw(l)^{-1})^n,
\end{align*}
which combined with Lemma \ref{lem_tool_for_decay_bound} implies that
\begin{align}
&|\widetilde{C_{ont,l}}(\beta_1)(\bX)-\widetilde{C_{ont,l}}(\beta_2)(\bX)|\label{eq_covariance_continuous_continuous_decay}\\
&\le
 c(M,d)f_{\bt}^{-\frac{d}{2}}M^{(d-1)l}\beta_1^{-\frac{1}{2}}\alpha^{-2}
 e^{-\sum_{j=0}^d\left(\frac{\fw(l)\hat{d}_j(\bX)}{4\pi(d+1)^2(c_{\chi}+(1+\sqrt{2})^2(8c_{\chi}+4\pi))}\right)^{1/2}}.\notag
\end{align}

On the assumption \eqref{eq_crucial_condition_M}, the inequalities
 \eqref{eq_covariance_continuous_normal_decay},
 \eqref{eq_covariance_continuous_continuous_decay} yield that
\begin{align}
&|\widetilde{C_{l}}(\beta_1)(\bU)(R_{\beta_1}(\bX))-\widetilde{C_{l}}(\beta_2)(\bU)(R_{\beta_2}(\bX))|\label{eq_IR_covariance_decay_difference_pre}\\
&\le
 c(d,M,c_w,c_{\chi})\beta_1^{-\frac{1}{2}}f_{\bt}^{-\frac{d}{2}}M^{(d-1)l}
 e^{-\sqrt{2}\sum_{j=0}^d\left(\frac{1}{\pi}\fw(l-1)\hat{d}_j(\bX)\right)^{1/2}},\notag\\
&(\forall \bX\in \hat{I}^2,\bU\in\overline{D}),\notag
\end{align}
and thus,
\begin{align*}
&|\widetilde{C_l}(\beta_1)(\bU)-\widetilde{C_l}(\beta_2)(\bU)|_{l-1}\le 
 c(d,M,c_w,c_{\chi})f_{\bt}^{-\frac{d}{2}}
\beta_1^{-\frac{1}{2}}M^{-2l},\quad (\forall \bU\in\overline{D}).
\end{align*}
Therefore, the inequality \eqref{eq_IR_covariance_decay_difference}
 holds for $c_0\ge c(d,M,c_w,c_{\chi})f_{\bt}^{-\frac{d}{2}}$. 
By using \eqref{eq_IR_covariance_determinant_pre},
 \eqref{eq_IR_covariance_decay_difference_pre} and applying the
 Cauchy-Binet formula as in the proof of Lemma
 \ref{lem_UV_covariance_properties}
 \eqref{item_UV_covariance_determinant_difference} we can prove that
 \eqref{eq_IR_covariance_determinant_difference} holds for $c_0\ge
 c(d,M,c_w,c_{\chi})$ $f_{\bt}^{-\frac{d}{2}}$.
\end{proof}

We conclude this subsection by describing the recursive structure of
the infrared integration in terms of the scale-dependent sets of
Grassmann polynomials and covariances introduced so far. 
The proof of the following lemma is essentially based on the
general results  \cite[\mbox{Lemma 3.9, Proposition 5.6, Proposition
5.9}]{K15}. See  \cite[\mbox{Subsection 2.2}]{K15} for the meaning of
uniform convergence of a sequence of Grassmann polynomials.

\begin{lemma}\label{lem_IR_recursive_structure}
There exists a constant
 $c\in\R_{>0}$ independent of any parameter such that if 
\begin{align}
M^{d-\frac{3}{2}}\ge c,\quad \alpha\ge c
 M^{d+\frac{3}{2}},\label{eq_condition_M_recursive}
\end{align}
the following statements hold true.
\begin{enumerate}
\item\label{item_IR_recursive_structure}
If $l\in\Z_{<0}$, 
$$J^{l+1}\in\cS(D,c_0,\alpha,M)(l+1),\quad
 C_{l+1}\in\cR(D,c_0,M)(l+1),$$ 
then, 
\begin{align*}
\sum_{n=0}^{\infty}\frac{1}{n!}\left(\frac{d}{d z}\right)^n\Big|_{z=0}
\log\left(\int
     e^{z\sum_{m=4}^NJ_m^{l+1}(\bU)(\psi+\psi^1)}d\mu_{C_{l+1}(\bU)}(\psi^1)\right)
\end{align*}
uniformly converges with $\bU\in\overline{D}$. Let $J^l$ denote it. Then, 
$$
J^l\in\cS(D,c_0,\alpha,M)(l).
$$
\item\label{item_IR_recursive_structure_difference}
In addition, assume that \eqref{eq_basic_beta_h_assumption} holds and 
\begin{align*}
&(J^{l+1}(\beta_1),J^{l+1}(\beta_2))\in\hat{\cS}(D,c_0,\alpha,M)(l+1),\\
&(C_{l+1}(\beta_1),C_{l+1}(\beta_2))\in\hat{\cR}(D,c_0,M)(l+1).
\end{align*}
Then,
$$
(J^{l}(\beta_1),J^{l}(\beta_2))\in\hat{\cS}(D,c_0,\alpha,M)(l).
$$
\end{enumerate}
\end{lemma}
\begin{proof}
\eqref{item_IR_recursive_structure}: 
Let us define $F^l,\ T^{l,(n)}\in \Map(\overline{D},\bigwedge \cV)$
 $(n\in \N_{\ge 2})$ by
\begin{align*}
&F^{l}(\bU)(\psi):=\frac{d}{dz}\Big|_{z=0}
\log\left(\int
     e^{z\sum_{m=4}^NJ_m^{l+1}(\bU)(\psi+\psi^1)}d\mu_{C_{l+1}(\bU)}(\psi^1)\right),\\
&T^{l,(n)}(\bU)(\psi):=
\frac{1}{n!}\left(\frac{d}{d z}\right)^n\Big|_{z=0}
\log\left(\int
     e^{z\sum_{m=4}^NJ_m^{l+1}(\bU)(\psi+\psi^1)}d\mu_{C_{l+1}(\bU)}(\psi^1)\right).
\end{align*}
It is implied by \cite[\mbox{Proposition 5.6}]{K15}
 with $a_1=d$, $a_2=1$, $a_3=1$, $a_4=1/2$ that on the assumption
 \eqref{eq_condition_M_recursive}
\begin{align*}
&\frac{h}{N}\Bigg(|F_0^l(\bU)|+\sum_{n=2}^{\infty}|T_0^{l,(n)}(\bU)|\Bigg)
\le M^{(d+\frac{3}{2})l}\alpha^{-1},\\
&M^{-(d+\frac{3}{2})l+tl}\sum_{m=2}^Nc_0^{\frac{m}{2}}\alpha^mM^{\frac{d}{2}lm}
\Bigg(\|F_m^l(\bU)\|_{l,t}+\sum_{n=2}^{\infty}\|T_m^{l,(n)}(\bU)\|_{l,t}
\Bigg)\le 1,\\
&(\forall \bU\in\overline{D},\ t\in \{0,1\}).
\end{align*}
Moreover, it is clear from the derivation of the inequalities
 ``(5.63)'', ``(5.66)'' in the proof of \cite[\mbox{Proposition
 5.6}]{K15} that
\begin{align*}
\sum_{n=2}^{\infty}\sup_{\bU\in
 \overline{D}}\Bigg(|T_0^{l,(n)}(\bU)|+\sum_{m=2}^N\|T_m^{l,(n)}(\bU)\|_{l,0}\Bigg)<\infty.
\end{align*}
Thus, $F^l+\sum_{n=2}^{\infty}T^{l,(n)}$ uniformly converges with
 respect to $\bU\in \overline{D}$.
By the definition of the free integration and the tree formula (see,
 e.g. \cite[\mbox{Theorem 3}]{SW}), $F^l,T^{l,(n)}$ $(n\in\N_{\ge 2})$
 consist of finite sums and products of $J^{l+1}, C_{l+1}$. Thus,
$F^l,T^{l,(n)}\in C(\overline{D};\bigwedge \cV)\cap C^{\o}(D;\bigwedge
 \cV)$ $(\forall n\in \N_{\ge 2})$. Therefore, the uniform convergent
 property ensures that $J^l\in C(\overline{D};\bigwedge \cV)\cap
 C^{\o}(D;\bigwedge \cV)$. The above inequalities imply the bound properties
 \eqref{eq_IR_boundedness_0th}, \eqref{eq_IR_boundedness}. We can apply \cite[\mbox{Lemma 3.9}]{K15}
 to prove that $J^l$ inherits the invariant properties claimed in the
 items \eqref{item_IR_invariances}, \eqref{item_IR_invariances_complex}
 of $\cS(D,c_0,\alpha,M)(l)$ from $J^{l+1}$ and $C_{l+1}$. Therefore,
 $J^l\in \cS(D,c_0,\alpha,M)(l)$. 

\eqref{item_IR_recursive_structure_difference}: On the assumption
 \eqref{eq_condition_M_recursive} the claim follows from \\
\cite[\mbox{Proposition 5.9}]{K15} with $a_1=d$,
 $a_2=1$, $a_3=1$, $a_4=1/2$. 
\end{proof}

\subsection{Completion of the infrared integration}\label{subsec_completion_IR}
Here we implement the infrared integration scheme to prove Theorem
\ref{thm_main_theorem}. Most of the necessary tools for justifying the
multi-scale integration have already been prepared in the preceding
sections. By putting together these lemmas and a lemma separately made in
Appendix \ref{app_h_L_limit} we will reach the proof of Theorem
\ref{thm_main_theorem}. First of all let us describe properties of the
output of the Matsubara UV integration in terms of the sets 
$\cS(D,c_0,\alpha,M)(0)$, $\hat{\cS}(D,c_0,\alpha,M)(0)$.
In the following $C_{>0}^{\delta}:I^2_0\to \C$ $(\delta=+,-)$ are the
covariances defined in \eqref{eq_UV_covariance_+},
\eqref{eq_UV_covariance_-} with the cut-off function $\phi(M_{UV}^{-2}h^2|1-e^{i\o/h}|^2)$ in
place of $\chi(h|1-e^{i\o/h}|)$.

\begin{lemma}\label{lem_input_to_IR_integration}
There exist a constant $c\in\R_{>0}$ independent of any parameter and a
 constant $c(M,d)\in \R_{\ge 1}$ depending only on $M$, $d$ such that if
 \eqref{eq_UV_parameter_conditions} holds with $c$, the following
 statements hold for any $c_0\ge c(M,d)$ and $r\in\R_{>0}$ satisfying 
\begin{align}
r e^{d\fw(0)^{1/2}}\sum_{m=1}^{N_v}c_0^m\alpha^{2m}v_m(\fw(0))\le
 \frac{1}{2}.
\label{eq_initial_data_adjustment}
\end{align}
\begin{enumerate}
\item\label{item_connection_to_original_formulation}
There exist 
     $r(\beta,L)\in\R_{>0}$ dependent on $\beta$, $L$, independent of
     $h$ and $J^0\in\cS(D(r)^{n_v},c_0,\alpha,M)(0)$ such that 
\begin{align}
&J^0(\bU)(\psi)=\frac{1}{2}\sum_{\delta\in\{+,-\}}\log\left(\int
 e^{-V^{\delta}(\bU)(\psi+\psi^1)}d\mu_{C_{>0}^{\delta}}(\psi^1)\right)+\beta
 V_0^L(\bU),\label{eq_connection_to_original_formulation}\\
&(\forall \bU\in \overline{D(r(\beta,L))}^{n_v}).\notag
\end{align}
\item\label{item_input_anisothermal}
On the assumption \eqref{eq_basic_beta_h_assumption},
$$
(J^0(\beta_1),J^0(\beta_2))\in \hat{\cS}(D(r)^{n_v},c_0,\alpha,M)(0).
$$
\end{enumerate}
\end{lemma}

\begin{proof}
There exists a constant
 $c(M,d)\in \R_{\ge 1}$ depending only on $M$, $d$ such that the
 conclusions of Lemma \ref{lem_UV_covariance_properties} hold for any
 $c_0\ge c(M,d)$. Fix such $c_0$. Let
 $F^{N_h,\delta}(\psi),T^{N_h,\delta}(\psi),J^{N_h,\delta}(\psi),F^{l,\delta}(\psi),T^{l,\delta,(n)}(\psi),J^{l,\delta}(\psi)$
 $(\in
 \bigwedge \cV)$ $(\delta\in
 \{+,-\},l\in\{0,1,\cdots,N_h-1\},n\in\N_{\ge 2})$ be defined as in the beginning of Subsection \ref{subsec_isothermal}. Here we explicitly show the dependency on the parameter $\delta$, while
 we concealed it in Subsection \ref{subsec_isothermal} and Subsection
 \ref{subsec_anisothermal}. Then, set $J^0(\psi):=(J^{0,+}(\psi)+J^{0,-}(\psi))/2$. By Lemma
 \ref{lem_UV_integration} and Lemma \ref{lem_UV_integration_difference}
 there exists a constant $c\in\R_{>0}$ independent of any parameter such
 that if \eqref{eq_UV_parameter_conditions} holds with $c$, the
 following bounds hold with any $r\in\R_{>0}$
 satisfying \eqref{eq_initial_data_adjustment}.
\begin{align*}
&\frac{h}{N}|J_0^0(\bU)|\le \alpha^{-1},\\
&\sum_{m=2}^Nc_0^{\frac{m}{2}}\alpha^m\|J_m^0(\bU)\|_{0,t}\le 1,\quad
 (\forall t\in\{0,1\},\bU\in \overline{D(r)}^{n_v}).
\end{align*}
On the assumption \eqref{eq_basic_beta_h_assumption},
\begin{align*}
&\left|\frac{h}{N(\beta_1)}J_0^0(\beta_1)(\bU)-\frac{h}{N(\beta_2)}J_0^0(\beta_2)(\bU)\right|\le \beta_1^{-\frac{1}{2}}\alpha^{-1},\\
&\sum_{m=2}^{N(\beta_2)}c_0^{\frac{m}{2}}\alpha^m|J_m^0(\beta_1)(\bU)-J_m^0(\beta_2)(\bU)|_{0}\le \beta_1^{-\frac{1}{2}},\quad
 (\forall \bU\in \overline{D(r)}^{n_v}).
\end{align*}
Moreover, by Lemma \ref{lem_UV_integration}
\begin{align}
&\sum_{n=2}^{\infty}\sup_{\bU\in
 \overline{D(r)}^{n_v}}\sum_{\delta\in\{+,-\}}\left(|T_0^{l,\delta,(n)}(\bU)|+\sum_{m=2}^N\|T_m^{l,\delta,(n)}(\bU)\|_{0,0}\right)<\infty,\label{eq_uniform_convergence_UV_application}\\
&(\forall l\in \{0,1,\cdots,N_h-1\}).\notag
\end{align}
Let us prove that 
\begin{align}
J^{l,\delta}\in C\left(\overline{D(r)}^{n_v};\bigwedge \cV\right)\cap C^{\o}\left(D(r)^{n_v};\bigwedge
 \cV\right)\quad (\forall \delta\in \{+,-\})
\label{eq_uniform_convergence_UV_induction}
\end{align}
for any $l\in \{0,1,\cdots,N_h\}$.
It is clear from the definition that
 \eqref{eq_uniform_convergence_UV_induction} holds for $l=N_h$.
Assume that \eqref{eq_uniform_convergence_UV_induction} holds for $l+1$.
Then, by definition $F^{l,\delta},T^{l,\delta,(n)}\in
 C(\overline{D(r)}^{n_v};\bigwedge \cV)\cap C^{\o}(D(r)^{n_v};\bigwedge
 \cV)$ $(\forall \delta\in \{+,-\},n\in\N_{\ge 2})$.
The bound property \eqref{eq_uniform_convergence_UV_application} implies
 that $\sum_{n=2}^{\infty}T^{l,\delta,(n)}$ uniformly converges. Thus
\eqref{eq_uniform_convergence_UV_induction} holds for $l$. By induction,
\eqref{eq_uniform_convergence_UV_induction} holds for any $l\in
 \{0,1,\cdots,N_h\}$. We especially have
 that $J^{0}\in
 C(\overline{D(r)}^{n_v};\bigwedge \cV)\cap C^{\o}(D(r)^{n_v};\bigwedge
 \cV)$.

By the same argument as in the proof of 
\cite[\mbox{Proposition 6.4
 (3)}]{K15} we can conclude that there exists $r(\beta,L)\in\R_{>0}$ depending on $\beta$, $L$ and
 independent of $h$ such that
\begin{align}
&J^0(\bU)(\psi)=\frac{1}{2}\sum_{\delta\in\{+,-\}}\log\left(\int
 e^{-V^{\delta}(\bU)(\psi+\psi^1)+\beta V_0^L(\bU)}d\mu_{C_{>0}^{\delta}}(\psi^1)\right),\label{eq_connection_to_original_formulation_pre}\\
&(\forall \bU\in \overline{D(r(\beta,L))}^{n_v}).\notag
\end{align}
We can see from the properties of $V^{\delta}(\bU)(\psi)$, $V_0^L(\bU)$
 and the definition of logarithm of Grassmann polynomial (see,
 e.g. \cite[\mbox{Subsection 2.2}]{K15}) that the right-hand side of
 \eqref{eq_connection_to_original_formulation_pre} is equal to that of
\eqref{eq_connection_to_original_formulation} if
 $\max_{j\in\{1,2,\cdots,n_v\}}|U_j|$ is sufficiently small. 
The inequality \eqref{eq_grassmann_0th_with_1} implies that there exists
 $r(\beta,L)'\in \R_{>0}$ dependent on $\beta,L$, independent of $h$
 such that the right-hand side of
 \eqref{eq_connection_to_original_formulation} is analytic with $\bU$ in
 $D(r(\beta,L)')^{n_v}$. Thus, by using the identity theorem and
 continuity and taking $r(\beta,L)$ smaller independently of $h$ if
 necessary we obtain the equality
 \eqref{eq_connection_to_original_formulation} for $\bU\in
 \overline{D(r(\beta,L))}^{n_v}$.

It remains to check that $J^0$ satisfies the invariant properties. 
Recall the definitions \eqref{eq_Grassmann_interaction},
 \eqref{eq_Grassmann_interaction_plus_minus}. The invariance
\begin{align}
&-V^{\delta}(\bU)(\cR\psi)+\beta V_0^L(\bU)=-V^{\delta}(\bU)(\psi)+\beta
 V_0^L(\bU),\label{eq_initial_interaction_invariance}\\
&(\forall \bU\in \overline{D(r)}^{n_v},\delta \in \{+,-\})\notag
\end{align}
for $S:I\to I$, $Q:I\to \R$ defined in \eqref{eq_IR_particle_hole},
 \eqref{eq_IR_spin_up}, \eqref{eq_IR_spin_inversion},
 \eqref{eq_IR_translation}
follows from the
 definition of $V^{\delta}$, \eqref{eq_spin_parity}, \eqref{eq_spin_reflection},
\eqref{eq_periodicity} and \eqref{eq_translation} respectively. 
The properties \eqref{eq_periodicity}, \eqref{eq_inversion} and
 \eqref{eq_U1_invariance} imply
 \eqref{eq_initial_interaction_invariance} for $S,Q$ defined in
 \eqref{eq_IR_momentum_inversion} as well.
Moreover, the property
 \eqref{eq_hermiticity} ensures that
\begin{align}
&-\overline{V^{\delta}(\overline{\bU})}(\cR\psi)+\beta \overline{V_0^L(\overline{\bU})}=-V^{\delta}(\bU)(\psi)+\beta
 V_0^L(\bU),\label{eq_initial_interaction_invariance_complex}\\
&(\forall \bU\in \overline{D(r)}^{n_v},\delta \in \{+,-\})\notag
\end{align}
for $S:I\to I$, $Q:I\to \R$ defined in \eqref{eq_IR_adjoint}.
It also follows from the definition of $V^{\delta}(\psi)$, \eqref{eq_U1_invariance} and \eqref{eq_hermiticity} that
\begin{align}
&-\overline{V^{\delta}(\overline{\bU})}(\cR\psi)+\beta \overline{V_0^L(\overline{\bU})}=-V^{-\delta}(\bU)(\psi)+\beta
 V_0^L(\bU),\label{eq_initial_interaction_invariance_complex_opposite}\\
&(\forall \bU\in \overline{D(r)}^{n_v},\delta \in \{+,-\})\notag
\end{align}
for $S:I\to I$, $Q:I\to \R$ defined in \eqref{eq_IR_inversion}.

Next let us confirm some invariances involving the covariances
 $C_{>0}^+$, $C_{>0}^-$. 
Let $S:I\to I$, $Q:I\to \R$ be one of those defined in \eqref{eq_IR_particle_hole},
 \eqref{eq_IR_spin_up}, \eqref{eq_IR_spin_inversion},
 \eqref{eq_IR_translation}, \eqref{eq_IR_momentum_inversion}. It is the
 same procedure as in the proof of Lemma
 \ref{lem_IR_covariance_inclusion} \eqref{item_IR_covariance_inclusion}
 to prove that
\begin{align*}
e^{iQ_2(S_2(\bX))}\widetilde{C_{>0}^{\delta}}(S_2(\bX))=\widetilde{C_{>0}^{\delta}}(\bX),\quad
 (\forall \bX\in I^2,
 \delta\in\{+,-\}).
\end{align*}
Moreover, we can see that for any $\bX\in I^m$, $\delta\in \{+,-\}$, 
\begin{align}
\int \psi_{\bX}d\mu_{C_{>0}^{\delta}}(\psi)&=e^{-\sum_{\bY\in
 I^2}\widetilde{C_{>0}^{\delta}}(\bY)\frac{\partial}{\partial
 \psi_{\bY}}}\psi_{\bX}\Big|_{\psi=0}\label{eq_UV_covariance_invariance_integral}\\
&= e^{-\sum_{\bY\in
 I^2}e^{-iQ_2(S_2(\bY))}\widetilde{C_{>0}^{\delta}}(\bY)\frac{\partial}{\partial
 \psi_{S_2(\bY)}}}e^{iQ_m(S_m(\bX))}\psi_{S_m(\bX)}\Big|_{\psi=0}\notag\\
&=e^{-\sum_{\bY\in
 I^2}\widetilde{C_{>0}^{\delta}}(\bY)\frac{\partial}{\partial
 \psi_{\bY}}}(\cR \psi)_{\bX}\Big|_{\psi=0}\notag\\
&=\int(\cR\psi)_{\bX}d\mu_{C_{>0}^{\delta}}(\psi).\notag
\end{align}
Let $S:I\to I$, $Q:I\to \R$ be defined in \eqref{eq_IR_adjoint}. By
 repeating the argument parallel to the proof of \cite[\mbox{Lemma 7.13
 (4)}]{K15} we obtain that 
\begin{align*}
e^{-iQ_2(S_2(\bX))}\overline{\widetilde{C_{>0}^{\delta}}(S_2(\bX))}=\widetilde{C_{>0}^{\delta}}(\bX),\quad
 (\forall \bX\in I^2,
 \delta\in\{+,-\}).
\end{align*}
Moreover, for any $\bX\in I^m$, $\delta \in \{+,-\}$, 
\begin{align*}
\int \psi_{\bX}d\mu_{C_{>0}^{\delta}}(\psi)
&= e^{-\sum_{\bY\in
 I^2}e^{iQ_2(S_2(\bY))}\widetilde{C_{>0}^{\delta}}(\bY)\frac{\partial}{\partial
 \psi_{S_2(\bY)}}}e^{-iQ_m(S_m(\bX))}\psi_{S_m(\bX)}\Big|_{\psi=0}\\
&=e^{-\sum_{\bY\in
 I^2}\overline{\widetilde{C_{>0}^{\delta}}(\bY)}\frac{\partial}{\partial
 \psi_{\bY}}}e^{-iQ_m(S_m(\bX))}\psi_{S_m(\bX)}\Big|_{\psi=0}\\
&=\int
 e^{-iQ_m(S_m(\bX))}\psi_{S_m(\bX)}d\mu_{\overline{C_{>0}^{\delta}}}(\psi),
\end{align*}
or
\begin{align}
\int\psi_{\bX}d\mu_{\overline{C_{>0}^{\delta}}}(\psi)=\int (\cR\psi)_{\bX}d\mu_{C_{>0}^{\delta}}(\psi).\label{eq_UV_covariance_invariance_integral_complex}
\end{align}
For $S$, $Q$ defined in \eqref{eq_IR_inversion} and
 $(\rho,\bx,\s,x,\theta)$, $(\eta,\by,\tau,y,\xi)\in I$, 
\begin{align}
&e^{-iQ_2(S_2(\rho\bx \s x \theta,\eta \by\tau y
 \xi))}\overline{\widetilde{C_{>0}^{\delta}}(S_2(\rho\bx \s x \theta,\eta
 \by\tau y \xi))}\label{eq_inversion_decomposition}\\
&=-\frac{1}{2}e^{i\<b(\rho),\bpi\>+i\<b(\eta),\bpi\>}\notag\\
&\quad\cdot(1_{(\theta,\xi)=(1,-1)}
\overline{C_{>0}^{\delta}(\eta \by\tau y, \rho \bx \s x)}
-1_{(\theta,\xi)=(-1,1)}
\overline{C_{>0}^{\delta}(\rho\bx \s x ,\eta \by\tau y)}).\notag
\end{align}
Moreover, by Lemma \ref{lem_hopping_properties}
 \eqref{item_unitary_hopping_minus},
\begin{align*}
&e^{i\<b(\rho),\bpi\>+i\<b(\eta),\bpi\>}\overline{C_{>0}^{+}(\rho\bx \s x
 ,\eta \by\tau y)}\\
&=\frac{\delta_{\s,\tau}}{\beta L^d}\sum_{(\o,\bk)\in\cM_h\times
 \G(L)^*} e^{i\<\by-\bx,\bk\>+i(y-x)\o}(1-\chi_{h,0}(\o))\\
&\quad\cdot e^{i\<b(\rho),\bpi\>+i\<b(\eta),\bpi\>}h^{-1}(I_{2^d}-e^{i\frac{\o}{h}I_{2^d}+\cE(\bk)})^{-1}(\eta,\rho)\\
&=\frac{\delta_{\s,\tau}}{\beta L^d}\sum_{(\o,\bk)\in\cM_h\times
 \G(L)^*} e^{i\<\by-\bx,\bk\>+i(y-x)\o}(1-\chi_{h,0}(\o))h^{-1}\\
&\quad\cdot U_d(\bpi)(I_{2^d}-e^{i\frac{\o}{h}I_{2^d}+\cE(\bk)})^{-1}U_d(\bpi)^*(\eta,\rho)\\
&=-C_{>0}^-(\eta \by\tau y,\rho\bx \s x).
\end{align*}
This implies that for $\delta\in \{+,-\}$,
\begin{align*}
e^{i\<b(\rho),\bpi\>+i\<b(\eta),\bpi\>}\overline{C_{>0}^{\delta}(\rho\bx \s x
 ,\eta \by\tau y)}=-C_{>0}^{-\delta}(\eta \by\tau y,\rho\bx \s x).
\end{align*}
By substituting this equality into \eqref{eq_inversion_decomposition} we
 obtain 
\begin{align*}
e^{-iQ_2(S_2(\bX))}\overline{\widetilde{C_{>0}^{\delta}}(S_2(\bX))}=\widetilde{C_{>0}^{-\delta}}(\bX),\quad
 (\forall \bX\in I^2,\delta
 \in\{+,-\}).
\end{align*}
Furthermore, based on this equality, transformations parallel to the derivation of
 \eqref{eq_UV_covariance_invariance_integral_complex} yield
\begin{align}
\int\psi_{\bX}d\mu_{\overline{C_{>0}^{\delta}}}(\psi)=\int (\cR\psi)_{\bX}d\mu_{C_{>0}^{-\delta}}(\psi),\quad
 (\forall \bX\in I^m,\delta
 \in\{+,-\}).\label{eq_UV_covariance_invariance_integral_complex_opposite}
\end{align}

Fix $\bU\in \overline{D(r(\beta,L))}^{n_v}$. Let $S$, $Q$ be one of
 those defined in \eqref{eq_IR_particle_hole},
 \eqref{eq_IR_spin_up}, \eqref{eq_IR_spin_inversion},
 \eqref{eq_IR_translation}, \eqref{eq_IR_momentum_inversion}. 
By \eqref{eq_connection_to_original_formulation_pre}, 
\eqref{eq_initial_interaction_invariance} and
\eqref{eq_UV_covariance_invariance_integral}, 
\begin{align*}
J^0(\bU)(\cR\psi)&=\frac{1}{2}\sum_{\delta\in\{+,-\}}\log\left(\int
 e^{-V^{\delta}(\bU)(\cR\psi+\cR\psi^1)+\beta
 V_0^L(\bU)}d\mu_{C_{>0}^{\delta}}(\psi^1)\right)\\
&=J^0(\bU)(\psi).
\end{align*}
Let $S$, $Q$ be defined by \eqref{eq_IR_adjoint}. By 
\eqref{eq_connection_to_original_formulation_pre}, 
\eqref{eq_initial_interaction_invariance_complex} and 
\eqref{eq_UV_covariance_invariance_integral_complex}, 
\begin{align*}
\overline{J^0(\overline{\bU})}(\cR\psi)&=\frac{1}{2}\sum_{\delta\in\{+,-\}}\log\left(\int
 e^{-\overline{V^{\delta}(\overline{\bU})}(\cR\psi+\cR\psi^1)+\beta
\overline{V_0^L(\overline{\bU})}}d\mu_{C_{>0}^{\delta}}(\psi^1)\right)\\
&=J^0(\bU)(\psi).
\end{align*}
Finally, let $S$, $Q$ be defined by \eqref{eq_IR_inversion}. By 
\eqref{eq_connection_to_original_formulation_pre}, 
\eqref{eq_initial_interaction_invariance_complex_opposite} and 
\eqref{eq_UV_covariance_invariance_integral_complex_opposite},
\begin{align*}
\overline{J^0(\overline{\bU})}(\cR\psi)&=\frac{1}{2}\sum_{\delta\in\{+,-\}}\log\left(\int
 e^{-\overline{V^{\delta}(\overline{\bU})}(\cR\psi+\cR\psi^1)+\beta
\overline{V_0^L(\overline{\bU})}}d\mu_{C_{>0}^{-\delta}}(\psi^1)\right)\\
&=J^0(\bU)(\psi).
\end{align*}
Since $\bU\mapsto J^0(\bU)(\psi)$, $\bU\mapsto J^0(\bU)(\cR\psi)$, 
 $\bU\mapsto \overline{J^0(\overline{\bU})}(\cR\psi)$ are continuous in
 $\overline{D(r)}^{n_v}$ and analytic in $D(r)^{n_v}$, the identity
 theorem and the continuity ensure the claimed invariances for any
 $\bU\in \overline{D(r)}^{n_v}$.
\end{proof}

Note that we can choose a constant $c(d)\in\R_{>0}$ depending only on
$d$, a constant $c(d,c_{\chi},N_v)\in\R_{>0}$ depending only on $d$,
$c_{\chi}$, $N_v$ so that if 
$$
h\ge e^{4d},\quad L\ge \beta,\quad M\ge c(d,c_{\chi},N_v),\quad \alpha
\ge c(d)M^{d+\frac{3}{2}},
$$
all the assumptions imposed on $h$, $L$, $\beta$, $\alpha$, $M$ in 
Lemma \ref{lem_self_energy_properties},
Lemma \ref{lem_IR_final_output},
Lemma \ref{lem_IR_covariance_inclusion},
Lemma \ref{lem_IR_recursive_structure},
Lemma \ref{lem_input_to_IR_integration} are satisfied. 
Then, there exists a constant $c(d,M,c_w,c_{\chi})\in\R_{>0}$ depending
only on $d$, $M$, $c_w$, $c_{\chi}$ such that the conclusions of  
Lemma \ref{lem_self_energy_properties},
Lemma \ref{lem_IR_covariance_inclusion},
Lemma \ref{lem_input_to_IR_integration} hold for
$c_0:=c(d,M,c_w,c_{\chi})f_{\bt}^{-\frac{d}{2}}$. With this $c_0$, set
\begin{align*}
r_{max}:=\frac{1}{2}\left(e^{d\fw(0)^{1/2}}
\sum_{m=1}^{N_v}c_0^{m}\alpha^{2m}v_m(\fw(0))\right)^{-1}.
\end{align*}
To make clear, let us sum up these assumptions. From now we assume that
\begin{align}
&h\ge e^{4d},\quad L\ge \beta,\quad M\ge c(d,c_{\chi},N_v),\quad \alpha
\ge c(d)M^{d+\frac{3}{2}},\label{eq_condition_IR}\\
&c_0=c(d,M,c_w,c_{\chi})f_{\bt}^{-\frac{d}{2}},\quad 
r_{max}=\frac{1}{2}\left(e^{d\fw(0)^{1/2}}
\sum_{m=1}^{N_v}c_0^{m}\alpha^{2m}v_m(\fw(0))\right)^{-1}.\notag
\end{align}

Recall that in Lemma \ref{lem_grassmann_symmetric} we derived the
Grassmann integral formulation
$$
\log\left(\int e^{\frac{1}{2}(R^+(\psi)+R^-(\psi))}d\mu_{C^{\infty}_{\le
0}}(\psi)\right)
$$
assuming that the coupling constants are sufficiently small, where
$$
R^{\delta}(\psi)=\log\left(\int
e^{-V^{\delta}(\psi+\psi^1)}d\mu_{C^{\delta}_{>0}}(\psi^1)\right)
\quad (\delta =+,-).
$$
Here we consider the cut-off functions $\chi(h|1-e^{i\o/h}|)$, $\chi(|\o|)$
inside $C_{>0}^{\delta}$, $C_{\le
0}^{\infty}$ as $\phi(M_{UV}^{-2}h^2|1-e^{i\o/h}|^2)$,
$\phi(M_{UV}^{-2}\o^2)$ respectively. The next lemma shows how we can
analytically continue this Grassmann integral formulation 
by means of the iterative infrared multi-scale integration or the
renormalization group method.

\begin{lemma}\label{lem_complete_IR_integration}
The following statements hold true. 
\begin{enumerate}
\item\label{item_IR_full_expansion}
There exist 
\begin{align*}
&J_0^l\in\cS(D(r_{max})^{n_v},c_0,\alpha,M)(l)\cap
     C(\overline{D(r_{max})}^{n_v};\C)\\
&\quad (l=0,-1,\cdots,N_{\beta}-1),\\ 
&E_l\in \cK(D(r_{max})^{n_v},\alpha,M)(l)\quad (l=0,-1,\cdots,N_{\beta})
\end{align*}
and $r(\beta,L)\in\R_{>0}$ depending on $\beta$, $L$, independent of
     $h$ such that
\begin{align}
&E_l(\bU)(\o,\bk)-E_{l+1}(\bU)(\o,\bk)=O,\label{eq_kernel_support}\\
&(\forall l\in \{0,-1,\cdots,N_{\beta}\},
 \bU\in \overline{D(r_{max})}^{n_v},\notag\\
&\quad(\o,\bk)\in\R^{d+1}\text{ with
 }\hat{\chi}_{\le l}(\o,\bk)=0),\notag\\
&-\frac{1}{\beta L^d}\log\left(\int
 e^{\frac{1}{2}(R^+(\bU)(\psi)+R^-(\bU)(\psi))}d\mu_{C_{\le
 0}^{\infty}}(\psi)\right)\label{eq_IR_full_expansion}\\
&=-\frac{1}{\beta
 L^d}\sum_{l=0}^{N_{\beta}-1}J_0^l(\bU)\notag\\
&\quad -\sum_{l=0}^{N_{\beta}}\frac{2}{\beta
 L^d}\sum_{(\o,\bk)\in\cM\times\G(L)^*}\notag\\
&\qquad\cdot\log(\det(I_{2^d}-(i\o
 I_{2^d}-\cE(\bk)-E_{l+1}(\bU)(\o,\bk))^{-1}\notag\\
&\qquad\qquad\qquad\qquad\quad \cdot
 (E_{l}(\bU)(\o,\bk)-E_{l+1}(\bU)(\o,\bk))))\notag\\
&\quad +\frac{1}{L^d}V_0^L(\bU),\quad (\forall \bU\in
 \overline{D(r(\beta,L))}^{n_v}),\notag
\end{align}
     where we set $E_1:=0$.
\item\label{item_IR_full_difference}
In addition, assume \eqref{eq_basic_beta_h_assumption}. Then, 
$J_0^l(\beta_a)$, $E_l(\beta_a)$
     introduced in
     \eqref{item_IR_full_expansion} for $a=1,2$ satisfy that
\begin{align*}
&(J_0^l(\beta_1),J_0^l(\beta_2))\\
&\in\hat{\cS}(D(r_{max})^{n_v},c_0,\alpha,M)(l)\cap
 (C(\overline{D(r_{max})}^{n_v};\C)\times
 C(\overline{D(r_{max})}^{n_v};\C)),\\
&(\forall
 l\in\{0,-1,\cdots,N_{\beta_1}-1\}),\\
&(E_l(\beta_1),E_l(\beta_2))\in \hat{\cK}(D(r_{max})^{n_v},\alpha,M)(l)
,\quad (\forall
 l\in\{0,-1,\cdots,N_{\beta_1}\}).
\end{align*}
\end{enumerate}
\end{lemma}

\begin{proof}
\eqref{item_IR_full_expansion}: By Lemma
 \ref{lem_input_to_IR_integration}
 \eqref{item_connection_to_original_formulation} there exist
 $J^0\in\cS(D(r_{max})^{n_v},c_0,\alpha,M)(0)$ and
 $r(\beta,L)\in\R_{>0}$ depending on $\beta$, $L$, independent of $h$
 such that \eqref{eq_connection_to_original_formulation} holds. Let
 $l\in \{0,-1,\cdots,N_{\beta}\}$ and assume that we have \\
 $(J^0,J^{-1},\cdots,J^l)\in
 \prod_{j=0}^l\cS(D(r_{max})^{n_v},c_0,\alpha,M)(j)$.  
Define $W^j,E_l\in$\\ $\Map(\overline{D(r_{max})}^{n_v},\Map(\R^{d+1},\Mat(2^d,\C)))$ 
$(j=0,-1,\cdots,l)$ by \eqref{eq_effective_kernel_definition},\\
  \eqref{eq_effective_self_energy_definition} respectively.
By Lemma \ref{lem_self_energy_properties} \eqref{item_self_energy_property}, $E_l\in
 \cK(D(r_{max})^{n_v},\alpha,M)(l)$. Define $C_l\in\Map(\overline{D(r_{max})}^{n_v},\Map(I_0^2,\C))$ by
 \eqref{eq_effective_covariance_general_definition} with $E_l$ in place
 of $G_l$. We can apply Lemma \ref{lem_IR_covariance_inclusion}
 \eqref{item_IR_covariance_inclusion} to conclude that
 $C_l\in\cR(D(r_{max})^{n_v},c_0,M)(l)$. Define $J^{l-1}\in
 \Map(\overline{D(r_{max})}^{n_v},\bigwedge \cV)$ by
\begin{align*}
J^{l-1}(\bU)(\psi):=\sum_{n=0}^{\infty}\frac{1}{n!}\left(\frac{d}{dz}\right)^n\Big|_{z=0}\log\left(\int
 e^{z\sum_{m=4}^NJ_m^l(\bU)(\psi+\psi^1)}d\mu_{C_l(\bU)}(\psi^1)\right).
\end{align*}
By Lemma \ref{lem_IR_recursive_structure}
 \eqref{item_IR_recursive_structure},
 $J^{l-1}\in\cS(D(r_{max})^{n_v},c_0,\alpha,M)(l-1)$. Thus, we have
 inductively created 
$J^{l}\in\cS(D(r_{max})^{n_v},c_0,\alpha,M)(l)$
 $(l=0,-1,\cdots,N_{\beta}-1)$,
$E_{l}\in\cK(D(r_{max})^{n_v},\alpha,M)(l)$
 $(l=0,-1,\cdots,N_{\beta})$. 
It is clear from the definition
 \eqref{eq_effective_self_energy_definition} that $E_l$ satisfies
 \eqref{eq_kernel_support}.
Note that by \eqref{eq_connection_to_original_formulation} and taking
 $r(\beta,L)$ smaller independently of $h$ if necessary the left-hand
 side of \eqref{eq_IR_full_expansion} is equal to 
\begin{align*}
-\frac{1}{\beta L^d}\log\left(\int e^{J^0(\bU)(\psi)}d\mu_{C_{\le 0}^{\infty}}(\psi)\right)+\frac{1}{L^d}V_0^L(\bU)
\end{align*}
for any $\bU\in \overline{D(r(\beta,L))}^{n_v}$. Then, we can expand the
 first term in the same way
 as in the proof of \cite[\mbox{Lemma 7.18 (3)}]{K15} by taking
 $r(\beta,L)$ smaller if necessary again 
and obtain the equality \eqref{eq_IR_full_expansion}.

\eqref{item_IR_full_difference}: By Lemma
 \ref{lem_input_to_IR_integration} \eqref{item_input_anisothermal},
$(J^{0}(\beta_1),J^{0}(\beta_2))\in\hat{\cS}(D(r_{max})^{n_v},c_0,\alpha,M)(0)$.
 Assume that $l\in\{0,-1,\cdots,N_{\beta_1}\}$ and
 $(J^{j}(\beta_1),J^{j}(\beta_2))\in$\\
 $\hat{\cS}(D(r_{max})^{n_v},c_0,\alpha,M)(j)$
 for all $j\in \{0,-1,\cdots,l\}$. By Lemma
 \ref{lem_self_energy_properties}
 \eqref{item_self_energy_property_difference}, 
$(E_l(\beta_1),E_l(\beta_2))\in\hat{\cK}(D(r_{max})^{n_v},\alpha,M)(l)$. 
Then, by Lemma \ref{lem_IR_covariance_inclusion}
 \eqref{item_IR_covariance_inclusion_difference}, 
$(C_l(\beta_1),C_l(\beta_2))\in \hat{\cR}(D(r_{max})^{n_v},c_0,M)(l)$.
Then, by Lemma \ref{lem_IR_recursive_structure}
 \eqref{item_IR_recursive_structure_difference} 
$(J^{l-1}(\beta_1),J^{l-1}(\beta_2))\in
 \hat{\cS}(D(r_{max})^{n_v},c_0,\alpha,M)(l-1)$.
The induction with $l$ ensures that the result holds true.
\end{proof}

\begin{remark} In fact the derivation of \eqref{eq_IR_full_expansion}
 well describes how to update the covariance and integrate the Grassmann
 polynomial by using the updated covariance at every step of the IR
 integration. Despite its conceptional importance, here we only refer to
 the corresponding part of \cite[\mbox{Lemma 7.18 (3)}]{K15} without
 reproducing it, since this paper is intended to be a continuation of \cite{K15}.
 However, we should remark that we need to use the relation 
\begin{align*}
&J_2^l(\bU)(\psi)=\frac{1}{h^2}\sum_{(\rho,\bx,\s,x),(\eta,\by,\tau,y)\in
 I_0}\frac{\delta_{\s,\tau}}{\beta L^d}\sum_{(\o,\bk)\in \cM_h\times
 \G(L)^*}e^{i\<\bk,\bx-\by\>+i\o(x-y)}\\
&\qquad\qquad\qquad\cdot W^l(\bU)(\o,\bk)(\rho,\eta)\psi_{\rho\bx\s
 x}\opsi_{\eta\by \tau y},\\
&(\bU\in \overline{D(r_{max})}^{n_v})
\end{align*}
to update the covariance by substituting $W^l$. To derive this equality,
 we use the invariance $J_2^l(\bU)(\psi)=J_2^l(\bU)(\cR\psi)$ with
 $S:I\to I$, $Q:I\to \R$ defined in \eqref{eq_IR_particle_hole},
 \eqref{eq_IR_spin_up}, \eqref{eq_IR_spin_inversion},
 \eqref{eq_IR_translation}, embodied in
 $\cS(D(r_{max})^{n_v},c_0,\alpha,M)(l)$. 
See \cite[\mbox{Lemma 7.6 (1)}]{K15} for the derivation of the same relation.
\end{remark}

Define $J_{end}\in \Map(\overline{D(r_{max})}^{n_v},\C)$ by the
right-hand side of \eqref{eq_IR_full_expansion}. Analyticity and
convergent properties of the free energy density follow from the
properties of $J_{end}$. Let us summarize them in the next two lemmas.

\begin{lemma}\label{lem_final_term_properties}
There exists a constant $c'(d,M,c_w,c_{\chi})\in\R_{>0}$ depending only
 on $d$, $M$, $c_w$, $c_{\chi}$ such that the following statements hold
 true. 
\begin{enumerate}
\item\label{item_final_term_regularity}
$$
J_{end}\in C(\overline{D(r_{max})}^{n_v};\C)\cap C^{\o}(D(r_{max})^{n_v};\C).
$$
\item\label{item_final_term_upper_bound}
\begin{align*}
|J_{end}(\bU)|\le
 c'(d,M,c_w,c_{\chi})f_{\bt}^{-\frac{d}{2}}\alpha^{-1}+v_0r_{max},\quad
 (\forall \bU\in \overline{D(r_{max})}^{n_v}).
\end{align*}
\item\label{item_final_term_difference}
In addition, assume \eqref{eq_basic_beta_h_assumption}.
Then,
\begin{align*}
&|J_{end}(\beta_1)(\bU)-J_{end}(\beta_2)(\bU)|\le
 c'(d,M,c_w,c_{\chi})\beta_1^{-\frac{1}{2}}f_{\bt}^{-\frac{d}{2}}\alpha^{-1},\\
&(\forall \bU\in \overline{D(r_{max})}^{n_v}).
\end{align*}
\end{enumerate}
\end{lemma}

\begin{remark}
Since we have fixed the $(d,M,c_w,c_{\chi})$-dependent constant
 $c(d,M,c_w,c_{\chi})$ in \eqref{eq_condition_IR}, we use the different notation $c'(d,M,c_w,c_{\chi})$
 to express a positive constant depending only on $d,M,c_w,c_{\chi}$.
\end{remark}

\begin{proof}[Proof of Lemma \ref{lem_final_term_properties}]
Since \eqref{eq_kernel_support} holds,  Lemma
 \ref{lem_IR_final_output} \eqref{item_IR_final_analyticity} implies
 the claim \eqref{item_final_term_regularity}. Moreover,
 by \eqref{eq_IR_boundedness_0th} and Lemma \ref{lem_IR_final_output}
 \eqref{item_IR_final_bound},
\begin{align*}
&|J_{end}(\bU)|\\
&\le
 c(d)\sum_{l=0}^{N_{\beta}-1}M^{(d+\frac{3}{2})l}\alpha^{-1}+
c'(d,M,c_w,c_{\chi})\sum_{l=0}^{N_{\beta}}f_{\bt}^{-\frac{d}{2}}M^{(d+1)l}\alpha^{-2}+v_0
 r_{max}\\
&\le c'(d,M,c_w,c_{\chi})f_{\bt}^{-\frac{d}{2}}\alpha^{-1}+v_0
 r_{max},\quad (\forall \bU\in \overline{D(r_{max})}^{n_v}).
\end{align*}
Thus, the claim \eqref{item_final_term_upper_bound} is true.

To prove the claim \eqref{item_final_term_difference}, assume
 \eqref{eq_basic_beta_h_assumption}. By \eqref{eq_IR_boundedness_0th},
 \eqref{eq_IR_boundedness_difference_0th} and
 Lemma \ref{lem_IR_final_output} \eqref{item_IR_final_bound},\eqref{item_IR_final_bound_difference},
\begin{align*}
&|J_{end}(\beta_1)(\bU)-J_{end}(\beta_2)(\bU)|\\
&\le\left|\frac{1}{\beta_1L^d}\sum_{l=0}^{N_{\beta_1}-1}J_0^l(\beta_1)(\bU)-
\frac{1}{\beta_2L^d}\sum_{l=0}^{N_{\beta_1}-1}J_0^l(\beta_2)(\bU)\right|\\
&\quad +
 \frac{1}{\beta_2L^d}\sum_{l=N_{\beta_1}-2}^{N_{\beta_2}}|J_{0}^l(\beta_2)(\bU)|\\
&\quad + c'(d,M,c_w,c_{\chi})\beta_1^{-\frac{1}{2}}
f_{\bt}^{-\frac{d}{2}}\sum_{l=0}^{N_{\beta_1}}M^{dl}\alpha^{-2}\\
&\quad + c'(d,M,c_w,c_{\chi})
f_{\bt}^{-\frac{d}{2}}\sum_{l=N_{\beta_1}-1}^{N_{\beta_2}}M^{(d+1)l}\alpha^{-2}\\
&\le c'(d,M,c_w,c_{\chi})\beta_1^{-\frac{1}{2}}
f_{\bt}^{-\frac{d}{2}}\alpha^{-1},\quad (\forall \bU\in \overline{D(r_{max})}^{n_v}).
\end{align*}
\end{proof}

\begin{lemma}\label{lem_final_term_convergence}
Let $K$ be a non-empty compact set of $\C^{n_v}$ satisfying $K\subset
 D(r_{max})^{n_v}$. Then, the following statements hold.
\begin{enumerate}
\item\label{item_final_term_h_limit}
For any $\beta\in \R_{>0}$, $L\in\N$ with $L\ge \beta$,
     $J_{end}(\beta,L,h)$ converges
     in $C(K;\C)$ as $h\to\infty (h\in (2/\beta)\N)$.
\item\label{item_final_term_L_limit}
Set $J(\beta,L):=\lim_{h\to\infty, h\in
     (2/\beta)\N}J_{end}(\beta,L,h)$. Then, for any $\beta\in \R_{>0}$,
     $J(\beta,L)$ converges in $C(K;\C)$ as $L\to\infty(L\in\N)$.
\item\label{item_final_term_beta_limit}
Set $J(\beta):=\lim_{L\to \infty,L\in\N}J(\beta,L)$. Then,
     $J(\beta)$ converges in $C(K;\C)$ as $\beta\to \infty(\beta\in\N)$.
\end{enumerate}
\end{lemma}

\begin{proof}
Though the proof is parallel to the proof of \cite[\mbox{Lemma 7.20}]{K15},
we present it for completeness.
Take $r_0\in (0,r_{max})$ and $\eps\in (0,r_0)$. Since $J_{end}\in
 C^{\o}(D(r_{max})^{n_v};\C)$ by Lemma \ref{lem_final_term_properties}
\eqref{item_final_term_regularity}, we have that for any $\bU\in \overline{D(r_0-\eps)}^{n_v}
$, 
\begin{align}
J_{end}(\beta,L,h)(\bU)=\sum_{n=0}^{\infty}\frac{1}{n!}\left(\frac{\partial}{\partial
 z}\right)^nJ_{end}(\beta,L,h)(z\bU)\Big|_{z=0}.\label{eq_final_term_taylor_expansion}
\end{align}
By Lemma \ref{lem_final_term_properties}
 \eqref{item_final_term_upper_bound},
\begin{align}
&\sup_{\bU\in \overline{D(r_0-\eps)}^{n_v}}\left|\frac{1}{n!}\left(\frac{\partial}{\partial
 z}\right)^nJ_{end}(\beta,L,h)(z\bU)\Big|_{z=0}\right|\label{eq_final_term_each_bound}\\
&=\sup_{\bU\in \overline{D(r_0-\eps)}^{n_v}}\left|\frac{1}{2\pi
 i}\oint_{|z|=r_0/(r_0-\eps)}dz \frac{J_{end}(\beta,L,h)(z\bU)}{z^{n+1}}\right|\notag\\
&\le \left(\frac{r_0-\eps}{r_0}\right)^n(
 c'(d,M,c_w,c_{\chi})f_{\bt}^{-\frac{d}{2}}\alpha^{-1}+v_0r_{max}),\quad
 (\forall n\in \N\cup\{0\}).\notag
\end{align}
By Lemma \ref{lem_grassmann_symmetric}
 \eqref{item_grassmann_symmetric_positive} and
Lemma \ref{lem_complete_IR_integration} \eqref{item_IR_full_expansion}
 there exist $h$-independent constants $h_0$, $c_1\in\R_{>0}$ such
 that for any $h\in (2/\beta)\N$ with $h\ge h_0$ and
 $\bU\in\overline{D(r_0-\eps)}^{n_v}$,
\begin{align}
&\frac{1}{n!}\left(\frac{\partial}{\partial
 z}\right)^nJ_{end}(\beta,L,h)(z\bU)\Big|_{z=0}\label{eq_final_each_term_decomposition}\\
&=\frac{1}{2\pi i}\oint_{|z|=c_1}dz \frac{1}{z^{n+1}}
\Bigg(-\frac{1}{\beta L^d}\log\Bigg(\int
 e^{\frac{1}{2}(R^+(z\bU)(\psi)+R^-(z\bU)(\psi))}d\mu_{C_{\le 0}^{\infty}}(\psi)
\Bigg)
\notag\\
&\qquad\qquad\qquad\qquad\quad+\frac{1}{\beta L^d}\log\Bigg(\int
 e^{-V(z\bU)(\psi)}d\mu_{C}(\psi)\Bigg)\Bigg)\notag\\
&\quad -\frac{1}{\beta L^d n!}\left(\frac{\partial}{\partial
 z}\right)^n\log\Bigg(\int
 e^{-zV(\bU)(\psi)}d\mu_{C}(\psi)\Bigg)\Bigg|_{z=0}.\notag
\end{align}
Here we used that $V(\bU)(\psi)$ is linear with $\bU$. By Lemma
 \ref{lem_grassmann_symmetric} \eqref{item_logarithm_final_h_estimate}
 the first term of the right-hand side of
 \eqref{eq_final_each_term_decomposition} uniformly converges to 0 with
 respect to $\bU\in \overline{D(r_0-\eps)}^{n_v}$ as $h\to \infty$.
By Lemma \ref{lem_truncation_h_L_limit} proved in Appendix
 \ref{app_h_L_limit} the second term of the right-hand side of
 \eqref{eq_final_each_term_decomposition} uniformly converges with
 $\bU\in \overline{D(r_0-\eps)}^{n_v}$ as we send $h\to \infty$ and then
 send $L\to \infty$.
Therefore, 
\begin{align*}
&\lim_{h\to \infty\atop h\in (2/\beta)\N}\frac{1}{n!}\left(\frac{\partial}{\partial
 z}\right)^nJ_{end}(\beta,L,h)(z\cdot)\Big|_{z=0},\\
&\lim_{L\to\infty\atop L\in \N}
\lim_{h\to \infty\atop h\in (2/\beta)\N}\frac{1}{n!}\left(\frac{\partial}{\partial
 z}\right)^nJ_{end}(\beta,L,h)(z\cdot)\Big|_{z=0}
\end{align*}
converge in $C(\overline{D(r_0-\eps)}^{n_v};\C)$. Since the right-hand
 side of \eqref{eq_final_term_each_bound} is summable over $n\in \N\cup
 \{0\}$, we can apply the dominated convergence theorem in
 $l^1(\N\cup\{0\};C(\overline{D(r_0-\eps)}^{n_v};\C))$ to the expansion
 \eqref{eq_final_term_taylor_expansion} and conclude that
\begin{align*}
&\lim_{h\to \infty\atop h\in (2/\beta)\N}J_{end}(\beta,L,h),\quad
\lim_{L\to\infty\atop L\in \N}
\lim_{h\to \infty\atop h\in (2/\beta)\N}J_{end}(\beta,L,h)
\end{align*}
converge in $C(\overline{D(r_0-\eps)}^{n_v};\C)$. Set
$$
J(\beta):=\lim_{L\to\infty\atop L\in \N}
\lim_{h\to \infty\atop h\in (2/\beta)\N}J_{end}(\beta,L,h).
$$
By taking the limits in the inequality obtained in Lemma
 \ref{lem_final_term_properties} \eqref{item_final_term_difference} we
 see that $(J(\beta))_{\beta\in\N}$ is a Cauchy sequence in
 $C(\overline{D(r_0-\eps)}^{n_v};\C)$. Thus, $\lim_{\beta\to
 \infty,\beta\in\N}J(\beta)$ converges in this Banach space. For any
 compact set $K$ of $\C^{n_v}$ satisfying $K\subset D(r_{max})^{n_v}$ we
 can choose $r_0\in (0,r_{max})$ and $\eps\in (0,r_0)$ so that $K\subset
 \overline{D(r_0-\eps)}^{n_v}$. Thus, the
 claimed convergence results in $C(K;\C)$ follow from the above arguments. 
\end{proof}

Before proceeding to the proof of Theorem \ref{thm_main_theorem} we
state a couple of necessary lemmas, which are close to \cite[\mbox{Lemma
E.2, Lemma E.3}]{K15}. In the proofs of these lemmas
$\|\cdot\|_{\mathfrak{B}(F_f)}$ denotes the operator norm for linear
operators on $F_f(L^2(\cB\times \G(L)\times \spin))$. 

\begin{lemma}\label{lem_partition_function_difference}
For any $\beta \in\R_{\ge 1}$,
\begin{align*}
&\left|\frac{1}{\beta L^d}\log\left(\frac{\Tr e^{-\beta H}}{\Tr
 e^{-\beta H_0}}\right)- \frac{1}{\lfloor\beta\rfloor L^d}\log\left(\frac{\Tr e^{-\lfloor\beta\rfloor H}}{\Tr
 e^{-\lfloor\beta\rfloor H_0}}\right)\right|\\
&\le \int_{\lfloor\beta\rfloor}^{\beta}d\gamma\frac{1}{\gamma^2L^d}\left|\log\left(\frac{\Tr e^{-\gamma H}}{\Tr
 e^{-\gamma H_0}}\right)\right|\\
&\quad +\Bigg(v_0\max_{l\in\{1,2,\cdots,n_v\}}
 |U_{l}|+2^{d+1}\sum_{j=1}^{N_v}v_j(0)\max_{l\in
 \{1,2,\cdots,n_v\}}|U_{l}|\\
&\qquad\quad +2^{d+2}\sup_{\bk\in\R^d}\|E(\beps^L,\btheta)(\bk)\|_{2^d\times
 2^d}\Bigg)\log\left(\frac{\beta}{\lfloor \beta\rfloor}\right).
\end{align*}
\end{lemma}
\begin{proof}
Since $\|\psi_X^*\|_{\mathfrak{B}(F_f)}=\|\psi_X\|_{\mathfrak{B}(F_f)}=1$ $(\forall X\in
 \cB\times \G(L)\times \spin)$, 
\begin{align*}
\|V\|_{\mathfrak{B}(F_f)}\le L^dv_0\max_{l\in\{1,2,\cdots,n_v\}}|U_{l}|+2^{d+1}L^d\sum_{j=1}^{N_v}v_j(0)\max_{l\in\{1,2,\cdots,n_v\}}|U_{l}|.
\end{align*}
By using this inequality in place of the inequality ``(E.3)'' and
 straightforwardly following the proof of \cite[\mbox{Lemma E.2}]{K15}, we can
 derive the claimed inequality.
\end{proof}

We may consider $\bU$ inside the operator $H$ as complex variables.
\begin{lemma}\label{lem_partition_function_direct_analyticity}
For any $r\in\R_{>0}$ there exists a domain $O$ of $\C$ such that
 $(-r,r)\subset O$ and $\log(\Tr e^{-\beta H})$ is analytic with respect
 to $\bU$ in $O^{n_v}$.
\end{lemma}
\begin{proof}
Take any $r\in\R_{>0}$, $\ba\in [-1,1]^{n_v}$, $\bU\in [-r,r]^{n_v}$ and
 $\delta\in[0,1]$. Note that 
\begin{align*}
&|\Tr e^{-\beta (H_0+V(\bU+i\delta\ba))}-\Tr e^{-\beta (H_0+V(\bU))}|
\le \int_0^{\delta}d\eps
\left|\frac{d}{d\eps}\Tr e^{-\beta (H_0+V(\bU)+i\eps
 V(\ba))}\right|\\
&\le \delta \beta 2^{2^{d+1}L^d}\|V(\ba)\|_{\mathfrak{B}(F_f)}e^{\beta
 (\|H_0\|_{\mathfrak{B}(F_f)}+\|V(\bU)\|_{\mathfrak{B}(F_f)}
+\|V(\ba)\|_{\mathfrak{B}(F_f)})},
\end{align*}
where we used the linearity of $V$ with respect to the coupling constants. 
Thus,
\begin{align*}
&\Re \Tr e^{-\beta (H_0+V(\bU+i\delta\ba))}\\
&\ge \Tr e^{-\beta (H_0+V(\bU))}\\
&\quad- \delta \beta
 2^{2^{d+1}L^d}\|V(\ba)\|_{\mathfrak{B}(F_f)}
e^{\beta
 (\|H_0\|_{\mathfrak{B}(F_f)}+\|V(\bU)\|_{\mathfrak{B}(F_f)}+\|V(\ba)\|_{\mathfrak{B}(F_f)})}\\
&\ge e^{-\beta L^d r v_0}\\
&\quad 
-\delta \beta 2^{2^{d+1}L^d}\sup_{\bz\in[-r,r]^{n_v} \atop
 \bb\in[-1,1]^{n_v}}
\big(\|V(\bb)\|_{\mathfrak{B}(F_f)}
e^{\beta
 (\|H_0\|_{\mathfrak{B}(F_f)}+\|V(\bz)\|_{\mathfrak{B}(F_f)}+\|V(\bb)\|_{\mathfrak{B}(F_f)})}\big)\\
&>0,\quad (\forall \bU\in [-r,r]^{n_v},\ba\in[-1,1]^{n_v}),
\end{align*}
if $\delta$ is sufficiently small. Therefore, there exists $\delta
 \in\R_{>0}$ such that $\log (\Tr e^{-\beta H})$ is
 analytic with respect to $\bU$ in the domain 
$\{x+iy\ |\ x\in(-r,r),y\in (-\delta,\delta)\}^{n_v}$.
\end{proof}
 
Here we can give the proof of Theorem \ref{thm_main_theorem}.
\begin{proof}[Proof of Theorem \ref{thm_main_theorem}]
Assume that the condition \eqref{eq_condition_IR} holds. By \\
Lemma
 \ref{lem_final_term_convergence} there exist 
$$J(\beta,L), J(\beta),J\in 
 C(\overline{D(r_{max}/2)}^{n_v};\C)\cap C^{\o}(D(r_{max}/2)^{n_v};\C)$$
 such that 
\begin{align}
&\lim_{h\to \infty\atop h\in
 (2/\beta)\N}J_{end}(\beta,L,h)=J(\beta,L),\quad (\forall \beta\in
 \R_{>0}, L\in\N\text{ with }L\ge \beta),\label{eq_final_convergence_summary_h}\\
&\lim_{L\to \infty\atop L\in\N}J(\beta,L)=J(\beta),\quad (\forall \beta\in\R_{>0}),\label{eq_final_convergence_summary_L}\\
&\lim_{\beta\to \infty\atop \beta\in\N}J(\beta)=J\label{eq_final_convergence_summary_beta}
\end{align}
in $C(\overline{D(r_{max}/2)}^{n_v};\C)$. 
By Lemma \ref{lem_grassmann_formulation}
 \eqref{item_grassmann_formulation_convergence},
Lemma \ref{lem_grassmann_symmetric} 
\eqref{item_logarithm_final_h_estimate},
Lemma \ref{lem_complete_IR_integration}
 \eqref{item_IR_full_expansion} and \eqref{eq_final_convergence_summary_h}
there exists a constant $c_1\in \R_{>0}$
 independent of $h$ such that for any $\bU\in\overline{D(c_1)}^{n_v}\cap\R^{n_v}$,
\begin{align}
J(\beta,L)(\bU)
&=\lim_{h\to\infty\atop h\in (2/\beta)\N}\Bigg(
-\frac{1}{\beta L^d}\log\left(\int
 e^{\frac{1}{2}(R^+(\bU)(\psi)+R^-(\bU)(\psi))}d\mu_{C_{\le
 0}^{\infty}}(\psi)\right)\label{eq_final_output_partition_function}\\
&\qquad\qquad\quad +\frac{1}{\beta L^d}\log\left(\int
 e^{-V(\bU)(\psi)}d\mu_{C}(\psi)\right)\Bigg)\notag\\
&\quad -\lim_{h\to\infty\atop h\in (2/\beta)\N}
\frac{1}{\beta L^d}\log\left(\int
 e^{-V(\bU)(\psi)}d\mu_{C}(\psi)\right)\notag\\
&=-\frac{1}{\beta L^d}\log\left(\frac{\Tr e^{-\beta H}}{\Tr e^{-\beta
 H_0}}\right).\notag
\end{align}
By Lemma \ref{lem_partition_function_direct_analyticity} there exists a
 domain $O\subset \C^{n_v}$ such that $\overline{D(r_{max}/2)}^{n_v}\cap
 \R^{n_v}\subset O$ and the right-hand side of
 \eqref{eq_final_output_partition_function} is analytic with $\bU$ in
 $O$. Thus, by the identity theorem and continuity the equality
 \eqref{eq_final_output_partition_function} holds for any $\bU\in
 \overline{D(r_{max}/2)}^{n_v}\cap \R^{n_v}$, $\beta \in\R_{> 0}$,
 $L\in\N$ with $L\ge \beta$. Then, by Lemma \ref{lem_hopping_properties}
 \eqref{item_hopping_upper}, 
Lemma
 \ref{lem_final_term_properties} \eqref{item_final_term_upper_bound} and
 Lemma \ref{lem_partition_function_difference},
\begin{align}
&\left|\frac{1}{\beta L^d}\log\left(\frac{\Tr e^{-\beta H}}{\Tr
 e^{-\beta H_0}}\right)- \frac{1}{\lfloor\beta\rfloor L^d}\log\left(\frac{\Tr e^{-\lfloor\beta\rfloor H}}{\Tr
 e^{-\lfloor\beta\rfloor H_0}}\right)\right|\label{eq_partition_function_final_difference}\\
&\le \int_{\lfloor\beta\rfloor}^{\beta}d\gamma\frac{1}{\gamma}
\left(c'(d,M,c_w,c_{\chi})f_{\bt}^{-\frac{d}{2}}\alpha^{-1}+\frac{v_0}{2}r_{max}\right)\notag\\
&\quad +\left(
\frac{v_0}{2}r_{max}+2^d\sum_{j=1}^{N_v}v_j(0)r_{max}+
2^{d+3}d\right)\log\left(\frac{\beta}{\lfloor
 \beta\rfloor}\right)\notag\\
&\le \left(c'(d,M,c_w,c_{\chi})f_{\bt}^{-\frac{d}{2}}\alpha^{-1}+
v_0r_{max}+2^d\sum_{j=1}^{N_v}v_j(0)r_{max}+
2^{d+3}d\right)\notag\\
&\quad\cdot \log\left(\frac{\beta}{\lfloor
 \beta\rfloor}\right),\notag
\end{align}
for any $\bU\in \overline{D(r_{max}/2)}^{n_v}\cap \R^{n_v}$, $\beta
 \in\R_{\ge 1}$,
 $L\in\N$ with $L\ge \beta$. 

Let $\alpha_\rho^L(\bk)$ $(\rho\in\cB)$ denote the eigen values of
 $E(\beps^L,\btheta)(\bk)$ for $\bk\in \G(L)^*$. Then, by 
\cite[\mbox{Lemma E.1}]{K15}, 
\begin{align*}
-\frac{1}{\beta (2L)^d}\log(\Tr e^{-\beta H_0})&=-\frac{2}{\beta
 (2L)^d}\sum_{\rho\in \cB}\sum_{\bk\in\G(L)^*}\log(1+e^{-\beta
 \alpha_{\rho}^L(\bk)})\\
&=-\frac{2}{\beta
 (2L)^d}\sum_{\bk\in\G(L)^*}\log\det(I_{2^d}+e^{-\beta
 E(\beps^L,\theta)(\bk)}).
\end{align*}
We can deduce from the definition that
$\lim_{L\to \infty,L\in\N}E(\beps^L,\btheta)(\bk)(\rho,\eta)$ converges for
 any $\bk\in \R^d$, $\rho,\eta\in\cB$ and if we set
 $E(\btheta)(\bk):=\lim_{L\to\infty,L\in\N}$ $E(\beps^L,\btheta)(\bk)$, 
$E(\btheta)(\bk)^*=E(\btheta)(\bk)$ $(\forall \bk\in\R^d)$. 
Let $\alpha_{\rho}(\bk)$ $(\rho\in \cB)$ be the eigen values of
 $E(\btheta)(\bk)$ for $\bk\in \R^d$. Then, by Lemma
 \ref{lem_hopping_properties} \eqref{item_hopping_upper},
 $|\alpha_{\rho}(\bk)|\le 2d$ $(\forall \rho\in\cB,\bk\in\R^d)$. By
 considering these facts we can apply the dominated convergence theorem
 to prove that
\begin{align}
&\lim_{L\to \infty\atop L\in\N}\left(-\frac{1}{\beta (2L)^d}\log(\Tr
 e^{-\beta H_0})\right)\label{eq_free_free_energy_zero_temperature}\\
&=-\frac{2}{\beta
 (4\pi)^d}\int_{[0,2\pi]^d}d\bk\log\det(I_{2^d}+e^{-\beta
 E(\theta)(\bk)})\notag\\
&=-\frac{2}{\beta
 (4\pi)^d}\int_{[0,2\pi]^d}d\bk\sum_{\rho\in\cB}\log(1+e^{-\beta
\alpha_{\rho}(\bk)}),\notag\\
&\lim_{\beta\to \infty\atop \beta\in\R_{>0}}
\lim_{L\to \infty\atop L\in\N}\left(-\frac{1}{\beta (2L)^d}\log(\Tr
 e^{-\beta H_0})\right)=\frac{2}{(4\pi)^d}
\int_{[0,2\pi]^d}d\bk\sum_{\rho\in \cB}1_{\alpha_{\rho}(\bk)<0}\alpha_{\rho}(\bk).\notag
\end{align}

Let us define $F(\beta,L)$, $F(\beta)$, $F\in C(\overline{D(r_{max}/2)}^{n_v};\C)\cap$\\
 $C^{\o}(D(r_{max}/2)^{n_v};\C)$ $(\beta\in\R_{>0},L\in\N\text{
 with }L\ge \beta)$ by 
\begin{align*}
&F(\beta,L):=2^{-d}J(\beta,L)-\frac{1}{\beta (2L)^d}\log(\Tr
 e^{-\beta H_0}),\\
&F(\beta):=\lim_{L\to \infty\atop L\in \N}F(\beta,L),\\
&F:=\lim_{\beta\to \infty\atop \beta\in \N}F(\beta).
\end{align*}
By \eqref{eq_final_convergence_summary_L}, \eqref{eq_final_convergence_summary_beta},
 \eqref{eq_partition_function_final_difference} and the fact that
 \eqref{eq_final_output_partition_function} holds for any $\bU\in \overline{D(r_{max}/2)}^{n_v}\cap
 \R^{n_v}$, $\beta\in\R_{\ge 1}$, $L\in\N$ with $L\ge \beta$ we see that 
\begin{align*}
&|F(\beta)(\bU)-F(\bU)|\\
&\le 2^{-d}|J(\beta)(\bU)-J(\lfloor \beta \rfloor)(\bU)|
+ 2^{-d}|J(\lfloor \beta \rfloor)(\bU)-J(\bU)|\\
&\quad +\left|
\lim_{L\to\infty\atop L\in \N}\frac{1}{\beta (2L)^d}\log(\Tr e^{-\beta
 H_0})-\lim_{\beta\to\infty\atop \beta\in\R_{>0}}\lim_{L\to\infty\atop
 L\in \N}\frac{1}{\beta (2L)^d}\log(\Tr e^{-\beta H_0})\right|\\
&\le \left(c'(M,d,c_w,c_{\chi})2^{-d}f_{\bt}^{-\frac{d}{2}}\alpha^{-1}+2^{-d}v_0r_{max}
+\sum_{j=1}^{N_v}v_j(0)r_{max} +2^3d\right)\\
&\quad\cdot\log\left(\frac{\beta}{\lfloor
 \beta\rfloor}\right)\\
&\quad + 2^{-d}|J(\lfloor \beta \rfloor)(\bU)-J(\bU)|\\
&\quad +\left|
\lim_{L\to\infty\atop L\in \N}\frac{1}{\beta (2L)^d}\log(\Tr e^{-\beta
 H_0})-\lim_{\beta\to\infty\atop \beta\in\R_{>0}}\lim_{L\to\infty\atop
 L\in \N}\frac{1}{\beta (2L)^d}\log(\Tr e^{-\beta H_0})\right|,\\
&(\forall \bU\in \overline{D(r_{max}/2)}^{n_v}\cap
 \R^{n_v}, \beta\in\R_{\ge 1}).
\end{align*}
Then, \eqref{eq_final_convergence_summary_beta} and 
 \eqref{eq_free_free_energy_zero_temperature}
imply that $\lim_{\beta\to \infty,\beta\in\R_{>0}}F(\beta)=F$ in\\
 $C(\overline{D(r_{max}/2)}^{n_v}\cap \R^{n_v};\C)$. By the same basic
 argument as the final part of \cite[\mbox{Proof of Theorem 1.1, Section
 7}]{K15} we can deduce from the above convergence property that 
$\lim_{\beta\to \infty,\beta\in\R_{>0}}F(\beta)=F$ in
 $C(\overline{D(r_{max}/2)}^{n_v};\C)$. 

To conclude the proof of the
 theorem under the assumption \eqref{eq_hopping_amplitude_condition},
 we may conceal the dependency on the artificial parameters $\alpha$,
 $M$, $c_{w}$, $c_{\chi}$ in \eqref{eq_condition_IR}. Then, we can read
 from the conditions \eqref{eq_condition_IR} and $4f_{\bt}\le 1$ that
there exists a constant $c(d,N_v)\in\R_{>0}$ depending only on $d$,
 $N_v$ such that 
$$
\left(\sum_{m=1}^{N_v}c(d,N_v)^mv_m(c(d,N_v))\right)^{-1}(4f_{\bt})^{\frac{dN_v}{2}}\le \frac{r_{max}}{2}.
$$ 
The left-hand side of the above inequality is equal to $R$ set in
 Theorem \ref{thm_main_theorem} if
 \eqref{eq_hopping_amplitude_condition} holds.
By recalling Lemma \ref{lem_free_energy_equivalence} we see that the
 above arguments have proved the theorem under the assumption
 \eqref{eq_hopping_amplitude_condition}. 

Let us show that the theorem in the general case follows from the
 theorem proved under
 \eqref{eq_hopping_amplitude_condition}. Let us temporarily write $R_{\bt}$
 in place of $R$. Set $t_{max}:=\max_{j\in\{1,2,\cdots,d\}}t_j$. By the
 theorem for the Hamiltonian $\frac{1}{t_{max}}\sH_0+\sV$, there exist
$F(\beta,L)$, $F(\beta)$, $F\in C(\overline{D(R_{\bt/t_{max}})}^{n_v};\C)\cap
 C^{\o}(D(R_{\bt/t_{max}})^{n_v};\C)$ 
such that 
\begin{align*} 
&F(\beta,L)(\bU)=-\frac{1}{\beta (2L)^d}\log(\Tr
 e^{-\beta (\frac{1}{t_{max}}\sH_0+\sV)}),\\
&(\forall \bU\in \overline{D(R_{\bt/t_{max}})}^{n_v}\cap \R^{n_v},
 \beta\in\R_{>0},L\in\N\text{ with }L\ge \beta),\\
&\lim_{L\to \infty\atop L\in \N}F(\beta,L)=F(\beta)\text{ in
 }C(\overline{D(R_{\bt/t_{max}})}^{n_v};\C),\\
&\lim_{\beta\to \infty\atop \beta\in \R_{>0}}F(\beta)=F\text{ in
 }C(\overline{D(R_{\bt/t_{max}})}^{n_v};\C).
\end{align*}
Then, by the linearity of $V_m^{L}(\bU)(\cdot)$ with $\bU$,
\begin{align*} 
&F(t_{max}\beta,L)\left(\frac{1}{t_{max}}\bU\right)=-\frac{1}{t_{max}\beta (2L)^d}\log(\Tr
 e^{-\beta \sH}),\\
&(\forall \bU\in \overline{(t_{max}D(R_{\bt/t_{max}}))}^{n_v}\cap \R^{n_v},
 \beta\in\R_{>0},L\in\N\text{ with }L\ge t_{max}\beta).
\end{align*}
Since $t_{max}D(R_{\bt/t_{max}})=D(R_{\bt})$,
\begin{align*}
t_{max}F(t_{max}\beta,L)\left(\frac{1}{t_{max}}\cdot\right)
\in  
 C(\overline{D(R_{\bt})}^{n_v};\C)\cap C^{\o}(D(R_{\bt})^{n_v};\C)
\end{align*}
and
\begin{align*}
&\lim_{L\to \infty\atop L\in
 \N}t_{max}F(t_{max}\beta,L)\left(\frac{1}{t_{max}}\cdot\right)=t_{max}F(t_{max}\beta)\left(\frac{1}{t_{max}}\cdot\right),\\
&\lim_{\beta\to \infty\atop \beta\in \R_{>0}}
t_{max}F(t_{max}\beta)\left(\frac{1}{t_{max}}\cdot\right)
=t_{max}F\left(\frac{1}{t_{max}}\cdot\right)
\text{ in
 }C(\overline{D(R_{\bt})}^{n_v};\C).
\end{align*}
Thus, the theorem has been proved.
\end{proof}

\appendix
\section{Reordering in a non-commutative $\C$-algebra}\label{app_normal_order}
Here we prove a lemma which is used in the proof of Lemma \ref{lem_half_filled}.
 Though the actual
problem involves the Fermionic creation/annihilation operators, we set up
the problem in a non-commutative $\C$-algebra for simplicity. Let $n\in \N$. Let $A$ be a $\C$-algebra spanned by 
products of the elements $a_1,a_2,\cdots,a_n$,
$a_1^*,a_2^*,\cdots,a_n^*$ satisfying the relation
\begin{align}
&a_j^*a_k+a_ka_j^*=\delta_{j,k},\quad a_ja_k+a_ka_j=0,\label{eq_ACR}\\
&a_j^*a_k^*+a_k^*a_j^*=0.\quad (\forall j,k\in\{1,2,\cdots,n\}).\notag
\end{align}
Set $S:=\{1,2,\cdots,n\}$. We call a function $f_m:S^m\times S^m\to\C$
bi-anti-symmetric if 
\begin{align*}
&f_m((x_{\s(1)},x_{\s(2)},\cdots,x_{\s(m)}),(y_{\tau(1)},y_{\tau(2)},\cdots,y_{\tau(m)}))\\
&=\sgn(\s)\sgn(\tau)f_m((x_{1},x_{2},\cdots,x_{m}),(y_{1},y_{2},\cdots,y_{m})),\\
&(\forall \s,\tau\in\S_m,
 (x_1,x_2,\cdots,x_m),(y_1,y_2,\cdots,y_m)\in S^m).
\end{align*}
For $\bX=(x_1,x_2,\cdots,x_m)\in S^m$ let $a_{\bX}$, $a_{\bX}^*$ denote 
$a_{x_1}a_{x_2}\cdots a_{x_m}$,\\ $a_{x_1}^*a_{x_2}^*\cdots a_{x_m}^*$
respectively.
Moreover, let $\widetilde{\bX}$ denote $(x_m,x_{m-1},\cdots,x_1)$.

\begin{lemma}\label{lem_normal_order}
For any bi-anti-symmetric function $f_m:S^m\times S^m\to \C$, 
\begin{align*}
&\sum_{\bX,\bY\in
 S^m}f_m(\bX,\bY)a_{\bX}^*a_{\bY}\\
&=\sum_{l=0}^m\sum_{\bX,\bY\in
 S^{m-l}\atop \bZ\in S^l}(-1)^{m-l}\left(\begin{array}{c} m \\
			       l\end{array}\right)^2l!
f_m((\bX,\bZ),(\widetilde{\bZ},\bY))a_{\bY}a_{\bX}^*.
\end{align*}
\end{lemma}

\begin{proof}
We prove the claim by induction with $m$. The equality for $m=1$ follows
 from the relation \eqref{eq_ACR}. Assume that the claim is true for
 some $m\in \N$. Let $f_{m+1}:S^{m+1}\times S^{m+1}\to\C$ be a
 bi-anti-symmetric function. By the hypothesis of induction, 
\begin{align*}
&\sum_{\bX,\bY\in S^{m+1}}f_{m+1}(\bX,\bY)a_{\bX}^*a_{\bY}\\
&=\sum_{x,y\in S}\sum_{l=0}^m\sum_{\bX,\bY\in
 S^{m-l}\atop \bZ\in S^l}(-1)^{m-l}\left(\begin{array}{c} m \\
			       l\end{array}\right)^2l!
f_{m+1}((x,\bX,\bZ),(\widetilde{\bZ},\bY,y))a_{x}^*a_{\bY}a_{\bX}^*a_y.
\end{align*}
Moreover, by \eqref{eq_ACR} and the bi-anti-symmetric property of
 $f_{m+1}$,
\begin{align*}
&\sum_{\bX,\bY\in S^{m+1}}f_{m+1}(\bX,\bY)a_{\bX}^*a_{\bY}\\
&=\sum_{l=0}^m\sum_{\bX,\bY\in
 S^{m+1-l}\atop \bZ\in S^l}(-1)^{m-l}\left(\begin{array}{c} m \\
			       l\end{array}\right)^2l!
f_{m+1}((\bX,\bZ),(\widetilde{\bZ},\bY))\\
&\qquad\qquad\cdot a_{x_1}^*a_{y_1}\cdots
 a_{y_{m-l}}a_{x_2}^*\cdots a_{x_{m-l+1}}^*a_{y_{m-l+1}}\\
&=\sum_{l=0}^m\sum_{\bX,\bY\in
 S^{m-l}\atop \bZ\in S^{l+1}}\left(\begin{array}{c} m \\
			       l\end{array}\right)^2l!
f_{m+1}((\bX,\bZ),(\widetilde{\bZ},\bY))a_{y_1}\cdots
 a_{y_{m-l-1}}a_{\bX}^*a_{y_{m-l}}\\
&\quad +\sum_{l=0}^m\sum_{\bX,\bY\in
 S^{m+1-l}\atop \bZ\in S^{l}}(-1)^{m-l+1}\left(\begin{array}{c} m \\
			       l\end{array}\right)^2l!
f_{m+1}((\bX,\bZ),(\widetilde{\bZ},\bY))\\
&\qquad\qquad\quad\cdot 
a_{y_1}a_{x_1}^*a_{y_2}
\cdots
 a_{y_{m-l}}a_{x_2}^*\cdots a_{x_{m-l+1}}^*a_{y_{m-l+1}}\\
&=2\sum_{l=0}^m\sum_{\bX,\bY\in
 S^{m-l}\atop \bZ\in S^{l+1}}\left(\begin{array}{c} m \\
			       l\end{array}\right)^2l!
f_{m+1}((\bX,\bZ),(\widetilde{\bZ},\bY))a_{y_1}\cdots
 a_{y_{m-l-1}}a_{\bX}^*a_{y_{m-l}}\\
&\quad +\sum_{l=0}^m\sum_{\bX,\bY\in
 S^{m+1-l}\atop \bZ\in S^{l}}(-1)^{m-l+2}\left(\begin{array}{c} m \\
			       l\end{array}\right)^2l!
f_{m+1}((\bX,\bZ),(\widetilde{\bZ},\bY))\\
&\qquad\qquad\quad\cdot 
a_{y_1}a_{y_2}a_{x_1}^*a_{y_3}
\cdots
 a_{y_{m-l}}a_{x_2}^*\cdots a_{x_{m-l+1}}^*a_{y_{m-l+1}}\\
&=\sum_{l=0}^m\sum_{\bX,\bY\in
 S^{m-l}\atop \bZ\in S^{l+1}}\left(\begin{array}{c} m \\
			       l\end{array}\right)^2l!(m-l)
f_{m+1}((\bX,\bZ),(\widetilde{\bZ},\bY))\\
&\qquad\qquad\quad\cdot a_{y_1}\cdots
 a_{y_{m-l-1}}a_{\bX}^*a_{y_{m-l}}\\
&\quad +\sum_{l=0}^m\sum_{\bX,\bY\in
 S^{m+1-l}\atop \bZ\in S^{l}}\left(\begin{array}{c} m \\
			       l\end{array}\right)^2l!
f_{m+1}((\bX,\bZ),(\widetilde{\bZ},\bY))a_{y_1}\cdots a_{y_{m-l}}a_{\bX}^*a_{y_{m-l+1}}\\
&=\sum_{l=0}^{m+1}\sum_{\bX,\bY\in
 S^{m+1-l}\atop \bZ\in S^{l}}\Bigg(1_{l\ge 1}\left(\begin{array}{c} m \\
			       l-1\end{array}\right)^2(l-1)!(m+1-l)
+1_{l\le m}\left(\begin{array}{c} m \\
			       l\end{array}\right)^2 l!\Bigg)\\
&\qquad\qquad\quad \cdot f_{m+1}((\bX,\bZ),(\widetilde{\bZ},\bY))
a_{y_1}\cdots a_{y_{m-l}}a_{\bX}^*a_{y_{m-l+1}}.
\end{align*}
Set 
\begin{align*}
D(l,m):=
1_{l\ge 1}\left(\begin{array}{c} m \\
			       l-1\end{array}\right)^2(l-1)!(m+1-l)
+1_{l\le m}\left(\begin{array}{c} m \\
			       l\end{array}\right)^2 l!.
\end{align*}
Then, by considering that $D(m+1,m)=0$,
\begin{align*}
&\sum_{\bX,\bY\in S^{m+1}}f_{m+1}(\bX,\bY)a_{\bX}^*a_{\bY}\\
&=\sum_{l=0}^{m}\sum_{\bX,\bY\in S^{m+1-l}\atop \bZ\in S^l}D(l,m)
f_{m+1}((\bX,\bZ),(\widetilde{\bZ},\bY))
a_{y_1}\cdots a_{y_{m-l}}a_{\bX}^*a_{y_{m-l+1}}\\
&=\sum_{l=0}^m\sum_{\bX,\bY\in S^{m-l}\atop \bZ\in
 S^{l+1}}(-1)^{m-l}D(l,m)
f_{m+1}((\bX,\bZ),(\widetilde{\bZ},\bY))
a_{\bY}a_{\bX}^*\\
&\quad+ \sum_{l=0}^{m}\sum_{\bX,\bY\in S^{m+1-l}\atop \bZ\in
 S^{l}}(-1)D(l,m)
f_{m+1}((\bX,\bZ),(\widetilde{\bZ},\bY))\\
&\qquad\qquad\cdot a_{y_1}\cdots a_{y_{m-l}}a_{x_1}^*\cdots a_{x_{m-l-1}}^*
a_{x_{m-l}}^*a_{y_{m-l+1}}a_{x_{m-l+1}}^*\\
&=\sum_{l=0}^m\sum_{\bX,\bY\in S^{m-l}\atop \bZ\in
 S^{l+1}}(-1)^{m-l}(m-l+1)D(l,m)
f_{m+1}((\bX,\bZ),(\widetilde{\bZ},\bY))
a_{\bY}a_{\bX}^*\\
&\quad+ \sum_{l=0}^{m}\sum_{\bX,\bY\in S^{m+1-l}\atop \bZ\in
 S^{l}}(-1)^{m-l+1}D(l,m)
f_{m+1}((\bX,\bZ),(\widetilde{\bZ},\bY))a_{\bY}a_{\bX}^*\\
&= \sum_{l=0}^{m+1}\sum_{\bX,\bY\in S^{m+1-l}\atop \bZ\in
 S^{l}}(1_{l\ge 1}(-1)^{m+1-l}(m+2-l)D(l-1,m)\\
&\qquad\qquad\qquad\quad +(-1)^{m+1-l}D(l,m))
f_{m+1}((\bX,\bZ),(\widetilde{\bZ},\bY))a_{\bY}a_{\bX}^*.
\end{align*}
By calculation we can derive that  
\begin{align*}
&1_{l\ge 1}(-1)^{m+1-l}(m+2-l)D(l-1,m)+(-1)^{m+1-l}D(l,m)\\
&=(-1)^{m+1-l}\left(\begin{array}{c} m+1 \\
			       l\end{array}\right)^2 l!,
\end{align*}
which implies the claimed equality for $m+1$. The induction with
 $m$ concludes the proof.
\end{proof}

\section{The flux phase problem on a periodic hyper-cubic lattice}\label{app_flux_phase}
 
In order to deduce Corollary \ref{cor_minimum_energy} from Theorem
\ref{thm_main_theorem}, we need to know when the free energy density
is minimum in the flux phase problem on the hyper-cubic lattice
$\G(2L)$ with the periodic boundary condition. A sufficient condition was
essentially proved by Lieb in \cite{L}.  It was also claimed by Macris
and Nachtergaele in \cite{MN}. In \cite[\mbox{Appendix A}]{K15} we restated Lieb's
theorem on a periodic square lattice with
supplementary arguments which were not explicit in the letter
\cite{L}. In order to assist the readers in deriving Corollary
\ref{cor_minimum_energy} from Theorem \ref{thm_main_theorem}, here we
restate Lieb's theorem on the flux phase problem on the periodic
hyper-cubic lattice with explanations of how to extend the arguments in 
 \cite[\mbox{Appendix A}]{K15} into the $d$-dimensional case. Not to
 confuse the problem, we should make clear the dependency between the
 original article \cite{L}, the preceding section \cite[\mbox{Appendix
 A}]{K15} and this section. In this section we admit the basic lemmas 
 \cite[\mbox{Lemma A.2, Lemma A.4}]{K15} and the contents of the proof of 
 \cite[\mbox{Theorem A.5}]{K15} which was based on the original key
 lemma \cite[\mbox{Lemma}]{L}. For those who know how to apply the
 reflection positivity lemma \cite[\mbox{Lemma}]{L} well, there is no
 need to follow the proof of Theorem \ref{thm_flux_phase}
 below. However, we should remark that Lemma
 \ref{lem_partition_gauge_invariance} claimed below itself is necessary to prove
 not only Theorem \ref{thm_flux_phase} but also Theorem
 \ref{thm_main_theorem}. In fact Lemma
 \ref{lem_partition_gauge_invariance} is referred in the proof of 
Lemma \ref{lem_free_energy_equivalence} in Section \ref{sec_formulation}.

First let us extend \cite[\mbox{Lemma A.2}]{K15} into the
$d$-dimensional case. Assume that phases $\varphi_1$,
$\varphi_2:\Z^d\times\Z^d\to \R$ satisfy
\eqref{eq_phase_condition_general} and 
\begin{align*}
&\varphi_1(\bx+\be_j,\bx)+\varphi_1(\bx+\be_j+\be_k,\bx+\be_j)\\
&\quad+\varphi_1(\bx+\be_k,\bx+\be_j+\be_k)+\varphi_1(\bx,\bx+\be_k)\\
&=\varphi_2(\bx+\be_j,\bx)+\varphi_2(\bx+\be_j+\be_k,\bx+\be_j)\\
&\quad+\varphi_2(\bx+\be_k,\bx+\be_j+\be_k)+\varphi_2(\bx,\bx+\be_k)\quad(\text{mod }2\pi),\\ 
&\sum_{m=0}^{2L-1}\varphi_1(\bx+(m+1)\be_j,\bx+m\be_j)=\sum_{m=0}^{2L-1}\varphi_2(\bx+(m+1)\be_j,\bx+m\be_j)\\
&(\text{mod }2\pi),\quad(\forall \bx\in\Z^d,j,k\in\{1,2,\cdots,d\}).
\end{align*}

For $\bx,\by\in\Z^d$ we simply write $(\varphi_1-\varphi_2)(\bx,\by)$ in place
of $\varphi_1(\bx,\by)-\varphi_2(\bx,\by)$.

\begin{lemma}\label{lem_phase_circuit}
Assume that $n\ge 2$, $\bx_1,\bx_2,\cdots,\bx_n\in\G(2L)$ and for any
 $j\in \{1,2,\cdots,n\}$ there exists $p\in \{1,2,\cdots,d\}$ such that 
$\bx_j-\bx_{j+1}=\be_p$ or $-\be_p$ in $(\Z/2L\Z)^d$, where
 $\bx_{n+1}:=\bx_1$. Then, 
\begin{align}
\sum_{j=1}^n(\varphi_1-\varphi_2)(\bx_{j+1},\bx_j)=0\quad(\text{mod
 }2\pi).\label{eq_phase_circuit}
\end{align}
\end{lemma}

\begin{proof}
It follows from \cite[\mbox{Lemma A.2}]{K15} that if there are $p,q\in
 \{1,2,\cdots,$ $d\}$ such that for any $j\in\{1,2,\cdots,n\}$
 $\bx_j-\bx_{j+1}$ is equal to one of $\be_p$, $-\be_p$, $\be_q$,
 $-\be_q$ in $(\Z/2L\Z)^d$, then \eqref{eq_phase_circuit} holds. Let us
 call this property $(\star)$.

As hypothesis of induction, let us assume that there exists $l\in\{1,2,$
 $\cdots,d-1\}$ such that if for any $j\in\{1,2,\cdots,n\}$
 $\bx_j-\bx_{j+1}$ is equal to one of $\be_1$, $-\be_1$, $\be_2$,
 $-\be_2$, $\cdots$, $\be_l$, $-\be_l$ in $(\Z/2L\Z)^d$, then
 \eqref{eq_phase_circuit} holds.
Let $\bx_1,\bx_2,\cdots,\bx_n\in\G(2L)$ satisfy that 
for any $j\in\{1,2,\cdots,n\}$,
 $\bx_j-\bx_{j+1}$ is equal to one of $\be_1$, $-\be_1$, $\be_2$,
 $-\be_2$, $\cdots$, $\be_{l+1}$, $-\be_{l+1}$ in $(\Z/2L\Z)^d$.
Let us prove \eqref{eq_phase_circuit} for $\bx_1,\bx_2,\cdots,\bx_n$.
If $\bx_j-\bx_{j+1}$ is equal to $\be_{l+1}$ or $-\be_{l+1}$ in 
$(\Z/2L\Z)^d$ for any $j\in \{1,2,\cdots,n\}$, then
 \eqref{eq_phase_circuit} follows from $(\star)$. 
Assume that there exist $k_1,k_2,\cdots,k_m\in \{1,2,\cdots,n\}$ such
 that 
\begin{align*}
&k_1<k_2<\cdots<k_m,\\
&\bx_{k_p}-\bx_{k_{p}+1}\neq \be_{l+1},-\be_{l+1}\text{ in
 }(\Z/2L\Z)^d\quad(\forall p\in \{1,2,\cdots,m\}),\\
&\bx_j-\bx_{j+1}=\be_{l+1}\text{ or }-\be_{l+1}\text{ in
 }(\Z/2L\Z)^d\\
&(\forall j\in
 \{1,2,\cdots,n\}\backslash\{k_1,k_2,\cdots,k_m\}).
\end{align*}
If $m=1$, again \eqref{eq_phase_circuit} follows from
 $(\star)$. Assume that $m\ge 2$. Define the map
 $P:\G(2L)\to \G(2L)$ by 
$$
P(\bx):=(\bx(1),\cdots,\bx(l),\bx_1(l+1),\bx(l+2),\cdots,\bx(d)).
$$
For any $j\in \{1,2,\cdots,m-1\}$ we can choose
 $\by_{j,1},\by_{j,2},\cdots,\by_{j,q_j}\in \G(2L)$ so that
$\by_{j,1}=\bx_{k_j+1}$, $\by_{j,q_j}=P(\bx_{k_j+1})$, 
$\by_{j,p}-\by_{j,p+1}=\be_{l+1}$ or $-\be_{l+1}$ in $(\Z/2L\Z)^d$
 ($\forall p\in \{1,2,\cdots,q_j-1\}$). By $(\star)$,
\begin{align}
&(\varphi_1-\varphi_2)(P(\bx_{k_1+1}),\bx_1)\label{eq_next_remark}\\
&=\sum_{r=1}^{k_1}(\varphi_1-\varphi_2)(\bx_{r+1},\bx_r)+
\sum_{p=1}^{q_1-1}(\varphi_1-\varphi_2)(\by_{1,p+1},\by_{1,p})
\quad(\text{mod }2\pi).\notag
\end{align}
Moreover, by $(\star)$, for any $j\in \{1,2,\cdots,m-2\}$,
\begin{align}
&(\varphi_1-\varphi_2)(P(\bx_{k_{j+1}+1}),P(\bx_{k_{j}+1}))\label{eq_middle_remark}\\
&=-\sum_{p=1}^{q_j-1}(\varphi_1-\varphi_2)(\by_{j,p+1},\by_{j,p})+
\sum_{r=k_j+1}^{k_{j+1}}(\varphi_1-\varphi_2)(\bx_{r+1},\bx_{r})\notag\\
&\quad + \sum_{p=1}^{q_{j+1}-1}(\varphi_1-\varphi_2)(\by_{j+1,p+1},\by_{j+1,p})
\quad(\text{mod }2\pi),\notag\\
&(\varphi_1-\varphi_2)(\bx_{n+1},P(\bx_{k_{m-1}+1}))\label{eq_final_remark}\\
&=-\sum_{p=1}^{q_{m-1}-1}(\varphi_1-\varphi_2)(\by_{m-1,p+1},\by_{m-1,p})+
\sum_{r=k_{m-1}+1}^{n}(\varphi_1-\varphi_2)(\bx_{r+1},\bx_{r})\notag\\
&(\text{mod }2\pi).\notag
\end{align}
By adding \eqref{eq_next_remark}, \eqref{eq_middle_remark},
 \eqref{eq_final_remark} together,
\begin{align}
&(\varphi_1-\varphi_2)(P(\bx_{k_{1}+1}),\bx_{1})+\sum_{j=1}^{m-2}(\varphi_1-\varphi_2)(P(\bx_{k_{j+1}+1}),P(\bx_{k_{j}+1}))\label{eq_sum_chain}\\
&+(\varphi_1-\varphi_2)(\bx_{n+1},P(\bx_{k_{m-1}+1}))=
\sum_{r=1}^{n}(\varphi_1-\varphi_2)(\bx_{r+1},\bx_{r})\quad(\text{mod
 }2\pi).\notag
\end{align}
By the hypothesis of induction the left-hand side of
 \eqref{eq_sum_chain} is 0 (mod $2\pi$) and thus
 \eqref{eq_phase_circuit} holds.

The induction with $l\in\{1,2,\cdots,d\}$ concludes the proof.
\end{proof}

The next lemma is the $d$-dimensional version of \cite[\mbox{Lemma
A.3}]{K15}. However, the content is essentially same as
\cite[\mbox{Lemma 2.1}]{LL}.

\begin{lemma}(\cite[\mbox{Lemma 2.1}]{LL})
\label{lem_phase_gauge_invariance}
There exists a function $\theta:\G(2L)\to \R$ such that for any
 $\bx,\by\in\G(2L)$ satisfying that $\bx-\by$ is equal to one of
 $\be_1$, $-\be_1$, $\be_2$, $-\be_2$, $\cdots$, $\be_d$, $-\be_d$
in $(\Z/2L\Z)^d$,
$$
\varphi_1(\bx,\by)=\varphi_2(\bx,\by)+\theta(\bx)-\theta(\by)\quad(\text{mod
 }2\pi).
$$
\end{lemma}
\begin{proof}
Define $\theta:\G(2L)\to\R$ by 
\begin{align*}
&\theta((x_1,x_2,\cdots,x_d))\\
&:=1_{x_1\ge
 1}\sum_{j=0}^{x_1-1}(\varphi_1-\varphi_2)((j+1,0,\cdots,0),(j,0,\cdots,0))\\
&\quad +1_{x_2\ge
 1}\sum_{j=0}^{x_2-1}(\varphi_1-\varphi_2)((x_1,j+1,0,\cdots,0),(x_1,j,0,\cdots,0))+\cdots\\
&\quad +1_{x_d\ge
 1}\sum_{j=0}^{x_d-1}(\varphi_1-\varphi_2)((x_1,\cdots,x_{d-1},j+1),(x_1,\cdots,x_{d-1},j)).
\end{align*}
Then, Lemma \ref{lem_phase_circuit} implies that for any $\bx$,
 $\by\in\G(2L)$ satisfying that $\bx-\by$ is equal to one of 
$\be_1$, $-\be_1$, $\be_2$, $-\be_2$, $\cdots$, $\be_d$, $-\be_d$
in $(\Z/2L\Z)^d$,
$$
\theta(\bx)+(\varphi_1-\varphi_2)(\by,\bx)-\theta(\by)=0\quad(\text{mod
 }2\pi).
$$
\end{proof}

With a phase $\varphi:\Z^d\times\Z^d\to\R$ satisfying
\eqref{eq_phase_condition_general} we define the free Hamiltonian
$\sH_0(\varphi)$ by \eqref{eq_one_band_free_hamiltonian} and set
$\sH(\varphi)=\sH_0(\varphi)+\sV$ with the generalized interaction
$\sV$ defined in \eqref{eq_one_band_interaction}. 

\begin{lemma}\label{lem_partition_gauge_invariance}
$$
\Tr e^{-\beta \sH(\varphi_1)}=\Tr e^{-\beta \sH(\varphi_2)}.
$$
\end{lemma}
\begin{proof}
By using the function $\theta$ introduced in Lemma
 \ref{lem_phase_gauge_invariance} and following the proof of 
\cite[\mbox{Lemma A.4}]{K15} we can construct the unitary transform $B$
 on $F_f(L^2(\G(2L)\times\spin))$ so that
 $B\sH(\varphi_2)B^*=\sH(\varphi_1)$. 
Here we need the invariance \eqref{eq_U1_invariance} to ensure that
 $B\sV B^*=\sV$.
This implies the result.
\end{proof}

Here we can state the sufficient condition to be a minimizer of the flux
phase problem. In the following we restrict the interaction
$\sV$ to have the reflection positive form
\eqref{eq_reflection_positive_interaction}. 

\begin{theorem}(\cite{L})\label{thm_flux_phase}
Assume that the phase $\theta_L$ satisfies
 \eqref{eq_phase_condition_general}, \eqref{eq_flux_per_plaquette} 
with $\theta_{j,k}=\pi$ for any $j,k\in\{1,2,\cdots,d\}$ with $j<k$ and
\eqref{eq_flux_per_circle} with $\eps_l^L=1_{L\in 2\Z}$ for any
 $l\in\{1,2,\cdots,d\}$. Then,
\begin{align*}
&-\frac{1}{\beta}\log(\Tr e^{-\beta \sH(\theta_L)})\\
& =\min\left\{-\frac{1}{\beta}\log(\Tr e^{-\beta \sH(\varphi)})\
 \Big|\ \varphi:\Z^d\times\Z^d\to \R\text{ satisfying
 }\eqref{eq_phase_condition_general}\right\}.
\end{align*}
\end{theorem}
\begin{proof}
By Lemma
 \ref{lem_partition_gauge_invariance} it is sufficient to prove the
 existence of a phase with the claimed properties minimizing the free
 energy. For any $j,k\in \{1,2,\cdots,d\}$, $\bx\in \Z^d$, $s\in \Z$ and
 $\eta:\Z^d\times\Z^d\to \R$, set
\begin{align*}
&f_{j,k}(\eta)(\bx):=
\eta(\bx+\be_j,\bx)+\eta(\bx+\be_j+\be_k,\bx+\be_j)\\
&\qquad\qquad\quad+\eta(\bx+\be_k,\bx+\be_j+\be_k)+\eta(\bx,\bx+\be_k),\\
&f_j(\eta)(\bx):=\sum_{m=0}^{2L-1}\eta(\bx+(m+1)\be_j,\bx+m\be_j),\\
&H_j(s):=\Big\{\Big(y_1,\cdots,y_{j-1},s+\frac{1}{2},y_{j+1},\cdots,y_d\Big)\in\R^d\\
&\qquad\qquad\qquad  \Big|\ y_1,\cdots,y_{j-1},y_{j+1},\cdots,y_d\in\R\Big\}.
\end{align*}
Since the interaction $\sV$ is assumed to satisfy the positivity
 convention, we can apply the reflection positivity lemma
 \cite[\mbox{Lemma}]{L} with respect to the cutting hyper-plane
 $H_j(s)$. Recall that in the proof of \cite[\mbox{Theorem A.5}]{K15}
 first we did the reflection with the horizontal line $\{(x,1/2)\in\R^2\
 |\ x\in\R\}$ and secondly we did the reflections with the vertical
 lines $\{(s+1/2,y)\in\R^2\ |\ y\in\R\}$ $(s=0,1,\cdots,L-1)$.
The argument involving the reflections with the hyper-planes
 $H_2(0)$, $H_1(s)$ $(s=0,1,\cdots,L-1)$, parallel to the proof of 
 \cite[\mbox{Theorem A.5}]{K15}, proves that there exists a phase
 $\varphi$ satisfying \eqref{eq_phase_condition_general},
\begin{align*}
&f_{1,2}(\varphi)(\bx)=\pi,\quad f_1(\varphi)(\bx)=f_2(\varphi)(\bx)=
1_{L\in 2\N}\pi
\quad(\text{mod }2\pi),\ (\forall \bx\in\Z^d)
\end{align*}
and minimizing the free energy. This concludes the proof in the case
 $d=2$. Let us consider the case that $d\ge 3$. As hypothesis of
 induction, assume that $l\in \{2,3,\cdots,d-1\}$ and there exists a
 phase $\varphi$ satisfying \eqref{eq_phase_condition_general},
\begin{align}
&f_{j,k}(\varphi)(\bx)=\pi,\quad f_m(\varphi)(\bx)=1_{L\in 2\N}\pi
\quad(\text{mod }2\pi),\label{eq_phase_condition_induction}\\
&(\forall \bx\in\Z^d,j,k,m\in \{1,2,\cdots,l\}\text{ with
 }j<k)\notag
\end{align}
and minimizing the free energy.
For $s\in\{0,1,\cdots,L-1\}$ let us define the map 
$\Ref_s:\Z^d\to\Z^d$ by
\begin{align*}
&\Ref_s(\bx):=(x_1,\cdots,x_l,2s+1-x_{l+1},x_{l+2},\cdots,x_d).
\end{align*}
Then, define the transform $R_s$ on $\Map(\Z^d\times\Z^d,\R)$ by 
\begin{align*}
&R_s(\eta)(\bx,\by)\\
&:=\left\{\begin{array}{ll}\eta(\Ref_s(\bx),\Ref_s(\by))+\pi 
& \text{if }\exists j,k\in \{s+1,s+2,\cdots,s+L\}\text{
 s.t. }\\
 & \bx(l+1)=j,\ \by(l+1)=k\quad(\text{mod }2L),\\
\eta(\bx,\by)& \text{otherwise,}
\end{array}
\right.\\
&(\bx,\by\in\Z^d).
\end{align*}
Also, for any function $\theta:\Z^d\to \R$ satisfying that
 $\theta(\bx)=\theta(\by)$ for any $\bx,\by\in \Z^d$ with $\bx=\by$ in
 $(\Z/2L\Z)^d$ we define the transform $G_{\theta}$ on
 $\Map(\Z^d\times\Z^d,\R)$ by 
$$
G_{\theta}(\eta)(\bx,\by):=\eta(\bx,\by)+\theta(\bx)-\theta(\by),\quad
 (\bx,\by\in \Z^d).
$$
Note that if $\eta\in \Map(\Z^d\times\Z^d,\R)$ satisfies \eqref{eq_phase_condition_general} and
 \eqref{eq_phase_condition_induction}, so do $R_s(\eta)$,
 $G_{\theta}(\eta)$. We reform $\Tr e^{-\beta \sH(\varphi)}$ by
 repeating the reflection with respect to the hyper-planes $H_{l+1}(s)$
 $(s=0,1,\cdots,L-1)$ and the gauge transformations. This procedure is
 parallel to the part of the proof of \cite[\mbox{Theorem A.5}]{K15}
 demonstrating the reflections with respect to the vertical lines
$\{(s+1/2,y)\in\R^2\ |\ y\in\R\}$ $(s=0,1,\cdots,L-1)$. Here we consider
 the $l+1$-th coordinate, the $k$-th coordinate $(k\in \{1,2,\cdots,l\})$ as
 the first coordinate, the second coordinate respectively in the part
 of the proof of \cite[\mbox{Theorem A.5}]{K15} after the first
 reflection with the horizontal line $\{(x,1/2)\in\R^2\ |\ x\in\R\}$. By
 replacing $\be_1$, $\be_2$ by $\be_{l+1}$, $\be_{k}$ respectively there
 we can see that there exists a phase $\varphi'\in
 \Map(\Z^d\times\Z^d,\R)$ satisfying \eqref{eq_phase_condition_general},
\begin{align*}
&f_{k,l+1}(\varphi')(\bx)=\pi,\quad f_{l+1}(\varphi')(\bx)=1_{L\in 2\N}\pi
\quad(\text{mod }2\pi),\\
&(\forall \bx\in\Z^d,k\in \{1,2,\cdots,l\})
\end{align*}
and minimizing the free energy. In fact the phase $\varphi'$ is derived
 by repeatedly applying the transforms $R_s$, $G_{\theta}$ with
 $s\in\{0,1,\cdots,L-1\}$ and some periodic functions $\theta:\Z^d\to
 \R$ to the phase $\varphi$ given by the induction hypothesis. 
As remarked above, the phase
 $\varphi'$ still satisfies \eqref{eq_phase_condition_induction}. The induction with $l$ concludes
 the existence of a phase with the claimed properties.
\end{proof}

\section{Lemmas for the time-continuum, infinite-volume limit}\label{app_h_L_limit}
Here we prove that each term of the Taylor expansion of the free energy
density with respect to the amplitude of the interaction converges in
the time-continuum, infinite-volume limit. This fact is used
 to prove that the free energy
itself converges in these limits in Subsection \ref{subsec_completion_IR}. Basic ideas of this section are not
essentially different from those in \cite[\mbox{Appendix B}]{K10},
\cite[\mbox{Appendix D}]{K14},  \cite[\mbox{Appendix D}]{K15}. Since we introduced a class of interactions, which are different from
the interactions in the preceding papers and some properties of our
interactions are necessary to prove the fact of concern, we should again
demonstrate the major part of the proof.

For $n\in \N\cup \{0\}$ set
\begin{align*}
&a_n(L,h)(\bU):=-\frac{1}{\beta L^d n!}\left(\frac{\partial}{\partial
 z}\right)^n\log\left(\int e^{-zV(\bU)(\psi)}d\mu_{C}(\psi)
\right)\Big|_{z=0},
\end{align*}
where the Grassmann Gaussian integral is same as that considered in
Lemma \ref{lem_grassmann_formulation}. Our aim here is to prove the
uniform convergence property of $a_n(L,h)(\bU)$ with the coupling $\bU$
as $h,L\to\infty$. The covariance $C:(\cB\times\G(L)\times\spin
\times[0,\beta))^2\to \C$ was originally defined in
\eqref{eq_covariance_original_definition}. We can periodically extend the domain of
$C$ into $(\cB\times\Z^d\times\spin\times[0,\beta))^2$. Then, by taking
into account Lemma \ref{lem_hopping_properties}
\eqref{item_hopping_upper},\eqref{item_derivative_hopping_upper} we can
see that the same procedure as in the derivation of the inequalities 
\cite[\mbox{(D.3), (D.4), Appendix D}]{K15} yields the following
results.

\begin{lemma}\label{lem_original_covariance_naive_decay}
There exists a constant $c(\beta,d,(t_j)_{1\le j\le d})\in\R_{>0}$ depending
 only on $\beta$, $d$, $(t_j)_{1\le j\le d}$ such that the following
 inequalities hold. 
\begin{align}
&|C(\rho\bx\s x,\eta\by\tau y)|\le
 \frac{c(\beta,d,(t_j)_{1\le j\le d})}{1+\sum_{j=1}^d\left|\frac{L}{2\pi}
\big(e^{i\frac{2\pi}{L}\<\bx-\by,\be_j\>}-1\big)\right|^{d+1}},
\label{eq_original_covariance_naive_decay}\\
&(\forall (\rho,\bx,\s,x),(\eta,\by,\tau, y)\in \cB\times\Z^d\times
 \spin \times [0,\beta)),\notag\\
&|C(\rho\bx\s x,\eta\by\tau y)|\le
 \frac{c(\beta,d,(t_j)_{1\le j\le d})}{1+\left(\frac{2}{\pi}\right)^{d+1}
\sum_{j=1}^d|\<\bx-\by,\be_j\>|^{d+1}}
\label{eq_original_covariance_naive_decay_L_independent}\\
&(\forall (\rho,\bx,\s,x),(\eta,\by,\tau, y)\in \cB\times\Z^d\times
 \spin \times [0,\beta)\notag\\
&\quad\text{ with }|\<\bx-\by,\be_j\>|\le L/2\ (\forall
 j\in \{1,2,\cdots,d\})).\notag
\end{align}
\end{lemma}

For conciseness we set 
\begin{align*}
&J:=\cB\times \G(L)\times\spin\times\{1,-1\},\\
&J_c:=\cB\times
 \left[-\frac{L}{2},\frac{L}{2}\right)\times\spin\times \{1,-1\},\\
&J_0:=\cB\times \{\b0\}\times\spin\times\{1,-1\},\quad J_{\infty}:=\cB\times
 \Z^d\times\spin\times \{1,-1\}.
\end{align*}
Using the original kernels $V_m^L$ $(m=0,1,\cdots,N_v)$ of the
interaction, we define $V_0^0\in \Map(\C^{n_v},\C)$, $V_{2m}^0\in
\Map(\C^{n_v},\Map(J_{\infty}^{2m},\C))$ ($m=1,2,\cdots,N_v$) by 
\begin{align*}
&V_0^0(\bU):=V_0^L(\bU),\\
&V_{2m}^0(\bU)((\rho_1,\bx_1,\s_1,\theta_1),\cdots,(\rho_{2m},\bx_{2m},\s_{2m},\theta_{2m}))\\
&:=\frac{1}{(2m)!}\sum_{\xi\in\S_{2m}}\sgn(\xi)1_{(\theta_{\xi(1)},\cdots,\theta_{\xi(m)})=(1,\cdots,1),(\theta_{\xi(m+1)},\cdots,\theta_{\xi(2m)})=(-1,\cdots,-1)}\\
&\quad\cdot V_m^{L}(\bU)(((2\bx_{\xi(1)}+b(\rho_{\xi(1)}),\s_{\xi(1)}),\cdots,
(2\bx_{\xi(m)}+b(\rho_{\xi(m)}),\s_{\xi(m)})),\\
&\qquad\quad ((2\bx_{\xi(m+1)}+b(\rho_{\xi(m+1)}),\s_{\xi(m+1)}),\cdots,
(2\bx_{\xi(2m)}+b(\rho_{\xi(2m)}),\s_{\xi(2m)}))).
\end{align*}
The next lemma summarizes some properties of $V_0^0$, $V_{2m}^0$ which we
will use later.

\begin{lemma}\label{lem_temperature_independent_interaction}
For any $r\in\R_{>0}$, $m\in \{1,2,\cdots,N_v\}$ the following
 statements hold.
\begin{align}
&\sup_{\bU\in
 \overline{D(r)}^{n_v}}|V_{2m}^0(\bU)(\rho_1\b0\s_1\theta_1,\rho_2\bx_2\s_2\theta_2,\cdots,\rho_{2m}\bx_{2m}\s_{2m}\theta_{2m})|\label{eq_temperature_independent_interaction_decay}\\
&\le r
 e^dv_m(1)e^{-\frac{1}{2m-1}\sum_{p=2}^{2m}\sum_{j=1}^d(\frac{2}{\pi}|\<\bx_p,\be_j\>|)^{1/2}},\notag\\
&\quad(\forall (\rho_1,\s_1,\theta_1)\in\cB\times \spin\times \{1,-1\},\notag\\
&\qquad  (\rho_j,\bx_j,\s_j,\theta_j)\in J_{\infty}\cap J_c\
 (j=2,3,\cdots,2m)).\notag\\
&\frac{1}{L^d}V_0^0,V_{2m}^0(\bX)\text{ converge in
 }C(\overline{D(r)}^{n_v};\C)\text{ as }L\to \infty(L\in\N)\notag\\
&\text{ for
 any }\bX\in J_{\infty}^{2m}.\notag
\end{align}
\end{lemma}

\begin{proof}
Take any $\bU\in\overline{D(r)}^{n_v}$, $(\rho_j,\bx_j,\s_j,\theta_j)\in
 J_{\infty}\cap J_c$ $(j=1,2,\cdots,2m)$, $p\in\{2,3,\cdots,2m\}$.
Note that 
\begin{align*}
e^{\sum_{j=1}^d(\frac{L}{2\pi}|e^{i\frac{2\pi}{L}\<\bx_1-\bx_p,\be_j\>}-1|)^{1/2}}
\le
 e^{\sum_{j=1}^d(\frac{L}{\pi}|e^{i\frac{\pi}{L}\<2\bx_1+b(\rho_1)-2\bx_p-b(\rho_p),\be_j\>}-1|)^{1/2}+d}.
\end{align*}
By this inequality, the linearity with $\bU$,
 \eqref{eq_bi_anti_symmetric}, \eqref{eq_hermiticity} and 
 \eqref{eq_decay_bound},
\begin{align*}
&e^{\sum_{j=1}^d(\frac{L}{2\pi}|e^{i\frac{2\pi}{L}\<\bx_1-\bx_p,\be_j\>}-1|)^{1/2}}|V_{2m}^0(\bU)(\rho_1\bx_1\s_1\theta_1,\cdots,\rho_{2m}\bx_{2m}\s_{2m}\theta_{2m})|\\
&\le
 r e^{d}v_m(1).
\end{align*}
Thus, 
\begin{align*}
&\sup_{\bU\in
 \overline{D(r)}^{n_v}}|V^0_{2m}(\bU)(\rho_1\bx_1\s_1\theta_1,\cdots,\rho_{2m}\bx_{2m}\s_{2m}\theta_{2m})|\\
&\le r
 e^{d}v_m(1)e^{-\frac{1}{2m-1}\sum_{p=2}^{2m}\sum_{j=1}^d(\frac{L}{2\pi}|e^{i\frac{2\pi}{L}\<\bx_1-\bx_p,\be_j\>}-1|)^{1/2}}.
\end{align*}
The above inequality implies \eqref{eq_temperature_independent_interaction_decay}. The claimed convergence properties
 follow from that $\bU\mapsto
 V_0^L(\bU)$, $\bU\mapsto V_m^L(\bU)(\bX)$ are linear and $\frac{1}{L^d}\frac{\partial}{\partial
 U_j}V_0^L(\bU)$, $\frac{\partial}{\partial
 U_j}V_m^L(\bU)(\bX)$ converge as $L\to \infty$. 
\end{proof}

For convenience in the proof of the next lemma we introduce some more
notations. Define the transform $P_0$ on $J_{\infty}^m$ by 
\begin{align*}
&P_0(((\rho_1,\bx_1,\s_1,\theta_1),\cdots,(\rho_m,\bx_m,\s_m,\theta_m)))\\
&=((\rho_1,\b0,\s_1,\theta_1),\cdots,(\rho_m,\b0,\s_m,\theta_m)).
\end{align*}
Define the map $P_s$ from $J_{\infty}$ to $\Z^d$ by
$P_s((\rho,\bx,\s,\theta)):=\bx$. Moreover, define the map $P_L$ from
$J_{\infty}^m$ to $J^m$ by
\begin{align*}
&P_L(((\rho_1,\bx_1,\s_1,\theta_1),\cdots,(\rho_m,\bx_m,\s_m,\theta_m)))\\
&=((\rho_1,\bx_1',\s_1,\theta_1),\cdots,(\rho_m,\bx_m',\s_m,\theta_m)),
\end{align*}
where $\bx_j'\in\G(L)$ and $\bx_j=\bx_j'$ in $(\Z/L\Z)^d$
$(j=1,2,\cdots,m)$. We also define a map from $(\cB\times
\Z^d\times\spin\times[0,\beta)\times\{1,-1\})^m$ to 
 $(\cB\times
\G(L)\times\spin\times[0,\beta)\times\{1,-1\})^m$ in the same way as
above and let $P_L$ denote the map, though this is abuse of notation.
For any
$\bX=((\rho_1,\bx_1,\s_1,\theta_1),\cdots,(\rho_m,\bx_m,\s_m,\theta_m))\in
J_{\infty}^m$, $s\in [0,\beta)$ we define $(\bX|s)\in (\cB\times
\Z^d\times\spin\times[0,\beta)\times\{1,-1\})^m$ by
$$
(\bX|s):=((\rho_1,\bx_1,\s_1,s,\theta_1),\cdots,(\rho_m,\bx_m,\s_m,s,\theta_m)).$$
Furthermore, for any
$\bX=((\rho_1,\bx_1,\s_1,s_1,\theta_1),\cdots,(\rho_m,\bx_m,\s_m,s_m,\theta_m))\in
(\cB\times\Z^d\times\spin\times[0,\beta)\times\{1,-1\})^m$ and
$\bx\in\Z^d$ we define $\bX+\bx\in  (\cB\times
\Z^d\times\spin\times[0,\beta)\times\{1,-1\})^m$ by
$$
\bX+\bx:=((\rho_1,\bx_1+\bx,\s_1,s_1,\theta_1),\cdots,(\rho_m,\bx_m+\bx,\s_m,s_m,\theta_m)).$$
For any $\bY\in J_{\infty}^{m}$ and $\bx\in\Z^d$ we also define $\bY+\bx\in J_{\infty}^m$
in the same way. We use the same notational rules for different power
$m$ for simplicity. We should make clear that the notations introduced
above are used only in the rest of this section, not used anywhere else
in this paper.

Here let us note that
\begin{align*}
V(\psi)=\beta V_0^0+\frac{1}{h}\sum_{s\in
 [0,\beta)_h}\sum_{m=1}^{N_v}\sum_{\bX\in
 J^{2m}}V^0_{2m}(\bX)\psi_{(\bX|s)}.
\end{align*}

\begin{lemma}\label{lem_truncation_h_L_limit}
For any $r\in\R_{>0}$, $n\in \N\cup\{0\}$ the following statements hold
 true.
\begin{enumerate}
\item\label{item_truncation_h_limit}
$a_n(L,h)$ converges in $C(\overline{D(r)}^{n_v};\C)$ as $h\to \infty(h\in(2/\beta)\N)$.
\item\label{item_truncation_L_limit}
Set $a_n(L):=\lim_{h\to \infty, h\in(2/\beta)\N}a_n(L,h)$. Then,
$a_n(L)$ converges in $C(\overline{D(r)}^{n_v};\C)$ as $L\to \infty(L\in\N)$.
\end{enumerate}
\end{lemma}

\begin{proof}
Since $a_0(L,h)=0$, the claims are trivial for $n=0$. First let us prove
 the claims for $n=1$. By the translation invariance
 \eqref{eq_translation} and the periodicity \eqref{eq_periodicity},
\begin{align*}
a_1(L,h)
&=\frac{1}{\beta L^d}\int V(\psi)d\mu_{C}(\psi)\\
&=\frac{1}{L^d}V_0^0+\frac{1}{L^d}\sum_{m=1}^{N_v}\sum_{\bX\in
 J^{2m}}V_{2m}^0(\bX)\int \psi_{(\bX|0)}d\mu_{C}(\psi)\\
&=\frac{1}{L^d}V_0^0+\frac{1}{L^d}\sum_{m=1}^{N_v}\sum_{\bX\in
 J^{2m}}V_{2m}^0(\bX-P_s(X_{1}))\\
&\qquad\qquad\qquad\qquad\qquad\cdot \int
 \psi_{P_L((\bX|0)-P_s(X_{1}))}d\mu_{C}(\psi)\\
&=\frac{1}{L^d}V_0^0+\sum_{m=1}^{N_v}\sum_{X\in J_0}
\sum_{\bX\in
 J^{2m-1}}V_{2m}^0(X,\bX)
 \int\psi_{((X,\bX)|0)}d\mu_{C}(\psi)\\
&=\frac{1}{L^d}V_0^0+\sum_{m=1}^{N_v}\sum_{X\in J_0}
\sum_{\bX\in
 J_{\infty}^{2m-1}}1_{\bX\in J_c^{2m-1}}
V_{2m}^0(X,\bX)\\
&\qquad\qquad\qquad\qquad\qquad\cdot\int
 \psi_{P_L(((X,\bX)|0))}d\mu_{C}(\psi).
\end{align*}
Set 
\begin{align*}
F_m'(\bU)(X,\bX):=1_{\bX\in J_c^{2m-1}}
V_{2m}^0(\bU)(X,\bX)\int
 \psi_{P_L(((X,\bX)|0))}d\mu_{C}(\psi).
\end{align*}
Note that $F_m'$ is independent of $h$. Then, it follows from 
 \eqref{eq_original_covariance_naive_decay}, the convergence property of
 $C$ as $L\to \infty$ and Lemma
 \ref{lem_temperature_independent_interaction} that 
\begin{align*}
&\sup_{\bU\in \overline{D(r)}^{n_v}}|F_m'(\bU)(X,\bX)|\\
&\le
 m!c(\beta,d,(t_j)_{1\le j\le d})^mr e^dv_m(1)
e^{-\frac{1}{2m-1}\sum_{p=1}^{2m-1}\sum_{j=1}^d(\frac{2}{\pi}|\<P_s(X_p),\be_j\>|)^{1/2}}
\end{align*}
and $\lim_{L\to \infty, L\in\N}F_m'(\cdot)(X,\bX)$ converges in
 $C(\overline{D(r)}^{n_v};\C)$ for any $m\in \{1,2,$ $\cdots,N_v\}$, $X\in
 J_0$, $\bX\in J_{\infty}^{2m-1}$.
Therefore, by the dominated convergence theorem in $L^1(J_0\times
 J_{\infty}^{2m-1},C(\overline{D(r)}^{n_v};\C))$ and the convergence property of
 $(1/L^d)V_0^0$ we see that $a_1(L,h)$ has the claimed convergence
 properties.

Let $n\ge 2$. Here we need to recall the tree formula for $a_n(L,h)$. We
 adopt a version of the tree formula \cite[\mbox{Theorem
 3}]{SW}, which states that
\begin{align*}
&a_n(L,h)\\
&= \frac{(-1)^{n+1}}{n!\beta L^d}\sum_{T\in \T_n}\prod_{\{p,q\}\in
 T}(\D_{p,q}(C)+\D_{q,p}(C))ope(T,C)\prod_{j=1}^nV(\psi^j)\Bigg|_{\psi^j=0\atop(\forall
 j\in\{1,2,\cdots,n\})},
\end{align*}
where $\T_n$ is the set of all trees over the vertices
$\{1,2,\cdots,n\}$, 
\begin{align*}
\D_{r,s}(C)&:=-\sum_{X,Y\in I_0}C(X,Y)\frac{\partial}{\partial
 \opsi_X^r}\frac{\partial}{\partial \psi_Y^s},\quad (\forall r,s\in
 \{1,2,\cdots,n\})
\end{align*}
with the Grassmann left derivatives $\partial/\partial \opsi_X^r$,
 $\partial/\partial \psi_X^r$, and 
\begin{align*}
ope(T,C)&:=\int_{[0,1]^{n-1}}d\bs\sum_{\xi\in
 \S_n(T)}\varphi(T,\xi,\bs)e^{\sum_{r,s=1}^nM(T,\xi,\bs)(r,s)\D_{r,s}(C)},
\end{align*}
with a $T$-dependent subset $\S_n(T)$ of $\S_n$, a $(T,\xi)$-dependent
 function $\varphi(T,\xi,\cdot)\in C([0,1]^{n-1};\R_{\ge 0})$ satisfying 
\begin{align}\label{eq_ope_irrelevant_function}
\int_{[0,1]^{n-1}}d\bs\sum_{\xi\in \S_n(T)}\varphi(T,\xi,\bs)=1,\quad
 (\forall T\in\T_n),
\end{align}
and a $(T,\xi)$-dependent matrix-valued function $M(T,\xi,\cdot)\in$\\ 
 $C([0,1]^{n-1};\Mat(n,\R))$ satisfying 
\begin{align}
&|M(T,\xi,\bs)(r,s)|\le 1,\label{eq_ope_matrix_component_bound}\\
&(\forall T\in \T_n,\xi\in\S_n(T),\bs\in [0,1]^{n-1},r,s\in
 \{1,2,\cdots,n\}).\notag
\end{align}

The important bound property of the operator $ope(T,C)$ is that
\begin{align}
&\left|ope(T,C)\psi_{\bX_1}^1\psi_{\bX_2}^2\cdots\psi_{\bX_n}^n\Big|_{\psi^j=0\atop(\forall
 j\in\{1,2,\cdots,n\})}\right|\label{eq_ope_naive_bound}\\
&\le \left\lfloor
 \frac{1}{2}\sum_{k=1}^{n}m_k\right\rfloor!c(\beta,d,(t_j)_{1\le j\le d})^{\frac{1}{2}\sum_{k=1}^{n}m_k},\notag\\
&(\forall m_j\in\N\cup\{0\},\bX_j\in I^{m_j}\ (j=1,2,\cdots,n)),\notag
\end{align}
which follows from \eqref{eq_original_covariance_naive_decay},
 \eqref{eq_ope_irrelevant_function} and
 \eqref{eq_ope_matrix_component_bound}.
The proof of \cite[\mbox{Lemma 4.5}]{K09} essentially shows how to
 derive \eqref{eq_ope_naive_bound}.

Define the anti-symmetric function
 $\widetilde{C}:(\cB\times\Z^d\times\spin\times[0,\beta)\times
 \{1,-1\})^2\to\C$ by 
\begin{align*}
&\widetilde{C}((X,\theta),(Y,\xi)):=\frac{1}{2}(1_{(\theta,\xi)=(1,-1)}C(X,Y)-1_{(\theta,\xi)=(-1,1)}C(Y,X)),\\
&(\forall X,Y\in \cB\times\Z^d\times\spin\times [0,\beta),
 \theta,\xi\in\{1,-1\}).
\end{align*}
Then, we have that 
\begin{align*}
\D_{p,q}(C)+\D_{q,p}(C)=-2\sum_{\bX\in
 I^2}\widetilde{C}(\bX)\frac{\partial}{\partial
 \psi_{X_1}^p}\frac{\partial}{\partial \psi_{X_2}^q}.
\end{align*}
The term $a_n(L,h)$ can be expanded as follows.
\begin{align*}
&a_n(L,h)\\
&= \frac{2^{n-1}}{n!\beta}\sum_{T\in \T_n}\prod_{j=1}^n\left(\sum_{m_j=2}^{2N_v}\right)
\frac{1}{L^d}ope(T,C)
\prod_{\{p,q\}\in T}\left(\sum_{\bY\in I^2}\widetilde{C}(\bY)\frac{\partial}{\partial
 \psi_{Y_1}^p}\frac{\partial}{\partial \psi_{Y_2}^q}\right)\\
&\quad\cdot 
\prod_{k=1}^n\left(\left(\frac{1}{h}\right)^{m_k}\sum_{\bX_k\in I^{m_k}}V_{m_k}(\bX_k)\psi_{\bX_k}^k\right)\Bigg|_{\psi^j=0\atop(\forall
 j\in\{1,2,\cdots,n\})}.
\end{align*}
For $T\in\T_n$ and $j\in \{1,2,\cdots,n\}$ let $d_j(T)$ denote the
 degree of the vertex $j$ in $T$. Fix $m_j\in \{2,4,\cdots,2N_v\}$ $(j=1,2,\cdots,n)$.
If $d_j(T)$ is larger than $m_j$ for
 some $j$, the derivatives along the lines of $T$ erase the Grassmann
 polynomials completely and thus such a tree does not contribute to the result.
  Take any $T\in
 \T_n$ satisfying $d_j(T)\le m_j$ $(\forall j\in
 \{1,2,\cdots,n\})$. Then, set
\begin{align*}
a_n'(L,h)
&:=
\frac{1}{L^d}ope(T,C)
\prod_{\{p,q\}\in T}\left(\sum_{\bY\in I^2}\widetilde{C}(\bY)\frac{\partial}{\partial
 \psi_{Y_1}^p}\frac{\partial}{\partial \psi_{Y_2}^q}\right)\\
&\quad\cdot 
\prod_{k=1}^n\left(\left(\frac{1}{h}\right)^{m_k}\sum_{\bX_k\in I^{m_k}}V_{m_k}(\bX_k)\psi_{\bX_k}^k\right)\Bigg|_{\psi^j=0\atop(\forall
 j\in\{1,2,\cdots,n\})}.
\end{align*}
It suffices to prove the convergence properties of $a_n'(L,h)$ instead
 of $a_n(L,h)$. By changing the numbering if necessary, we may assume
 that if $\{p,q\}\in T$ and $p<q$ the length of the shortest path
 between 1 and $p$ in $T$ is shorter than that between 1 and $q$. Then,
 we can define the map $f:\{2,3,\cdots,n\}\to
 \{1,2,\cdots,n-1\}$
by $f(q):=p$ with $p\in \{1,2,\cdots,q-1\}$ satisfying $\{p,q\}\in T$.

 To shorten formulas, we use the notational
 convention that for integers $l,l+1,\cdots,l+m$ and objects
 $w_l,w_{l+1},\cdots,w_{l+m}$, 
$$
\prod_{j=l\atop order}^{l+m}w_j,\quad\prod_{j=l+m\atop order}^{l}w_j 
$$
denote
$$
w_lw_{l+1}\cdots w_{l+m},\quad w_{l+m}w_{l+m-1}\cdots w_l
$$
respectively. Also, it will be convenient to write $X\subset \bY$ for
 $X\in J$, $\bY\in J^n$ if there exists $j\in\{1,2,\cdots,n\}$ such that $X=Y_j$.
  By using the notations
 introduced so far, anti-symmetry, translation invariance and
 periodicity from part to part we can transform as follows.
\begin{align*}
&a_n'(L,h)\\
&=
\frac{1}{L^d}ope(T,C)\prod_{j=1}^n\left(\frac{1}{h}\sum_{s_j\in
 [0,\beta)_h}\right)\\
&\quad\cdot\prod_{q=n\atop order}^{2}\left(\sum_{Y,Z_q\in
 J}\widetilde{C}((Y|s_{f(q)}),(Z_q|s_q))\frac{\partial}{\partial
 \psi_{(Y|s_{f(q)})}^{f(q)}}\right)\sum_{\bX_1\in
 J^{m_1}}V_{m_1}^0(\bX_1)\psi^1_{(\bX_1|s_1)}\\
&\quad\cdot
\prod_{k=2\atop order}^n\left(m_k\sum_{\bX_k\in J^{m_k-1}}V_{m_k}^0(Z_k,\bX_k)\psi_{(\bX_k|s_k)}^k\right)\Bigg|_{\psi^j=0\atop(\forall
 j\in\{1,2,\cdots,n\})}\\
&=
\frac{1}{L^d}\prod_{j=1}^n\left(\frac{1}{h}\sum_{s_j\in
 [0,\beta)_h}\right)\sum_{X_0\in J\atop \bX_1\in
 J^{m_1-1}}V_{m_1}^0(X_0,\bX_1+P_s(X_0))\\
&\quad\cdot\prod_{q=2}^{n}\Bigg(m_q\sum_{Y_{q},Z_q\in
 J\atop \bX_q\in
 J^{m_q-1}}\widetilde{C}((Y_{q}|s_{f(q)})+P_s(X_0),(Z_q|s_q)+P_s(X_0))\\
&\qquad\qquad\qquad\qquad\quad\cdot
V_{m_q}^0(Z_q+P_s(X_0),\bX_q+P_s(X_0))\Bigg)\\
&\quad\cdot ope(T,C)\prod_{k=n\atop order}^2\left(\frac{\partial}{\partial
 \psi_{P_L((Y_{k}+P_s(X_0)|s_{f(k)}))}^{f(k)}}\right)\\
&\quad\cdot
\psi^1_{(X_0|s_1)}\psi_{P_L((\bX_1+P_s(X_0)|s_1))}^1
\prod_{l=2\atop order}^n\psi_{P_L((\bX_l+P_s(X_0)|s_l))}^l\Bigg|_{\psi^j=0\atop(\forall
 j\in\{1,2,\cdots,n\})}\\
&=
\prod_{j=1}^n\left(\frac{1}{h}\sum_{s_j\in
 [0,\beta)_h}\right)\sum_{X_0\in J_0\atop \bX_1\in
 J^{m_1-1}}V_{m_1}^0(X_0,\bX_1)\\
&\quad\cdot\prod_{q=2}^{n}\Bigg(m_q\sum_{Y_{q},Z_q\in
 J\atop \bX_q\in
 J^{m_q-1}}\widetilde{C}((Y_{q}|s_{f(q)}),(Z_q|s_q))
V_{m_q}^0(Z_q,\bX_q)\Bigg)\\
&\quad\cdot ope(T,C)\prod_{k=n\atop order}^2\left(\frac{\partial}{\partial
 \psi_{(Y_{k}|s_{f(k)})}^{f(k)}}\right)
\psi^1_{(X_0|s_1)}\psi_{(\bX_1|s_1)}^1
\prod_{l=2\atop order}^n\psi_{(\bX_l|s_l)}^l\Bigg|_{\psi^j=0\atop(\forall
 j\in\{1,2,\cdots,n\})}\\
&=
\prod_{j=1}^n\left(\frac{1}{h}\sum_{s_j\in
 [0,\beta)_h}\right)\sum_{X_0\in J_0\atop \bX_1\in
 J^{m_1-1}}V_{m_1}^0(X_0,\bX_1)\\
&\quad\cdot\prod_{q=2}^{n}\Bigg(m_q\sum_{Y_{q},Z_q\in
 J\atop \bX_q\in
 J^{m_q-1}}\widetilde{C}((Y_{q}|s_{f(q)}),(Z_q|s_q))
V_{m_q}^0(Z_q,\bX_q+P_s(Z_q))\Bigg)\\
&\quad\cdot ope(T,C)\prod_{k=n\atop order}^2\left(\frac{\partial}{\partial
 \psi_{(Y_{k}|s_{f(k)})}^{f(k)}}\right)\\
&\quad\cdot
\psi^1_{(X_0|s_1)}\psi_{(\bX_1|s_1)}^1
\prod_{l=2\atop order}^n\psi_{P_L((\bX_l+P_s(Z_l)|s_l))}^l\Bigg|_{\psi^j=0\atop(\forall
 j\in\{1,2,\cdots,n\})}\\
&=
\prod_{j=1}^n\left(\frac{1}{h}\sum_{s_j\in
 [0,\beta)_h}\right)\sum_{X_0\in J_0\atop \bX_1\in
 J^{m_1-1}}V_{m_1}^0(X_0,\bX_1)\\
&\quad\cdot\prod_{q=2}^{n}\Bigg(m_q\sum_{Y_{q},Z_q\in
 J\atop \bX_q\in
 J^{m_q-1}}\widetilde{C}((Y_{q}|s_{f(q)}),(Z_q|s_q))
V_{m_q}^0(P_0(Z_q),\bX_q)\Bigg)\\
&\quad\cdot ope(T,C)\prod_{k=n\atop order}^2\left(\frac{\partial}{\partial
 \psi_{(Y_{k}|s_{f(k)})}^{f(k)}}\right)\\
&\quad\cdot
\psi^1_{(X_0|s_1)}\psi_{(\bX_1|s_1)}^1
\prod_{l=2\atop order}^n\psi_{P_L((\bX_l+P_s(Z_l)|s_l))}^l\Bigg|_{\psi^j=0\atop(\forall
 j\in\{1,2,\cdots,n\})}\\
&=
\prod_{j=1}^n\left(\frac{1}{h}\sum_{s_j\in
 [0,\beta)_h}\right)\sum_{X_0\in J_0\atop \bX_1\in
 J^{m_1-1}}V_{m_1}^0(X_0,\bX_1)\\
&\quad\cdot\prod_{q=2}^{n}\Bigg(m_q\sum_{Y_{q},Z_q\in
 J\atop \bX_q\in
 J^{m_q-1}}\widetilde{C}((Y_{q}|s_{f(q)}),(Z_q+P_s(Y_{q})|s_q))
V_{m_q}^0(P_0(Z_q),\bX_q)\Bigg)\\
&\quad\cdot ope(T,C)\prod_{k=n\atop order}^2\left(\frac{\partial}{\partial
 \psi_{(Y_{k}|s_{f(k)})}^{f(k)}}\right)\\
&\quad\cdot
\psi^1_{(X_0|s_1)}\psi_{(\bX_1|s_1)}^1
\prod_{l=2\atop order}^n\psi_{P_L((\bX_l+P_s(Z_l)+P_s(Y_{l})|s_l))}^l\Bigg|_{\psi^j=0\atop(\forall
 j\in\{1,2,\cdots,n\})}\\
&=
\prod_{j=1}^n\left(\frac{1}{h}\sum_{s_j\in
 [0,\beta)_h}\right)\sum_{X_0\in J_0\atop \bX_1\in
 J^{m_1-1}}V_{m_1}^0(X_0,\bX_1)\\
&\quad\cdot\prod_{q=2}^{n}\Bigg(\sum_{Y_{q}\in J_0,Z_q\in
 J\atop \bX_q\in
 J^{m_q-1}}\widetilde{C}((Y_{q}|s_{f(q)}),(Z_q|s_q))
V_{m_q}^0(P_0(Z_q),\bX_q)\Bigg)\\
&\quad\cdot F((s_j)_{j=1}^n,X_0,\bX_1,(\bX_j)_{j=2}^n,(Y_{j})_{j=2}^n,(Z_j)_{j=2}^n),
\end{align*}
where $F$ is the function on 
$$
 [0,\beta)_h^n\times J_0\times
 J_{\infty}^{m_1-1}\times\prod_{j=2}^nJ_{\infty}^{m_j-1}\times
 J_0^{n-1}\times J_{\infty}^{n-1}
$$
defined by
\begin{align*}
&F((s_j)_{j=1}^n,X_0,\bX_1,(\bX_j)_{j=2}^n,(Y_{j})_{j=2}^n,(Z_j)_{j=2}^n)\\
&:=\prod_{q=2}^{n}\Bigg(m_q\sum_{\by_{q}\in \G(L)}\Bigg) ope(T,C)\prod_{k=n\atop order}^2\left(\frac{\partial}{\partial
 \psi_{(Y_{k}+\by_{k}|s_{f(k)})}^{f(k)}}\right)\\
&\qquad\cdot
\psi^1_{(X_0|s_1)}\psi_{P_L((\bX_1|s_1))}^1
\prod_{l=2\atop order}^n\psi_{P_L((\bX_l+P_s(Z_l)+\by_{l}|s_l))}^l\Bigg|_{\psi^j=0\atop(\forall
 j\in\{1,2,\cdots,n\})}.
\end{align*}
Note that 
\begin{align*}
&F((s_j)_{j=1}^n,X_0,\bX_1,(\bX_j)_{j=2}^n,(Y_{j})_{j=2}^n,(Z_j)_{j=2}^n)\\
&=\prod_{j=2}^nm_j\prod_{r\in f^{-1}(1)}\Bigg(\sum_{\by_{r}\in
 \G(L)}1_{Y_{r}+\by_{r}\subset P_L((X_0,\bX_1))}\Bigg)\\
&\quad\cdot\prod_{q=2\atop order}^{n-1}\Bigg(\prod_{p\in f^{-1}(q)}
\Bigg(\sum_{\by_{p}\in
 \G(L)}1_{Y_{p}+\by_{p}\subset P_L(\bX_q+P_s(Z_q)+\by_{q})}\Bigg)\Bigg)\\
&\quad\cdot  ope(T,C)\prod_{k=n\atop order}^2\left(\frac{\partial}{\partial
 \psi_{(Y_{k}+\by_{k}|s_{f(k)})}^{f(k)}}\right)\\
&\quad\cdot
\psi^1_{(X_0|s_1)}\psi_{P_L((\bX_1|s_1))}^1
\prod_{l=2\atop order}^n\psi_{P_L((\bX_l+P_s(Z_l)+\by_{l}|s_l))}^l\Bigg|_{\psi^j=0\atop(\forall
 j\in\{1,2,\cdots,n\})},
\end{align*}
and thus by \eqref{eq_ope_naive_bound},
\begin{align}
&|F((s_j)_{j=1}^n,X_0,\bX_1,(\bX_j)_{j=2}^n,(Y_{j})_{j=2}^n,(Z_j)_{j=2}^n)|\label{eq_ope_aposteriori_bound}\\
&\le \prod_{j=2}^nm_j\prod_{r\in
 f^{-1}(1)}m_1\prod_{q=2}^{n-1}\prod_{p\in f^{-1}(q)}(m_q-1)
\left(\frac{1}{2}\sum_{l=1}^nm_l-n+1\right)!\notag\\
&\quad\cdot c(\beta,d,(t_j)_{1\le j\le d})^{\frac{1}{2}\sum_{l=1}^nm_l-n+1}\notag\\
&\le  \prod_{j=1}^nm_j^{d_T(j)}
\left(\frac{1}{2}\sum_{l=1}^nm_l-n+1\right)!c(\beta,d,(t_j)_{1\le j\le d})^{\frac{1}{2}\sum_{l=1}^nm_l-n+1},\notag\\
&(\forall
 ((s_j)_{j=1}^n,X_0,\bX_1,(\bX_j)_{j=2}^n,(Y_{j})_{j=2}^n,(Z_j)_{j=2}^n)\in
 [0,\beta)_h^n\times \bJ_{\infty}).\notag
\end{align}
Here we set 
$$
\bJ_{\infty}:=J_0\times
 J_{\infty}^{m_1-1}\times\prod_{j=2}^nJ_{\infty}^{m_j-1}\times
 J_0^{n-1}\times J_{\infty}^{n-1}.
$$

For any $s\in
 [0,\beta)$ we let $\hat{s}$ denote an element of $[0,\beta)_h$
 satisfying $s\in [\hat{s},\hat{s}+1/h)$. By periodicity we can rewrite
 $a_n'(L,h)$ as follows.
\begin{align*}
a_n'(L,h)
&=
\prod_{j=1}^n\left(\int_{0}^{\beta}ds_j\right)
\sum_{X_0\in J_0\atop \bX_1\in
 J^{m_1-1}_{\infty}}V_{m_1}^0(X_0,\bX_1)\\
&\quad\cdot\prod_{q=2}^{n}\Bigg(\sum_{Y_{q}\in J_{0},Z_q\in
 J_{\infty}\atop \bX_q\in
 J_{\infty}^{m_q-1}}\widetilde{C}((Y_{q}|\hat{s}_{f(q)}),(Z_q|\hat{s}_q))
V_{m_q}^0(P_0(Z_q),\bX_q)\Bigg)\\
&\quad\cdot 1_{((\bX_j)_{j=1}^n,(Z_j)_{j=2}^n)\in
 \prod_{j=1}^nJ_c^{m_j-1}\times J_c^{n-1}}\\
&\quad\cdot F((\hat{s}_j)_{j=1}^n,X_0,\bX_1,(\bX_j)_{j=2}^n,(Y_{j})_{j=2}^n,(Z_j)_{j=2}^n).
\end{align*}
By \eqref{eq_original_covariance_naive_decay_L_independent},
 \eqref{eq_temperature_independent_interaction_decay} and
 \eqref{eq_ope_aposteriori_bound}, 
\begin{align}
&\sup_{\bU\in \overline{D(r)}^{n_v}}\Bigg|
V_{m_1}^0(\bU)(X_0,\bX_1)\prod_{q=2}^{n}(\widetilde{C}((Y_{q}|\hat{s}_{f(q)}),(Z_q|\hat{s}_q))
V_{m_q}^0(\bU)(P_0(Z_q),\bX_q))\label{eq_all_integrant_bound}\\
&\qquad\qquad\cdot 1_{((\bX_j)_{j=1}^n,(Z_j)_{j=2}^n)\in
 \prod_{j=1}^nJ_c^{m_j-1}\times J_c^{n-1}}\notag\\
&\qquad\qquad\cdot
 F((\hat{s}_j)_{j=1}^n,X_0,\bX_1,(\bX_j)_{j=2}^n,(Y_{j})_{j=2}^n,(Z_j)_{j=2}^n)\Bigg|\notag\\
&\le 
\left(\frac{1}{2}\sum_{l=1}^nm_l-n+1\right)!\prod_{q=1}^n\Bigg(\frac{ r
 e^dv_{m_q/2}(1)c(\beta,d,(t_j)_{1\le j\le d})^{m_q/2}m_q^{d_T(q)}}{1+1_{q\neq
 1}(\frac{2}{\pi})^{d+1}\sum_{j=1}^d|\<P_s(Z_q),\be_j\>|^{d+1}}\notag\\
&\qquad\qquad\qquad\qquad\qquad\qquad\quad\cdot e^{-\frac{1}{m_q-1}\sum_{l=1}^{m_q-1}\sum_{j=1}^d(\frac{2}{\pi}|\<P_s(X_{q,l}),\be_j\>|)^{1/2}}\Bigg),\notag\\
&(\forall
 (({s}_j)_{j=1}^n,X_0,\bX_1,(\bX_j)_{j=2}^n,(Y_{j})_{j=2}^n,(Z_j)_{j=2}^n)\in [0,\beta)^n\times \bJ_{\infty}).\notag
\end{align}
The right-hand side of \eqref{eq_all_integrant_bound} is integrable with
 respect to 
 $$(({s}_j)_{j=1}^n,X_0,\bX_1,(\bX_j)_{j=2}^n,(Y_{j})_{j=2}^n,(Z_j)_{j=2}^n)$$
 over $[0,\beta)^n\times \bJ_{\infty}$. 

Since $F$ becomes a finite sum of products of the covariance $C$
 after applying all the Grassmann derivatives to the monomial, the
 domain of $F$ can be naturally extended into $[0,\beta)^n\times
 \bJ_{\infty}$. Moreover, we can see that the function
 $F:[0,\beta)^n\times \bJ_{\infty}\to\C$ is independent of $h$.
Since $(s,t)\mapsto
 \widetilde{C}((X|s),(Y|t))$ is continuous a.e. in $[0,\beta)^2$ for any
 $X,Y\in J_{\infty}$, so is $\bs\mapsto F(\bs,Z)$ a.e. in $[0,\beta)^n$
 for any $Z\in \bJ_{\infty}$. Thus, for any $X,Y\in J_{\infty}$, $Z\in \bJ_{\infty}$,
\begin{align*}
&\lim_{h\to\infty\atop
 h\in(2/\beta)\N}\widetilde{C}((X|\hat{s}),(Y|\hat{t}))=\widetilde{C}((X|s),(Y|t))\text{
 a.e. }(s,t)\in[0,\beta)^2,\\
&\lim_{h\to\infty\atop
 h\in(2/\beta)\N}F((\hat{s}_j)_{j=1}^n,Z)=F((s_j)_{j=1}^n,Z)\text{
 a.e. }(s_j)_{j=1}^n\in[0,\beta)^n.
\end{align*}
Furthermore, by using the fact that $\lim_{L\to
 \infty,L\in\N}C(\bX)$ converges for any $\bX\in(\cB\times
 \Z^d\times\spin\times [0,\beta))^2$ we can check that 
 $\lim_{L\to \infty,L\in\N}F(\bs,Z)$ converges for any $(\bs,Z)\in[0,\beta)^n\times
 \bJ_{\infty}$.

Now we can apply the dominated convergence theorem in \\
$L^1([0,\beta)^n\times \bJ_{\infty},C(\overline{D(r)}^{n_v};\C))$
 to prove that $a_n'(L,h)$ converges
 in \\
$C(\overline{D(r)}^{n_v};\C)$ as $h\to \infty(h\in(2/\beta)\N)$ and 
\begin{align}
\lim_{h\to\infty\atop
 h\in(2/\beta)\N} a_n'(L,h)
&=
\prod_{j=1}^n\left(\int_{0}^{\beta}ds_j\right)
\sum_{X_0\in J_0\atop \bX_1\in
 J^{m_1-1}_{\infty}}V_{m_1}^0(X_0,\bX_1)\label{eq_all_integrant_h_limit}\\
&\quad\cdot\prod_{q=2}^{n}\Bigg(\sum_{Y_{q}\in J_{0},Z_q\in
 J_{\infty}\atop \bX_q\in
 J_{\infty}^{m_q-1}}\widetilde{C}((Y_{q}|{s}_{f(q)}),(Z_q|{s}_q))
V_{m_q}^0(P_0(Z_q),\bX_q)\Bigg)\notag\\
&\quad\cdot 1_{((\bX_j)_{j=1}^n,(Z_j)_{j=2}^n)\in
 \prod_{j=1}^nJ_c^{m_j-1}\times J_c^{n-1}}\notag\\
&\quad\cdot F(({s}_j)_{j=1}^n,X_0,\bX_1,(\bX_j)_{j=2}^n,(Y_{j})_{j=2}^n,(Z_j)_{j=2}^n).\notag
\end{align}
Set $a_n'(L):=\lim_{h\to\infty,h\in (2/\beta)\N}a_n'(L,h)$. 
By sending $h\to\infty$ we obtain the inequality \eqref{eq_all_integrant_bound} with
 a.e. $(s_j)_{j=1}^n\in[0,\beta)^n$ in place of
 $(\hat{s}_j)_{j=1}^n$ in the left-hand side. Then, by the convergence property of $V_{2m}^0$
 proved in Lemma \ref{lem_temperature_independent_interaction}, the
 convergence properties of $F$ and $\widetilde{C}$ in the limit $L\to
 \infty$ and that
\begin{align*}
&\lim_{L\to \infty\atop L\in \N}1_{((\bX_j)_{j=1}^n,(Z_j)_{j=2}^n)\in
 \prod_{j=1}^nJ_c^{m_j-1}\times J_c^{n-1}}=1,\\
&\Bigg(\forall ((\bX_j)_{j=1}^n,(Z_j)_{j=2}^n)\in
 \prod_{j=1}^nJ_{\infty}^{m_j-1}\times J_{\infty}^{n-1}\Bigg),
\end{align*}
 we can again apply the dominated convergence theorem in
 $L^1([0,\beta)^n\times \bJ_{\infty},C(\overline{D(r)}^{n_v};\C))$
 to deduce from \eqref{eq_all_integrant_h_limit} that
$a_n'(L)$ converges
 in \\$C(\overline{D(r)}^{n_v};\C)$ as $L\to\infty(L\in \N)$.
\end{proof}

\section*{Acknowledgments}
This work was supported by JSPS KAKENHI Grant Number 
26870110.

\section*{Supplementary list of notations}
\subsection*{Parameters and constants}
\begin{center}
\begin{longtable}{p{4cm}|p{7.5cm}|p{3.5cm}}
Notation & Description & Reference \\
\hline
$\theta_{j,k}$ & flux per plaquette & Subsection
 \ref{subsec_hamiltonian}\\
$\ (1\le j<k\le d)$ &  & \\ 
$\eps^L_{j}$  & flux per large circles around & Subsection
 \ref{subsec_hamiltonian}\\
$\ (j=1,2,$ $\cdots,d)$  &  periodic lattice  & \\ 
$t_{j}$  & hopping amplitude & Subsection
 \ref{subsec_hamiltonian}\\
$\ (j=1,2,\cdots,d)$ &  & \\ 
$n_v$  & number of coupling constants & Subsection
 \ref{subsec_hamiltonian}\\
$N_v$ & maximal degree of interacting part of Hamiltonian & Subsection
 \ref{subsec_hamiltonian}\\
$v_0,v_m(c)$ & integral of kernels of interaction
 & beginning of\\
$\ (m =1,\cdots,N_v)$ &   &  Subsection
 \ref{subsec_examples}\\ 
$\bpi$ & $(\pi,\pi,\cdots,\pi)\in\R^d$ & Subsection
 \ref{subsec_multi_band_hamiltonian}\\
$\beps^L$ & $(\eps_j^L)_{1\le j\le d}$ & Subsection
 \ref{subsec_multi_band_hamiltonian}\\
$\btheta$ & $(\theta_{j,k})_{1\le j<k\le d}$ & Subsection
 \ref{subsec_multi_band_hamiltonian}\\
$N$ & $2^{d+2}L^d\beta h$, cardinality of the index set $I$ & Subsection
 \ref{subsec_grassmann}\\
$M_{UV}$ & $\frac{2\sqrt{6}}{\pi}(2d+1)$ & Subsection
 \ref{subsec_UV_covariance}\\ 
$\fw(0)$ & $c_w(d+1)^{-2}M^{-2}$ & Subsection \ref{subsec_isothermal}\\ 
$f_{\bt}$ & parameter depending only on $(t_j)_{1\le j\le d}$, $(\theta_{j,k})_{1\le j<k\le d}$ & Subsection
 \ref{subsec_general_lemma}\\ 
$c_{\chi}$ & constant appearing in an upper bound on derivatives of cut-off
 functions & Lemma \ref{lem_IR_cut_off_properties} \eqref{item_IR_cut_off_derivative}
\end{longtable}
\end{center}

\subsection*{Sets and spaces}
\begin{center}
\begin{longtable}{p{4cm}|p{7.5cm}|p{3.5cm}}
Notation & Description & Reference \\
\hline
$\G(L)$ & $\{0,1,\cdots,L-1\}^d$ & Subsection \ref{subsec_hamiltonian}\\
$\Map(A,B)$ & set of maps from $A$ to $B$  & Subsection
 \ref{subsec_hamiltonian}\\
$D(c)$  &  $\{z\in\C\ |\ |z|<c\}$ & Subsection \ref{subsec_results}\\
$C(\overline{D};\C)$ & set of continuous functions on $\overline{D}$ &
 Subsection \ref{subsec_results}\\
$\Mat(n,\C)$ & set of $n\times n$ complex matrices &
 Subsection \ref{subsec_multi_band_hamiltonian}\\
$\cB$ & $\{1,2,3,\cdots,2^d\}$ &
 Subsection \ref{subsec_multi_band_hamiltonian}\\
$\G(L)^*$ & $\{0,\frac{2\pi}{L},\cdots,\frac{2\pi}{L}(L-1)\}^d$ &
 Subsection \ref{subsec_multi_band_hamiltonian}\\
$C(\overline{D};\bigwedge \cV)$ & set of Grassmann polynomials
 continuous with $\bz\in\overline{D}$ &
 Subsection \ref{subsec_general_lemma}\\
$C^{\o}(D;\bigwedge \cV)$ & set of Grassmann polynomials
 analytic with $\bz\in D$ &
 Subsection \ref{subsec_general_lemma}\\
$\cS(D,c_0,\alpha,M)$ $(l)$ & subset of $C(\overline{D};\bigwedge \cV)\cap C^{\o}(D;\bigwedge \cV)$
& Subsection \ref{subsec_general_lemma}\\
$\hat{\cS}(D,c_0,\alpha,M)$ $(l)$ & subset of
 $\cS(D,c_0,\alpha,M)(l)(\beta_1)\times$ $\cS(D,c_0,\alpha,M)(l)(\beta_2)$ & Subsection \ref{subsec_general_lemma}\\
$\cK(D,\alpha,M)(l)$ & subset of 
& Subsection \ref{subsec_general_lemma}\\
 & $\Map(\overline{D},C^{\infty}(\R^{d+1};\Mat(2^d,\C)))$ & \\
$\hat{\cK}(D,\alpha,M)(l)$ & subset of
 $\cK(D,\alpha,M)(l)(\beta_1)\times \cK(D,\alpha,M)(l)(\beta_2)$
& Subsection \ref{subsec_general_lemma}\\
$C^{\o}(D;\C)$ & set of analytic functions in $D$ &
 Subsection \ref{subsec_general_lemma}\\
$\cR(D,c_0,M)(l)$ & subset of $\Map(\overline{D},\Map(I_0^2,\C))$
& Subsection \ref{subsec_general_lemma}\\
$\hat{\cR}(D,c_0,M)(l)$ & subset of $\cR(D,c_0,M)(l)(\beta_1)\times\cR(D,c_0,M)(l)(\beta_2)$
& Subsection \ref{subsec_general_lemma}
\end{longtable}
\end{center}

\subsection*{Functions and maps}
\begin{center}
\begin{longtable}{p{4cm}|p{7.5cm}|p{3.5cm}}
Notation & Description & Reference \\
\hline
$\sH$ & 1-band Hamiltonian  &
 Subsection \ref{subsec_hamiltonian} \\
$\sH_0$ & kinetic part of $\sH$  &
 Subsection \ref{subsec_hamiltonian} \\
$\sV$ & interacting part of $\sH$  &
 Subsection \ref{subsec_hamiltonian} \\
$b$ & bijection from $\cB$ to $\{0,1\}^d$  &
 Subsection \ref{subsec_multi_band_hamiltonian} \\
$U_d((\xi)_{1\le j\le d})$ & $2^d\times 2^d$ diagonal unitary matrix &
 Subsection \ref{subsec_multi_band_hamiltonian} \\
$\nu$ & bijection from $\cB\times\G(L)$ to $\G(2L)$ &
 Subsection \ref{subsec_multi_band_hamiltonian} \\
$H$ & $2^d$-band Hamiltonian &
 Subsection \ref{subsec_multi_band_hamiltonian} \\
$H_0$ & kinetic part of $H$  &
Subsection \ref{subsec_multi_band_hamiltonian} \\
$V$ & interacting part of $H$  &
 Subsection \ref{subsec_multi_band_hamiltonian} \\
$\cE$ & $2^d\times 2^d$-matrix-valued function  & Subsection \ref{subsec_grassmann}
\end{longtable}
\end{center}

\subsection*{Other notations}
\begin{center}
\begin{longtable}{p{4cm}|p{7.5cm}|p{3.5cm}}
Notation & Description & Reference \\
\hline
$\be_j$ 
& standard basis of $\R^d$ & Subsection \ref{subsec_hamiltonian} \\
$\ (j=1,2,\cdots,d)$ & & \\ 
$\|\cdot\|_{n\times n}$ & operator norm for $n\times n$-matrices & Subsection \ref{subsec_multi_band_hamiltonian}
\end{longtable}
\end{center}

\end{document}